\documentclass[a4paper,12pt]{article}
\usepackage{macros}
\graphicspath{{./Figures/}}

\title{Quantum toroidal algebras and \\ solvable structures in gauge/string theory}

\date{}

\author[a]{Yutaka Matsuo}
\author[b]{Satoshi Nawata}
\author[c]{Go Noshita}
\author[d]{Rui-Dong Zhu}
\affil[a]{Department of Physics \& Trans-scale Quantum Science Institute \& Mathematics and Informatics Center, University of Tokyo, Hongo 7-3-1, Bunkyo-ku, Tokyo 113-0033, Japan}
\affil[b]{Department of Physics and Center for Field Theory and Particle Physics, Fudan University, 2005, Songhu Road, 200438 Shanghai, China}
\affil[c]{Department of Physics, University of Tokyo, Hongo 7-3-1, Bunkyo-ku, Tokyo 113-0033, Japan}
\affil[d]{Institute for Advanced Study \& School of Physical Science and Technology, Soochow University, Suzhou 215006, China}

\begin{document}
\maketitle\thispagestyle{empty}\setcounter{page}{0}
\abstract{This is a review article on the quantum toroidal algebras, focusing on their roles in various solvable structures of 2d conformal field theory, supersymmetric gauge theory, and string theory. Using $\mathcal{W}$-algebras as our starting point, we elucidate the interconnection of affine Yangians, quantum toroidal algebras, and double affine Hecke algebras.

Our exploration delves into the representation theory of the quantum toroidal algebra of $\frakgl_1$ in full detail, highlighting its connections to partitions, $\mathcal{W}$-algebras, Macdonald functions, and the notion of intertwiners. Further, we also discuss integrable models constructed on Fock spaces and associated $\cR$-matrices, both for the affine Yangian and the quantum toroidal algebra of $\frakgl_1$. 

The article then demonstrates how quantum toroidal algebras serve as a unifying algebraic framework that bridges different areas in physics. Notably, we cover topological string theory and supersymmetric gauge theories with eight supercharges, incorporating the AGT duality.  Drawing upon the representation theory of the quantum toroidal algebra of $\frakgl_1$, we provide a rather detailed review of its role in the algebraic formulations of topological vertex and $qq$-characters. Additionally, we briefly touch upon the corner vertex operator algebras and quiver quantum toroidal algebras.\blfootnote{ Email addresses: ${}^\textrm{a}$\href{mailto:matsuo@hep-th.phys.s.u-tokyo.ac.jp}{matsuo@hep-th.phys.s.u-tokyo.ac.jp}, ${}^\textrm{b}$\href{mailto:snawata@gmail.com}{snawata@gmail.com}, ${}^\textrm{c}$\href{mailto:noshita-go969@g.ecc.u-tokyo.ac.jp}{noshita-go969@g.ecc.u-tokyo.ac.jp}, ${}^\textrm{d}$\href{mailto:rdzhu@suda.edu.cn}{rdzhu@suda.edu.cn}}}

\bigskip


\newpage

\tableofcontents
\allowdisplaybreaks

\newpage

\section{Introduction}\label{sec:intro}
\noindent Since its birth \cite{Belavin:1984vu}, Virasoro algebra—the two-dimensional conformal algebra—has firmly established its influence across a myriad of research areas, including string theory, 2d statistical models, quantum gravity in 2d and 3d, and the overarching domain of mathematical physics. It provides an ideal interdisciplinary playground between Physics and Mathematics. Notably, the classification of the Virasoro minimal models found an immediate application to compute the critical exponents of 2d statistical systems. Building on BPZ's seminal work, extensive research has been devoted to 2d CFTs, revealing that certain models possess larger symmetries.  The celebrated examples are affine Lie algebras \cite{kac1990infinite} in the Wess-Zumino-Witten (WZW) models \cite{Knizhnik:1984nr} and $\mathcal{W}$-algebras—nonlinear algebras described by higher spin currents \cite{Zamolodchikov:1985wn,Fateev:1987vh}. The presence of infinite-dimensional symmetry is a crucial element for exact solvability in 2d CFTs.

Emerging in parallel, the quantum inverse scattering method arises as a powerful technique for solving a specific class of integrable systems—systems possessing enough conserved quantities to be exactly solvable. The cornerstone of the quantum inverse scattering method is the  $\cR$-matrix, a solution to the Yang-Baxter equation. The $\cR$-matrices turned out to be associated with quantum groups (and their representations) \cite{Drinfeld1987,QG-Jimbo1}, which are deformations of the universal enveloping algebras of Lie algebras. A seminal series of Drinfeld's works revealed that quantum groups provide the underlying mathematical structure behind the solvability of integrable systems.

Quantum field theory provides the unified viewpoint of quantum 3-manifolds/knots invariants \cite{Witten-Jones}, where the relationship between Chern-Simons theory and the WZW models plays an essential role. Mathematically, this is interpreted as the relation between affine Lie algebras and quantum groups. The braid group representation on correlation functions of the WZW models is equivalent to that of the universal $\scR$-matrices of quantum groups \cite{kohno1987monodromy,drinfeld1989quasi}. Also, an equivalence of categories between highest-weight integrable representations of affine Lie algebras and finite-dimensional representations of quantum groups at a root of unity was established in \cite{kazhdan1993tensor}. Consequently, the modular tensor category $\Rep(U_{q=e^{2\pi i/\kappa}}(\frakg))$ has proven to be a natural description of 3d topological quantum field theory (TQFT) \cite{reshetikhin1990ribbon,reshetikhin1991invariants}.

These developments have triggered a surge of active study of infinite-dimensional symmetry and quantum groups both in physics and mathematics. Relevant to this note are the geometric interpretations of infinite-dimensional symmetry and quantum groups. The cohomology groups of certain instanton moduli spaces are connected to Heisenberg algebras \cite{nakajima1997heisenberg}, affine Lie algebras \cite{Nakajima:1994nid}, Yangians \cite{varagnolo2000quiver}, and quantum affine algebras \cite{nakajima2001quiver}. Motivated by the geometric constructions, ``toroidal'' generalizations of quantum groups, named quantum toroidal algebras, were introduced in \cite{ginzburg1995langlands}. Later, their degenerations, termed affine Yangians, were constructed in \cite{guay2005cherednik,guay2007affine}. Importantly, quantum toroidal algebras occupy a central position in this hierarchy, as their various limits give rise to the aforementioned algebras. Within the scope of this note, it is the quantum toroidal algebras that stand as the primary subjects of exploration.

In the realms of TQFT and topological string theory, many observables can be computed exactly, and the origins of exact results can often be attributed to infinite-dimensional symmetries. Nakajima's geometric construction \cite{Nakajima:1994nid} of affine Lie algebras has found its interpretation in 4d TQFT \cite{Vafa:1994tf}. Also, Chern-Simons TQFT is embedded into topological string theory \cite{Witten-CS,Gopakumar:1998ki}, which results in the powerful framework called the topological vertex \cite{AKMV}. As demonstrated by these works, duality plays a significant role in identifying the connections to infinite-dimensional Lie algebras.

In a parallel line of development, exact low-energy effective actions for 4d $\cN=2$ supersymmetric gauge theories have been obtained by using both supersymmetry and the theory of Riemann surfaces \cite{Seiberg:1994aj,Seiberg:1994rs}. Inspired by this groundbreaking work, techniques for calculating exact partition functions in supersymmetric theories on the $\Omega$-background—referred to as supersymmetric localizations—were subsequently developed \cite{Nekrasov:2002qd,Pestun:2007rz}. Additionally, realizations in M-theory and relations to Hitchin systems have substantially broadened the landscape of 4d $\cN=2$ supersymmetric theories \cite{Gaiotto:2009we}. This set the stage for the Alday-Gaiotto-Tachikawa (AGT) duality \cite{Alday:2009aq} (see \cite{LeFloch:2020uop} for a review), stating that the instanton partition function of a 4d $\cN=2$ theory is identical to the conformal block of a 2d Toda CFT.

The AGT duality has fundamentally reshaped our perspectives on both supersymmetric theories and 2d CFTs. To understand the AGT duality for 4d $\cN=2$ super Yang-Mills theory, it is necessary to find an orthogonal frame labeled by sets of Young diagrams as the instanton partition function. Remarkably, this orthogonal frame has been found to be related to (generalized) Jack functions—the eigenfunctions of the Calogero-Sutherland model \cite{Alba:2010qc,Awata:1994xd}. 
The mathematical proof of the AGT duality by Schiffmann and Vasserot \cite{Schiffmann:2012aa} illuminated the connection more explicitly. Their analysis originated from the symmetrized version (spherical) of the degenerate double affine Hecke algebra (degenerate DAHA), which describes the underlying algebraic structure of the Calogero-Sutherland model. The large rank limit of the spherical degenerate DAHA is indeed equivalent to the affine Yangian $\AY$ of $\mathfrak{gl}_1$, containing all the $\mathcal{W}$-algebras as subalgebras. Moreover, $\AY$ is endowed with a coproduct that naturally defines the generalized Jack functions labeled by tuples of Young diagrams, whose norm gives the instanton partition function. Consequently, $\AY$ serves as a universal symmetry underpinning the AGT duality, furnishing a natural mathematical framework to prove the duality.

The AGT duality can be extended into a five-dimensional context, often referred to as the K-theoretic version or 5d AGT. This duality posits that the 5d $\mathcal{N}=1$ theory is dual to the correlators of $q$-deformed $\mathcal{W}_{N}$  algebras \cite{Awata:2009ur,Awata:2010yy}. Originally, the $q$-Virasoro and $q$-$\mathcal{W}_{N}$ algebras were introduced to study the relation with Macdonald symmetric functions \cite{macdonald1998symmetric}. Since Macdonald symmetric functions appear in diverse contexts, such as integrable systems and the topological vertex,  this signals the existence of a broader algebraic underlying framework.  

Similar to the relation between $\mathcal{W}$-algebras and the affine Yangian $\AY$, the $q$-$\mathcal{W}$ algebras are contained in the quantum toroidal algebra of $\mathfrak{gl}_{1}$, denote by $\QTA$. It is a natural trigonometric generalization of $\AY$. Various ways of constructing the quantum toroidal $\mathfrak{gl}_1$ have been explored in the literature; 
\begin{itemize}[nosep]
    \item the double deformation of the $\mathcal{W}_{1+\infty}$ algebra \cite{ding1997generalization,miki2007q}
    \item a central extension of the spherical $\mathfrak{gl}_\infty$ double affine Hecke algebra \cite{cherednik2005double,schiffmann2011elliptic} 
\item the Hall algebra of coherent sheaves on an elliptic curve \cite{burban2012hall}
\item  Drinfeld's quantum double of the shuffle algebra \cite{feigin2011equivariant,Tsymbaliuk:2022bqx}
\end{itemize}
 These diverse realizations indicate the versatility and richness of the algebra, and $\QTA$ indeed exhibits greater symmetry than the affine Yangian in many respects.  Specifically, its generators are graded by $\mathbb{Z}^2$ and have an $\SL(2,\mathbb{Z})$ automorphism \cite{miki2007q}. Due to this, it has two types of representations:  vertical and horizontal. The vertical representations are described by using sets of (generalized) Young diagrams, whereas the horizontal ones are described by vertex operators using ($q$-deformed) free boson. This note reviews these representations, emphasizing the connection with the representations of $\mathcal{W}$-algebras, the relation with the (degenerate) DAHA, and the applications to gauge/string theory.

Building upon the AGT duality, we introduce a more encompassing theoretical framework known as the BPS/CFT correspondence \cite{Nekrasov:2015wsu,Kimura:2015rgi}. This framework serves as a conceptual bridge, linking certain BPS states in a supersymmetric quantum field theory (often in higher dimensions) to representations of infinite-dimensional symmetry (2d CFT). This connection has profound implications for understanding the non-perturbative aspects of supersymmetric quantum field theories, especially in terms of partition functions, correlation functions, and the spectra of BPS states. This note aims to elucidate the integral role played by the quantum toroidal $\frakgl_1$ in the BPS/CFT correspondence.

From this vantage point, the topological string theory emerges as an exciting area.  5d supersymmetric theories can be realized as low-energy effective theories arising from compactifying M-theory on particular Calabi-Yau three-folds through a duality known as geometric engineering \cite{Katz:1996fh,Katz:1997eq}. The topological vertex is an indispensable tool for evaluating the topological string partition function of a Calabi-Yau three-fold, thereby also yielding the instanton partition function of the corresponding 5d theory. The key advantage of using the topological vertex is that it enables one to compute the instanton partition function in a piecewise manner. This property allows one to reformulate the topological vertex in terms of intertwiners of the quantum toroidal $\frakgl_1$ \cite{Awata:2011ce}. This novel formulation broadens the scope of connections between BPS states in M-theory and representations of the quantum toroidal $\frakgl_1$, extending beyond the K-theoretic version of AGT duality and into the realm of enumerative invariants of Calabi-Yau three-folds.

 One of the most important features of $\QTA$ and $\AY$ is that they are equipped with a universal $\scR$-matrix as the quantum groups \cite{miki2007q,Feigin:2015raa,Maulik:2012wi,negut2014shuffle,Prochazka:2019dvu}. Using the techniques of the quantum inverse scattering method, one can construct an infinite number of mutually commuting conserved charges and Bethe states via the transfer matrix and Bethe ansatz equations. The conserved charges are related to integrals of motions in integrable field theories characterized by intermediate long-wave equations \cite{Litvinov:2013zda}. The biggest difference from the well-known integrable models such as the spin chains is that the degree of freedom at each site of the lattice model is infinite-dimensional, described by a Fock space.
 The connection between quantum toroidal algebras or affine Yangians and supersymmetric gauge theories then further enables us to connect the integrable systems to gauge theories.
 
Yet another connection of supersymmetric gauge theories to integrable systems can be found through the Bethe/Gauge correspondence. It was first realized in the seminal works \cite{Nekrasov:2009rc,Nekrasov:2009ui,Nekrasov:2009uh} that both the Seiberg-Witten curve, which describes the low-energy effective physics of 4d $\cN=2$ gauge theories \cite{Seiberg:1994aj,Seiberg:1994rs}, and its corresponding classical integrable systems (Hitchin systems) can be quantized in the so-called Nekrasov-Shatashivili limit of the gauge theory. In this respect,  we introduce the physical observable called the $qq$-character \cite{Nekrasov:2015wsu} as a manifestation of the BPS/CFT correspondence, which incorporates quantum toroidal algebras, the Bethe/Gauge correspondence, and the BPS defects. The $qq$-character appears as the second quantization of the $TQ$-relation for quantum integrable systems such as spin chain systems, which eventually give the algebraic explanation of the Bethe/Gauge correspondence.

In the context of the BPS/CFT correspondence, intriguing infinite-dimensional algebras emerge when exploring the algebras of BPS operators in various supersymmetric theories. One big playground is provided by the corner vertex operator algebra (corner VOA) \cite{Gaiotto:2017euk}, which is an algebra of BPS operators at a junction of interfaces of topologically twisted 4d $\cN=4$ super Yang-Mills theory. Free field realizations \cite{Prochazka:2018tlo}, screening currents \cite{bershtein2018plane,Litvinov:2016mgi}, and the triality \cite{Creutzig:2020zaj} of the algebra enabled us to study these algebras explicitly. Moreover, since the corner VOA can be realized in an NS5-D5-D3 brane system in Type IIB string theory, leveraging the vantage point of string theory,  one can easily generalize the brane setup, resulting in a broader class of $\cW$-algebras, termed the web of $\cW$-algebras \cite{Prochazka:2017qum}. These developments have created a stream of research activity that continues to the present day, spanning topics such as $q$-deformations and connections to quantum toroidal algebras \cite{Harada:2021xnm}, and intricate geometric constructions \cite{Rapcak:2018nsl,Rapcak:2020ueh}. 

From the algebraic viewpoint, the corner VOA is regarded as a truncation of the affine Yangian $\frakgl_1$ \cite{Prochazka:2015deb}. This naturally raises the question of whether one can generalize this framework by starting with ``parent'' algebras and subsequently deriving other algebras by truncations. In a notable advancement, a broader class of affine Yangians, termed quiver Yangians, has recently been constructed \cite{Li:2020rij}. These algebras are associated with $(Q, W)$ pairs dual to toric Calabi-Yau three-folds, where $Q$ is a quiver diagram and $W$ is a superpotential. The introduction of quiver Yangians largely extends the existing landscape of algebras and representation theory. Especially noteworthy is the construction of representations of the algebra on BPS crystals, naturally given by this physics setup.
Subsequent studies have explored various aspects such as its derivation from supersymmetric quantum mechanics \cite{Galakhov:2020vyb}, representations \cite{Galakhov:2021xum}, characters \cite{Li:2023zub}, and trigonometric generalizations \cite{Galakhov:2021vbo,Noshita:2021ldl}.  
Undoubtedly, the synergy of the mathematical and physical approaches in this manner will continue to uncover novel algebras in the future. This will, in turn, unravel new dualities and non-perturbative aspects of quantum field theory and string theory. More broadly, this evolution holds the potential to make substantial contributions to the wider arenas of both physics and mathematics.

Indeed, these advancements represent just the tip of the iceberg; a plethora of intriguing questions and unexplored territories await future investigation. In summary, this article aims to offer a comprehensive review of the quantum toroidal algebras as the overarching symmetry, their ties with $\mathcal{W}$-algebras,  the (degenerate) DAHA, and their impacts on string/gauge theories. We wish to navigate this rich tapestry of mathematical and physical connections, laying the groundwork for future explorations and discoveries.

\bigskip

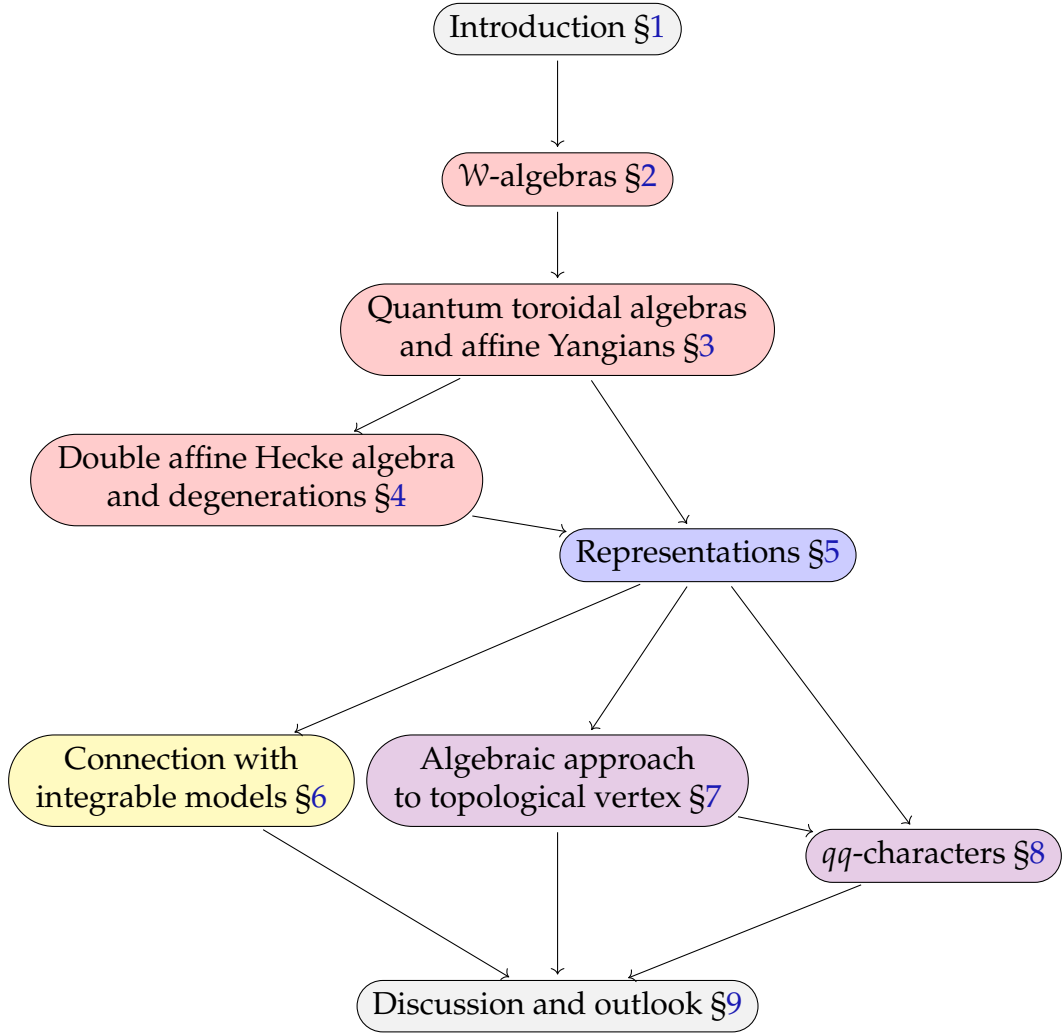
\begin{figure}[ht] \centering
\begin{tikzpicture}[scale=1.0]
\node(A)[rounded rectangle,draw,align=center,scale=1.0,fill=gray!10] at (0,4) {Introduction \S\ref{sec:intro}};

\node(B)[rounded rectangle,draw,align=center,scale=1.0,fill=red!20] at (0,2) {$\mathcal{W}$-algebras \S\ref{sec:W}};

\node(C)[rounded rectangle,draw,align=center,scale=1.0,fill=red!20] at (0,0) {Quantum toroidal algebras\\ and affine Yangians \S\ref{sec:QTA-deg}};

\node(D)[rounded rectangle,draw,align=center,scale=1.0,fill=red!20] at (-4,-2) {Double affine Hecke algebra\\and degenerations \S\ref{sec:DAHA-deg}};

\node(E)[rounded rectangle,draw,align=center,scale=1.0,fill=blue!20] at (2,-3) {Representations \S\ref{sec:QTrep}};

\node(F)[rounded rectangle,draw,align=center,scale=1.0,fill=yellow!30] at (-5,-6) {Connection with\\ integrable models \S\ref{s:int-model}};

\node(G)[rounded rectangle,draw,align=center,scale=1.0,fill=violet!20] at (0,-6) {Algebraic approach\\ to topological vertex \S\ref{sec:algebraic_topvertex}};

\node(H)[rounded rectangle,draw,align=center,scale=1.0,fill=violet!20] at (5,-7) {$qq$-characters \S\ref{sec:qq}};

\node(I)[rounded rectangle,draw,align=center,scale=1.0,fill=gray!10] at (0,-9) {Discussion and outlook \S\ref{sec:final}};
\draw[->,shorten >= 2pt,shorten <= 2pt] (A) to (B);
\draw[->,shorten >= 2pt,shorten <= 2pt] (B) to (C);
\draw[->,shorten >= 2pt,shorten <= 2pt] (C) to (D);
\draw[->,shorten >= 2pt,shorten <= 2pt] (C) to (E);
\draw[->,shorten >= 2pt,shorten <= 2pt] (D) to (E);
\draw[->,shorten >= 2pt,shorten <= 2pt] (E) to (F);
\draw[->,shorten >= 2pt,shorten <= 2pt] (E) to (G);
\draw[->,shorten >= 2pt,shorten <= 2pt] (E) to (H);
\draw[->,shorten >= 2pt,shorten <= 2pt] (F) to (I);
\draw[->,shorten >= 2pt,shorten <= 2pt] (G) to (I);
\draw[->,shorten >= 2pt,shorten <= 2pt] (G) to (H);
\draw[->,shorten >= 2pt,shorten <= 2pt] (H) to (I);

\end{tikzpicture}\caption{Flow chart of the entire note. The sections colored in red, blue, and purple are topics related to algebras, representations, and applications in physics, respectively.}\label{fig:flowchart}\end{figure}

Let us outline the structure of the note below. The interrelation of the sections within this note is summarized in Figure \ref{fig:flowchart}.
\begin{itemize}
\item In \S\ref{sec:W}, we delve into the $\mathcal{W}$-algebras, an important family of infinite-dimensional Lie algebras that hold a central position in 2d CFT. We begin by elucidating the Virasoro algebra and its representations, particularly highlighting representations using free bosons and vertex operators. We also shed light on the Virasoro minimal models, which stand as foundational in 2d CFT. Building on this, we transition to the $\mathcal{W}_N$-algebras approached through the viewpoint of the Miura transformation. The construction of screening currents paves the way to describe the $\mathcal{W}_N$ minimal models. In \S\ref{sec:Winf}, our focus shifts to the $\cW_{1+\infty}$-algebra as a central extension of the Lie algebra of differential operators. In \S\ref{sec:Winf[mu]}, we present the $\mathcal{W}_{\infty}[\mu]$-algebra as a deformation of the $\cW_{\infty}$ algebra.
Lastly, under certain constraints, we bridge the $\mathcal{W}_{\infty}[\mu]$-algebra to the corner vertex operator algebra. 

\item \S\ref{sec:QTA-deg} is dedicated to an introduction of the quantum toroidal $\frakgl_1$ and affine Yangian $\frakgl_1$. In \S\ref{sec:QTA}, we delineate the generators and relations of the quantum toroidal $\frakgl_1$. Thereafter, \S\ref{sec:AY} delves into the affine Yangian $\mathfrak{g l}_1$, detailing its manifestation as a degenerate limit and its intrinsic connection with the $\cW_{1+\infty}$-algebra.
\S\ref{sec:quiver-QTA} illuminates recent advancements and broader constructions of quantum toroidal algebras, particularly those associated with quivers $(Q,W)$.

\item \S\ref{sec:DAHA-deg} offers a comprehensive understanding of the double affine Hecke algebra (DAHA), with a special emphasis on the connection between the large rank limit of spherical DAHA and the quantum toroidal $\mathfrak{g l}_1$.
\S\ref{sec:dDAHA} focuses on the degenerate limit of DAHA, elucidating the relationship binding the spherical degenerate DAHA, the affine Yangian $\frakgl_1$, and the $\mathcal{W}_N$-algebra.

\item In \S\ref{sec:QTrep}, we review the representations of quantum toroidal $\mathfrak{gl}_{1}$ and discuss some applications of them. There are two types of representations: vertical and horizontal. The former is based on Young diagrams and their generalization, plane partitions. The latter one is roughly a vertex operator representation and is described by free bosons. Vertical representations are discussed in \S\ref{sec:vertical-rep} while horizontal representations are discussed in \S\ref{sec:horizontalrep}. We also discuss the representations of the affine Yangian $\mathfrak{gl}_{1}$ in \S\ref{sec:repAY}. We then move on to applications of these representations in the rest of the subsections (see Figure \ref{fig:flowchart2}). We discuss how the Macdonald polynomials arise from the quantum toroidal $\mathfrak{gl}_{1}$ in \S\ref{sec:QTA-Mac}. Both vertical and horizontal representations give different perspectives of the Macdonald polynomials. We then move on to the free field description of the deformed $\mathcal{W}$-algebras in \S\ref{sec:deformedW} using the coproduct structure on $\QTA$ in \S\ref{sec:deformedW} and the screening currents in \S\ref{sec:QTA-screening}. $\mathcal{W}$-algebras can also be described by adding a pit to the MacMahon representation and truncating the representation (see \S\ref{sec:AYminimalmodel}). We further introduce another pit and show that they give minimal models. Finally, we use both the vertical and horizontal representations and introduce algebraic quantities called intertwiners in \S\ref{sec:intertwiner}. 

\item In \S\ref{s:int-model}, we peek at the connection to an integrable field theory.
In \S\ref{s:MO-Rmatrix}, we provide the definition and properties of the Maulik-Okounkov $\mathcal{R}$-matrix, highlighting its operation on the tensor product of two Fock spaces of free bosons. It further elucidates the role of this matrix in the exchange of Miura operators. In \S\ref{sec:Monodromy}, we construct an infinite set of commuting operators derived from the $\mathcal{R}$-matrix. 
In \S\ref{s:Bethe}, we identify the integrable field theory associated with the infinite conserved charges with an extra deformation parameter $p$. Depending on the limit for $p$, one obtains either the conserved charges coming from local field theory \cite{bazhanov1996integrable,bazhanov1997integrable} or the Calogero system associated with the affine Yangian. In \S\ref{s:univR}, we discuss the $q$-deformed version of the Maulik-Okounkov $\mathcal{R}$-matrix and the universal $\mathcal{R}$-matrix intrinsic to the quantum toroidal algebra.

\item In \S\ref{sec:algebraic_topvertex}, we consider applications of representations of the algebra to string theory and supersymmetric gauge theories. At its core, the  Alday-Gaiotto-Tachikawa (AGT) duality unravels the relationship between supersymmetric theories and 2d CFT. The story is further enriched by an interplay of the topological vertex with the intertwiners of $\QTA$. We start by explaining instanton partition functions (\S\ref{sec:AFStopvertex}). Upon a comprehensive introduction of the AGT duality (\S\ref{sec:AGT}), we transition to the description of the topological vertex in terms of the intertwiners of  $\QTA$ (\S\ref{sec:TV-intertwiners}). The section culminates with illustrating how to use the intertwiners to compute the instanton partition functions (\S\ref{sec:int-partition}). 

\item In \S\ref{sec:qq}, we deal with the concept of $qq$-characters, serving as a  primary illustration of the BPS/CFT correspondence and acting as a bridge between symmetries, integrabilities, and BPS defects. The physical setup of a 1/2-BPS codimension-four defect in supersymmetric theories is introduced in \S\ref{sec:defect}. Moving forward, \S\ref{sec:qq-alg} elucidates the algebraic construction of the $qq$-character through vertical representations of the Drinfeld current of affine Yangian. Finally, 
\S\ref{sec:NS} highlights the connection between $qq$-characters and the integrability aspects of supersymmetric gauge theories.

\item In \S\ref{sec:final}, we briefly touch upon recent advancements in quantum toroidal algebras not explored in detail within this note. We provide cursory discussions on several key topics including Hilbert schemes of points, elliptic Hall algebra, quiver $\cW$-algebra, web of $\mathcal{W}$-algebra, representation theory of quantum toroidal algebras, integrable systems, and Knizhnik-Zamolodchikov equations.

\item Appendix \ref{app:notations} summarizes the notations and conventions used in the note. 

\item Appendix \ref{app:QG} provides a quick overview of the basics of quantum groups. In \S\ref{app:Hopf}, we explain Hopf algebras as underlying algebraic structures of quantum groups. Here we particularly highlight the Drinfeld's quantum double and universal $\mathscr{R}$-matrices, both essential to the Yang-Baxter equation. Following this, we explore key examples including quantum groups, Yangians, and quantum affine algebras in subsequent subsections. Finally, the toroidal generalizations are briefly discussed so that the appendix lays a foundation for the algebraic structures underpinning quantum toroidal algebras and affine Yangians. 

\item 
Appendix \ref{app:symmetric-functions} elucidates the fundamental properties of symmetric functions pivotal to the representations of quantum toroidal algebras, affine Yangians, and topological strings. In \S\ref{app:Macdonald}, Macdonald functions are expounded, covering their definition, orthogonality, Pieri formulas, Cauchy formulas, and their relationship with $q$-Heisenberg algebras. As a limit of Macdonald functions, \S\ref{app:Jack} delves into Jack functions where the link to the Calogero-Sutherland model is particularly emphasized. Lastly, \S\ref{app:Schur} briefly explains the properties of Schur functions, which play a central role in topological vertex formalism. 

\item In addition, multiple appendices are included to supplement the main text and to offer specific computational details.

\end{itemize}

Recently, it has become evident that the quantum toroidal algebras play an important role in describing symmetries across diverse fields, including mathematical physics such as superstring theory, supersymmetric theory, conformal field theory, and integrable systems. While it possesses a remarkably rich mathematical structure, many of its facets are yet to be understood, and it deserves extensive research in the 21st century. Nevertheless, the defining relations of quantum toroidal algebras are quite involved. Also, because it is initially rooted in geometric representation theory, many pioneering papers on quantum toroidal algebras are written from this perspective. Consequently, understanding this subject requires a solid foundation in the geometry of moduli spaces and geometric representation theory, which can be challenging for beginners. One of the motivations behind writing this note was to bridge the representation theory of quantum toroidal algebra and its applications to physics with more accessible frameworks like the free-field realizations and $\cW$-algebras. Therefore, this note starts from a more elementary viewpoint, emphasizing its physical aspects. We hope that this approach will resonate with a wider readership. 

\bigskip
\begin{footnotesize}
\noindent\textbf{Acknowledgments} 
First of all, we would like to express sincere gratitude to our collaborators for joint work on related topics. We want to thank Taro Kimura and Kilar Zhang for carefully reading the manuscript and providing insightful comments. 
YM is partly supported by KAKENHI Grant-in-Aid No.18K03610, 23K03380, and 21H05190, and in part by grant NSF PHY-1748958 to the Kavli Institute for Theoretical Physics (KITP). He would like to appreciate the workshop "Integrability in String, Field, and Condensed Matter Theory" held at KITP, UC Santa Barbara 2022, and "String Math 2023" held at the University of Melbourne, where some part of the result was presented.
The research of S.N. is supported by the National Natural Science Foundation of China under Grant No.12050410234 and Shanghai Foreign Expert grant No. 22WZ2502100.
G.N. is supported by JSPS KAKENHI Grant-in-Aid for JSPS fellows Grant No. JP22J20944, JSR Fellowship, and FoPM (WINGS Program), the University of Tokyo. R.Z. is supported by National Natural Science Foundation of China No. 12105198 and the High-level personnel project of Jiangsu Province (JSSCBS20210709). 
\end{footnotesize}

\section{\texorpdfstring{$\cW$}{W}-algebras}\label{sec:W}

Infinite-dimensional symmetry is of crucial importance in 2d conformal field theories (CFTs) as it provides a powerful framework for understanding the structure and properties of the theory. It plays a significant role in connecting exactly solvable models, vertex operator algebras, and modular tensor categories, mathematical structures associated with CFTs.

The most fundamental infinite-dimensional symmetry in 2d CFT is the Virasoro algebra. The Virasoro algebra has a central role in CFT as it captures the conformal symmetry of the theory and governs the behavior of the energy and the momentum. Moreover, the operator product expansions and correlation functions of the energy-momentum tensor provide crucial information about the theory, including determining the spectrum of states. 

Extending beyond the Virasoro algebra, the $\cW$-algebra incorporates additional generators corresponding to higher-spin currents. It goes beyond conformal symmetry and encompasses additional conserved quantities associated with these higher-spin currents.  One of the primary connections between the $\cW$-algebras and exactly solvable models lies in the study of minimal models, a class of CFTs that possess a finite number of primary fields and exhibit a discrete spectrum of conformal dimensions. Moreover, the $\cW$-algebras also play an important role in supersymmetric theories, and of particular relevance to this note is the AGT relation. These aspects will be explored in the subsequent sections.

Furthermore, the $\cW$-algebra finds relevance in the large rank limit, known as the $\cW_\infty$-algebra. 
In the large rank limit, the $\cW_\infty$-algebra captures an infinite number of higher-spin symmetries in the theory, giving rise to an infinite set of conserved charges. The $\cW_\infty$-algebra exhibits interesting features, such as additional central extensions and higher commutation relations beyond those of the $\cW$-algebra.

The central extension of the $\cW_{1+\infty}$-algebra and its $q$-deformation are affine Yangians and quantum toroidal algebra of $\frakgl_1$, which are the main targets of this note. In the context of studying affine Yangians and quantum toroidal algebras, the $\cW$-algebra and their representations consequently appear in various instances. Therefore, this section serves as preparation, providing a review of the $\cW$-algebras and their generalizations. By delving into the study of $\cW$-algebras and their extended versions, this section establishes the foundation for the subsequent exploration of affine Yangians and quantum toroidal algebras.

For a thorough review of $\cW$-algebras and their super extensions, we refer to \cite{Bouwknegt:1992wg} and the references in \cite{bouwknegt1995w}.

\subsection{Virasoro algebra and its representation}\label{sec:Virasoro}

The Virasoro algebra, also known as 2d conformal symmetry, plays a central role in the study of string theory and critical phenomena in two dimensions. In two dimensions, the conformal symmetry encompasses an infinite set of transformations, denoted as $l_n$ with $n \in \mathbb{Z}$, acting on the holomorphic coordinate $z$. These transformations are defined as
\begin{equation}
l_n=-z^{n+1} \partial_z\,,
\end{equation}
which satisfies a commutation relation known as the Witt algebra
\begin{equation}
    [l_n, l_m]=(n-m) l_{n+m}\,.
\end{equation}
The Virasoro algebra possesses a unique central extension expressed as
\begin{equation}\label{Virasoro}
    [\sfL_n, \sfL_m] = (n-m)\sfL_{n+m}+\frac{c}{12}n(n^2-1)\delta_{n+m,0}\,,
\end{equation}
where we denote the generator $l_n$ as $\sfL_n$ with $n \in \mathbb{Z}$, and $c$ is the central element of the Virasoro algebra. In the following, we treat it as a number. By reorganizing the operators, we can express them in terms of the chiral field, which is a quantum field that solely depends on $z$ but not on $\bar{z}$. The chiral field $T(z)$ can be defined as
\begin{equation}\label{EM}
    T(z) =\sum_{n\in \mathbb{Z}} \sfL_n z^{-n-2}\,,
\end{equation}
where $\sfL_n$ represents the generator. The energy-momentum tensor field consists of $T(z)$ and its anti-chiral counterpart $\overline{T}(\bar{z})$. The Virasoro algebra expressed in  (\ref{Virasoro}) can be rewritten using the operator product expansion (OPE) of $T(z)$ as
\begin{equation}\label{TT}
    T(z) T(w) \sim \frac{c/2}{(z-w)^4}+\frac{2}{(z-w)^2}T(w) + \frac{1}{z-w}\partial_w T(w)+\cdots,
\end{equation}
where the radial ordering, i.e. $|w/z|<1$ is always assumed in the OPEs in this note.

The representation of the Virasoro algebra is constructed from the highest weight state, denoted as $|\Delta\rangle$, which satisfies the conditions:
\begin{equation}\label{HWC_Virasoro}
    \sfL_n|\Delta\rangle = 0,\quad (n>0),\qquad
    \sfL_0|\Delta\rangle =\Delta|\Delta\rangle\,,
\end{equation}
where $\Delta$ is the conformal dimension (or weight or spin).
The representation space $\mathcal{H}$, also known as the Verma module, is spanned by states of the form
\begin{equation}
    (\sfL_{-n})^{m_n}(\sfL_{-n+1})^{m_{n-1}}\cdots (\sfL_{-1})^{m_1}|\Delta\rangle=:\mathbb{L}_{-\lambda}|\Delta\rangle,
    \quad \lambda=(n^{m_n},\cdots,1^{m_1}),\quad 
    m_i\geq 0.
\end{equation}
where $\lambda$ can be identified with a partition or a Young diagram. (See Appendix \ref{sec:appendix-Youngdiagram} for the notation.) The notation $\mathbb{L}_{-\lambda}$ represents the product of Virasoro generators.
The Hilbert space $\mathcal{H}$ is characterized by the eigenvalue of the operator $\sfL_0$. Using the commutation relation $[\sfL_0, \sfL_{-n}] = n \sfL_{-n}$, we can compute the eigenvalue as:
\begin{equation}
    \sfL_0 \mathbb{L}_{-\lambda}|\Delta\rangle=
    (\Delta +|\lambda|)\mathbb{L}_{-\lambda}|\Delta\rangle,
\end{equation}
where $|\lambda|$ represents the number of boxes in the Young diagram $\lambda$.
When there are no null states in the Hilbert space, the character can be expressed as:
\begin{equation}\label{Dedekind}
    \chi(\mathfrak{q})=\mbox{Tr}_{\mathcal{H}}(\mathfrak{q}^{\sfL_0})=\frac{\mathfrak{q}^\Delta}{\prod_{n=1}^\infty (1-\mathfrak{q}^n)}\,.
\end{equation}

One may relate the highest weight state $|\Delta\rangle$ with the primary field $\Phi_\Delta(z)$ by the state-operator correspondence
\begin{equation}
    |\Delta\rangle \quad \leftrightarrow \quad \Phi_\Delta(z)\,,
\end{equation}
where the primary field $\Phi_\Delta(z)$ is defined as an operator satisfying the c operator product expansion (OPE) with the energy-momentum tensor $T(z)$:
\begin{equation}
T(z) \Phi_\Delta(w) \sim \frac{\Delta}{(z-w)^2} \Phi_\Delta(w) + \frac{1}{z-w} \partial_w \Phi_\Delta(w) + \cdots.
\end{equation}
Here, $\partial_w$ denotes the holomorphic derivative with respect to $w$. This OPE provides a fundamental relation between the primary field and the energy-momentum tensor.
The state $|\Delta\rangle$ can be obtained from the field $\Phi_\Delta(z)$ by taking the limit as $z$ approaches $0$:
\begin{equation}
|\Delta\rangle = \lim_{z\to 0} \Phi_\Delta(z) |0\rangle.
\end{equation}
Here, $|0\rangle$ represents the vacuum state. This limit allows us to extract the highest weight state from the primary field.

\subsubsection{Coulomb gas representation}
The analysis of the Virasoro Hilbert space can be carried out using the Coulomb Gas representation, which provides an explicit construction in terms of free boson oscillators. In this representation, the energy-momentum tensor $T(z)$ is given by
\begin{equation}\label{VirasoroCG}
T(z) = \frac{1}{2} :(\partial\phi(z))^2: + \frac{Q}{\sqrt{2}}\partial_z^2 \phi,
\end{equation}
where $\phi(z)$ is defined by the mode expansion:
\begin{equation}\label{boson-modes}
\phi(z) = \sfq_0 + \sfJ_0 \ln(z) - \sum_{n\neq 0} \frac{\sfJ_n}{n} z^{-n-1}.
\end{equation}
The algebra for the boson oscillator $\sfJ_n$ is given by:
\begin{equation}\label{Heisenberg}
[\sfJ_n, \sfJ_m] = n\delta_{n+m,0}, \quad [\sfJ_n, \sfq_0] = \delta_{n,0}.
\end{equation}
which follows from the operator product expansion (OPE) of $\phi(z)$
\begin{equation}\label{free_boson_ope}
\phi(z) \phi(w) \sim \log(z-w).
\end{equation}
A normal ordering $:\ :$ is introduced for the free boson, which ensures that the mode expansion of $T(z)$ in (\ref{VirasoroCG}) is properly ordered:
\begin{equation}
    :\sfJ_n \sfJ_m:=\left\{
    \begin{array}{ll}
    \sfJ_n \sfJ_m \quad & n\leq m\\
    \sfJ_m \sfJ_n \quad & n> m
    \end{array}
    \right.\, . 
\end{equation}
The Virasoro generator $\sfL_n$ can be obtained from the expansion of $T(z)$ and is given by
\begin{equation}\label{CGLn}
    \sfL_n =\frac12\sum_{m}:\sfJ_{n-m}\sfJ_m: -\frac{(n+1)Q}{\sqrt{2}} \sfJ_n\,.
\end{equation}
The generators $\sfL_n$ satisfy the Virasoro algebra with the central charge 
\begin{equation}\label{central-charge}
    c=1-6Q^2\,.
\end{equation}
Also, they have commutation relations with the Heisenberg modes given by
\be 
\left[\sfL_n, \sfJ_m\right]=-m \sfJ_{n+m}
\ee 
The parameter $Q$ is known as the background charge.
In the Coulomb gas representation, the primary operator is described by the vertex operator:
\begin{align}
    \cV(t,z)& = :\exp\left(t\phi(z)\right):=\cV_+(t,z)\cV_-(t,z)\cV_0(t,z),\\
    \cV_+(t,z)& =\exp\left(t\sum_{n=1}^\infty \frac{\sfJ_{-n}}{n}z^n\right),
    \quad \cV_-(t,z)=\exp\left(-t\sum_{n=1}^\infty \frac{\sfJ_{n}}{n}z^{-n}\right),\quad \cV_0(t,z)=e^{t\sfq_0}z^{t\sfJ_0}\,,\label{vertex-op-pm}
\end{align}
which has the conformal dimension
\begin{equation}\label{Delta-t}
    \Delta(t) = \frac{t^2}{2}-\frac{Q}{\sqrt{2}} t\,.
\end{equation}
The highest weight state is introduced as
\begin{equation}\label{HWS_CG}
|t\rangle = \lim_{z\to 0} \cV(t,z)|0\rangle = e^{t\sfq_0}|0\rangle,
\end{equation}
where $|0\rangle$ is the Fock vacuum defined by $\sfJ_n|0\rangle = 0$ for $n \geq 0$. The state $|t\rangle$ satisfies the following properties
\begin{align}
\sfJ_n|t\rangle = 0, \quad \sfL_n|t\rangle = 0 \quad \text{for } n > 0,  \quad \textrm{ and } \quad
\sfJ_0|t\rangle = t|t\rangle, \quad \sfL_0|t\rangle = \Delta(t)|t\rangle,
\end{align}
where $\Delta(t)$ is the conformal dimension given by (\ref{Delta-t}).
The vertex operator $\cV(t,z)$ satisfies the following product formula
\begin{equation}\label{prodVV}
\cV(t_1, z_1) \cV(t_2,z_2) = (z_1 - z_2)^{t_1 t_2} :\cV(t_1, z_1) \cV(t_2,z_2):.
\end{equation}
This formula expresses the vertex operators in terms of a normal-ordered product.

\paragraph{Screening charge and Singular vector}
The Coulomb gas representation provides a framework for constructing the singular vector that appears in the minimal model. We will follow the description presented in \cite{Kato:1985vq}.

To simplify the notation, we introduce a new parameter $\beta$, which allows us to rewrite the background charge $Q$ as
\begin{equation}
Q = \beta^{1/2} - \beta^{-1/2}~,
\end{equation}
where $c<1$ for $\beta>0$ and $c>1$ for $\beta<0$.
In the Coulomb gas representation, the vertex operator with conformal dimension one plays a special role. We find that $\Delta(t) = 1$ is satisfied for the values
\begin{equation}\label{t-Virasoro}
t = t_\pm = \pm \sqrt{2}\beta^{\pm \frac12}.
\end{equation}
We define the screening current operators $\cS_\pm$ as follows
\begin{equation}
\cS_\pm = \oint\frac{{\rm d}z}{2\pi i}\cV(t_\pm,z).
\end{equation}
These operators have the following properties
\begin{equation}\label{cS_property}
[\sfL_n, \cS_\pm] = 0, \quad [\sfJ_0, \cS_\pm] = t_\pm \cS_\pm.
\end{equation}

Next, we construct the Hilbert space of the Virasoro algebra in the free boson Fock space generated from the highest weight state $|t\rangle$ given by  (\ref{HWS_CG}). The general state in this Hilbert space is denoted as $\mathbb{L}_{-\lambda}|t\rangle$, where $\mathbb{L}_{-\lambda}$ contains $\sfL_n$ operators defined as in  (\ref{CGLn}).
Using the properties of the screening operators in (\ref{cS_property}), we find that
\begin{equation}
    \cS_\pm \mathbb{L}_{-\lambda} |t\rangle\propto 
    \mathbb{L}_{-\lambda}|t+t_\pm\rangle\,.
\end{equation}
This means that the screening operators $\cS_\pm$ shift the background charge $t$ by $t_\pm$ while leaving the Virasoro operators unchanged. In particular, it maps one highest weight state to another with a shifted value of $t$.

To study the application of the screening operator in more detail, let us consider the action of $r$ screening currents $\cS_\pm$ on the state $|t\rangle$. The states obtained from this action, denoted as $|t\rangle_r^\pm$, are given by:
\begin{equation}\label{singular_state}
|t\rangle_r^\pm = (\cS_\pm)^r|t-r t_\pm\rangle.
\end{equation}
These states satisfy the highest weight condition of the Virasoro algebra with the following properties:
\begin{equation}
\sfJ_0 |t\rangle_r^\pm = t |t\rangle_r^\pm, \quad
\sfL_0 |t\rangle_r^\pm = \Delta(t-rt_\pm)|t\rangle_r^\pm,
\end{equation}
assuming that the states are well-defined. If the expression $\Delta(t-rt_\pm) -\Delta(t)$ corresponds to a positive integer $N$, then the state $|t\rangle_r^\pm$ can be identified as a level $N$ state within the Verma module of $|t\rangle$. As it satisfies the highest weight condition, it becomes a null state.
It is important to note that the parameter $t$ needs to satisfy the on-shell condition
\begin{equation}\label{onshell_c}
    \Delta(t-rt_\pm) -\Delta(t)= \frac{r^2t_\pm^2}{2}-rtt_\pm+\frac{1}{\sqrt{2}}(\beta^{1/2}-\beta^{-1/2})rt_\pm=N\,,
\end{equation}
which determines the values of $t$ for which the integral of the screening operator becomes non-zero. For further details and a comprehensive discussion, we refer to Appendix \ref{a:sing}.

The solution to the on-shell condition (\ref{onshell_c}) is given by
\begin{equation}\label{trs}
    t=t_{r,s}=\frac{1}{\sqrt{2}}((r+1)\beta^{1/2}-(s+1)\beta^{-1/2})=\frac12((r+1)t_++(s+1)t_-),\quad s=N/r\,.
\end{equation}
This equation determines the values of $t$ for which we obtain a null state at level $N = rs$. The corresponding conformal weight is given by (\ref{Delta-t}) as
\begin{equation}
    \Delta=\Delta(t_{r,s})=\frac{(r\beta^{1/2}-s\beta^{-1/2})^2-(\beta^{1/2}-\beta^{-1/2})^2}{4}\,.
\end{equation}
We denote the null state that appears here as
\begin{equation}\label{singular-state}
    |\Delta_{r,s}\rangle = |t_{r,s}\rangle^+_r \propto  |t_{r,s}\rangle^-_s\,.
\end{equation}

To obtain the fully degenerate module, we need to impose the existence of another null state in the module. To achieve this, we observe that the relation $\Delta(\sqrt{2}Q-t)=\Delta(t)$ implies $\Delta(t_{-r,-s})=\Delta(t_{r,s})$. We impose the condition
\begin{equation}\label{betapq}
    \beta=q/p,
\end{equation}
where $(p,q)$ are mutually co-prime positive integers. With this choice, we have
\begin{equation}
    \Delta(t_{r,s}) = \Delta(t_{-r,-s})=\Delta(t_{p-r,q-s})=\frac{(qr-ps)^2-(q-p)^2}{4pq}\,.
\end{equation}
If we further impose $p-r>0$ and $q-s>0$, we find that there is an additional null state at level $(p-r)(q-s)$.

\subsubsection{Minimal models}
The analysis above means that when the central charge \eqref{central-charge} with the condition \eqref{betapq} is
\begin{equation}
    c=1-\frac{6(p-q)^2}{pq}\,,
\end{equation}
for mutually coprime integers $p,q>1$, a special class of the representations of Virasoro algebra arises. For this choice of central charge, we obtain a set of highest states with conformal dimensions
\begin{equation}\label{conf_grid}
\Delta_{r,s}=\frac{(rq-ps)^2-(q-p)^2}{4pq}, \quad
    1\leq r< p,\quad 1\leq s< q\,.
\end{equation}
The seminal paper \cite{Belavin:1984vu} claims that the finite set of primary fields $\Phi_{rs}(z)$  corresponding to these highest states defines a 2d CFT. The theory satisfies various consistency relations, including modular invariance and crossing symmetry. We note that there is a relation $\Delta_{p-r,q-s}=\Delta_{rs}$, and we identify $\Phi_{p-r,q-s}$ with $\Phi_{rs}$.

As seen above, there are an infinite number of null states in the Hilbert space generated from $|\Delta_{rs}\rangle$. For example, in the Hilbert space generated from the highest weight state $|\Delta_{rs}\rangle$, there exist null states (singular states) at level $rs$ and $(p-r)(q-s)$, and all states generated from these two states also become null. These two singular vectors satisfy the highest weight state condition (\ref{HWC_Virasoro}), which is crucial in proving that these states have the zero norm when paired with any arbitrary state in the Hilbert space. \footnote{To prove this statement, we introduce an arbitrary bra state of the form $\langle\Delta|\mathbb{L}_{\lambda}$ with the same level of the extra highest weight state $|\chi\rangle$. The inner product $\langle\Delta|\mathbb{L}_{\lambda}|\chi\rangle$ vanishes identically for any $\lambda$ due the highest weight condition.}

 Importantly, we observe that $\Delta_{r, s}+rs=\Delta_{-r, s}$, indicating that this null state $|\Delta_{-r, s}\rangle$ generates a submodule. Furthermore, by considering the relation $\Delta_{r,s}+(p-r)(q-s)=\Delta_{2p-r,s}$, we identify another submodule obtained from $|\Delta_{2p-r, s}\rangle$. Furthermore, these submodules themselves contain additional submodules such as $|\Delta_{-2p-r, s}\rangle$ and $|\Delta_{2p+r, s}\rangle$. This pattern repeats iteratively, leading to a hierarchical structure of submodules within the Verma module. Hence, to obtain the character of the irreducible representation $L(c, \Delta_{r, s})$ of the Virasoro algebra, we need to perform subtraction and addition of the characters \eqref{Dedekind} corresponding to the hierarchical structure of submodules. By applying these operations, we can derive the irreducible character for the minimal model
\be\label{Rocha-Caridi}
\begin{aligned}
\chi_{r, s}(\mathfrak{q}) =&\operatorname{Tr}_{L\left(c, \Delta_{r, s}\right)} \mathfrak{q}^{\sfL_0} \\
=&\frac{1}{\prod_{n=1}^\infty (1-\mathfrak{q}^n)} \sum_{k \in \mathbb{Z}}\left[\mathfrak{q}^{\Delta_{r+2 p k, s}}-\mathfrak{q}^{\Delta_{2 p k-r, s}}\right]~,
\end{aligned}
\ee
which is called the Rocha-Caridi formula.

\subsection{\texorpdfstring{$\cW_N$}{WN}-algebras}

Following the groundbreaking work \cite{Belavin:1984vu}, extensive research has been devoted to 2d CFTs, revealing that certain models possess larger symmetries.  One notable generalization involves the inclusion of higher-spin currents, giving rise to $\mathcal{W}$-algebras. The pioneering discovery in this direction was the $\widetilde\cW_3$-algebra by Zamolodchikov \cite{Zamolodchikov:1985wn}, which introduced an additional current, $W^{(3)}(z)$, alongside the conventional Virasoro current $T(z)$. The commutation relations of $\mathcal{W}$-algebras involve non-linear terms. Consequently, the bootstrap construction of higher-spin algebras and finding their OPEs become increasingly complicated. Nonetheless, this difficulty was overcome by Fateev and Lukyanov through a generalized Coulomb gas representation \cite{Fateev:1987zh}.

The heart of the Coulomb gas realizations, which provide a powerful tool for studying these intricate quantum algebras, lies the concept of the Miura transformation \cite{miura1968korteweg1,miura1968korteweg2}. The $\cW$-algebra generators of $\fraksl_N$-type can be defined via the Miura transformation \cite{Fateev:1987zh,Lukyanov:1987xg}\footnote{We use an additional accent $\,\tilde{ }\,$ for $\mathcal{W}_{\fraksl_N}$ generators to distinguish them from those of $\mathcal{W}_{\mathfrak{gl}_N}$, which contains an extra $\U(1)$ current.
}
\begin{equation}\label{Miura_tr}
\widetilde{W}_{\mathfrak{sl}_N}(z)\equiv-\sum_{k=0}^{N} \widetilde{W}^{(k)}(z)(Q \partial)^{N-k}\equiv\prod_{i=1}^N\left(Q \partial-\boldsymbol{\nu}_{i} \cdot \partial \boldsymbol{\phi}(z)\right) ~~,
\end{equation}
where $\boldsymbol{\nu}_i=\boldsymbol{e}_i-\frac{1}{N}\sum_{j=1}^{N}\boldsymbol{e}_j$ are the weights of the fundamental representation. (Here we write the simple roots of $A_{N-1}$ as $\boldsymbol{\alpha}_i=\boldsymbol{e}_i-\boldsymbol{e}_{i+1}$.) We call it the $\widetilde{\cW}_{N}$-algebra, and it is generated by the currents $\widetilde{W}^{(k)}(z)$ for $k=2,\dots,N$. It is easy to see that $\widetilde{W}^{(0)}=-1$. Because $\sum_{i=1}^{N}\boldsymbol{\nu}_i=0$, the $\cW$-algebra associated to $\mathfrak{sl}_{N}$ has trivial spin-one current
\begin{equation}
\widetilde{W}^{(1)}(z)=\sum_{i=1}^{N}\boldsymbol{\nu}_i\cdot\partial\boldsymbol{\phi}(z)=0.
\end{equation}
The first non-trivial current appears at spin two
\begin{equation}
\begin{aligned}
T(z)=\widetilde{W}^{(2)}(z)=&-\sum_{i<j}(\boldsymbol{\nu}_i\cdot \partial\boldsymbol{\phi})(\boldsymbol{\nu}_j\cdot \partial\boldsymbol{\phi})+Q\sum_{i=1}^{N}(i-1)\boldsymbol{\nu}_i\cdot\partial^2\boldsymbol{\phi}\cr
=&\frac{1}{2}\sum_i(\boldsymbol{\nu}_i\cdot \partial\boldsymbol{\phi})(\boldsymbol{\nu}_i\cdot \partial\boldsymbol{\phi})+Q\rho^\vee \cdot\partial^2\boldsymbol{\phi}
\end{aligned}\label{Virasoro_g}
\end{equation}
where $\boldsymbol{\rho}^\vee:=\sum_i(i-1)\boldsymbol{\nu}_i=\frac{1}{2}\sum_{j>i}(\boldsymbol{e}_j-\boldsymbol{e}_i)$ is the Weyl vector. This expression coincides with the Sugawara construction of the energy-momentum tensor. Using the OPEs of free bosons
\begin{equation}\label{free_boson_ope_N}
    \phi_i(z)\phi_j(w)\sim \delta_{i,j}\log(z-w),\quad \phi_i:=\boldsymbol{e}_i\cdot\boldsymbol{\phi}
\end{equation}
the $TT$ OPE obey \eqref{TT} with the central charge of the $\widetilde{\cW}_N$-algebra is given by
\be\label{cc-Wsln}
c=(N-1)-N(N^2-1)Q^2~.
\ee
The next current, $\widetilde{W}^{(3)}(z)$, can be written as
\begin{equation}
\begin{aligned}
\widetilde{W}^{(3)}(z)=&\sum_{i<j<k}(\boldsymbol{\nu}_i\cdot \partial\boldsymbol{\phi})(\boldsymbol{\nu}_j\cdot \partial\boldsymbol{\phi})(\boldsymbol{\nu}_k\cdot \partial\boldsymbol{\phi})-Q\sum_{i<j}(i-1)\partial\lt((\boldsymbol{\nu}_i\cdot \partial\boldsymbol{\phi})(\boldsymbol{\nu}_j\cdot \partial\boldsymbol{\phi})\rt)\cr
&-Q\sum_{i<j}(j-i-1)(\boldsymbol{\nu}_i\cdot \partial\boldsymbol{\phi})(\boldsymbol{\nu}_i\cdot \partial^2\boldsymbol{\phi})+\frac{Q^2}{2}\sum_i(i-1)(i-2)(\boldsymbol{\nu}_i\cdot \partial^3\boldsymbol{\phi}).
\end{aligned}
\end{equation}
To make it a primary field, we need to modify it along with the normalization \cite{Bouwknegt:1992wg}:
\begin{equation}
    W^{(3)}(z)=\mathrm{i}\sqrt{3\beta}\left( \widetilde{W}^{(3)}(z) -\frac{N-2}{2} Q\partial \widetilde{W}^{(2)}(z)\right)\,
\end{equation}
For the case of $n=3$, we can confirm the following OPEs
\begin{align}
T(z)T(w)\sim& \frac{c/2}{(z-w)^4}+\frac{2T(w)}{(z-w)^2}+\frac{\partial T(w)}{z-w},\\
\label{opeTW}
T(z) {W}^{(3)}(w)\sim& \frac{3{W}^{(3)}(w)}{(z-w)^2}+\frac{\partial {W}^{(3)}(w)}{z-w},\\
{W}^{(3)}(z) {W}^{(3)}(w)\sim& \frac{c/3}{(z-w)^6}+\frac{2T(w)}{(z-w)^4}+\frac{\partial T(w)}{(z-w)^3}\nonumber\\
&+\frac{1}{(z-w)^2}\lt(2\beta\Lambda_4(w)+\frac{3}{10}\partial^2 T(w)\rt)\nonumber\\
&+\frac{1}{z-w}\lt(\beta\partial\Lambda_4(w)+\frac{1}{15}\partial^3T(w)\rt),\label{opeWW}
\end{align}
where $\beta={16}/(5c+22)$ and 
\begin{equation}
\Lambda_4(z)=:TT:(z)-\frac{3}{10}\partial^2T(z).
\end{equation}
We note that the second equation (\ref{opeTW}) implies that $W^{(3)}(z)$ is a primary field of spin~3. The OPE in the third line (\ref{opeWW}) is determined from the conformal symmetry, and it shows the non-linearity due to the presence of $\Lambda_4$, which is a characteristic feature of the $\cW$-algebra.

In a similar way, we can  define the $\cW$-algebra associated to $\mathfrak{gl}_{N}$, denoted as the $\cW_N$-algebra, by the Miura transformation
\begin{equation}\label{Miura-glN}
W_{\frakgl_N}(z):=-\sum_{k=0}^{N} W^{(k)}(z)(Q \partial)^{N-k}:=\left(Q \partial-\boldsymbol{e}_{1} \cdot \partial \boldsymbol{\phi}(z)\right) \cdots\left(Q \partial-\boldsymbol{e}_{N} \cdot \partial \boldsymbol{\phi}(z)\right).
\end{equation}
The first several non-trivial currents are given by
\bea
    W^{(1)}(z)=&\sum_i\partial\phi_i,\cr
    W^{(2)}(z)=&-\sum_{i<j}\partial\phi_i\partial\phi_j+Q\sum_i(i-1)\partial^2\phi_i,\label{Virasoro_W}\cr
    W^{(3)}(z)=&\sum_{i<j<k}\partial\phi_i\partial\phi_j\partial\phi_k-Q\sum_{i<j}(i-j)\partial(\partial\phi_i\partial\phi_j)\cr
    &-Q\sum_{i<j}(j-i-1)\partial\phi_i\partial^2\phi_j+\frac{Q^2}{2}\sum_i(i-1)(i-2)\partial^3\phi_i.
\eea
After the following redefinition
\be
T(z)=W^{(2)}+\frac12(W^{(1)})^2-\frac{Q(N-1)}2 \partial W^{(1)}=\frac12(\partial\boldsymbol{\phi})^2+Q\rho^\vee\cdot\boldsymbol{\phi},
\ee
the OPEs of $W^{(1)}(z)$ and $T(z)$ becomes,
\begin{align}
W^{(1)}(z)W^{(1)}(w)\sim & \frac{N}{(z-w)^2}. \cr 
 T(z)W^{(1)}(w)\sim & \frac{W^{(1)}(w)}{(z-w)^2}+\frac{\partial W^{(1)}(w)}{(z-w)}. \cr 
 T(z)T(w)\sim & \frac{c / 2}{(z-w)^4}+\frac{2T(w)}{(z-w)^2}+\frac{\partial T(w)}{(z-w)}
\end{align}
where the central charge is given by \eqref{cc-Wsln} and the contribution from one free boson
\be \label{cc-Wgln}
c=1+(N-1)-N(N^2-1)Q^2~.
\ee 

Now, let us introduce the operators
\begin{equation}
    R_i:=Q \partial-\boldsymbol{e}_{i} \cdot \partial \boldsymbol{\phi}=Q\partial-\partial\phi_i=Q\partial-\boldsymbol{\nu}_i\cdot\partial\boldsymbol{\phi}-\frac{1}{N}\partial\phi
\end{equation}
where $\phi:=\sum_i\phi_i$. Since we can express $R_i$ as
\begin{equation}
    R_i=e^{\frac{\phi}{QN}}\lt(Q\partial-\boldsymbol{\nu}_i\cdot\partial\boldsymbol{\phi}\rt)e^{-\frac{\phi}{QN}}~,
\end{equation}
the $\cW$-algebras of type $\fraksl_N$ and $\frakgl_N$ are related through the conjugation by $e^{-\frac{\phi}{QN}}$.

A set of screening charges corresponding to the simple roots $\boldsymbol{\alpha}_i=\boldsymbol{e}_i-\boldsymbol{e}_{i+1}$ ($i=1,\cdots, N-1$)  of type $A_{N-1}$ can be introduced as
\begin{equation}
    \cS_\pm^{(i)}:=\oint\frac{{\rm d}z}{2\pi i}:\exp\lt(t_\pm \boldsymbol{\alpha}_i\cdot \boldsymbol{\phi}(z)\rt):~,\qquad t_\pm =\pm \beta^{\pm 1/2}\,,
    \label{eq:WNscreening}
\end{equation}
for $i=1,2,\dots,n$.\footnote{Here, $t_\pm$ differs from \eqref{t-Virasoro} by a factor of $\sqrt{2}$, which arises due to the different normalization of free bosons. Nonetheless, as both versions fundamentally play the same role, we use the same symbol for simplicity, albeit with a slight difference in normalization. The same caution is also applied to the notation of $t$ in the subsequent $\cW_N$ minimal model as well. } When $Q=\sqrt{\beta}-\sqrt{\beta}^{-1}$, the screening charges $\cS_+^{(i)}$ commute with a part of the Miura transformation,
\begin{small}\bea
   & R_i(z)R_{i+1}(z)e^{\sqrt{\beta}(\phi_i(w)-\phi_{i+1}(w))}\cr 
   \sim &-\big[\tfrac{\beta}{(z-w)^2}+\partial_z\tfrac{Q\sqrt{\beta}}{z-w}\big]e^{\sqrt{\beta}(\phi_i(w)-\phi_{i+1}(w))}
    -\tfrac{\sqrt{\beta}}{z-w}:\lt(\partial\phi_i(w)-\partial\phi_{i+1}(w)\rt)e^{\sqrt{\beta}(\phi_i(w)-\phi_{i+1}(w))}:\cr
    =&-\frac{1}{(z-w)^2}e^{\sqrt{\beta}(\phi_i(w)-\phi_{i+1}(w))}-\frac{\sqrt{\beta}}{z-w}:\bigg[\partial\phi_i(w)-\partial\phi_{i+1}(w)\bigg]e^{\sqrt{\beta}(\phi_i(w)-\phi_{i+1}(w))}:\cr
    =&-\partial_w\lt(\frac{1}{z-w}e^{\sqrt{\beta}(\phi_i(w)-\phi_{i+1}(w))}\rt)~.
\eea\end{small}
 This implies that the screening charges $\cS^{(i)}_+$ commute with the Miura transformation as well as all the currents in the $\cW_N$-algebra. One can similarly confirm that $\cS^{(i)}_-$ also serves the role of the screening currents. It is worth mentioning that the screening charges $\cS^{(i)}_\pm$ also commute with the vertex operator $e^{-\frac{\phi}{QN}}$. Hence, the screening charges are common to both the $\cW_N$-algebra and the $\widetilde{\cW}_N$-algebra.

\subsubsection{Screening currents and \texorpdfstring{$\mathcal{W}_N$}{WN} minimal model}\label{sec:minimal}

In the previous subsection, we construct the Virasoro minimal model by considering the null states from the screening current. Here, we generalize the construction to the $\widetilde\cW_N$-algebra. To this end, we introduce $(N-1)$ free bosons and use the Fock module generated from the highest-weight state
\begin{equation}
    |\boldsymbol{t}\rangle =\lim_{z\rightarrow 0} :\exp\left(\boldsymbol{t}\cdot \boldsymbol{\phi}\right):|0\rangle~,
\end{equation}
which has a conformal dimension $\Delta = \frac12 \boldsymbol{t}\cdot (\boldsymbol{t}-2Q\boldsymbol{\rho})$.

We define a singular state 
as in (\ref{singular_state})
\begin{equation}\label{sing_vectW_n}
    |\boldsymbol{t}\rangle^\pm_{i,r_i}\sim  (\cS^{(i)}_\pm)^{r_i}|\boldsymbol{t}-r_i t_\pm \boldsymbol{\alpha}_i\rangle \,.
\end{equation}
It satisfies the highest weight condition since the screening charge commutes with all the generators of $\cW$-currents.

Analogous to (\ref{onshell_c}), the condition that the difference of the conformal dimensions becomes an integer can be written as
\begin{equation}
    \Delta(\boldsymbol{t}-r_i t_+ \boldsymbol{\alpha}_i) -\Delta(\boldsymbol{t}) = r_i s_i\,
\end{equation}
which leads to
\begin{equation}
    \boldsymbol{t}\cdot\boldsymbol{\alpha_i}=
    (r_i+1)t_++(s_i+1)t_-\,.
\end{equation}
We introduce the fundamental weight $\boldsymbol{\omega}_j$ satisfying $\boldsymbol{\alpha}_i\cdot \boldsymbol{\omega}_j=\delta_{ij}$ and
$\boldsymbol{\omega}_i\cdot\boldsymbol{\omega}_j =\frac{i(N-j)}{N}$.
By expanding in terms of $\boldsymbol{\omega}_i$, one obtains
\begin{equation}
\boldsymbol{t}=\sum_{i=1}^{N-1}\left((r_i+1)t_++ (s_i+1)t_-\right)\boldsymbol{\omega}_i\,.
\end{equation}
To obtain the minimal model, one has to impose an additional condition that we have a similar singular vector in the direction
\begin{equation}\label{affine_condition}
    \boldsymbol{\alpha}_N=-\boldsymbol{\alpha}_1-\cdots - \boldsymbol{\alpha}_{N-1}\,.
\end{equation}
This requirement leads us to the following condition for the existence of a similar singular vector
\begin{equation}
    \Delta(\boldsymbol{t}-r_N t_+ \boldsymbol{\alpha}_N) -\Delta(\boldsymbol{t})=r_N s_N\,.
\end{equation}
Using (\ref{affine_condition}), this is equivalent to
\begin{equation}
    t_+\sum_{i=1}^N r_i +t_- \sum_{i=1}^N s_i=0\,.
\end{equation}
Therefore, writing 
\begin{equation}\label{WN-pq}
    \sum_{i=1}^N r_i=p,\quad \sum_{i=1}^N s_i=q\,,
\end{equation}
one must impose, by combining (\ref{eq:WNscreening}),
\begin{equation}\label{rationality}
   \sqrt{\beta}=t_+ =-1/t_-=\sqrt{q/p}\,.
\end{equation}

By imposing these conditions and constraints, we obtain the $\widetilde\cW_N$ minimal model. The $\widetilde\cW_N$ minimal model is characterized by its central charge $c$
\begin{equation}\label{Wminimal}
c=(N-1)\left(1-\frac{(p-q)^2}{pq}N(N+1)\right)\,.
\end{equation}
It is endowed with finitely many primary states with conformal dimensions
\begin{equation}\label{eq:hwmm}
\Delta(\boldsymbol{r}, \boldsymbol{s}) = \frac{1}{24pq} \left(12(\sum_i(q r_i -p s_i)   \boldsymbol{\omega}_i )^2 -N(N^2-1)(p-q)^2\right)\,,
\end{equation}
where the sets of positive integers $r_i$ and $s_i$ ($i=1,\ldots, N-1$) satisfy the constraints
\begin{equation}
    \sum_{i=1}^{N-1} r_i=p-r_N<p,\quad \sum_{i=1}^{N-1} s_i=q-s_N<q\,.\label{eq:minimalprimarycond}
\end{equation}

\subsection{\texorpdfstring{$\mathcal{W}_{1+\infty}$}{W1+inf}-algebra: Lie algebra of differential operators}\label{sec:Winf}
The primary purpose of this note is to provide an explanation of the quantum toroidal algebra and its degeneration, the affine Yangian, which are deformations of the $\mathcal{W}_{1+\infty}$ algebra (for the original references, see for example, \cite{Bakas:1989mz,Pope:1989ew,Pope:1990be}.) As these algebras require sophisticated techniques to describe, it would be helpful to begin with the undeformed versions of these symmetries to gain some intuition about them. For a more in-depth exploration, we refer to the review paper \cite{Awata:1994tf}, which provides extensive details on the subject.

The $\mathcal{W}_{1+\infty}$-algebra is a central extension of the Lie algebra of higher-order differential operators generated by $X$ and $D_X=X\partial_X$.  Using two arbitrary polynomials $f(y), g(y)\in \mathbb{C}[y]$, the commutation relation between general combinations of $X, D$ can be summarized in the following equation:
\begin{equation}
\left[
X^n f(D), X^m g(D)
\right]= X^{n+m} \left(f(D+m)g(D) -f(D) g(D+n)\right),\quad n,m\in \mathbb{Z}\,.
\end{equation}
To introduce the central extension to this algebra, we define a map from differential operators to quantum operators as
\begin{equation}
X^n f(D) \rightarrow W[X^n f(D)],
\end{equation}
where the argument in the bracket gives the label of the generator.
The commutation relation among $W[X^n f(D)]$ is
\begin{align}
&\left[
W[X^n f(D)], W[X^m g(D)]
\right]= W[X^{n+m}f(D+m)g(D)] - W[f(D) g(D+n)]\cr
&\hspace{5cm}+C \Psi\left(X^n f(D), X^m g(D)\right),\\
&\Psi\left[X^n f(D), X^m g(D)\right)=\delta_{n+m,0} \big(
\theta(n\geq 1) \sum_{j=1}^n f(-j)g(n-j) -\theta(m\geq 1)\sum_{j=1}^m f(m-j)g(-j)
\big],\nonumber
\end{align}
where $C$ is the central extension parameter, and $\theta$ is the step function, which is equal to 1 when its argument is true and 0 otherwise. This algebra contains the U(1) current and the Virasoro generators, represented by $\sfJ_n =W[X^n]$ and $\sfL_n =-W[X^nD]$, respectively. The central charge of the Virasoro algebra in this context is given by $c_{\textrm{Vir}}=-2C$. Additionally, the higher spin currents ${\sfW}^{(s)}_n$ are associated with the higher-order derivative $W[X^n D^{s-1}]$.

The $\mathcal{W}_{1+\infty}$ algebra contains an infinite number of commuting operators $W[D^n]$ for $n=0,1,2,\cdots$. A highest-weight state representation is specified by 
\begin{align}
W[X^nD^m]|\Delta\rangle =& 0,\qquad n\geq 1, \quad m\geq 0\\
W[D^m]|\Delta\rangle =& \Delta_m |\Delta\rangle, \qquad  m\geq 0
\end{align}
with an infinite number of weight parameters $\Delta_m$.
The Hilbert space is generated by the multiplication of $W[X^{-n}D^m]$ for $n>1, m\geq 0$ on the highest weight state $|\Delta\rangle$. Similar to the Virasoro algebra, states can be labeled by the eigenvalues of $\sfL_0$, which we refer to as the level. However, a novelty in this context is that the number of states at each level is infinite. For example, at level 1, the Hilbert space is spanned by $\left\{
W(X^{-1}D^m)\right\}$ for $m\geq 0$.

Kac and Radul \cite{kac1993quasifinite} found conditions under which the number of states at each level becomes finite, referred to as the quasi-finite representation. We introduce a generating function of the weights as
\begin{equation}
\Delta(\epsilon)=-\sum_{n=0}^\infty \frac{\Delta_n}{n!}\epsilon^n\,.
\end{equation}
One can interpret the generating function $\Delta(\epsilon)$ as the eigenvalue of the highest weight states $|\Delta\rangle$ with respect to the operator:
\begin{equation}
W[-e^{\epsilon D}]:=-\sum_{n=1}^\infty \frac{\epsilon^n}{n!}W[D^n]\,.
\end{equation}
The condition for the quasi-finiteness is that there exists a monic polynomial $p(\epsilon)$ such that $\Delta(\epsilon)$ satisfies the following differential equation:
\begin{equation}\label{QFrep}
p\left(\frac{d}{d\epsilon}\right) \left(
(e^\epsilon-1)\Delta(\epsilon) +C
\right)=0\,.
\end{equation}
This equation and its derivation are detailed in \cite{kac1993quasifinite,Awata:1994tf}.
If we factorize $p(\epsilon)=\prod_{i=1}^K (\epsilon-\lambda_i)^{m_i}$ with $m_i>0$, the solution to (\ref{QFrep}) is:
\begin{equation}
\Delta(\epsilon)=\frac{\sum_{i=1}^K P_i(\epsilon) e^{\lambda_i \epsilon}-C}{e^\epsilon-1}
\end{equation}
where $P_i(\epsilon)$ is a polynomial of degree $m_i-1$.

The unitary representation of the $\mathcal{W}_{1+\infty}$ algebra was studied in \cite{kac1993quasifinite}. In this context, it is assumed that $C$ is a non-negative integer ($C=N$) and that the function $\Delta(\epsilon)$ takes the form
\begin{equation}\label{Delta_unitary}
\Delta(\epsilon) = \sum_{i=1}^N \frac{e^{\lambda_i \epsilon}-1}{e^\epsilon-1}\,.
\end{equation}
This equation is essential to understand the unitary representation of $\mathcal{W}_{1+\infty}$ algebra.

The unitary representation of the $\mathcal{W}_{1+\infty}$ algebra can also be described using free fermions, known as $bc$-ghosts. $N$ pairs of free fermions, $b^{(i)}(z), c^{(i)}(z)$ ($i=1,\cdots, N$) are introduced with the mode expansion and the anti-commutation relations,
\begin{equation}\label{mod_expansion}
b^{(i)}(z) =\sum_{s\in \mathbb{Z}+1/2} \sfb^{(i)}_s z^{-s-\frac12-\lambda_i}\,,\ 
c^{(i)}(z) =\sum_{r\in \mathbb{Z}+1/2} \sfc^{(i)}_r z^{-s+\lambda_i-\frac12}\,,\quad
\left\{b^{(i)}_r, c^{(j)}_s\right\}=\delta_{i,j}\delta_{r+s,0}\,.
\end{equation}
The Fock space is specified by
\begin{equation}\label{fermion_vacuum}
\sfb^{(i)}_s|\lambda^{(i)}\rangle=\sfc^{(i)}_s|\lambda^{(i)}\rangle=0\quad\mbox{for}\qquad s>0\,.
\end{equation}
We define the generators of the algebra as
\begin{align}
W[X^n e^{\epsilon D}]=&\mbox{Reg.}\left(\sum_{i=1}^N \oint \frac{dz}{2\pi i} b^{(i)}(z) z^n e^{\epsilon D_z}c^{(i)}(z)\right)\cr
=&\sum_{i=1}^N \oint\frac{dz}{2\pi i}:b^{(i)}(z) z^n e^{\epsilon D_z}c^{(i)}(z):-\sum_i^N \frac{e^{\lambda_i \epsilon}-1}{e^\epsilon-1}\delta_{n,0}
\label{W-generator2}
\end{align}
The ``Reg.'' in the first line takes only the regular part, $D_z=z\partial_z$, and the normal ordering $:~:$ in the second line is defined by 
\begin{equation}\label{fermion_normal_ordering}
:\sfb_r \sfc_s:=\left\{
\begin{array}{ll}
\sfb_r \sfc_s \quad & r\leq s\\ -\sfc_s \sfb_r \quad & r>s 
\end{array}
\right.
\end{equation}
The vacuum state $|\boldsymbol{\lambda}\rangle=\otimes_i |\lambda^{(i)}\rangle$ satisfies the highest weight condition,
\begin{equation}\label{highest_weight}
W[X^n e^{\epsilon D}]|\boldsymbol{\lambda}\rangle =
\begin{dcases}
0\quad & n>0\\
-\sum_i\frac{e^{\lambda_i \epsilon}-1}{e^\epsilon-1}|\boldsymbol{\lambda}\rangle
& n=0
\end{dcases}
\end{equation}
It coincides with (\ref{Delta_unitary}), which implies that the unitary representation of $\cW_{1+\infty}$ admits the realization by the free fermions.

An alternative way to describe the unitary representation of the $\mathcal{W}_{1+\infty}$ algebra is by replacing the free fermions with $N$ pairs of symplectic bosons (or $\upbeta\upgamma$ ghosts) $\sfb_r^{(i)}\to \upbeta^{(i)}_r$, $\sfc_s^{(i)}\to \upgamma_s^{(i)}$ with the commutation relation:
\begin{equation}
\left[\upbeta_r^{(i)}, \upgamma^{(j)}_s\right]=-\delta_{r+s,0}\delta_{i,j}\,.
\end{equation}
The replacements of the fermion oscillators to $\upbeta,\upgamma$ in (\ref{mod_expansion}-\ref{fermion_normal_ordering}) can be done by keeping similar forms of equations. One difference is the sign change in (\ref{W-generator2}), which implies the sign change in (\ref{fermion_normal_ordering}) (\ref{highest_weight}) for $n=0$. This in turn implies that the central charge in this case is given by the negative integer $C=-N$.

\paragraph{Relation with $\Hat{\mathfrak{gl}}_\infty$}
From now on, we set $N=1$. Let us start from the simplest case where the weight function takes the form,
\begin{equation}\label{simplest_weight}
\Delta(\epsilon) = C \frac{e^{\lambda \epsilon}-1}{e^\epsilon-1}\,,
\end{equation}
For $C=\pm 1$, the algebra has a representation by the free fermion or the symplectic boson. For $C=\pm 1$ case, 
it is possible to rewrite the first term in (\ref{W-generator2}) as
\begin{equation}
\sum_{s\in \mathbb{Z}+1/2}e^{\epsilon(\lambda-s-1/2)}E_{n-s,s}
\end{equation}
where $E_{r,s}$ equals to
\begin{equation}
E_{r,s}=\begin{cases} :\sfb_r \sfc_s:\quad& C=1\\
:\upbeta_r\upgamma_s:\quad & C=-1\end{cases}
\,.
\end{equation}
The operators $E_{r,s}$ satisfy the $\Hat{\mathfrak{gl}}_\infty$ algebra, as defined by the commutation relations:
\begin{equation}\label{glinfty}
\left[E_{r,s}, E_{r',s'}\right]=
\delta_{r'+s,0}E_{r,s'}-\delta_{r+s',0}E_{r',s}
+C \delta_{r+s',0}\delta_{r'+s,0}(\theta(r>0)-\theta(r'>0))
\end{equation}
where $C=\pm 1$. The highest weight state $|\Delta\rangle$ satisfies
\begin{equation}\label{HWC_glinfty}
E_{r,s}|\Delta\rangle=0,\quad
r>0\quad\mbox{or}\quad s>0\,.
\end{equation}
It is possible to express the generator of $\mathcal{W}_{1+\infty}$ in terms of $\Hat{\mathfrak{gl}}_\infty$ for any value of $C$ while preserving the highest weight condition (\ref{HWC_glinfty}).
Thus, for this particular choice \eqref{simplest_weight} of the weight, $\mathcal{W}_{1+\infty}$ is equivalent to $\Hat{\mathfrak{gl}}_\infty$. 

\paragraph{Generalized character}
Since we have infinitely many commuting operators, $W[D^n]$, the generalized character is defined using these operators as
\begin{equation}\label{q-character}
\boldsymbol{Z}([g]):=\mathrm{Tr}_{\mathcal{H}}\, e^{\sum_{n=0}^\infty g_n W[D^n]}\,,
\end{equation}
This type of summation with the most general weight is an analog of $q$-character \cite{frenkel1999algebras} (see also \cite{Awata:1994tf} where it is referred to as ``full character''.)

As an exercise, let us consider the representation associated with the weight (\ref{simplest_weight}), which is an example of a $\Hat{\mathfrak{gl}}_\infty$-module. For a general $C$, the Hilbert space is generated by applying the $E_{r,s}$ operators with $r,s<0$ arbitrary times. By using the commutation relation\footnote{To derive the following equation, it is useful to consider the free fermion representation. By observing that $W(e^{\epsilon D})=\sum_s :\sfb_{-s}\sfc_s: e^{\e(-s+\lambda-1/2)}$, we can deduce the commutation relations
$$
\left[W[e^{\epsilon D}], \sfb_r\right]=\sfb_r e^{\epsilon(r+\lambda-1/2)},\quad
\left[W[e^{\epsilon D}], \sfc_r\right]=-\sfc_r e^{\epsilon(-r+\lambda-1/2)}\,.
$$
By taking the expansion with respect to $\epsilon$, we obtain
$$
\left[W[D^n],\sfb_r\right] =(r+\lambda-1/2)^n \sfb_r,\quad
\left[W[D^n],\sfc_r\right] =-(-r+\lambda-1/2)^n \sfc_r.
$$
This gives the commutation relation (\ref{Ersweight}) for $E_{r,s}=:\sfb_r \sfc_s:$. The derivation for a general $C$ follows from a similar computation for $E_{r,s}$.
}
\begin{align}\label{Ersweight}
&\exp(\sum_{n=0}^\infty g_n W[D^n])E_{r,s}=x_r y_s E_{r,s} \exp(\sum_{n=0}^\infty g_n W[D^n])~, \cr
& x_r= \exp(\sum_{n=0}^{\infty} g_n (-r+\lambda-1/2)^n)~,\qquad y_r= \exp(-\sum_{n=0}^{\infty} g_n (r+\lambda-1/2)^n)~,
\end{align}
we obtain
\begin{equation}\label{generalized_character}
\boldsymbol{Z}([g]) = \prod_{r,s>0}\frac{1}{1-x_r y_s}= \sum_{\lambda}s_\lambda (x_{1/2}, x_{3/2},\cdots) \, s_\lambda (y_{1/2}, y_{3/2},\cdots)\,.
\end{equation}
Here we use the Cauchy formula \eqref{Schur-nor-id}, which expresses the product as a sum of Schur functions,  $s_\lambda (x_{1/2}, x_{3/2},\cdots), s_\lambda (y_{1/2}, y_{3/2},\cdots)$, with infinite arguments where $\lambda$ runs over all partitions.

The $q$-character \eqref{q-character} can be simplified to the usual character $\Tr_\mathcal{H} \,q^{\sfL_0}=\Tr_\mathcal{H}\, q^{-W[D]}$ by setting 
\begin{equation}
e^{-g_1}=q,\quad g_0=g_2=g_3=\cdots=0~.
\end{equation}
Consequently, the simplified character becomes
\begin{equation}
Z(q)=\prod_{r,s\in \mathbb{Z_{\geq 0}}+1/2} \frac{1}{(1-q^{r+s})}=\prod_{n=1}^\infty \frac{1}{(1-q^n)^n}\,,
\end{equation}
which is identical to the character of the MacMahon module.

As emphasized before, when $C$ is an integer, there occurs the degeneration of the Hilbert space.
For instance, when $C=1$, the free fermion representation implies $E_{r,s} E_{r,s'}=E_{r,s} E_{r',s}=0$ due to the fermionic statistics. This implies that the summation in the third formula in (\ref{generalized_character}) should be restricted to Young diagrams of one row. In general, when $C=N$ (resp. $C=-N$) with $N>0$, $(N+1)$-th anti-symmetrization (resp. symmetrization) annihilates the state. Hence, the summation should be restricted to Young diagrams of height (resp. width) $N$ \cite{Awata:1994xm}.

\paragraph{Toroidal generalization} The $q$-deformation $q$-$\mathcal{W}_{1+\infty}$-algebra can be constructed from \emph{quantum torus algebra} generated by $X$ and $Y$ satisfying
  $$
  XY = qYX\,.
  $$
This can be understood as the algebra formed by the polynomial variable $X$ and $q$-difference operator
\begin{equation}
Y=q^{D_X}=\exp\left(\e D_X\right)\,,
\end{equation}
with $q =e^\e$. (This is called the \emph{polynomial representation} of the quantum torus algebra.)
  The free field representation of the $q$-$\mathcal{W}_{1+\infty}$-algebra is constructed by sandwiching the generators by free fields as in \eqref{W-generator2}. The quantum torus algebra enjoys the $\SL(2,\mathbb{Z})$ transformation
\begin{equation}
\mathrm{SL}(2, \mathbb{Z}) \ni \begin{pmatrix}
a & b \\
c & d
\end{pmatrix}:(X, Y) \mapsto\left(X^a Y^b, X^c Y^d\right)~,
\end{equation}
which will survive after the deformation.

\subsection{\texorpdfstring{$\mathcal{W}_\infty[\mu]$}{Winf[mu]}-algebra}\label{sec:Winf[mu]}
The $\mathcal{W}_\infty[\mu]$-algebra was introduced by Gaberdiel and Gopakumar \cite{Gaberdiel:2010pz} (see also the review \cite{Gaberdiel:2012uj}) to describe the asymptotic symmetry algebra of higher spin theories on $AdS_3$. It is composed of higher spin currents, $\widetilde{W}^{(k)}(z)$, with $k=2,3,4,\cdots$. The parameter $\mu$ is associated with the rank of the $\mathcal{W}$-algebra, so when it is an integer, $\mu=N$, the $\mathcal{W}_\infty[\mu]$-algebra is identified with the $\widetilde{\cW}_N$-algebra. The algebra includes the Virasoro current, $\widetilde{W}^{(2)}(z)=T(z)$, which has the central extension $c$ while the parameter $\mu$ plays the role of the second deformation parameter. The commutation relations between the currents are generally complicated, as in (\ref{opeWW}), but we can schematically write them as
\begin{equation}
\widetilde{W}^{(a)}\widetilde{W}^{(b)}=\sum_c C_{ab}^c [\widetilde{W}^{(c)}],
\end{equation}
where $[\widetilde{W}^{(c)}]$ represents $\widetilde{W}^{(c)}$ and its descendants, which can be nonlinear in general. The coefficients $C_{ab}^c$ are the structure constants, which depend on the spins $a,b,c$. It has been claimed that once $C_{33}^4$ is fixed, the other coefficients, such as $C_{44}^4, C_{34}^5, C_{45}^5$, can be determined in terms of $c$ and $C_{33}^4$. The coefficient $C_{33}^4$ is related to $c$ and $\mu$ as \cite{Gaberdiel:2012ku}
\begin{equation}\label{C334}
(C_{33}^4)^2=\frac{64(c+2)(\mu-3)(c(\mu+3)+2(4\mu+3)(\mu-1))}{(5c+22)(\mu-2)(c(\mu+2)+(3\mu+2)(\mu-1))}
\end{equation}
Since $C_{33}^4$ is the parameter that appears in the algebra, and the other structure constants are written in terms of it, it is more fundamental than $\mu$. Since (\ref{C334}) gives a cubic equation for $\mu$, there are three solutions for each choice of $c$ and $C_{33}^4$:
\begin{equation}
\mu=\mu_a(c,C_{33}^4), \quad a=1,2,3,.
\end{equation}
Since the cubic equation does not contain a linear term in $\mu$, these three solutions satisfy the constraint
\begin{equation}\label{mu_constraint}
\sum_{a=1,2,3}\frac{1}{\mu_a}=0,.
\end{equation}
The permutation symmetry among $\mu_a$ is referred to as \emph{triality} in \cite{Gaberdiel:2010pz,Gaberdiel:2012ku,Gaberdiel:2012uj}.
They are related to parameters of the affine Yangian, $\epsilon_a$, as $\mu_a = A/\epsilon_a$ ($a=1,2,3$) with a constant $A$.

When one of $\mu$ (say $\mu_1$) is a positive integer $N$, $\mathcal{W}_\infty[\mu]$ is reduced to the $\widetilde\cW_N$-algebra. If we parameterize the central charge as
\begin{equation}
    c_{N,k}= (N-1)\left(
    1-\frac{N(N+1)}{(N+k)(N+k+1)}
    \right)
\end{equation}
that corresponds to $p=N+k$ and $q=N+k+1$ in (\ref{Wminimal}), the second and the third $\mu$ becomes
\begin{equation}
    \mu_2 = \frac{N}{N+k},\quad \mu_3= -\frac{N}{N+k+1}.
\end{equation}

This parameterization gives the relation between the triality and the level-rank duality in the $\widetilde\cW$ minimal models. We define new parameters as
\begin{equation}
    M=\mu_2=\frac{N}{N+k}\,.
\end{equation}
We can give the relation between two coset models
\begin{equation}
    \mathcal{W}_{N,k}=\frac{\mathfrak{su}(N)_k\oplus \mathfrak{su}{(N)}_1}{\mathfrak{su}(N)_{k+1}}\cong
    \frac{\mathfrak{su}(M)_l\oplus \mathfrak{su}{(M)}_1}{\mathfrak{su}(M)_{l+1}}=\mathcal{W}_{M,l}
\end{equation}
where
\begin{equation}
    k=\frac{N}{M}-N,\quad l=\frac{M}{N}-M\,.
\end{equation}
This correspondence between the coset models is referred to as the level-rank duality.

\subsubsection{Corner VOA}\label{sec:corner}
A generalization of the reduction of $\cW_\infty[\mu]$ is given in the context of the affine Yangian.
Instead of putting ``quantization condition" $\mu_1=N$, we impose
\begin{equation}\label{restriction}
    \sum_a \frac{N_a}{\mu_a}=1\,
\end{equation}
where $N_a$ ($a=1,2,3$) are positive integers.
Since $\mu_a$ satisfies (\ref{mu_constraint})
there is a freedom to shift the integer set $(N_1, N_2, N_3)$ simultaneously as
\begin{equation}\label{shift}
    N_1\rightarrow N_1+m,\quad N_2\rightarrow N_2+m,\quad N_3\rightarrow N_3+m
\end{equation}
where $m$ is an arbitrary integer. With this freedom, one may choose one of $N_a$ to be zero. This is called the \emph{shift symmetry}.

The algebra with the restriction \eqref{restriction} and an additional $\U(1)$ current is called the \emph{corner vertex operator algebra} (in short, corner VOA), denoted by $Y_{N_1 N_2 N_3}$.

It was originally defined in \cite{Gaiotto:2017euk} as an algebra of BPS operators at the corner of three intersecting 5-branes with $N_1, N_2, N_3$ D3-branes attached to the 5-branes as follows:
\begin{align}
    \adjustbox{valign=c}{\begin{tikzpicture}[thick]
		\begin{scope}[scale=1.5,xscale = 1]
		\node[above,scale=1] at (-1,1) {NS5};
		\node[right,scale=1] at (0,0) {D5};
		  \node[scale=1] at (-0.5,0.5){$N_{3}$};
            \node[scale=1] at (-0.5,-0.4){$N_{2}$};
            \node[scale=1] at (-1.5,0.2){$N_{1}$};
		\draw[] (0,0) -- (-1,0) -- (-1.7,-0.7);
		\draw[] (-1,1) -- (-1,0);
		\end{scope}
    \end{tikzpicture}}\hspace{2cm} {\footnotesize
\begin{tabular}{|c|c|c|c|c|c|c|c|c|c|c|}
\hline
& 0 & 1 & 2 & 3 & 4& 5 & 6 & 7 & 8 & 9 \\
\hline
D5 & $\bullet$ & $\bullet$ & $\bullet$ &$-$  & $\bullet$ & $\bullet$ & $\bullet$ & $-$ & $-$ & $-$ \\
\hline
NS5 & $\bullet$ & $\bullet$ & $-$ & $\bullet$ & $\bullet$ & $\bullet$ & $\bullet$ & $-$ & $-$ & $-$ \\
\hline
D3 & $\bullet$ & $\bullet$ & $\bullet$ & $\bullet$ & $-$ & $-$ & $-$ & $-$ & $-$ & $-$ \\
\hline
\end{tabular}}\label{eq:CVOAfigure}
\end{align}
Note that the trivalent vertex of the 5-branes is drawn in the $(2,3)$-plane. Algebraically, it has a realization through the Drinfeld-Sokolov reduction of affine super-Lie algebras. 
Later, Prochazka and Rapcak \cite{Prochazka:2018tlo} found that it defines an algebra associated with the MacMahon module of the affine Yangian with a pit and also provided a representation using free bosons. In this paper, we reverse the order and start directly from the free boson representation and then explain the connection to the affine Yangian later.

\paragraph{Fractional Miura transformation}
We first note that the corner VOA $Y_{N_1 N_2 N_3}$ includes the $\U(1)$ current along with higher spin currents, and the subalgebra of $\cW_\infty[\mu]$ can be obtained by removing the $\U(1)$ current. In the free boson construction, we need to follow this path.

In order to describe $Y_{N_1 N_2 N_3}$, we need to use $N_1+N_2+N_3=:N$ free bosons $\phi_i$ ($i=1,\cdots,N$), which satisfy the OPE relation (\ref{free_boson_ope}) for each $\phi_i=\boldsymbol{e}_i\cdot\boldsymbol{\phi}$. We note that we will use the notation of \cite{Harada:2021xnm} instead of the original reference \cite{Prochazka:2018tlo}.
We introduce a generalized Miura operator as
\begin{equation}\label{GeneralMiura}
    R^{(c)}(\phi) = (\partial_z + \varpi_c\partial\phi)^{\upsilon_c}=e^{-\varpi_c\phi}\partial_z^{\upsilon_c}\, e^{\varpi_c\phi}\,
\end{equation}
for $c=1,2,3$. The parameters are set to
\begin{align}
    & \upsilon_a=\frac{\mu_3}{\mu_a}~,
    \qquad \varpi_a = \sqrt{\frac{\mu_a \mu_3}{\mu_1\mu_2}}\,.
\end{align}
We can identify the parameter $Q$ in (\ref{Miura_tr}) as $Q=\varpi_3^{-1}=\frac{\sqrt{\mu_1\mu_2}}{\mu_3}$ since $R^{(3)}$ is proportional to the Miura operator for $\mathcal{W}_N$-algebra. For $R^{(1)}$ and $R^{(2)}$, there appears the fractional-order differential operator, which is defined by the analytic continuation of the binomial expansion
\begin{equation}
    \partial^\upsilon_z f(z) =\sum_{n=0}^\infty \frac{(-1)^n(-\upsilon)_n}{n!}(\partial_z^n f(z))\partial_z^{\upsilon-n}\,,\quad
    (\upsilon)_n=\upsilon(\upsilon+1)\cdots(\upsilon+n-1).
\end{equation}
With the help of the formula, one obtains
\begin{align}
    (\partial_z+\varpi\partial_z\phi)^\upsilon=&e^{-\varpi\phi}\partial_z^{\upsilon}\, e^{\varpi\phi}\nonumber\\
    =&\sum_{n=0}^\infty \frac{(-1)^n(-\upsilon)_n}{n!}e^{-\varpi\phi}(\partial_z^n e^{\varpi\phi})\partial_z^{\upsilon-n}\nonumber\\
    =&\sum_{n=0}^\infty (-1)^n \frac{(-\upsilon)_n}{n!} P_n[\varpi\partial\phi]\partial_z^{\upsilon-n}\,
\end{align}
where we use
\begin{align}
    e^{-\phi}(\partial_z^n e^{\phi})=(\partial_z+\partial_z\phi)^n\cdot 1=P_n[\partial_z \phi],\quad
    P_n[J]=:(\partial+J)^n:\mathbb{1}\,.
\end{align}
It gives
\begin{equation}
    R^{(c)}(\phi) = \sum_{n=0}^\infty U_n^{(c)}\partial_z^{\upsilon_c-n}, \quad U_n^{(c)}=(-1)^n \frac{(-\upsilon_c)_n}{n!} P_n[\mu_c\partial\phi]\,.
\end{equation}
To define $Y_{N_1 N_2 N_3}$, we combine $N_1$ $R^{(1)}$'s, $N_2$ $R^{(2)}$'s, $N_3$ $R^{(3)}$'s with some order. We introduce a sequence of 1, 2, 3 as $\boldsymbol{c}=(c_1, c_2,\cdots, c_{N})$ which contains $N_1$ $1$'s, $N_2$ $2$'s and $N_3$ $3$'s to describe the Miura transformation
\begin{align}\label{Miura-tranformation}
    R(\boldsymbol{\phi}) =& R^{(c_1)}(\phi_1)R^{(c_2)}(\phi_2)\cdots R^{(c_{N})}(\phi_{N})
    = \sum_{s=0}^\infty U_s(\boldsymbol{\phi})\partial^{N-s}\,
\end{align}
where $U_s(\boldsymbol{\phi})$ ($s=1,2,\cdots$) describes the spin $s$ current. Then, the algebra of the currents
$U_s(\boldsymbol{\phi})$ $(s=1,2,3,\cdots)$ define the corner VOA $Y_{N_1 N_2 N_3}$, which gives a subalgebra of $\mathcal{W}_{1+\infty}[\mu]$, namely $\mathcal{W}_{\infty}[\mu]$ algebra with an extra $\U(1)$ current.
\footnote{To obtain $\mathcal{W}_{\infty}[\mu]$ current, we need to remove the $\U(1)$ current. The easiest way is to modify the Miura transformation $R(\phi)$ as follows
\begin{align}
    \widetilde{R}(\phi):= e^{\boldsymbol{\gamma}\cdot\boldmath{\phi}}R(\phi)
    e^{-\boldsymbol{\gamma}\cdot\boldmath{\phi}}= \sum_{s=0}^\infty \widetilde{U}_s(\boldsymbol{\phi})\partial^{N_1\upsilon_1+N_2\upsilon_2+N_3\upsilon_3-s}\,
    \quad \boldsymbol{\gamma}=\frac{\sum_{i=1}^N \upsilon_{c_i} \varpi_{c_i} \boldsymbol{e_i}}{\sum_{i=1}^N \upsilon_{c_i}}\,
\end{align}
where $\widetilde{U}_1\equiv 0$.
The higher spin currents of $\mathcal{W}_{\infty}[\mu]$ are written by combining $\widetilde{U}_s$ ($s=2,3,4,\cdots$).
In the following, we will use the $\mathcal{W}_{1+\infty}[\mu]$ currents since they are directly related to the affine Yangian, which will be discussed in the subsequent sections.}
For example, we can explicitly write them as
\begin{align*}
    U_0 =& 1\,,\\
    U_1 =& \sum_{i=1}^{N_1+N_2+N_3} \varpi_{c_i}\upsilon_{c_i}\partial_z\phi_i,\\
    U_2 =& \sum_{i=1}^{N_1+N_2+N_3}U_2^{(c_i)}+\sum_{i<j} (U_1^{(c_i)}U_1^{(c_j)}+\upsilon_{c_i}\partial U^{(c_j)}_1)\,.
\end{align*}
While the current $U_s$ depends on the choice of $\boldsymbol{c}$, it has been conjectured that the representation of $Y_{N_1 N_2 N_3}$ is independent of this choice and depends only on $N_1, N_2, N_3$. This conjecture has been supported by the existence of an $R$-matrix that can exchange the order of Miura operators, as demonstrated in \cite{Prochazka:2019dvu}. In the case of the $q$-deformed version, the conjecture is proven in \cite{Harada:2021xnm,Kojima2019,Kojima:2020vtc}. 

The OPE of $U_1$ themselves is given by
\begin{equation}\label{ope_U1}
    U_1(z)U_1(w) =\frac{\sum_{a=1,2,3}N_a\upsilon_a}{Q^2(z-w)^2}\,,
\end{equation}
where we use $\varpi_c^2\upsilon_c^2=\frac{\upsilon_c}{Q^2}$.
From these expressions, we can derive the energy-momentum tensor as $T\propto U_2-\frac12 (U_1)^2$
\begin{equation}\label{EMTY}
    T(z) = -\frac{1}{2}(\partial \boldsymbol{\phi})^2+\boldsymbol{\tau}\cdot\partial^2\boldsymbol{\phi},\quad
    \boldsymbol{\tau}=\sum_{i=1}^N \tau_i \boldsymbol{e}_i,\quad {\tau}_i =\frac{\sum_{j<i}\upsilon_{c_j}-\sum_{j>i}\upsilon_{c_j}}{2\varpi_{c_i}}\,.
\end{equation}
The central charge for this energy-momentum tensor is
\begin{equation}\label{EMTYc}
    c=N_1+N_2+N_3+Q^2\left((\sum_{a=1}^3 N_a\upsilon_a)^3-\sum_{a=1}^3 N_a\upsilon^3_a\right).
\end{equation}
We note that the OPE relation (\ref{ope_U1}) and the central charge (\ref{EMTYc}) are invariant under the shift symmetry (\ref{shift}), which can be verified by using
\be
\sum_a\upsilon_a=0, \quad \textrm{and} \quad \sum_a \upsilon_a^3=\frac{3}{Q^2}~.
\ee
In particular, $Y_{111}$, expressed in terms of three bosons, is (conjectured to be) trivial: all the states generated by the currents $U_s$ have vanishing norm except for the vacuum.

\paragraph{Screening current} The screening current for the corner VOA can be derived by imposing the commutability of the adjacent operators
\begin{equation}
    R^{(c_i)}(\phi_i) R^{(c_{i+1})}(\phi_{i+1})=\partial^{\upsilon_{c_i}+\upsilon_{c_{i+1}}}+ U_1^{c_ic_{i+1}}\partial^{\upsilon_{c_i}+\upsilon_{c_{i+1}}-1}+U_2^{c_ic_{i+1}}\partial^{\upsilon_{c_i}+\upsilon_{c_{i+1}}-2}+\cdots\,.
\end{equation}
For a sequence of operators $R^{(a)}(\phi_1)R^{(b)}(\phi_2)$, the corresponding screening current(s) takes the form
\begin{align}
\mathcal{S}_{ab}=\oint \frac{dz}{2\pi \mathrm{i}}:\exp\left(\boldsymbol{k}_{ab}\cdot \boldsymbol{\phi}(z)\right):\label{eq:degenerate-screening}
\end{align}
with $\boldsymbol{\phi}=(\phi_1,\phi_2)$ and
\begin{align}
    \boldsymbol{k}_{aa}=&\sqrt{\frac{\mu_{c}}{\mu_{b}}}\left(1,-1\right),\\
    \boldsymbol{k}_{ab}=&\left(\sqrt{\frac{\mu_{c}}{\mu_{b}}},-\sqrt{\frac{\mu_{c}}{\mu_{a}}}\right),\quad \textrm{for } a\neq b\,.
\end{align}
where $a,b,c\in \left\{1,2,3\right\}$ are all distinct. Then, it can be shown that the screening operators $\mathcal{S}_{ab}$ commute with $R^{(a)}(\phi_1)R^{(b)}(\phi_2)$ and the commutation with the other operators $ R^{(c_i)}(\phi_i) $ automatically follows. 
We note that we have two choices for $k_{ab}$ for $a=b$ while a unique choice for $a\neq b$.
Also, one can show that
$(k_{ab})^2=-1$ for $a\neq b$ from (\ref{mu_constraint}), which implies that $:\exp\left(\boldsymbol{k}_{ab}\cdot \boldsymbol{\phi}(z)\right):$ is anti-commuting with itself, and $(\mathcal{S}_{ab})^2=0$. We call such a screening current to be ``fermionic''. On the other hand, the screening current $\mathcal{S}_{aa}$ will be referred to be ``bosonic'' \cite{bershtein2018plane,Litvinov:2016mgi}.

\paragraph{Flip transformation}
Depending on the order of $\boldsymbol{c}$, the set of the screening charges changes. It is conjectured, however, that the algebraic structure of $Y_{N_1 N_2 N_3}$ depends only on $N_1, N_2, N_3$. When the order of the adjacent Miura operator is flipped, the statistics of the screening charges are modified. This is referred to as the flip transformation.

Using the shift symmetry, one may choose one of $N_c$ to be zero. Let $a$ and $b$ be the other colors satisfying $a\neq b\neq c\neq a$.  The Miura transformation (\ref{Miura-tranformation}) is composed of $N_a$ $R^{(a)}$'s and $N_b$ $R^{(b)}$'s. Between two choices of bosonic screening operators, we pick
\begin{equation}
    \boldsymbol{k}_{aa}=\sqrt{\frac{\mu_{c}}{\mu_{b}}}\left(1,-1\right),\quad
    \boldsymbol{k}_{bb}=\sqrt{\frac{\mu_{c}}{\mu_{a}}}\left(1,-1\right)\,.
\end{equation}
For an arbitrary sequence of colors $c=(c_1,\cdots, c_{N_a+N_b})$, one may change the order of the adjacent pair $c_{i} c_{i+1}$ ($i=1,\cdots, N_a+N_b-1$), which results in changing the order of the product $R^{(c_i)}(\phi_i)R^{(c_i+1)}(\phi_{i+1})$. It induces the redefinition of the exponents of the screening currents
\begin{equation}
    \tilde{\boldsymbol{k}}'_j=\left\{
    \begin{array}{ll}
    -\tilde{\boldsymbol{k}}_i \quad & j=i\\
    \tilde{\boldsymbol{k}}_j+\tilde{\boldsymbol{k}}_i & j=i-1, i+1\\
    \tilde{\boldsymbol{k}}_j & j\neq i-1, i, i+1
    \end{array}
    \right.
\end{equation}
where $\tilde{\boldsymbol{k}}_i =k_{c_i c_{i+1}}^1 \boldsymbol{e}_i + k_{c_ic_{i+1}}^2 \boldsymbol{e}_{i+1}$ ($k^j_{ab}$ ($j=1,2$) refers to the $j$-th component of $\boldsymbol{k}_{ab}$.) 
In \cite{Litvinov:2016mgi}, this modification of the screening currents is called ``flip transformation". We note that when the flip is made on the fermionic screening currents, the statistics (bosonic or fermionic) of the adjacent screening currents are flipped.

After a series of such flip transformations (or the exchange of the adjacent Miura operators), one can arrive at the sequence $c$ to be
\begin{equation}
    c_i=\left\{
    \begin{array}{ll}
    a\quad & i=1,\cdots, N_a\\
    b\quad & i=N_{a}+1,\cdots, N_a+N_b
    \end{array}
    \right.
\end{equation}
With this choice, the statistics of the screening currents $\mathcal{S}_{c_i c_{i+1}}$ are bosonic for $i=1,\cdots, N_a-1$ and $i=N_{a}+1,\cdots, N_b-1$ while it is fermionic only for $i=N_a$. It gives a Dynkin diagram of $\mathfrak{sl}(N_a-1| N_b-1)$. The corner VOA associated with such $c$ is referred to as $\mathcal{W}$-algebra for the super Lie algebra, denoted by $\mathcal{W}(\mathfrak{sl}_{N_a-1| N_b-1})$ \cite{bershtein2018plane,Litvinov:2016mgi}.

\section{Quantum toroidal algebras and affine Yangians}\label{sec:QTA-deg}

In the previous section, we explored the deformation of $\cW_{1+\infty}$-algebra. This can indeed be viewed as the affine Yangian $\AY$ of $\frakgl_1$. Similarly, the deformation of $q$-$\cW_{1+\infty}$-algebras corresponds to the quantum toroidal algebra $\QTA$ of $\frakgl_1$. These algebras are the central focus of this note. Thus, in this section, we will introduce the quantum toroidal $\frakgl_1$ and affine Yangian $\frakgl_1$ by presenting their generators and relations explicitly.

These algebras can be perceived as the toroidal (a.k.a. double affine or double loop) extensions of quantum groups. For a succinct overview of quantum groups and their generalizations, readers are referred to Appendix \ref{app:QG}. In essence, a toroidal algebra serves as a central extension of a two-variable Laurent polynomial ring $\frakg[u^{\pm1},v^{\pm1}]$ \cite{moody1990toroidal}, and the quantum toroidal algebra is a deformation of its universal enveloping algebra. Consequently, we can construct the quantum toroidal algebra associated to a Lie (super) algebra $\frakg$ \cite{ginzburg1995langlands,ding1997generalization,nakajima2001quiver,jing1998quantum,miki2007q,FFJMM1,Feigin2012gln,Bezerra2019QuantumTA}. Recently, broader generalizations of quantum toroidal algebras associated to quivers $(Q,W)$ were provided in \cite{Li:2023zub,Galakhov:2021vbo,Noshita:2021ldl,Noshita:2021dgj}. We will expand upon these broader constructions of quantum toroidal algebras in this section.

\subsection{Quantum toroidal algebra of \texorpdfstring{$\mathfrak{gl}_{1}$}{gl(1)}}\label{sec:QTA}
The quantum toroidal $\mathfrak{gl}_{1}$ denoted as $U_{q_1,q_2,q_3}(\ddt\frakgl_{1})$ has two independent parameters $q_{c}\in\mathbb{C}^{\times}\,(c=1,2,3)$ with the constraint $q_{1}q_{2}q_{3}=1$. The algebra is generated by Drinfeld currents
\bea\label{Drinfeld}
    E(z)=\sum_{m\in\mathbb{Z}}\frac{\sfE_{m}}{z^{m}},\quad F(z)=\sum_{m\in\mathbb{Z}}\frac{\sfF_{m}}{z^{m}},\quad K^{\pm}(z)=K^{\pm}\exp\left(\sum_{r>0}\mp\frac{\kappa_{r}}{r}\sfH_{\pm r}z^{\mp r}\right)
\eea
and central elements
\bea
    C,\quad K^{-}=(K^{+})^{-1}
\eea
where we impose the second condition.
The current representation of $K^{\pm}(z)$ can also be expanded as
\bea\label{K-modes}
    K^{\pm}(z)=\sum_{r\geq 0}\frac{\sfK^\pm_{\pm r}}{z^{\pm r}}
\eea
where from now on, we omit the superscript $\pm$ for $r\neq0$ modes. Note that $\sfK^{\pm}_{0}=K^{\pm}$.
The defining relations are
\bea
E(z)E(w)=g(z/w)E(w)E(z), &\quad F(z)F(w)=g(z/w)^{-1}F(w)F(z),\cr
K^\pm(z)K^\pm(w) = K^\pm(w)K^\pm (z), \quad&K^-(z)K^+ (w)=\frac{g(C^{-1}z/w)}{g(Cz/w)}K^+(w)K^-(z),\cr
K^\pm(C^{(1\mp1)/2}z)E(w)=&g(z/w)E(w)K^\pm(C^{(1\mp1) /2}z),\cr
K^\pm(C^{(1\pm1)/2}z)F(w)=&g(z/w)^{-1}F(w)K^\pm(C^{(1\pm1)/2}z)\,,\cr
[E(z),F(w)]=\tilde{g} (\delta&\left(\frac{Cw}{z}\right)K^{+}(z)-\delta\left(\frac{Cz}{w}\right)K^{-}(w)),\cr
\mathop{\mathrm{Sym}}_{z_1,z_2,z_3}z_2z_3^{-1}
[E(z_1),[E(z_2),E(z_3)]]=0,&\quad \mathop{\mathrm{Sym}}_{z_1,z_2,z_3}z_2z_3^{-1}
[F(z_1),[F(z_2),F(z_3)]]=0
\label{eq:DIMdef}
\eea
where
\bea\label{str-fn-QTA}
g(z)=\frac{\prod_{c=1}^{3}(1-q_{c}z)}{\prod_{c=1}^{3}(1-q_{c}^{-1}z)},\quad \kappa_{r}=\prod_{c=1}^{3}(q_{c}^{r/2}-q_{c}^{-r/2})=\sum_{c=1}^{3}(q_{c}^{r}-q_{c}^{-r})
\eea
and\footnote{We note that this multiplicative delta function has the property $f(z)\delta(z/a)=f(a)\delta(z/a)$.}
\bea
    \delta(z)=\sum_{m\in\mathbb{Z}}z^{m}.\label{eq:deltafunctiondef}
\eea
The function $g(z)$ is called the structure function of $\QTA$ and obeys the following property
\bea
    g(z)g(z^{-1})=1.\label{eq:reflectionproperty}
\eea
We keep the normalization $\Tilde{g}$ arbitrary.

We also have two grading operators $d_{1},d_{2}$ such that
\bea
    &[d_{1},E(z)]=E(z),\quad [d_{1},F(z)]=-F(z),\quad [d_{1},K^{\pm}(z)]=0,\quad [d_{1},C]=0,\\
    &[d_{2},E(z)]=z\partial_{z}E(z),\quad [d_{2},F(z)]=z\partial_{z}F(z),\quad [d_{2},K^{\pm}(z)]=z\partial_{z}K^{\pm}(z),\quad [d_{2},C]=0. \label{eq:gradingop1}
\eea
Given a generic parameter $\frakq$, we may rewrite the above relations in the product form as\footnote{Sometimes, the degree operators are written as $D=\frakq^{d_{2}},D^{\perp}=\frakq^{d_{1}}$ in the literature. The $\frakq$ here is usually fixed to a specific central charge such as $\frakq=q_{3}^{1/2}$ (in \cite{Feigin:2015raa,Feigin:2016pld} it is $\frakq=q_{2}^{1/2}$). Moreover, the central elements $C,K^{-}$ are written as $C,C^{\perp}=K^{-}$ in the literature.  This notation comes from the fact that they are dual to each other after Miki automorphism (see \S\ref{sec:QTrep}). }
\bea
    &\frakq^{d_{1}}E(z)=\frakq E(z)\frakq^{d_{1}},\quad \frakq^{d_{1}}F(z)=\frakq^{-1}F(z)\frakq^{d_{1}},\quad \frakq^{d_{1}}K^{\pm}(z)=K^{\pm}(z)\frakq^{d_{1}},\quad \frakq^{d_{1}}C=C\frakq^{d_{1}},\\
    &\frakq^{d_{2}}E(z)=E(\frakq z)\frakq^{d_{2}},\quad \frakq^{d_{2}}F(z)=F(\frakq z)\frakq^{d_{2}},\quad \frakq^{d_{2}}K^{\pm}(z)=K^{\pm}(\frakq z)\frakq^{d_{2}},\quad \frakq^{d_{2}}C=C\frakq^{d_{2}}.   \label{eq:gradingop2}
\eea

\paragraph{Hopf algebraic structure}
The quantum toroidal $\mathfrak{gl}_{1}$ is known to have a Hopf algebraic structure \cite{miki2007q}.
The definition of a Hopf algebra is given in Definition \ref{Hopf}. The coproduct is defined as
\bea
    \Delta E(z)=&E(z)\otimes 1+K^{-}(C_{1}z)\otimes E(C_{1}z),\\
    \Delta F(z)=&F(C_{2}z)\otimes K^{+}(C_{2}z)+1\otimes F(z),\\
    \Delta K^{+}(z)=&K^{+}(z)\otimes K^{+}(C_{1}^{-1}z),\\
    \Delta K^{-}(z)=&K^{-}(C_{2}^{-1}z)\otimes K^{-}(z),\\
    \Delta(X)=&X\otimes X,\quad X=C,K^{-},
\label{eq:coproduct}
\eea
where $C_{1}=C\otimes 1$ and $C_{2}=1\otimes C$. The counit and antipode are defined as
\bea
    \varepsilon(E(z))=&\varepsilon(F(z))=0,\\
    \varepsilon(K^{\pm}(z))=&\varepsilon(C)=1,\\
    S(E(z))=&-(K^{-}(z))^{-1}E(C^{-1}z),\\
    S(F(z))=&-F(C^{-1}z)(K^{+}(z))^{-1},\\
    S(K^{\pm}(z))=&(K^{\pm}(Cz))^{-1},\\
    S(C)=&C^{-1}.
    \label{eq:counit-antipode}
\eea
One can easily show that these maps preserve the defining relations (\ref{eq:DIMdef}) and obey the conditions (\ref{eq:Hopf1}) and (\ref{eq:Hopf2}). Note that the maps are extended to products of generators using algebra (anti-)homomorphisms. Since the currents are infinite sum, actually the formulas above are understood under completion.

\subsection{Affine Yangian of \texorpdfstring{$\mathfrak{gl}_{1}$}{gl(1)}}\label{sec:AY}
 Let us introduce the affine Yangian of $\mathfrak{gl}_{1}$. We follow the notation in \cite{Tsymbaliuk}. The affine Yangian of $\mathfrak{gl}_{1}$ is an infinite-dimensional algebra with generators
\be\label{AY-generators}
    \sfe_{n},\quad \sff_{n},\quad \uppsi_{n},\quad (n\in\mathbb{Z}_{\geq 0})
\ee
obeying the following algebraic relations:
\bea
        0=&[\uppsi_{m},\uppsi_{n}],\\
        0=&[\sfe_{m+3},\sfe_{n}]-3[\sfe_{m+2},\sfe_{n+1}]+3[\sfe_{m+1},\sfe_{n+2}]-[\sfe_{m},\sfe_{n+3}]\\
        &+\sigma_{2}[\sfe_{m+1},\sfe_{n}]-\sigma_{2}[\sfe_{m},\sfe_{n+1}]-\sigma_{3}\{\sfe_{m},\sfe_{n}\},\\
        0=&[\sff_{m+3},\sff_{n}]-3[\sff_{m+2},\sff_{n+1}]+3[\sff_{m+1},\sff_{n+2}]-[\sff_{m},\sff_{n+3}]\\
        &+\sigma_{2}[\sff_{m+1},\sff_{n}]-\sigma_{2}[\sff_{m},\sff_{n+1}]+\sigma_{3}\{\sff_{m},\sff_{n}\},\\
        0=&[\sfe_{m},\sff_{n}]-\uppsi_{m+n},\\
        0=&[\uppsi_{m+3},\sfe_{n}]-3[\uppsi_{m+2},\sfe_{n+1}]+3[\uppsi_{m+1},\sfe_{n+2}]-[\uppsi_{m},\sfe_{n+3}]\\
        &+\sigma_{2}[\uppsi_{m+1},\sfe_{n}]-\sigma_{2}[\uppsi_{m},\sfe_{n+1}]-\sigma_{3}\{\uppsi_{m},\sfe_{n}\},\\
        0=&[\uppsi_{m+3},\sff_{n}]-3[\uppsi_{m+2},\sff_{n+1}]+3[\uppsi_{m+1},\sff_{n+2}]-[\uppsi_{m},\sff_{n+3}]\\
        &+\sigma_{2}[\uppsi_{m+1},\sff_{n}]-\sigma_{2}[\uppsi_{m},\sff_{n+1}]+\sigma_{3}\{\uppsi_{m},\sff_{n}\}
    \label{eq:AYgl1mode}
\eea
We also need to impose boundary conditions
\bea
    \relax[\uppsi_{0},\sfe_{m}]=0,\quad [\uppsi_{1},\sfe_{m}]=0,\quad [\uppsi_{2},\sfe_{m}]=2\sfe_{m},\\
    [\uppsi_{0},\sff_{m}]=0,\quad [\uppsi_{1},\sff_{m}]=0,\quad [\uppsi_{2},\sff_{m}]=-2\sff_{m}
\eea
and Serre relations
\bea
    \text{Sym}_{(m_{1},m_{2},m_{3})}[\sfe_{m_{1}},[\sfe_{m_{2}},\sfe_{m_{3}+1}]]=0,\\ \text{Sym}_{(m_{1},m_{2},m_{3})}[\sff_{m_{1}},[\sff_{m_{2}},\sff_{m_{3}+1}]]=0
\eea
where Sym is the complete symmetrization over all indicated indices. The algebra is parametrized by three complex parameters $\epsilon_{1},\epsilon_{2},\epsilon_{3}$ with the following constraint
\bea
    \sigma_{1}\equiv \epsilon_{1}+\epsilon_{2}+\epsilon_{3}=0.\label{eq:AYgl1parametercond}
\eea
The other parameters $\sigma_{2},\sigma_{3}$ are defined as
\bea
    \sigma_{2}=\epsilon_{1}\epsilon_{2}+\epsilon_{1}\epsilon_{3}+\epsilon_{2}\epsilon_{3},\quad \sigma_{3}=\epsilon_{1}\epsilon_{2}\epsilon_{3}.
\eea
The defining relations can be written in a compact form using Drinfeld currents
\be \label{AY-current}
e(z)=\sum_{j=0}^{\infty} \frac{\sfe_{j}}{z^{j+1}}, \quad f(z)=\sum_{j=0}^{\infty} \frac{\sff_{j}}{z^{j+1}}, \quad \psi(z)=1+\sigma_{3} \sum_{j=0}^{\infty} \frac{\uppsi_{j}}{z^{j+1}}
\ee
where the parameter $z$ is the spectral parameter.
\be
\begin{gathered}
e(z) e(w) \sim \varphi(z-w) e(w) e(z),\quad f(z) f(w) \sim \varphi(w-z) f(w) f(z), \\
\psi(z) e(w) \sim \varphi(z-w) e(w) \psi(z),\quad \psi(z) f(w) \sim \varphi(w-z) f(w) \psi(z), \\
\psi(z)\psi(w)=\psi(w)\psi(z),\quad [e(z),f(w)]\sim -\frac{1}{\sigma_{3}}\frac{\psi(z)-\psi(w)}{z-w}
\label{eq:AY}
\end{gathered}
\ee
where $\varphi(z)$ is the structure function of affine Yangian
\begin{equation}\label{str-fn-AY}
\varphi(z)=\frac{\left(z+\epsilon_{1}\right)\left(z+\epsilon_{2}\right)\left(z+\epsilon_{3}\right)}{\left(z-\epsilon_{1}\right)\left(z-\epsilon_{2}\right)\left(z-\epsilon_{3}\right)}
\end{equation}
We note that ``$\sim$" implies both sides are equal up to regular terms at $z=0$ or $w=0$.  $\uppsi_{0}$ is the central element of the algebra. The structure function \eqref{str-fn-AY} is invariant under the scale transformation $\epsilon_{i} \rightarrow \gamma \epsilon_{i}, \uppsi_{0} \rightarrow \gamma^{-2} \uppsi_{0}, z \rightarrow \gamma z$. It implies that we have two independent parameters. Note also that the structure function obeys the associativity condition
\bea
    \varphi(z)\varphi(-z)=1.
\eea

It is worth noting that the Hopf algebra structure of the affine Yangian is not explicitly apparent in this presentation. The Hopf algebra structure will become more manifest once we establish its equivalence with the spherical degenerate double affine Hecke algebra in \S\ref{sec:central-ext}.

\paragraph{Degenerate limit} Actually, we can understand the affine Yangian $\mathfrak{gl}_{1}$ in (\ref{eq:AY}) as the degenerate limit of the quantum toroidal $\mathfrak{gl}_{1}$ in (\ref{eq:DIMdef}). We will not discuss this process rigorously so we refer the reader to \cite{Tsymbaliuk} for more details. To take the degenerate limit, we set $C=1$ and consider the limit $q_{c}=e^{R\epsilon_{c}}\rightarrow 1+R\epsilon_{c}$ $(R\rightarrow 0)$. The algebra homomorphism under the limit is given as follows \cite{Tsymbaliuk}:
\bea
    &\Upsilon: U_{e^{R\epsilon_{1}},e^{R\epsilon_{2}},e^{R\epsilon_{3}}}(\ddt\frakgl_{1})\to Y_{R\epsilon_{1},R\epsilon_{2},R\epsilon_{3}}(\dt\frakgl_{1})\\
   & \sfE_{j}\mapsto \sum_{n=0}^{\infty}\frac{j^{n}}{n!}\sfe_{n},\quad \sfF_{j}\mapsto \sum_{n=0}^{\infty}\frac{j^{n}}{n!}\sff_{n},\quad \sfH_{r}\mapsto -\frac{1}{R^{3}\sigma_{3}}\sum_{n=0}^{\infty}\frac{r^{n-2}}{n!}\sfk_{n}\quad (j\in\mathbb{Z},\, r\in\mathbb{Z}_{\neq 0}),\\
   & K\rightarrow e^{R^{3}\uppsi_{0}\sigma_{3}}
\eea
where $k(z)=\sum_{n=0}^{\infty}\sfk_{n}z^{-n-1}=\log \psi(z)$. In Drinfeld currents, the relation will be much clearer. For example, let us consider the relationship between the Drinfeld currents $E(z)$ and $e(u)$. We first divide the current $E(z)$ into positive and negative modes
\bea
 E^{\pm}(z)=\sum_{j=0}^{\infty}\sfE_{\pm j}z^{\mp j}.
\eea
Let us consider only the positive modes $E^{+}(z)$:
\bea
    E^{+}(z)=&\sum_{j=0}^{\infty}\frac{\sfE_{j}}{z^{j}}=\sum_{j=0}^{\infty}\sum_{n=0}^{\infty}\frac{j^{n}\sfe_{n}}{n!z^{j}}=\sum_{n=0}^{\infty}\frac{\sfe_{n}}{n!}(w\frac{d}{dw})^{n}\left.\frac{1}{1-w}\right|_{w=z^{-1}}\\
    =&\sum_{n=0}^{\infty}\frac{\sfe_{n}}{n!}(-R^{-1}\partial_{u})^{n}\frac{1}{1-e^{-Ru}}\\
    &\xrightarrow{R\rightarrow 0} \sum_{n=0}^{\infty}\frac{\sfe_{n}}{R^{n+1}u^{n+1}}=e(Ru)
\eea
where the relation between the spectral parameters are $z=w^{-1}e^{Ru}$.
Actually, we can use the rescaling automorphism of the parameters $\epsilon_{i}$ in the affine Yangian and obtain
\bea
        &\Upsilon :  U_{e^{R\epsilon_{1}},e^{R\epsilon_{2}},e^{R\epsilon_{3}}}(\ddt\frakgl_{1})\to Y_{\epsilon_{1},\epsilon_{2},\epsilon_{3}}(\dt\frakgl_{1}),\\
        &\sfE_{j}\mapsto R^{-1}\sum_{n=0}^{\infty}\frac{(Rj)^{n}}{n!}\sfe_{n},\quad \sfF_{j}\mapsto R^{-1}\sum_{n=0}^{\infty}\frac{(Rj)^{n}}{n!}\sff_{n},\quad \sfH_{r}\mapsto -\frac{1}{\sigma_{3}}\sum_{n=0}^{\infty}\frac{(Rr)^{n-2}}{n!}\sfk_{n},\\
        &K\mapsto e^{R\uppsi_{0}\sigma_{3}}.\label{eq:QTAAYmap}
\eea
Using this, the Drinfeld currents will transform as
\bea
E^{\pm}(z)\mapsto \pm R^{-2}e(u),\quad F^{\pm}(z)\mapsto \pm R^{-2}f(u),\quad K^{\pm}(z)\mapsto \psi(u).
\eea
We note the structure function will transform as
 \bea
    g(z)=\prod_{c=1}^{3}\frac{1-q_{c}z}{1-q_{c}^{-1}z}\rightarrow \prod_{c=1}^{3}\frac{u+\epsilon_{c}}{u-\epsilon_{c}}=\varphi(u),\quad z=e^{u}\rightarrow 1+u
\eea
which gives the structure function of the affine Yangian $\mathfrak{gl}_{1}$. Thus, the degenerate limit of the quantum toroidal $\mathfrak{gl}_{1}$ can be understood as two copies of affine Yangian $\mathfrak{gl}_{1}$. In other words, we can see the affine Yangian structure by restricting to the positive (or negative) mode part of the quantum toroidal algebra.

\paragraph{Relation to $\cW_{1+\infty}$-algebra} 
The affine Yangian of $\mathfrak{g l}_1$ is a one-parameter deformation of the $\cW_{1+\infty}$-algebra, and it is expected to be equivalent to the $\cW_{1+\infty}[\mu]$-algebra (the $\cW_{\infty}[\mu]$ algebra with the $\U(1)$ current) in  \S\ref{sec:Winf[mu]}.
Although the relation between the elements of $\AY$ and those of $\cW_{1+\infty}$ is highly nontrivial, we can find it for the Heisenberg and Virasoro modes \cite{Prochazka:2015deb}. If we set
\bea
& \sfJ_0=\uppsi_1, \quad \sfJ_{-1}=\sfe_0, \quad \sfJ_1=-\sff_0, \quad \sfL_{-1}=\sfe_1, \quad \sfL_1=-\sff_1, \quad \sfL_0=\frac{1}{2} \uppsi_2, \\
& \sfL_{-2}=\frac{1}{2}\left[\sfe_2, \sfe_0\right]-\frac{\uppsi_0 \sigma_3}{2}\left[\sfe_1, \sfe_0\right], \quad \sfL_2=-\frac{1}{2}\left[\sff_2, \sff_0\right]+\frac{\uppsi_0 \sigma_3}{2}\left[\sff_1, \sff_0\right],\label{eq:AYCFTcorr}
\eea
they satisfy the commutation relations of the Heisenberg and Virasoro algebra,
$$
\begin{aligned}
&[\sfJ_n, \sfJ_m]=\uppsi_0 n \delta_{n+m, 0}, \\
& [\sfL_n, \sfL_m]=(n-m) \sfL_{n+m}+\frac{c}{12}(n+1) n(n-1) \delta_{n+m, 0}, \\
& [\sfL_n, \sfJ_m]=-m \sfJ_{n+m},
\end{aligned}
$$
where the central charge is given by
\bea
c=-\sigma_2 \uppsi_0-\sigma_3^2 \uppsi_0^3=1-\left(\frac{\uppsi_0 \sigma_3}{\e_1}+1\right)\left(\frac{\uppsi_0 \sigma_3}{\e_2}+1\right)\left(\frac{\uppsi_0 \sigma_3}{\e_3}+1\right)~. \label{eq:AYcentralcharge}
\eea

To see that the affine Yangian of $\mathfrak{g l}_1$ is a one-parameter deformation of the linear $\cW_{1+\infty}$, let us consider the case of $\e_1=-\e_2=1$, $\e_3=0$. Under this condition, $\AY$ can be realized by the algebra of the differential operators as follows \cite{Prochazka:2015deb}:
$$
\sfe_i \rightarrow(-1)^i W(X^{-1} D^i), \quad \sff_i \rightarrow(-1)^{i+1} W(D^i X), \quad \uppsi_i \rightarrow(-1)^{i+1} W((D+1)^i-D^i)
$$
Thus, the affine Yangian of $\mathfrak{g l}_1$ degenerates to the $\cW_{1+\infty}$-algebra in  \S\ref{sec:Winf} if we set $\e_3=0$.

\subsection{Quiver quantum toroidal algebras}\label{sec:quiver-QTA}
Looking at the definition and the properties of the quantum toroidal $\mathfrak{gl}_{1}$ studied in \S\ref{sec:QTA}, one can see that the fundamental function determining the algebra is the structure function $g(z)$. One can also show that as long as the structure function obeys the associativity condition $g(z)g(z^{-1})=1$, the algebra is also endowed with a formal Hopf algebraic structure as in (\ref{eq:coproduct}) and (\ref{eq:counit-antipode}).
In fact, quantum toroidal $\mathfrak{gl}_{1}$ can be associated to the quiver diagram of $\bC^3$, and its structure function can be read off from the quiver diagram as
\bea
\begin{tikzcd}
{\Huge\color{red}\bullet } \arrow[out=60,in=120,loop,swap,"q_{1}"]
  \arrow[out=195,in=255,loop,swap,"q_{2}"]
  \arrow[out=-70,in=-15,loop,swap,"q_{3}"]
\end{tikzcd}
\rightsquigarrow \quad g(z)=\frac{\prod_{c=1}^{3}(q_{c}^{1/2}z-q_{c}^{-1/2})}{\prod_{c=1}^{3}(q_{c}^{-1/2}z-q_{c}^{1/2})}
\eea
where the $q$-parameters are associated with each arrow of the quiver diagram with a condition $q_{1}q_{2}q_{3}=1$ and they determine the zero and pole structures of the structure function.

Therefore, taking this vantage point, it is natural to ask if we can generalize quantum toroidal algebras to a general quiver diagram associated with some geometry \cite{ginzburg1995langlands,Rapcak:2018nsl,Rapcak:2020ueh,Li:2020rij,Galakhov:2020vyb,Galakhov:2021vbo,Galakhov:2021xum,Noshita:2021dgj,Noshita:2021ldl}.
To consider general quivers, we need to add other nodes to the diagram. We denote the quiver diagram as $Q=(Q_{0},Q_{1})$, where $Q_{0}$ is the set of nodes and $Q_{1}$ is the set of arrows connecting the nodes\footnote{The quivers here originate from 4d $\mathcal{N}=1$ quiver gauge theories which are oriented quivers.}. For each arrow of the quiver, we assign a $q$-deformation parameter $q_{I},\,(I\in Q_{1})$. The information of these $q$-deformation parameters is determined from a polynomial $W$ associated with the geometry which is called the superpotential. This superpotential gives conditions such as the condition $q_{1}q_{2}q_{3}=1$ for the quantum toroidal $\mathfrak{gl}_{1}$ case.

Given a quiver diagram $Q=(Q_{0},Q_{1})$ and a superpotential $W$, we can associate an algebra as follows:
\bea\label{quiver-QTA}
E_{i}(z)=\sum_{k\in\mathbb{Z}}\mathsf{E}_{i,k}z^{-k},\quad F_{i}(z)=&\sum_{k\in\mathbb{Z}}\mathsf{F}_{i,k}z^{-k},\quad K_{i}^{\pm}(z)=K_{i}^{\pm}\exp\left(\pm\sum_{r=1}^{\infty}\mathsf{H}_{i,\pm r}z^{\mp r}\right),\\
K_{i}^{\pm}(z)K_{j}^{\pm}(w)=&K_{j}^{\pm}(w)K_{i}^{\pm}(z),\\
    K_{i}^{-}(z)K_{j}^{+}(w)=&\frac{\varphi^{i\Leftarrow j}(z,Cw)}{\varphi^{i\Leftarrow j}(Cz,w)}K_{j}^{+}(w)K_{i}^{-}(z),\\
    K_{i}^{\pm}(C^{\frac{1\mp1}{2}}z)E_{j}(w)=&\varphi^{i\Leftarrow j}(z,w)E_{j}(w)K_{i}^{\pm}(C^{\frac{1\mp1}{2}}z),\\
    K_{i}^{\pm}(C^{\frac{1\pm1}{2}}z)F_{j}(w)=&\varphi^{i\Leftarrow j}(z,w)^{-1}F_{j}(w)K_{i}^{\pm}(C^{\frac{1\pm1}{2}}z),\\
    [E_{i}(z),F_{j}(w)]=\delta_{i,j}&\left(\delta\left(\frac{Cw}{z}\right)K_{i}^{+}(z)-\delta\left(\frac{Cz}{w}\right)K_{i}^{-}(w)\right),\\
    E_{i}(z)E_{j}(w)=&\varphi^{i\Leftarrow j}(z,w)E_{j}(w)E_{i}(z),\\
    F_{i}(z)F_{j}(w)=&\varphi^{i\Leftarrow j}(z,w)^{-1}F_{j}(w)F_{i}(z)
\eea
where $C$ is the central element and
\bea
\varphi^{j\Leftarrow i}(z,w)=(-1)^{\chi_{i\rightarrow j}}\frac{\prod_{I\in\{j\rightarrow i\}}(q_{I}^{1/2}z-q_{I}^{-1/2}w)}{\prod_{I\in\{i\rightarrow j\}}(q_{I}^{-1/2}z-q_{I}^{1/2}w)}.
\eea
The sign factor $\chi_{i\rightarrow j}$ is determined by imposing the associativity condition
\bea
\varphi^{j\Rightarrow i}(z,w)\varphi^{i\Rightarrow j}(w,z)=1.
\eea
Here we will not go into the details of the sign factor. It is worth noting that even though the $q$-deformation parameters are arbitrary, as long as the associativity condition is met, the algebra is endowed with a formal Hopf algebraic structure, as outlined as follows.  However, whether it has a meaningful representation depends on these parameters. In this note, we will not provide the method to determine the $q$-parameters; rather, they will be imposed. For further details on the derivation of these parameters, we refer to \cite{Li:2020rij,Galakhov:2021vbo,Noshita:2021ldl} for the derivations (see also \cite{Yamazaki:2022cdg} for a nice summary).

We can further impose the condition $K_{i}^{+}K_{i}^{-}=K_{i}^{-}K_{i}^{+}$ and then we have another central element
\begin{align}
    \kappa=\prod_{i\in Q_{0}}K_{i}^{-}.
\end{align}

Namely, we introduce Drinfeld currents $E_{i}(z),F_{i}(z),K^{\pm}_{i}(z)$ for each quiver node, and the structure functions $\varphi^{i\Leftarrow j}(z,w)$ are determined from the arrows and the associate $q$-parameters of the quiver.

The algebra above actually has a formal Hopf algebraic structure:
\begin{align}
\begin{split}
    &\Delta E_{i}(z)=E_{i}(z)\otimes 1+K_{i}^{-}(C_{1}z)\otimes E_{i}(C_{1}z),\\
&\Delta F_{i}(z)=F_{i}(C_{2}z)\otimes K_{i}^{+}(C_{2}z)+1\otimes F_{i}(z),\\
&\Delta K_{i}^{+}(z)=K_{i}^{+}(z)\otimes K_{i}^{+}(C_{1}^{-1}z),\\
&\Delta K_{i}^{-}(z)=K_{i}^{-}(C_{2}^{-1}z)\otimes K_{i}^{-}(z),\\
&\Delta C=C\otimes C
\end{split}
\begin{split}
&\varepsilon(E_{i}(z))=\varepsilon(F_{i}(z))=0,\\
&\varepsilon(K_{i}^{\pm}(z))=\varepsilon(C)=1,\\
&S(E_{i}(z))=-(K_{i}^{-}(z))^{-1}E_{i}(C^{-1}z),\\
&S(F_{i}(z))=-F_{i}(C^{-1}z)(K_{i}^{+}(z))^{-1},\\
&S(K_{i}^{\pm}(z))=(K_{i}^{\pm}(Cz))^{-1},\\
&S(C)=C^{-1}.
\end{split}
\end{align}

Additionally, we need to add extra nontrivial relations called Serre relations. See \cite{Li:2020rij,Galakhov:2021vbo,Galakhov:2020vyb} for the discussion on the explicit form of the Serre relations. See also \cite{Galakhov:2022uyu} for discussions on the Miki automorphism for these algebras.

It is important to note that not all combinations of quivers and $q$-parameters result in well-defined representations or useful physical applications. While investigations of the various quivers and their physical and mathematical implications have been conducted in recent years, a comprehensive study is still lacking. In what follows, we will present a few examples of algebras that have been previously studied in the literature.

\paragraph{Affine $A$-type quiver \cite{ginzburg1995langlands,Feigin2012gln}}
The quantum toroidal algebra associated with the $\hat{A}_{n-1}$ quiver is referred to as the quantum toroidal $\mathfrak{gl}_{n}$. This quantum toroidal algebra is associated with the geometry $\mathbb{C}^{2}/\mathbb{Z}_{n}\times \mathbb{C}$. Let us consider the generic case when $n\geq 3$.  The quiver diagram takes the form of a cyclic quiver, with $n$ nodes and arrows connecting adjacent nodes. The structure functions are determined as
\bea\label{A-quiver}
\adjustbox{valign=c}{\includegraphics[width=7cm]{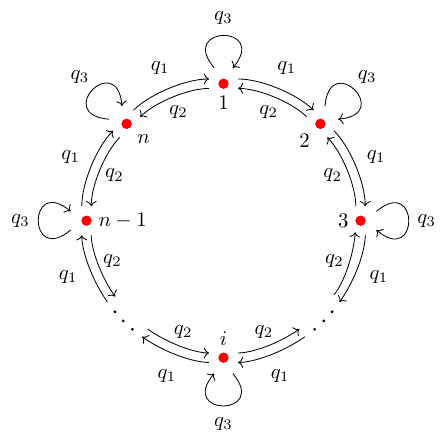}}\quad\rightsquigarrow\quad \begin{dcases}\varphi^{i+1\Leftarrow i}(z,w)=\frac{q_{2}^{1/2}z-q_{2}^{-1/2}w}{q_{1}^{-1/2}z-q_{1}^{1/2}w},\\
\varphi^{i\Leftarrow i+1}(z,w)=\frac{q_{1}^{1/2}z-q_{1}^{-1/2}w}{q_{2}^{-1/2}z-q_{2}^{1/2}w},\\
\varphi^{i\Leftarrow i}(z,w)=\frac{q_{3}^{1/2}z-q_{3}^{-1/2}w}{q_{3}^{-1/2}z-q_{3}^{1/2}w},\\
\varphi^{i\Leftarrow j}(z,w)=1,\quad (|i-j|\neq 0,1)
\end{dcases}
\eea
for $i,j\in\mathbb{Z}_{n}$, where the parameters obey the following condition
\bea
q_{1}q_{2}q_{3}=1.
\eea
After specializing the $q$-deformation parameters as $q_{1}=q_{2}=q^{-1},\,q_{3}=q^{2}$, we can see the relation with the affine Cartan matrix $\dt{\mathtt{a}}_{ij}$:
\bea\label{bond}
\varphi^{j\Leftarrow i}(z,w)=\frac{q^{\dt{\mathtt{a}}_{ij}}z-w}{z-q^{\dt{\mathtt{a}}_{ij}}w},\quad \dt{\mathtt{a}}_{ij}=2\delta_{ij}-\delta_{i,j+1}-\delta_{i,j-1}\,(i,j\in\mathbb{Z}_{n}).
\eea
Note that after this specialization, the structure functions will be symmetric in their indices due to the symmetric Cartan matrix.

\paragraph{Affine DE-type quiver and other affine Lie algebras \cite{ginzburg1995langlands}}
Based on the correspondence with the Cartan matrix of the affine $A$ quiver case, the structure functions can be generalized to the affine $DE$-type case:
\bea
\varphi^{j\Leftarrow i}(z,w)=\frac{q^{\dt{\mathtt{a}}_{ij}}z-w}{z-q^{\dt{\mathtt{a}}_{ij}}w}
\eea
where $\dt{\mathtt{a}}_{ij}$ is the Cartan matrix of $\hat{D}_{n}$ and $\hat{E}_{6,7,8}$ respectively. For example, the quiver of $\hat{D}_{n}$ is
\bea\label{D-quiver}
\includegraphics{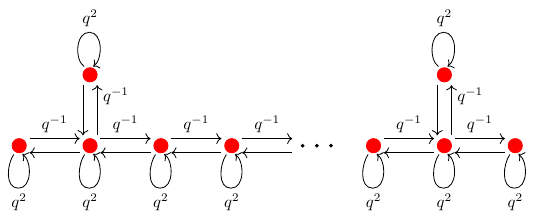}
\eea
and it comprises $(n+1)$ nodes, thus requiring $(n+1)$ sets of Drinfeld currents with the aforementioned structure function. Recall that the Cartan matrix can be derived directly from the quiver as
\bea
\dt{\mathtt{a}}_{ij}=\begin{dcases}
2,\quad i=j\\
-1,\quad \exists\, e_{j\rightarrow i}\in Q_{1}\\
0,\quad \text{otherwise}.
\end{dcases}
\eea
The affine $A$-type and quantum toroidal $\mathfrak{gl}_{1}$ have two independent deformation parameters, while the algebra for $DE$-type is known to admit only one deformation parameter. 

Quivers of $ADE$ type are called simply-laced Dynkin diagrams. For these cases, the Cartan matrix is symmetric by itself. However, for other Lie algebras, the Cartan matrix is not symmetric. Nonetheless, it can still be symmetrized, thereby allowing the definition of a quantum toroidal algebra as follows. For a semisimple Lie algebra $\mathfrak{g}$, let $\dt{\mathtt{a}}_{ij}, i,j=0,\ldots,n$ be the corresponding affine Cartan matrix. There exist mutually prime integers $d_{i},\,i=0,\ldots,n$ such that the matrix $b_{ij}=d_{i}\dt{\mathtt{a}}_{ij}$ is symmetric. Using this symmetrized matrix, the structure function is defined as
\bea
\varphi^{j\Leftarrow i}(z,w)=\frac{q^{b_{ij}}z-w}{z-q^{b_{ij}}w}.
\eea

\bigskip

The algebra of type $ADE$ is associated with the geometry $\mathbb{C}^{2}/\Gamma\times \mathbb{C}$, where  $\Gamma$ is the dihedral group for $\hat{D}_{n}$, tetrahedral group for $\hat{E}_{6}$, octahedral group for $\hat{E}_{7}$, and icosahedral group for $\hat{E}_{8}$. Combined with the affine $A$-type, $\Gamma$ is the discrete subgroup of SU(2). Interestingly, the minimal resolution of an ALE space $\mathbb{C}^{2}/\Gamma$ can be described as a Nakajima quiver variety \cite{kronheimer1989construction}, which is closely related to the one appearing here (e.g., \eqref{A-quiver} and \eqref{D-quiver}), and the intersection matrix of its second homology group is (the minus of) the corresponding Cartan matrix $-\mathtt{a}_{ij}$. Also, the affine Cartan matrix $\dt{\mathtt{a}}_{ij}$ also appears in the irreducible decomposition of a tensor product of representations of $\G$ \cite{mckay1983graphs}
$$
Q \otimes \rho_j=\bigoplus_i a_{ij} \rho_i,\qquad \dt{\mathtt{a}}_{ij}=2\delta_{ij}-a_{ij}~,
$$
where $\left\{\rho_i\right\}_{i=0}^r$ be the set of (isomorphism classes of) irreducible representations of $\G$ and $Q$ is the 2-dimensional representation given by the inclusion $\G \subset \SU(2)$.
Moreover, the Grothendieck ring of vector bundles over (the minimal resolution of) $\mathbb{C}^{2}/\Gamma$ is isomorphic to the representation ring of the corresponding finite group $\Gamma$. This is known as the \emph{McKay correspondence}. (See \cite[\S4]{nakajima1999lectures} and \cite[\S4]{aspinwall2009dirichlet} for details and references therein.)
\begin{table}[ht]
   \centering
    \begin{tabular}{c|c}
        quiver & discrete subgroup of SU(2)  \\\hline\hline
        $\hat{A}_{n}$ & cyclic\\
        $\hat{D}_{n}$ & dihedral\\
        $\hat{E}_{6}$ & tetrahedral\\
        $\hat{E}_{7}$ & octahedral\\
        $\hat{E}_{8}$ & icosahedral
    \end{tabular}
    \caption{McKay correspondence}
    \label{tab:McKaycorrespondence}
\end{table}

The moduli spaces of instantons on the ALE space are constructed as Hyper-K\"ahler quotients associated with quivers of type $ADE$, which are referred to as Nakajima quiver varieties. This construction is detailed in \cite{kronheimer1990yang} and \cite{Nakajima:1994nid}. By taking the union of all the instanton moduli spaces, denoted $\cM_{\textrm{inst}}=\cup_k \cM_k$, one can introduce an algebra structure on the equivariant K-homology groups of holomorphic Lagrangian submanifolds in $\cM_{\textrm{inst}}\times \cM_{\textrm{inst}}$ (branes of type $(B,A,A)$ in the language of \cite[below Fig.6]{Gukov:2006jk}) based on the technique of ``Hecke correspondence'' or ``convolution algebra''. This algebra is isomorphic to the quantum toroidal algebra of type $ADE$ \cite{nakajima2001quiver}. 
A similar geometric construction of $\QTA$ and $\AY$ is presented in \cite{feigin2011equivariant,Schiffmann:2009aa}. (See also \S\ref{sec:final}.) One of the advantages of this geometric approach is that it allows for the natural construction of representations of these algebras.
It is believed that this construction provides the mathematical framework for the algebra of D0-$\overline{\text{D0}}$ branes on D4-branes over the ALE space with the $\Omega$-background. However, to the best of the authors' knowledge, there remains a gap between the physical picture and the mathematical formulation, and further exploration is needed to bridge this gap. A deeper understanding of the geometric constructions in \cite{Nakajima:1994nid,nakajima1997heisenberg,nakajima2001quiver,feigin2011equivariant,Schiffmann:2009aa} from a physical perspective would be particularly valuable in this regard.

\paragraph{Orbifold quivers \cite{Li:2020rij,Noshita:2021ldl,Bourgine:2019phm}}
To generalize the discussion to other quivers, one tractable class is the orbifold quivers, where the corresponding geometry is $\mathbb{C}^{3}/\Gamma$, with $\Gamma$ being the discrete subgroup of SU(3). For simplicity, we will focus on Abelian orbifold quivers.\footnote{Relatively less is known about non-Abelian orbifold quivers, especially about physical applications, so the authors would encourage further research in this area. While preparing this article, three papers related \cite{Li:2023zub,Bao:2023ece,Bao:2023kkh} appeared in arXiv. } For example, in the case of $\mathbb{C}^{3}/\mathbb{Z}_{n}$, the orbifold action is given by
\begin{align}
    (z_{1},z_{2},z_{3})\mapsto (\omega^{\nu_{1}}z_{1},\omega^{\nu_{2}}z_{2},\omega^{\nu_{3}}z_{3}),\quad \nu_{1}+\nu_{2}+\nu_{3}\equiv0\,\,(\bmod n).\label{eq:orbifoldaction}
\end{align}
The quiver diagram comprises $n$ nodes in the quiver $Q_{0}={1,\ldots,n}$ (therefore indices are understood mod $n$). The quiver and structure functions will take the form of
\bea
\adjustbox{valign=c}{
\includegraphics{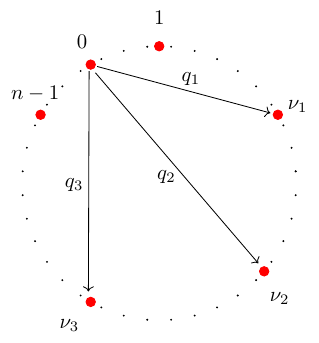}} \rightsquigarrow\quad \varphi^{j\Leftarrow i}(z,w)=(-1)^{\chi_{i\rightarrow j}}\frac{\prod_{k=1}^{3}(q_{k}^{1/2}z-q_{k}^{-1/2}w)^{\delta_{j,i-\nu_{k}}}}{\prod_{k=1}^{3}(q_{k}^{-1/2}z-q_{k}^{1/2}w)^{\delta_{j,i+\nu_{k}}}}.
\eea
Namely, $n$ nodes arise from the $\mathbb{Z}_{n}$ action and  three arrows emanate from each node $i$ with parameters $q_{1},q_{2},q_{3}$ to the nodes $i+\nu_{1},i+\nu_{2},i+\nu_{3}$ respectively. Note that after setting $\nu_{3}=0$ and the condition $\nu_{1}+\nu_{2}=0$, the orbifold action is embedded in SU(2) and thus will be reduced to the affine $A$ quiver case. We also can generalize this discussion to orbifolds
$\mathbb{C}^{3}/(\mathbb{Z}_{n}\times \mathbb{Z}_{m})$ (for example see \cite{Noshita:2021ldl,Noshita:2021dgj}). 

It is natural to interpret that the parameters $q_{1},q_{2},q_{3}$ have $\mathbb{Z}_{n}$-charges $\nu_{1},\nu_{2},\nu_{3}$ for these orbifold quivers and the structure functions can be understood as the $\mathbb{Z}_{n}$-invariant parts of the structure function of the quantum toroidal $\mathfrak{gl}_{1}$. For instance, for $\varphi^{j\Leftrightarrow i}(z,w)$, the variables $z$ and $w$ can be assigned charges $j,i\in\mathbb{Z}_{n}$ respectively. Therefore, for the numerator of the structure function
\begin{align}
q_{k}^{1/2}z-q_{k}^{-1/2}w
\end{align}
to be $\mathbb{Z}_{n}$-invariant, it is required that $i-\nu_{k}\equiv j\,\,(\bmod n)$, which leads to
\begin{align}
(q_{k}^{1/2}z-q_{k}^{-1/2}w)^{\delta_{j,i-\nu_{k}}}.
\end{align}
The denominator part can be understood similarly. This property also appears in the vertical representations of these algebras. For more information, see \cite{Noshita:2021dgj,Noshita:2021ldl,Bourgine:2019phm,Feigin2012gln}. 

\paragraph{Toric Calabi-Yau three-fold quivers }
Recently, another tractable class of quivers, arising from toric Calabi-Yau three-folds, has been studied in various works, including \cite{Rapcak:2018nsl,Rapcak:2020ueh,Rapcak:2021hdh,Li:2020rij,Galakhov:2020vyb,Diaconescu:2020rnt,Galakhov:2021vbo,Galakhov:2021xum,Noshita:2021dgj,Noshita:2021ldl}. The geometric structure of this class is encapsulated in a \emph{toric} diagram, which is explained in \S\ref{sec:TV-intertwiners}. The ``inverse algorithm'' is an approach for deriving the quiver and superpotential from the toric diagram, while the ``forward algorithm'' allows one to derive the toric diagram from the quiver and superpotential. The ``brane tilings'' technique \cite{Ooguri:2009ijd,Franco:2005rj,Franco:2005sm} is a convenient way to perform both algorithms (see \cite{Yamazaki:2008bt} for a review).  Using these brane tilings techniques, one can obtain the quiver and superpotential which define the quantum toroidal algebra (see Figure \ref{fig:qqta}). Moreover, an interesting fact is that these brane tilings actually directly give representations with $C=1$ of the associate quantum algebra \cite{Li:2020rij,Galakhov:2021vbo,Galakhov:2021xum,Noshita:2021ldl,Noshita:2021dgj}. Note that the $\mathbb{C}^{3}$ geometry and the Abelian orbifold quivers are also included in this class, and toric diagrams can represent the geometry.

\begin{figure}[t]
    \centering
    \includegraphics[width=17cm]{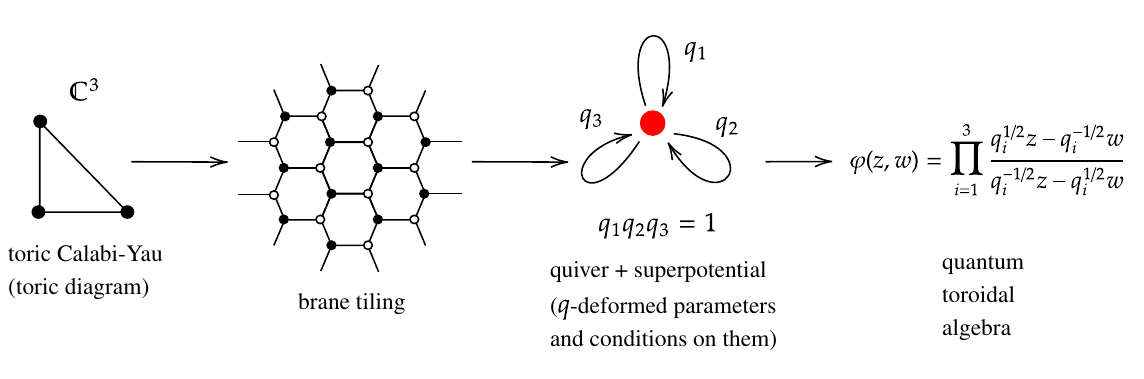}
    \caption{Flow chart of defining quantum toroidal algebras associated with toric quivers. The example drawn is the algebra associated with $\mathbb{C}^{3}$, which is simply the quantum toroidal $\mathfrak{gl}_{1}$.}
    \label{fig:qqta}
\end{figure}

Let us provide an example of a new case where the geometry is the conifold. The associated quiver has two nodes with the following parameters
\bea
\adjustbox{valign=c}{
\includegraphics{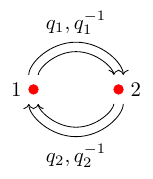}} \rightsquigarrow\quad\begin{dcases}
\varphi^{1\Leftarrow1}(z,w)=&\varphi^{2\Leftarrow 2}(z,w)=1,\\
\varphi^{2\Leftarrow 1}(z,w)=&\frac{(q_{2}^{1/2}z-q_{2}^{-1/2}w)(q_{2}^{-1/2}z-q_{2}^{1/2}w)}{(q_{1}^{1/2}z-q_{1}^{-1/2}w)(q_{1}^{-1/2}z-q_{1}^{1/2}w)},\\
\varphi^{1\Leftarrow 2}(z,w)=&\frac{(q_{1}^{1/2}z-q_{1}^{-1/2}w)(q_{1}^{-1/2}z-q_{1}^{1/2}w)}{(q_{2}^{1/2}z-q_{2}^{-1/2}w)(q_{2}^{-1/2}z-q_{2}^{1/2}w)}.
\end{dcases}
\eea
The quiver diagram is also known as the quiver of the affine superalgebra $\widehat{\mathfrak{gl}}_{1|1}$, and as a result, the corresponding algebra is referred to as quantum toroidal  $\mathfrak{gl}_{1|1}$. It is straightforward to generalize it to quantum toroidal $\mathfrak{gl}_{m|n}$. 

\bigskip 

All the algebras here are formulated as ``algebras of BPS states'' \cite{Harvey:1996gc}. A general idea in physics proposes that a particular subspace of BPS states forms an algebra, which subsequently acts on a larger (or even the entire) space of BPS states, thereby providing a representation. As a result, the character of this representation becomes the generating function for enumerative invariants, such as Gopakumar-Vafa, Donaldson-Thomas, and Pandharipande-Thomas invariants, associated with the corresponding Calabi-Yau manifold. Although this idea serves as the general picture, only a handful of cases have been explicitly developed to exemplify this picture \cite{Bao:2022oyn,Li:2023zub}. (As we will see in \S\ref{sec:MacMahonrep}, the MacMahon representation is the simplest example of such kind.) Hence, it is a promising avenue for future research to link the representation theory of quantum toroidal algebras to the enumerative geometry of Calabi-Yau manifolds.

The cohomological Hall algebra associated with the quiver data $(Q,W)$ constructed in \cite{kontsevich2010cohomological} is expected to constitute the mathematical foundation for the algebras of BPS states. Utilizing the techniques of the cohomological Hall algebra, corner vertex operator algebras (VOAs), quantum algebras and shifted Yangians are geometrically constructed in \cite{sala2018cohomological,yang2018cohomological,schiffmann2020cohomological,Rapcak:2020ueh,Diaconescu:2020rnt,Davison:2020xkr,kapranov2019cohomological,Porta:2019pzk,Davison:2022cxv,Davison:2022rwy,davison2023bps}. Furthermore, the representation theory for these algebras can be developed through the "correspondence," which in turn establishes a connection to enumerative invariants in the references therein.

\section{Double affine Hecke algebras and degenerations}\label{sec:DAHA-deg}
Slightly prior to the introduction of the quantum toroidal algebra \cite{ginzburg1995langlands}, Cherednik had introduced the double affine generalization of Hecke algebras. Intriguingly, the quantum toroidal algebra and the double affine Hecke algebra are related by the Schur-Weyl duality \cite{varagnolo1996schur}.

In this section, we will elucidate double affine Hecke algebra (DAHA). Of particular interest in DAHA is its distinguished subalgebra, denoted as the \emph{spherical subalgebra}, which is symmetrized by the Hecke algebra (or gauge invariant subalgebra). The seminal work \cite{schiffmann2011elliptic} by Schiffmann and Vasserot proved that a central extension of the large rank limit of the spherical DAHA is equivalent to the quantum toroidal $\frakgl_1$.

By its construction, the spherical DAHA receives $\SL(2, \mathbb{Z})$ action as an automorphism. Therefore, the realization of \QTA~as the stable limit of spherical DAHA is one way to see that \QTA~enjoys the $\SL(2, \mathbb{Z})$ symmetry. The $\SL(2, \mathbb{Z})$ symmetry plays an important role in the representation theory of \QTA~, and its various connections to physics. 

In this section, we also delve into the degeneration limit of these relationships. Namely, we will explore another groundbreaking work \cite{Schiffmann:2012aa} by Schiffmann and Vasserot, which reveals the relationships among the spherical degenerate DAHA, the affine Yangian $\frakgl_1$ and $\cW$-algebras. Establishing these relationships ultimately led to the proof of the AGT relation. 

\subsection{DAHA}\label{sec:DAHA}

In this section, we will introduce double affine Hecke algebra (DAHA)  of $\GL(N,\bC)$ type. DAHA was introduced by Cherednik to study Macdonald polynomials in the pioneering papers \cite{cherednik1992double,cherednik1995macdonald,cherednik1995double,cherednik1995non,cherednik1996intertwining}, and the basic references of DAHA are  \cite{kirillov1997lectures,macdonald2003affine,cherednik2005double,haiman2006cherednik,Gukov:2022gei}.
Let us start with a very familiar one. The Weyl group of $\GL(N,\bC)$ is the symmetric group $W\cong \frakS_N$, which is generated by transpositions $s_{i}:=s_{i,i+1}:=(i,i+1)\in \frakS_N$. They satisfy the braid relations
$$
s_{i} s_{i+1} s_{i} =s_{i+1} s_{i} s_{i+1} ~, \qquad s_{i} s_{j}=s_{j} s_{i} \quad (|i-j|>1)
$$
and the quadratic relations
$$
s_{i}^{2}=1~.
$$
The Hecke algebra is obtained from the group algebra of the Weyl group by deforming the quadratic relation
$$
(T_i-t^{\frac{1}{2}})(T_i+t^{-\frac{1}{2}})=0~
$$
while keeping the braid relations as they are. To doubly affinize the Hecke algebra, we will introduce the polynomial variables $X_i,$ $Y_i$ of two types ($i=1,\ldots, N$).

\begin{definition}
DAHA $\HH_N$ of $\GL(N,\bC)$ type consists of three types of generators
$T_i,$ ($i=1,\ldots, N-1$), $ X_i^{\pm 1},Y_i^{\pm 1}$ ($i=1,\ldots, N$)
with the following relations
\bea\label{DAHA}
\textrm{Braid relations} \quad& T_i T_{i+1} T_{i} =T_{i+1}T_i T_{i+1},\qquad
T_i T_j=T_j T_i \quad (|i-j|>1),\cr
\textrm{Hecke (quadratic) relations}\quad&(T_i-t^{\frac{1}{2}})(T_i+t^{-\frac{1}{2}})=0\cr
\textrm{Laurent relations} \quad& X_i X_j=X_j X_i ,\quad Y_i Y_j=Y_j Y_i\\
\textrm{Action relations}\quad& T_i X_i T_i=X_{i+1},\quad T_i^{-1} Y_i T_i^{-1}=Y_{i+1}\cr
& T_i X_j=X_j T_i,\quad T_i Y_j=Y_j T_i,\quad (j\neq i,i+1)\cr
\textrm{Other relations}\quad&Y_{2}^{-1} X_{1} Y_{2} X_{1}^{-1}=T_{1}^{2},\quad \widetilde{Y} X_{j}=q X_{j} \widetilde{Y} \quad \textrm{and } \quad \widetilde{X} Y_{j}=q^{-1} Y_{j} \widetilde{X}\nonumber
\eea
Here $\widetilde{X}=\prod_{i=1}^{N} X_{i}$ and $\widetilde{Y}=\prod_{i=1}^{N} Y_{i}$. This algebra is endowed with two deformation parameters $t,q\in \bC^\times$. The labelings are assumed to be modulo $N$.
\end{definition}

If we drop the Hecke relations and set $q=1$ in \eqref{DAHA}, then it is known as \emph{elliptic braid group} or \emph{double affine braid group}. This admits a topological interpretation.
 Let $E$ be a 2-torus. Let us consider the $N$-fold product $E^N$, and we define the set of regular elements by $(E^N)^{\textrm{reg}} := \{(x_1, \ldots, x_N) \in E^N : x_i\neq  x_j \quad \textrm{if}\quad  i\neq j\}$. Then, we consider the symmetric product $C := (E^N)^{\textrm{reg}}/S_N$ where we divide  $(E^N)^{\textrm{reg}}$ by the action of the symmetric group. Consequently, the orbifold fundamental group $\pi_1(C)$ is what we refer to as the elliptic braid group \cite{birman1969braid}.
The generator $X_i$ corresponds to the ``meridian'' loop going around the $i$-th torus $E^{(i)}$ while the generator $Y_{i}$ corresponds to the ``longitude'' loop going around the $i$-th torus $E^{(i)}$. The generator $T_i$ corresponds to the transposition of the $i$-th and ($i+1$)-th points. In this way, the origins of the "toroidal" characteristics inherent in DAHA become more transparent. 

\begin{figure}[ht]
    \centering
    \includegraphics{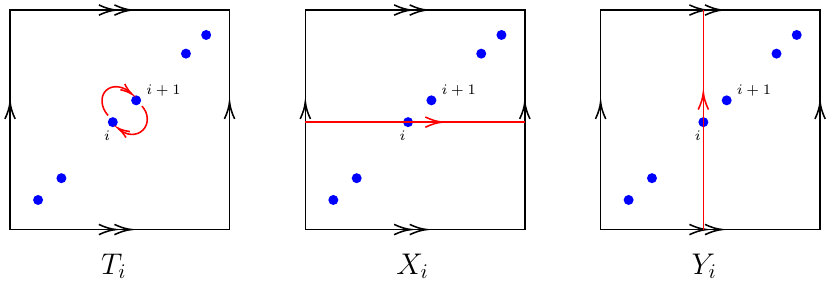}
    \caption{Generators of the elliptic braid group, which is the orbifold fundamental group of $(E^N)^{\textrm{reg}}/S_N$. The $i$-th blue dot represents a point in the $i$-th torus $E^{(i)}$. The braid group generator $T_i$ can be understood as the exchange of positions between the $i$-th and $(i+1)$-th  blue dots. Additionally, the generators $X_i$ and $Y_{i}$ represent the meridian (horizontal) and longitude (vertical) cycles of the $i$-th torus $E^{(i)}$.} 
    \label{fig:EBGrelation1}
\end{figure}

Therefore, DAHA can be understood as the $q$-deformation of the group algebra of the elliptic braid group with the Hecke relations. From this construction, it is natural to expect that it is endowed with the braid group on three strands as the automorphism group
\be B_3= \langle\tau_\pm\text{ : } \tau_+\tau^{-1}_-\tau_+ = \tau^{-1}_-\tau_+\tau^{-1}_-\rangle~.\ee
The braid group $B_3$ on three strands is an extension of $\SL(2,\bZ)$ by $\bZ$
\be\label{SES} 0\to \bZ\to B_3\to \SL(2,\bZ)\to1~,\ee
where the kernel of the projection $B_3\to \SL(2,\bZ)$ is $\bZ$ generated by $(\tau_{+} \tau_{-}^{-1} \tau_{+})^{4}$.
The action on the generators can be explicitly written as
\begin{equation}\label{modular-T}
\tau_-: \begin{cases}
      X_i\mapsto X_i Y_i  (T_{i-1}\cdots T_1)(T_{1}\cdots T_{i-1})\\
      T_i\mapsto T_i\\
      Y_i\mapsto Y_i
   \end{cases}\  \tau_+: \begin{cases}
      X_i\mapsto X_i\\
      T_i\mapsto T_i\\
      Y_i\mapsto Y_iX_i(T_{i-1}^{-1}\cdots T_1^{-1})(T_{1}^{-1}\cdots T_{i-1}^{-1})
   \end{cases}
\end{equation}
Of particular importance is the generator
\begin{equation}\label{modular-S}
\sigma:=  \tau_+ \tau_-^{-1}\tau_+= \tau_-^{-1}\tau_+\tau_-^{-1}: \begin{cases}
      X_i\mapsto  Y_i^{-1} \\
      T_i\mapsto T_i\\
      Y_i\mapsto  Y_i X_i Y_i^{-1}
   \end{cases}
\end{equation}
In fact, the generator $\tau_{-}$ of $B_3$ action on DAHA is obtained from the anti-involution $\epsilon$ of the DAHA acting on generators and parameters as:
$$\epsilon:\left(X_{i}, Y_{i}, T_{i}, q, t\right) \mapsto (Y_{i}, X_{i}, T_{i}^{-1}, q^{-1}, t^{-1})$$
and such that
$$
\tau_{-}=\epsilon \tau_{+} \epsilon
$$

The definition of $\HH_N$ in \eqref{DAHA} is called the \emph{Bernstein presentation}. There is another representation called the \emph{Coxeter presentation} by introducing an element
\begin{equation}
\pi:=T_{1}^{-1} \cdots T_{i-1}^{-1} Y_{i}^{-1} T_{i} \cdots T_{N-1}~.
\end{equation}
Thanks to the action relations, the element $\pi$ is well-defined; namely, it is independent of the choice of $i$. (See \cite[\S1.4.3]{cherednik2005double} for more details.)

\begin{definition}
  The Coxeter presentation of $\HH_N$ is given by the generators $$T_{1}, \ldots, T_{N-1}, X_{1}^{\pm 1}, \ldots, X_{N}^{\pm 1}, \pi^{\pm 1}$$ with relations:
\bea
  &\textrm{the braid and quadratic relations for } T_{1}, \ldots, T_{N-1}\cr
& X_{i} X_{j}=X_{j} X_{i}\quad (1 \leq i, j \leq N), \quad T_{i} X_{i} T_{i}=X_{i+1}\quad (1 \leq i<N) \cr
&\pi X_{i}=X_{i+1} \pi\quad (i=1, \ldots N-1) \quad \textrm{and} \quad \pi^{N} X_{i}=q^{-1} X_{i} \pi^{N}\quad (i=1, \ldots, N)\\
&\pi T_{i}=T_{i+1} \pi\quad (i=1, \ldots N-2) \quad  \textrm{and}  \quad \pi^{N} T_{i}=T_{i} \pi^{N}\quad (i=1, \ldots, N-1).\nonumber
\eea
\end{definition}

There is an important subalgebra called the \emph{spherical subalgebra} of DAHA. To introduce it
we define the idempotent
\be\boldsymbol{e}=\left(\prod_{i=1}^N\frac{1-t}{1-t^{i}}\right) \sum_{w\in W} t^{\frac{l(w)}{2}} T_{w}\ee
where the length $l(w)$ of $w\in W$ is the minimal length of a product $w=s_{i_1} \cdots s_{i_{l(w)}}$ of simple reflections, referred to as a reduced expression, and we define the corresponding generator $T_w:=T_{i_1} \cdots T_{i_{l(w)}}$ in the Hecke algebra. Due to the braid relations, it is well-defined. It is straightforward to check from the Hecke relations that it is the idempotent $\boldsymbol{e}^2=\boldsymbol{e}$.
Then, the spherical DAHA is defined by
\be\SH_N:=\boldsymbol{e} \HH_N \boldsymbol{e}~.\ee
There is another idempotent element
\be\widetilde{\boldsymbol{e}}=\left(\prod_{i=1}^N\frac{1-t^{-1}}{1-t^{-i}}\right) \sum_{w\in W} (-t)^{\frac{-l(w)}{2}} T_{w}\ee
in $\HH_N$, which defines another spherical subalgebra in the same way. Nonetheless, the two spherical subalgebras are isomorphic.

The action of $B_3$ preserves $\SH_N$ and moreover it factors through $\SL(2,\bZ)$. Namely, the action of
$(\tau_{+} \tau_{-}^{-1} \tau_{+})^{4}$ is trivial on $\SH_N$ \cite{cherednik2005double}.
Therefore, the elements in \eqref{modular-T} and \eqref{modular-S} correspond to the following matrices of $\SL(2,\bZ)$ on $\SH_N$:
\begin{equation}
\begin{pmatrix}
1 & 1\\
0 & 1\end{pmatrix}\leftrightarrow \tau_+~, \qquad  \begin{pmatrix}
1 & 0\\
1 & 1\end{pmatrix}\leftrightarrow \tau_-~,\qquad   \begin{pmatrix}
  0 & 1\\
  -1 & 0\end{pmatrix}\leftrightarrow \sigma~.
\end{equation}

Indeed, the sandwich by the idempotent $\boldsymbol{e}$ leads to the symmetrization. Hence, $\SH_N$ is generated by
\be\label{SH-e}
\boldsymbol{e} e_r(X) \boldsymbol{e}~,\qquad \boldsymbol{e} e_r(Y) \boldsymbol{e}~,\qquad(r=0,1, \ldots, N)~
\ee
where $e_{r}$ is the elementary symmetric function of degree $r$.
In the physical setting, the spherical DAHA $\SH_N$ is realized as an algebra of loop operators in 4d U($N$) $\cN=2^*$ theory on $S^1\times \bR^3$ with $\Omega$-background where $q$ corresponds to the $\Omega$-deformation parameter, and $t$ corresponds to the mass of the adjoint hypermultiplet. Specifically, $\boldsymbol{e} e_r(X) \boldsymbol{e}$ correspond to the Wilson loops on $S^1$ while $\boldsymbol{e} e_r(Y) \boldsymbol{e}$  to the 't Hooft loops. This algebra can be understood as the deformation quantization of the coordinate ring of the moduli space of flat $\GL(N,\bC)$-connections on a once-punctured torus \cite{Oblomkov:aa}, which is the Coulomb branch of the corresponding theory. For a comprehensive exploration of the physical interpretations of both $\HH_N$ and $\SH_N$, as well as their representations, readers are directed to \cite{Gukov:2022gei}.

The \emph{polynomial representation} of $\SH_N$ is the action on the ring of symmetric functions over the field $\bC(q,t)$  of rational functions of $q,t$.
\be\label{DAHA-pol}
\textrm{pol: }\SH_N \  \rotatebox[origin=c]{-90}{$\circlearrowright$} \ \bC(q,t)[X_1^{\pm1},\ldots,X_N^{\pm1}]^{\frakS_N}~.
\ee
Under the polynomial representation, the generators are mapped as
\bea
\boldsymbol{e} e_r(X) \boldsymbol{e} \mapsto &e_r(X) \cr
\boldsymbol{e} e_r(Y) \boldsymbol{e} \mapsto &D^{(r)}=\sum_{\substack{I \subset[1,\ldots ,N]\\ |I|=r}} \prod_{\substack{i \in I\\ j\not\in I}} \frac{t^{\frac12} X_{i}-t^{-\frac12}X_{j}}{X_{i}-X_{j}} T_{q,X_i}
\eea
where the $q$-shift operators act as $T_{q,X_i} X_{j}=q^{\delta_{ij}} X_{j}$. In fact, $D^{(r)}$ are called \emph{Macdonald difference operators}, and symmetric Macdonald polynomials $P_{\lambda}$ are eigenfunction of these difference operators \cite{macdonald1998symmetric}
$$
D^{(r)} P_{\lambda}(X)=e_{r}(t^{\frac{N-1}{2}} q^{\lambda_{1}}, \ldots, t^{\frac{1-N}{2}} q^{\lambda_{N}}) P_{\lambda}(X) \quad(r=0,1, \ldots, N)
$$
where $\lambda$ is a partition (or Young diagram) of length (at most) $N$. We refer the reader to Appendix \ref{app:Macdonald} for the definitions and properties of Macdonald polynomials.

\subsubsection{Connection to quantum toroidal algebra}\label{sec:connection-QTA}

It is shown in \cite{schiffmann2011elliptic} that the large $N\to\infty$ limit of the spherical DAHA $\SH_N$ is equivalent to the quantum toroidal algebra $\QTA$.
To see it, first, we present mode expansions of the Drinfeld current of $\QTA$ slightly different from \eqref{Drinfeld}.
\bea
(1- q_1)E(z)=&\sum_{b \in \mathbb{Z}} \sfP_{1, b} z^{-b}\cr
(1- q_1)F(z)=&\sum_{b \in \mathbb{Z}} \sfP_{-1, b} z^{-b}\cr
K^\pm(z)=&K^\pm\exp \left(\sum_{k>0} \frac{\left(1-q_{2}^{k}\right)\left(1-q_{3}^{k}\right)}{k} \sfP_{0, \pm k} z^{\mp k}\right)\cr
\eea
 The generators are $\sfP_{\pm1, b}$ for $b \in \mathbb{Z}, \sfP_{0, k}$ for $k \in \mathbb{Z}_{\neq 0}$, and central elements $C$, $K^{\pm}$.

Then, the relations of modes \cite{miki2007q,burban2012hall,schiffmann2012drinfeld} are given as follows. For $k, l \in \mathbb{Z}$
\be
\left[\sfP_{0, k}, \sfP_{0, l}\right]=k \frac{(1-q_{1}^{|k|})(C^{|k|}-C^{-|k|})}{(1-q_{2}^{|k|})(1-q_{3}^{|k|})} \delta_{k+l, 0}
\ee
For $k \in \mathbb{Z}_{>0}$ and $l \in \mathbb{Z}$
\be
\begin{aligned}
{\left[\sfP_{0, k}, \sfP_{1, l}\right] } =&C^{-k}(q_{1}^{k}-1) \sfP_{1, l+k,} & {\left[\sfP_{0,-k}, \sfP_{1, l}\right] } =&(1-q_{1}^{k}) \sfP_{1, l-k} \\
{\left[\sfP_{0, k}, \sfP_{-1, l}\right] } =&(1-q_{1}^{k}) \sfP_{-1, l+k}, & {\left[\sfP_{0,-k}, \sfP_{-1, l}\right] } =&(q_{1}^{k}-1) C^{k} \sfP_{-1, l-k}
\end{aligned}
\ee
For $k+l>0$
\bea
\left[\sfP_{1, k}, \sfP_{-1, l}\right]=&\frac{(1-q_{1}) C^{k} }{(1-q_{2})(1-q_{3})} \sfK_{k+l}\cr
\left[\sfP_{1,-k}, \sfP_{-1,-l}\right]=&-\frac{(1-q_{1}) C^{-l}}{(1-q_{2})(1-q_{3})} \sfK_{-k-l}
\eea
For $k \in \mathbb{Z}$
\be
\begin{aligned}
&{\left[\sfP_{1, k}, \sfP_{-1,-k}\right]=\frac{(1-q_{1})(C^{k} K^+-C^{-k}K^{-})}{(1-q_{2})(1-q_{3})}} \\
&{\left[\sfP_{1, k},\left[\sfP_{1, k-1}, \sfP_{1, k+1}\right]\right]=0} \\
&{\left[\sfP_{-1, k},\left[\sfP_{-1, k-1}, \sfP_{-1, k+1}\right]\right]=0 .}
\end{aligned}
\ee
 The relations for the other generators are
defined recursively from the relations above. Consequently, the quantum toroidal algebra $\QTA$ is generated by four elements $\sfP_{\pm1, 0}, \sfP_{0,\pm1}$.

Therefore, $\QTA$ has a presentation via generators $\sfP_{a, b}$ for $(a, b) \in \mathbb{Z}^{2} \backslash\{(0,0)\}$ and the central elements $C,K^{\pm}$. As we will see at the beginning of \S\ref{sec:QTrep}, these generators $\sfP_{a, b}$ correspond to the lattice points in \eqref{eq:DIMsubalgebra}.
In fact, $\QTA$ receives an action of $\SL(2,\bZ)$ as outer automorphisms \cite{burban2012hall,miki2007q}.
Taking $C=K=1$, $\SL(2,\bZ)$ naturally acts on the two-dimensional lattice $\mathbb{Z}^{2} \backslash\{(0,0)\}$. Namely, they transform the generators as
\bea\label{SL2Z}
\gamma(\sfP_{m,n})=\sfP_{am+bn, cm+dn}~,\qquad \gamma =\begin{pmatrix}
a&b\\
c&d
\end{pmatrix}\in \SL(2,\bZ)~.
\eea
For generic central elements $C,K$, $\SL(2,\bZ)$ acts projectively on the generators and non-trivially on the central elements. (See \cite[\S6.4]{burban2012hall} for more detail.)

To connect $\SH_N$ to $\QTA$, we define the generators of $\SH_N$ by the power sum symmetric polynomials (instead of the elementary symmetric functions in \eqref{SH-e}) as follows:
\be\label{SH-gen}
\begin{array}{ll}
\sfP_{k,0}^{(N)}=q^{k}\boldsymbol{e}\sum_i X_i^k \boldsymbol{e}, & \sfP_{-k,0}^{(N)}= \boldsymbol{e}\sum_i X_i^{-k} \boldsymbol{e}, \\
\sfP_{0,k}^{(N)}= \boldsymbol{e}\sum_i Y_i^k \boldsymbol{e}, & \sfP_{0,-k}^{(N)}=q^{k}\boldsymbol{e}\sum_i Y_i^{-k} \boldsymbol{e}, \\
\sfP_{1, b}^{(N)}=q \left(\prod_{i=1}^N\frac{1-t^{i}}{1-t}\right) \boldsymbol{e} X_{1} Y_{1}^{b} \boldsymbol{e}, & \sfP_{-1, b}^{(N)}=\left(\prod_{i=1}^N\frac{1-t^{-i}}{1-t^{-1}}\right)\boldsymbol{e} Y_{1}^{b} X_{1}^{-1} \boldsymbol{e}, \\
\sfP_{ b,1}^{(N)}=\left(\prod_{i=1}^N\frac{1-t^{-i}}{1-t^{-1}}\right)\boldsymbol{e} Y_{1} X_{1}^{b} \boldsymbol{e}, & \sfP_{b,-1}^{(N)}=q \left(\prod_{i=1}^N\frac{1-t^{i}}{1-t}\right) \boldsymbol{e} X_{1}^{b} Y_{1}^{-1} \boldsymbol{e}
\end{array}
\ee
where $k \in \mathbb{Z}_{>0}, b \in \mathbb{Z}$. Defining generators of $\SH_N$ in this way, $\SL(2,\bZ)$ acts on the two-dimensional lattice as in \eqref{SL2Z}. Moreover, by showing a surjective homomorphism $\QTA \twoheadrightarrow \SH_N$ as follows, it is proved in \cite{schiffmann2011elliptic} that $\QTA$ at $C=K^\pm=1$ is a stable limit of $\SH_{N}$ as $N\to\infty$.

\begin{theorem}\cite{schiffmann2011elliptic}
  For any $N$, there is a surjective algebra homomorphism $$\Phi_N:\QTA \twoheadrightarrow \SH_N; \quad \sfP_{a,b}\mapsto \sfP_{a,b}^{(N)}~.$$
\end{theorem}

\subsection{Degenerate DAHA}\label{sec:dDAHA}

We have seen the toroidal property of $\HH_N$, and we now take the trigonometric degeneration. To this end, we set $$Y_i =e^{R y_i}~, \quad q=e^{R}~, \quad  t=e^{R \beta }~, \quad  T_i =s_i e^{R \beta s_i}~,$$
and take the leading order in $R$ from \eqref{DAHA}. We can think the generators $s \in \frakS_N$, $X\in (\bC^\times)^N,y\in \bC^N$, and they satisfy the following relations:
\be\begin{split}
&s X_{i}^{\pm 1}=X_{s(i)}^{\pm 1} s ~, \cr
&s_{i} y=s_{i}(y) s_{i}-\beta\left\langle y^\vee_{i}-y^\vee_{i+1}, y\right\rangle,  \cr
&\left[y_{i}, X_{j}\right]= \begin{cases}-\beta X_{i} s_{ij} & (i<j) \cr
X_{i}+\beta\left(\sum_{k<i} X_{k} s_{i k}+\sum_{k>i} X_{i} s_{i k}\right) & (i=j) \cr
-\beta X_{j} s_{ij} & (i>j)\end{cases}
\end{split}\ee
where $y^\vee_{i}$ is dual to $y_i$ and $\langle\ , \ \rangle$ is the natural pairing. Here, $s_{i}:=s_{i,i+1}:=(i,i+1)$ and $s_{ij}:=(ij)$ are transpositions in $\frakS_N$. This algebra is called \emph{degenerate double affine Hecke algebra} \cite[\S1.6]{cherednik2005double} (a.k.a. \emph{trigonometric Cherednik algebra} or \emph{graded Cherednik algebra} \cite{oblomkov2016geometric}), which we denote $\dH_N$. Note that, despite the name, the algebra is \emph{not} endowed with a Hecke algebra as a subalgebra.

The polynomial representation of $\dH_N$ is
\be
\textrm{pol: } \dH_N \  \rotatebox[origin=c]{-90}{$\circlearrowright$} \ \bC(\beta)[X_1^{\pm1},\ldots,X_N^{\pm1}]~.
\ee
where the generators are mapped as
\begin{equation}\label{SdH-pol}
\begin{aligned}
s\mapsto &s \cr
X_{i}^{\pm 1}\mapsto &X_{i}^{\pm 1} \cr
y_{i}\mapsto & X_{i} \partial_{X_{i}}+\beta \sum_{k \neq i} \frac{1-s_{i k}}{1-X_{k} / X_{i}}+\beta \sum_{k<i} s_{i k} .
\end{aligned}
\end{equation}

As in DAHA, we introduce a spherical subalgebra of $\dH_N$ by defining the idempotent. Abusing the notation, we define the idempotent as
\bea
\boldsymbol{e}:=\frac{1}{N!}\sum_{s\in \frakS_N}s~
\eea
which is subject to $\boldsymbol{e}^2=\boldsymbol{e}$. Then
the \emph{spherical degenerate double affine Hecke algebra} is defined as
$$\SdH_N:= \boldsymbol{e} \dH_N \boldsymbol{e}~.$$
Now let us define the generators $\sfD_{a,b}$ of $\SdH_N$ where $(a,b)\in \bZ\times \bZ_{\ge0}$.  First, $\sfD_{0, 0}^{(N)}$ is the central element, and we define
\begin{equation}
\sfD_{\pm \ell, 0}^{(N)}:=\boldsymbol{e} p_{\ell}(X_{1}^{\pm 1}, \ldots, X_{N}^{\pm 1}) \boldsymbol{e}~,\quad \ell>0~.
\end{equation}
We can define $\sfD_{0,k}^{(N)}$ simply by the power-sum polynomial $\boldsymbol{e} p_{k}(y) \boldsymbol{e}$. However, to connect the large $N$ limit of $\SdH_N$ to the affine Yangian, we modify the definition of $\sfD_{0,k}^{(N)}$.
To this end, we introduce the polynomial in $r$
\begin{equation}
T_{k, i}(r):=\sum_{j=1}^{r}((j-1)-\beta(i-1))^{k}~.
\end{equation}
Then, we define a symmetric function
\begin{equation}
B_{k}(y_1,\ldots,y_N):=\sum_{i=1}^N T_{k, i}(y_i+\beta i)~
\end{equation}
which takes the form
\be
B_{k-1}(y_1,\ldots,y_N)=\frac{p_{k}(y_1,\ldots,y_N)}{k}+\sum_{i=0}^{k-1} c_i(\beta) p_i(y_1,\ldots,y_N)~.
\ee
Then, we define the other generators of $\SdH_N$ as
\be
\sfD_{0,k}^{(N)}=\boldsymbol{e} B_{k-1}(y_1-N\beta,\ldots,y_N-N\beta) \boldsymbol{e}~.
\ee
Under the polynomial representation
\be
\textrm{pol: } \SdH_N \  \rotatebox[origin=c]{-90}{$\circlearrowright$} \ \bC(\beta)[X_1^{\pm1},\ldots,X_N^{\pm1}]^{\frakS_N}~
\ee
the eigenfunctions of $\sfD_{0,k}^{(N)}$ are the Jack polynomials \cite[\S VI.10]{macdonald1998symmetric}
\be\label{y-pol}
\textrm{pol}(\sfD_{0,k}^{(N)})\cdot J_\lambda(X)=\sum_{s \in \lambda} (a'(s)-\beta l'(s))^{k-1}J_\lambda(X)~.
\ee
In fact, $\sfD_{0,k}^{(N)}$ are mutually commuting operators in the Calogero-Sutherland system under the polynomial representation.
Note that $\textrm{pol}(\sfD_{0,2}^{(N)})$ is proportional to the Laplace-Beltrami operator $\square_{N}^{\beta^{-1}}$ introduced in \cite[\S VI.4, Ex. 3]{macdonald1998symmetric}, and it is the Calogero-Sutherland Hamiltonian \eqref{CS-Hamiltonian} deformed in such a way that the eigenvalue is independent of $N$.

It was shown in \cite{Schiffmann:2012aa} that the generators of $\SdH_N$ satisfy the following relations
\bea\label{SdH_N}
\left[\sfD_{0, \ell}^{(N)}, \sfD_{1, k}^{(N)}\right]=&\sfD_{1, \ell+k-1}^{(N)}, \cr
\left[\sfD_{0, \ell}^{(N)}, \sfD_{-1, k}^{(N)}\right]=&-\sfD_{-1, \ell+k-1}^{(N)}, \cr
\left[\sfD_{-1, k}^{(N)}, \sfD_{1, \ell}^{(N)}\right]=&\mathsf{\Psi}_{k+\ell}^{(N)}
\eea
where the elements $\mathsf{\Psi}_{k+\ell}^{(N)}$ are determined through the formula
\bea
1+(1-\beta) \sum_{\ell \ge 0} \mathsf{\Psi}_{\ell}^{(N)} z^{\ell+1}=&K\left(\beta, \sfD_{0,0}^{(N)}, z\right) \exp \left(\sum_{\ell \ge 0} \sfD_{0, \ell+1}^{(N)} \varphi_{\ell}(z)\right) ~.
\eea
Here each function is defined as
\bea
G_{0}(z)=&-\log (z), \cr
G_{\ell}(z)=&(s^{-\ell}-1) / \ell, \quad \ell \neq 0, \cr
\varphi_{\ell}(z)=&\sum_{\alpha=1,-\beta,\beta-1} z^{\ell}\left(G_{\ell}(1-\alpha z)-G_{\ell}(1+\alpha z)\right) \cr
K(\beta, \omega, z)=&\frac{(1+(1-\beta) z)(1+\beta \omega z)}{1+(1-\beta) z+\beta \omega z}~.
\eea
It is straightforward to show that $\SdH_{N}$ is generated by the collection of the generators $\left\{\sfD_{1, \ell}^{(N)}, \sfD_{0, \ell}^{(N)}, \sfD_{-1, \ell}^{(N)}\right\}_{\ell\in \bZ_{>0}}$. Therefore, the relations \eqref{SdH_N} are sufficient for explicitly realizing the algebra $\SdH_{N}$ by generators and relations.

\subsubsection{Connection to affine Yangian}\label{sec:central-ext}

To connect to affine Yangian, we will take the large $N$ limit of $\SdH_{N}$ with central extension, and we denote it by $\SdH^{\boldsymbol{c}}$ \cite{Schiffmann:2012aa}. The algebra $\SdH^{\boldsymbol{c}}$ consists of generators are $\sfD_{a,b}$ of $\SdH_N$ where $(a,b)\in \bZ\times \bZ_{\ge0}$, and a family $\boldsymbol{c}=\left(c_{0}, c_{1}, \ldots\right)$ of central elements. The commutation relations are defined in terms of $\sfD_{\pm1,\ell},\sfD_{0,\ell}$ as
\bea\label{SdH-com}
[\sfD_{0,\ell} , \sfD_{1,k} ] =& \sfD_{1,\ell+k-1} & \ell \geq 1 \,,\cr
[\sfD_{0,\ell},\sfD_{-1,k}]=&-\sfD_{-1,\ell+k-1} & \ell \geq 1 \,,\cr
[\sfD_{-1,k},\sfD_{1,\ell}]=&\mathsf{\Psi}_{k+\ell} & k,\ell \geq 1\,,\cr
[\sfD_{0,\ell} , \sfD_{0,k} ] =& 0 & k,\ell\geq 0\, ,
\eea
where
$\mathsf{\Psi}_k$ is a nonlinear combination of $\sfD_{0,k}$ determined in the form of a generating function
\bea
1+(1-\beta)\sum_{\ell= 0}^\infty\mathsf{\Psi}_\ell z^{\ell+1}= \exp(\sum_{\ell= 0}^\infty(-1)^{\ell+1}c_\ell \pi_\ell(z))\exp(\sum_{\ell= 0}^\infty\sfD_{0,\ell+1} \varphi_\ell(z)) \,,\label{com0}
\eea
with
\be
\pi_\ell(z)=s^\ell G_\ell(1+(1-\beta)z)~.
\ee
The parameters $c_\ell$ ($\ell\geq 0$) are central elements. The relations for the other generators are defined
recursively by \cite{arbesfeld2013presentation}
\bea\label{DDAHAPhi}
\ell \sfD_{\ell+1,0}=\left[\sfD_{1,1}, \sfD_{\ell, 0}\right]~, &\qquad \ell\sfD_{-\ell-1,0}=\left[\sfD_{-\ell, 0}, \sfD_{-1,1}\right], \cr
\sfD_{r, \ell}=\left[\sfD_{0, \ell+1}, \sfD_{r, 0}\right]~, &\qquad \sfD_{-r, \ell}=\left[\sfD_{-r, 0}, \sfD_{0, \ell+1}\right]
\eea
for $\ell \geq 0, r>0$. Therefore, $\SdH^{\boldsymbol{c}}$ is generated by three elements  $\sfD_{0,1}$, $\sfD_{\pm1,1}$.
The expression for $\mathsf{\Psi}_\ell$ in terms of $\sfD_{0,\ell}$ is nonlinear, where the first few relations can be expressed explicitly as
$$
\begin{aligned}
\mathsf{\Psi}_0=&c_0 \\
\mathsf{\Psi}_1=&-c_1+c_0\left(c_0-1\right) (1-\beta)/ 2, \\
\mathsf{\Psi}_2=&c_2+c_1\left(1-c_0\right) (1-\beta)+c_0\left(c_0-1\right)\left(c_0-2\right) (1-\beta)^2 / 6+2 \beta \sfD_{0,1}, \\
\mathsf{\Psi}_3=&6 \beta \sfD_{0,2}+2 c_0 \beta (1-\beta)\sfD_{0,1}+\cdots, \\
\mathsf{\Psi}_4=&12 \beta \sfD_{0,3}+6 c_0 \beta (1-\beta)\sfD_{0,2}\cr 
&+\left(c_0(c_0 -1) \beta (1-\beta)^2-2 c_1 \beta (1-\beta)+2 \b (1 - \b + \b^2)\right) \sfD_{0,1}+\cdots .
\end{aligned}
$$
It is noteworthy that the specialization of $\SdH^{(1,0, \ldots)}$ at $\beta=1$ is isomorphic to the universal enveloping algebra of $\cW_{1+\infty}$ appeared in \S\ref{sec:Winf}.

Furthermore, it was shown in \cite[Theorem 7.9]{Schiffmann:2012aa} that $\SdH^{\boldsymbol{c}}$ is endowed  with a Hopf algebra structure where the coproduct is determined by the following formulas
\begin{equation}\label{SdH-coproduct}
\begin{aligned}
 \Delta(c_m)=&c_m\otimes1+1\otimes c_m \qquad \textrm{for } m \ge 0 \\
 \Delta(\sfD_{m, 0})=&\sfD_{m, 0}\otimes1+1\otimes\sfD_{m, 0} \qquad \textrm{for } m \neq 0 \\
 \Delta(\sfD_{0,1})=&\sfD_{0,1}\otimes1+1\otimes\sfD_{0,1} \\
 \Delta(\sfD_{0,2})=&\sfD_{0,2}\otimes1+1\otimes\sfD_{0,2}+(1-\beta) \sum_{m >0} m \beta^{1-m} \sfD_{-m, 0} \otimes \sfD_{m, 0}, \\
 \Delta(\sfD_{1,1})=&\sfD_{1,1}\otimes1+1\otimes\sfD_{1,1}+(1-\beta) c_0 \otimes \sfD_{1,0} \cr \Delta(\sfD_{-1,1})=&\sfD_{-1,1}\otimes1+1\otimes\sfD_{-1,1}+(1-\beta) \sfD_{-1,0} \otimes c_0 ~.
\end{aligned}
\end{equation}

Now let us compare $\SdH^{\boldsymbol{c}}$ with affine Yangian in \S\ref{sec:AY}.
To this end, we first make change of variable $s=1/z$ in \eqref{com0}, and define the holomorphic fields
\be
\Psi(z):=1+(1-\beta) \sum_{\ell=0}^{\infty} \frac{\mathsf{\Psi}_{\ell}}{z^{\ell+1}}~.
\ee
Similarly, it is useful to define holomorphic fields by $\sfD_{\pm 1,\ell}$ \cite{Bourgine:2014tpa,Bourgine:2015szm}
\be
D_{\pm 1}(z):=\sum_{\ell=0}^{\infty} \frac{\sfD_{\pm 1,\ell}}{z^{\ell+1}}~ .
\ee
Then, the commutation relation \eqref{SdH-com} yields
\be \label{DD-E}
\left[D_{-1}(z), D_{1}(w)\right]=\frac{1}{1-\beta}\frac{\Psi(w)-\Psi(z)}{z-w}~.
\ee

To write $\Psi(z)$ in terms of $\sfD_{0, \ell}$, we introduce
\bea
\Phi(z)&:=\sfD_{0,1}\log (z) -\sum_{\ell=1}^{\infty} \frac{\sfD_{0, \ell+1}}{\ell z^{\ell}} \cr
c(z)&:=c_{0} \log (z)-\sum_{\ell=1}^{\infty} \frac{(-1)^\ell c_{\ell}}{\ell z^{\ell}}\label{def-cPhi}
\eea
This definition of the field $\Phi(z)$ is similar to the mode expansion \eqref{boson-modes} of a holomorphic free bosonic field in 2d CFT. In addition, it is easy to see that $[\Phi(z), \Phi(w)]=0$. Then, the right-hand side of \eqref{com0} can be written as
\begin{equation}
\Psi(z)=\exp \bigg[c(z+1-\beta)-c(z) +\sum_{\alpha=1,-\beta,\beta-1} (\Phi(z-\alpha)-\Phi(z+\alpha))\bigg]~.\label{def-Psi}
\end{equation}
Note that the deformation parameter of $\SdH^{\boldsymbol{c}}$ is related to $\beta=-\e_2/\e_1$ where $\epsilon_{i}$ ($i=1,2,3$) are the deformation parameters of affine Yangian $\AY$ as in \S\ref{sec:QTA-deg}. Therefore, the summation over the set $\{1,-\beta,\beta-1\}$ corresponds to $\{\e_1,\e_2,\e_3\}$.

We now check the OPE of these holomorphic fields. It is easy from \eqref{SdH-com} to derive
\begin{equation}
    \left[\Phi(z), D_\pm (w) \right]= \left[\pm\ln(z-w) D_\pm(w)\right]_+
\end{equation}
where $[\bullet]_+$ implies the projection into the negative power of $w$. It implies
\begin{equation}
    e^{\Phi(z)} D_\pm (w) e^{-\Phi(z)}=\left[(z-w)^{\pm 1}D_\pm(w)\right]_+\,
\end{equation}
and \eqref{eq:AY}
\begin{equation}
    \Psi(z) D_\pm(w) \sim \tilde\varphi(z-w)^{\mp 1} D_\pm(w)\Psi(z), \quad
  \,
\end{equation}
where we define the structure function analogous to \eqref{str-fn-AY} as
\be
\tilde\varphi(z):=\frac{(z+1)(z-\beta)(z-1+\beta)}{(z-1)(z+\beta)(z+1-\beta)}~.
\ee
In addition, it was shown in \cite{arbesfeld2013presentation} that the $D_\pm D_\pm$ OPEs are
\be
D_\pm(z)D_\pm(w)= \tilde\varphi(z-w)^{\mp 1}D_\pm(w)D_\pm(z)~.
\ee

Therefore, we can compare these commutators with affine Yangian in \eqref{AY-current}, and their relations are
\bea
D_{-1}(z) &\leftrightarrow  e(z)\,,\cr
D_{+1}(z) &\leftrightarrow  f(z)\,,\cr
\Psi(z) &\leftrightarrow  \psi(z)\,.\label{SHD-aY}
\eea

\subsubsection{Connection to \texorpdfstring{$\cW_N$}{WN}-algebra}
\label{sec:SHc-AY-W}

At the end of \S\ref{sec:AY}, the relation between the affine Yangian and $\cW_{1+\infty}$-algebra is briefly explained. As seen above, the affine Yangian and $\SdH^{\boldsymbol{c}}$ are equivalent, suggesting a connection between $\SdH^{\boldsymbol{c}}$ and the $\cW$-algebra. Here, we explore this relationship.

As seen in \eqref{y-pol}, $\sfD_{0, k}^{(N)}$ constitutes a set of mutually commuting Hamiltonians for the Calogero-Sutherland system under the polynomial representations. Furthermore, as touched upon in Appendix \ref{app:Jack}, the degenerate DAHA $\SdH_N$ encapsulates the symmetry of the quantum Calogero-Sutherland system. Additionally, as demonstrated in \eqref{beta-iso}, the ring $\scR^\beta_N$ of symmetric functions, characterized by $\bQ(\beta)$ coefficients, is isomorphic to the Fock space of the Heisenberg algebra. Consequently, the Hamiltonian's action on $\scR^\beta_N$ can be represented through Heisenberg modes as in \eqref{CS-free-field}. In fact, $\SdH^{\boldsymbol{c}}$ supports a free field realization as presented in \cite[\S8]{Schiffmann:2012aa}. When specializing the central elements ${\boldsymbol{c}}$ of $\SdH^{\boldsymbol{c}}$ by that $\sfJ_0$ of the Heisenberg algebra 
\be 
c_m=(-\sqrt{\beta}\sfJ_0)^m~,
\ee 
$\SdH^{\boldsymbol{c}}$ is equivalent to the universal enveloping algebra $U(\cW_1)$ of the Heisenberg algebra. The explicit expressions of generic generators $\sfD_{n,m}$ in terms of the Heisenberg modes are rather complicated under the equivalence. Nonetheless, some simple generators can be expressed as
\begin{align} 
\sfD_{ \pm r,0} =&(-1)^r(\sqrt{\beta})^{r-1} \sfJ_{\mp r} \qquad \textrm{for } r>0\cr  
\sfD_{0,1} =& \sum_{r>0}\sfJ_{-r} \sfJ_r\cr 
\sfD_{0,2} =&\frac{\sqrt{\beta}}{6} \sum_{r, s \in \mathbb{Z}}: \sfJ_r \sfJ_s \sfJ_{-r-s}:+\frac{1-\beta}{2} \sum_{r>0}(r-1) \sfJ_{-r} \sfJ_r~.
\label{D02}
\end{align}

Upon establishing the free field realization mentioned above, we can apply the coproduct \eqref{SdH-coproduct} $N$ times to generate $N$ copies of the Heisenberg algebras. Considering that the $\cW_N$-algebra can be derived from the Miura transformation \eqref{Miura-glN} involving $N$ free bosons, it is reasonable to anticipate a connection between $\SdH^{\boldsymbol{c}}$ and the $\cW_N$-algebra.  Indeed, setting the central elements 
\be  \label{central-SdH-W}
c_m=p_m(-\sqrt{\beta}\sfJ_0^{(1)},(1-\beta)-\sqrt{\beta}\sfJ_0^{(2)},2(1-\beta)-\sqrt{\beta}\sfJ_0^{(3)},\cdots,(N-1)(1-\beta)-\sqrt{\beta}\sfJ_0^{(N)})
\ee
where $p_m$ are the power-sum polynomials, it is shown in \cite[\S8.9]{Schiffmann:2012aa} that  $\SdH^{\boldsymbol{c}}$ is equivalent to the universal enveloping algebra $U(\cW_N)$. For instance, the following generators can be expressed by the Heisenberg modes 
\begin{align} 
\sfD_{ \pm r,0} =&(-1)^r(\sqrt{\beta})^{r-1}\sum_{i=1}^N \sfJ_{\mp r}^{(i)} \qquad \textrm{for  }  r>0~,\\
\sfD_{0,2} =&\sum_{i=1}^N\left[\frac{\sqrt{\beta}}{6} \sum_{r, s \in \mathbb{Z}}: \sfJ_r^{(i)} \sfJ_s^{(i)} \sfJ_{-r-s}^{(i)}:+\frac{1-\beta}{2} \sum_{r>0}(r+1-2 i) \sfJ_{-r}^{(i)} \sfJ_r^{(i)}\right]+(1-\beta) \sum_{i<j}^N \sum_{r>0} r \sfJ_{-r}^{(i)} \sfJ_r^{(j)} ~.\label{D02a}\end{align}
where the last term in $\sfD_{0,2}$ represents the nontrivial cross-term due to the coproduct \eqref{SdH-coproduct}.

In this equivalence, the Virasoro modes are expressed by the modes of $\SdH^{\boldsymbol{c}}$ as 
\begin{align}
\sfL_{\pm r} =&(-\sqrt{\beta})^{-r}\left[\frac{\sfD_{\mp r,1}}{r} +\frac{(1-r)c_0(1-\beta)}{2}\sfD_{\mp r,0}\right]\qquad \textrm{for } r>0~,\\
\sfL_{0} =&\frac12 [\sfL_1, \sfL_{-1}]=\sfD_{0,1}+\frac{1}{2\beta}\left[
c_2+c_1(1-c_0)(1-\beta) +\frac{(1-\beta)^2}{6}c_0(c_0-1)(c_0-2)
\right]\,.\nonumber
\end{align}
While the derivation is laborious, one can verify that these are consistent with the Miura transformation \eqref{Virasoro_g}, and the commutation relation  \eqref{Virasoro} of the Virasoro algebra can be derived recursively. For example,  the explicit computation of $\left[\sfL_2, \sfL_{-2}\right]$ \cite[Appendix C]{Kanno:2013aha} yields the desired central charge \eqref{cc-Wgln}
\begin{equation}
c=\frac{1}{\beta}\left(-c_0^3 (1-\beta)^2+c_0-c_0 (1-\beta)+c_0 (1-\beta)^2\right)=1+(N-1)\left(1-Q^2N(N+1)\right)
\end{equation}
where we use $c_0=N$ from \eqref{central-SdH-W}.

In principle, using the coproduct \eqref{SdH-coproduct} of $\SdH^{\boldsymbol{c}}$ and the Miura transformation \eqref{Miura-glN} of the $\mathcal{W}_N$-algebra, one can show that  $\SdH^{\boldsymbol{c}}$ contains the $\mathcal{W}_N$-algebra as a subalgebra. Detailed explanations can be found in \cite{Schiffmann:2012aa,Matsuo:2014rba}, but we provide the outcome as follows:
\begin{align}
    {\sfW}^{(d)}_{\pm r} =& (-1)^{d-2+r}\beta^{(2-r-d)/2} \frac{\sfD_{\mp r,d-1}}{r}+\cdots
\end{align}
where $\cdots$ represent nonlinear terms of $\sfD_{m,n}$ where $n<d-1$. Consequently, the modes of the $W^{(d)}(z)$ current can be expressed using generators $\sfD_{m,n}$ when $n<d$.  (See Figure \ref{fig:SdH-W}.)
This relationship between $\SdH^{\boldsymbol{c}}$ and the $\mathcal{W}_N$-algebra is essential in proving the Whittaker condition of the Gaiotto state given in \cite{Schiffmann:2012aa}. (See the end of \S\ref{sec:AGT}.)
\begin{figure}[ht]
    \centering
    \includegraphics{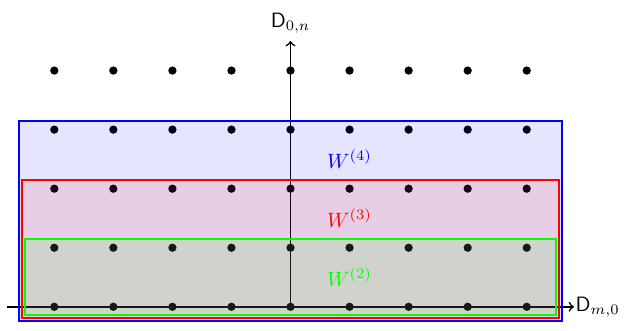}
    \caption{The modes for the $W^{(d)}(z)$ current are expressed by generators $\sfD_{m,n}$ with $n<d$.}
    \label{fig:SdH-W}
\end{figure}

\section{Representations}\label{sec:QTrep}
Up to this point, our focus has been on introducing the algebras. The heart of our exploration, however, lies in the representation theory of these algebras. Hence, this section studies the representation theory of \QTA. Given its profound richness, this section is quite extensive, so we provide a succinct overview below.

We will first discuss some basic properties of the $\QTA$. In particular, as seen in \eqref{SL2Z}, $\QTA$ enjoys the $\SL(2,\bZ)$ automorphism, and the S-transformation is often referred to as the Miki automorphism \cite{miki2007q}. The existence of this automorphism leads to two classes of representations: the vertical and the horizontal representations. Detailed discussions on these will be presented in \S\ref{sec:vertical-rep} and \S\ref{sec:horizontalrep}, respectively.

Subsequently, our discussion will transition to applications of these representations. This includes exploring topics like Macdonald functions in \S\ref{sec:QTA-Mac}, deformed $\mathcal{W}$-algebras in \S\ref{sec:deformedW}, screening currents in \S\ref{sec:QTA-screening}, intertwiners in \S\ref{sec:intertwiner}, representations of $\AY$ in \S\ref{sec:repAY}, and the minimal models of $\mathcal{W}$-algebras in \S\ref{sec:AYminimalmodel}. 

The structure and logical progression of our discussion can be visualized as delineated in the following figure. For a coherent understanding, readers are advised to start with \S\ref{sec:vertical-rep} and \S\ref{sec:horizontalrep}. Once familiarized, the subsequent sections can be read independently. To elaborate further on the interconnections, both \S\ref{sec:deformedW} and \S\ref{sec:QTA-screening} resonate with horizontal representations. In contrast, \S\ref{sec:repAY} and \S\ref{sec:AYminimalmodel} pertain more to vertical representations. Lastly, \S\ref{sec:QTA-Mac} and \S\ref{sec:intertwiner} relate elements from both the horizontal and vertical representations.
\begin{figure}[ht]\centering
\begin{tikzpicture}
\node(A)[rounded rectangle,draw,align=center,scale=1.1,fill=gray!10] at (0,4.8) {Representations of $\QTA$ \S\ref{sec:vertical-rep},\S\ref{sec:horizontalrep}};
\node(B)[rounded rectangle,draw,align=center,scale=0.9,fill=violet!20] at (-6,3) {Macdonald functions \S\ref{sec:QTA-Mac}, \S\ref{app:Macdonald}};
\node(C)[rounded rectangle,draw,align=center,scale=0.9,fill=blue!20] at (-3.5,1) {\begin{varwidth}{\linewidth}\begin{itemize}[nosep]
    \item[] Deformed $\mathcal{W}$-algebras
    \item Generators and quadratic relations \S\ref{sec:deformedW}
    \item Screening currents \S\ref{sec:QTA-screening}
\end{itemize}\end{varwidth}};
\node(D)[rounded rectangle,draw,align=center,scale=0.9,fill=violet!20] at (6,3) {Intertwiners \S\ref{sec:intertwiner}};
\node(E)[rounded rectangle,draw,align=center,scale=0.9,fill=red!20] at (4,1.5) {Representations
of $\AY$ \S\ref{sec:repAY}};
\node(F)[rounded rectangle,draw,align=center,scale=0.9,fill=red!20] at (4,-.5) {$\mathcal{W}_N$ minimal models \S\ref{sec:AYminimalmodel}};
\draw[->,shorten >= 2pt,shorten <= 2pt,violet] (A) to (B);
\draw[->,shorten >= 2pt,shorten <= 2pt,blue] (A) to (C);
\draw[->,shorten >= 2pt,shorten <= 2pt,violet] (A) to (D);
\draw[->,shorten >= 2pt,shorten <= 2pt,red] (A) to (E);
\draw[->,shorten >= 2pt,shorten <= 2pt,red] (E) to (F);
\end{tikzpicture}\caption{Flow chart of \S\ref{sec:QTrep}. The subsections colored in red, blue, and purple are topics related to vertical, horizontal, and both representations, respectively. }\label{fig:flowchart2}
\end{figure}
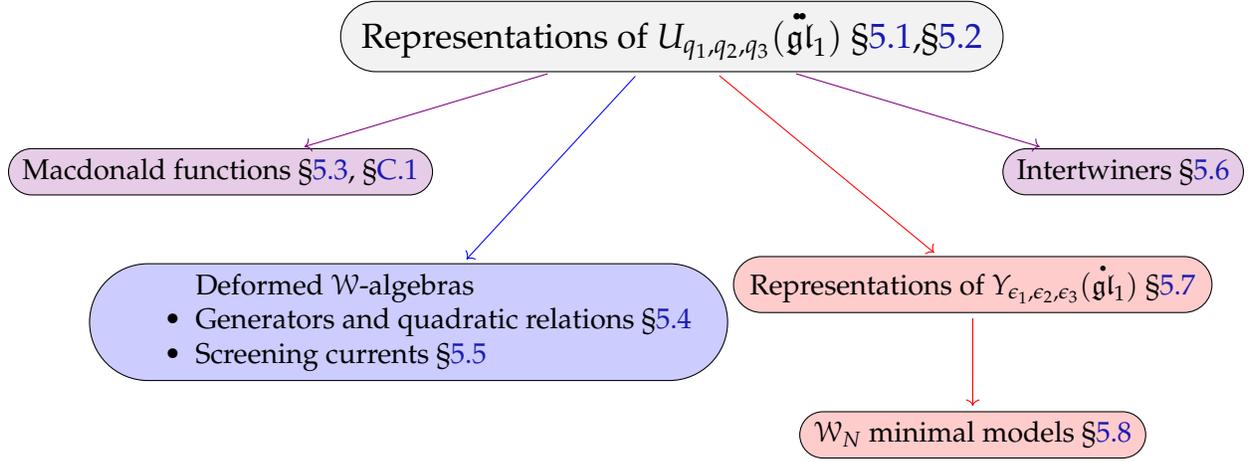

The first will be the relation with Macdonald functions in \S\ref{sec:QTA-Mac}. We will show that the Macdonald symmetric functions are related to the vertical and horizontal Fock representations, respectively. In the vertical representation viewpoint, the bases of the representation space have a one-to-one correspondence with the Macdonald functions. The Drinfeld currents correspond to the elementary symmetric function and the Macdonald difference operators. Similarly, from the horizontal representation viewpoint, the vertex operator representation of the Drinfeld currents gives the free field realization of the Macdonald difference operators. 

In \S \ref{sec:deformedW} and \S\ref{sec:QTA-screening}, we show that the $\QTA$ gives the deformed $\mathcal{W}$-algebra. As discussed in \S\ref{sec:W}, there are two ways to describe the free field realization of $\mathcal{W}$-algebras. One is to explicitly write down the free field realization of the generators using Miura transformation and find a close set of algebraic relations, which will be called the \textit{quadratic relations}. The other way is to introduce screening currents, establishing the $\cW$-algebra as the commutant of these currents. We will give the free field realization of the generators in \S\ref{sec:deformedW}, detailing how to derive the quadratic relations. This realization stems from the horizontal representations with coproduct structure, which links directly to the trigonometric version of the Miura transformation. Subsequently, \S\ref{sec:QTA-screening} defines the screening currents. Using the coproduct and the defining relations of $\QTA$, one can show that the generators introduced in \S\ref{sec:deformedW} indeed commute with the screening currents.

\S\ref{sec:intertwiner} is dedicated to exploring algebraic structures called intertwiners. Intertwiners are obtained from the horizontal and vertical representations and the coproduct structure. The intertwiners find applications in subsequent sections to describe instanton partition functions.

In \S\ref{sec:repAY}, we study the degenerate version of the $\QTA$ which is the affine Yangian $\mathfrak{gl}_{1}$. The representations of the affine Yangian $\mathfrak{gl}_{1}$ are simply obtained from the representations of $\QTA$ by taking the degenerate limit. Using the degenerate version, we will show in \S\ref{sec:AYminimalmodel} as an application that the $\cW_{N}$ minimal model can be realized from MacMahon representations with constraints. We will show that introducing two pits to the plane partitions gives the minimal model condition.  

\bigskip

Before moving on to the explicit representations of $\QTA$, let us review some basic properties of $\QTA$. The algebra $\QTA$ has two subalgebras. First, recall that the mode expansions of (\ref{eq:DIMdef}) including the currents $K^{\pm}(z)$ are written as
\bea
&[\sfH_{r},\sfH_{s}]=\delta_{r+s,0}\frac{C^{r}-C^{-r}}{r}\kappa_{r},\\
&[\sfH_{r},\sfE_{m}]=\frac{C^{\frac{r-|r|}{2}}\kappa_r}{r}\sfE_{r+m},\quad [\sfH_{r},\sfF_{m}]=-\frac{C^{\frac{r+|r|}{2}}\kappa_{r}}{r}\sfF_{r+m},\\
&[\sfE_{m},\sfF_{n}]=\begin{dcases}
\tilde{g}C^{n}\sfK_{m+n}\quad (m+n>0),\\
\tilde{g}\left(C^{n}K^{+}-C^{n}K^{-}\right)\quad (m+n=0),\\
-\tilde{g}C^{-m}\sfK_{m+n}\quad (m+n<0).
\end{dcases}\label{eq:DIM-mode}
\eea
The two subalgebras are
\begin{itemize}
    \item Subalgebra from $C,\sfH_{\pm 1}$:
    \begin{align}
    \begin{split}
    [\sfH_{1},\sfH_{-1}]=&\kappa_{1}(C-C^{-1}),\\
    [C,\sfH_{\pm 1}]=&0.
    \end{split}\label{eq:verticalsubalg}
    \end{align}
    \item Subalgebra from $\sfE_{0},\sfF_{0}, K^{-}$:
    \begin{align}
    \begin{split}
    [\sfE_{0},\sfF_{0}]=&-\tilde{g}(K^{-}-(K^{-})^{-1}),\\
    [K^{-},\sfE_{0}]=&[K^{-},\sfF_{0}]=0.
    \end{split}\label{eq:horizontalsubalg}
    \end{align}
\end{itemize}
The first subalgebra is called \emph{vertical subalgebra} while the second one is called \emph{horizontal subalgebra} in the literature. The elements of the algebra are actually $\mathbb{Z}^{2}$-graded \cite{FFJMM1} and the defining relations of the algebra preserve this grading:
\begin{align}
\begin{split}
    &\text{deg}(\sfE_{n})=(1,n),\quad \text{deg}(\sfF_{n})=(-1,n),\quad \text{deg}(\sfK_{r})=(0,r),\\
    &\text{deg}(C)=\text{deg}(K^{-})=(0,0).\label{eq:phdegreee}
\end{split}
\end{align}
Note that this grading come from the grading operators $d_{1},d_{2}$ in (\ref{eq:gradingop1}) and (\ref{eq:gradingop2}). For a mode element $x\in\QTA$ with degree $(n_{1},n_{2})$, we have 
\begin{align}
    q^{d_{1}}xq^{-d_{1}}=q^{n_{1}}x,\quad q^{d_{2}}xq^{-d_{2}}=q^{-n_{2}}x.
\end{align}
We say that $x$ has a principal degree $n_{1}$ and homogeneous degree $n_{2}$ (see for example \cite{Feigin:2015raa}).

We can summarize this $\mathbb{Z}^{2}$-grading as follows:
\begin{align}
\includegraphics[width=10cm]{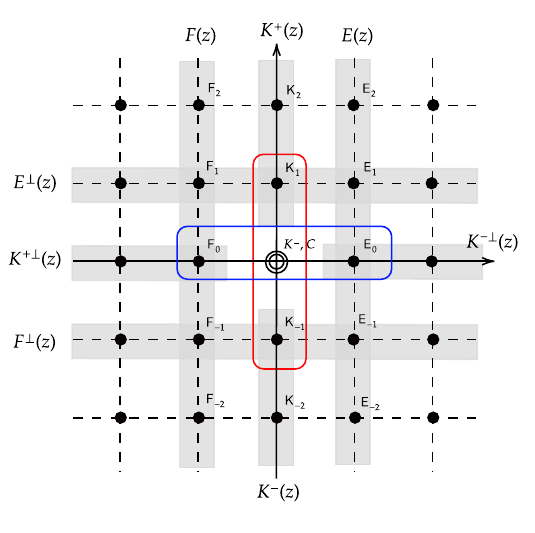}
\label{eq:DIMsubalgebra}
\end{align}
The modes surrounded by the red and blue boxes are the vertical and horizontal subalgebra, respectively. Obviously, the terminology ``vertical" and ``horizontal" come from (\ref{eq:DIMsubalgebra}). Note that after expanding the exponential of the currents $K^{\pm}(z)$, the modes $\sfH_{\pm 1}$ are related to $\sfK_{\pm 1}$ as
\begin{align}
  \sfK_{\pm 1}=\mp \kappa_{1}K^{\pm}\sfH_{\pm 1}\label{eq:KHcorrespondence}
\end{align}
and that $\sfH_{\pm 1}$ will be exactly on the same place where $\sfK_{\pm1}$ is drawn in (\ref{eq:DIMsubalgebra}). We will not distinguish them unless it is necessary. Note that the degree operators $d_{1},d_{2}$ are placed in the origin with the central elements $K^{-},C$, though they are not drawn.

Looking at (\ref{eq:verticalsubalg}) and (\ref{eq:horizontalsubalg}), it is obvious to see that the subalgebras are isomorphic under
\begin{align}
\begin{split}
E_0\to H_{-1}\to F_0\to H_1 \to E_0,\\
K^-\to C\to K^+ \to C^{-1} \to K^{-}\\
d_1\to d_2\to -d_1 \to -d_2\to d_1
   \label{eq:Miki-duality}
\end{split}
\end{align}
where we are schematic, not being careful with the factors scaling the modes. 
This is rotating the modes in (\ref{eq:DIMsubalgebra}) by 90 degrees in the counterclockwise direction. Actually, this automorphism extends to the entire algebra, which corresponds to the $S$-transformation in \eqref{SL2Z}, called Miki-automorphism \cite{miki2007q}. We denote the currents rotated by the Miki-automorphism as 
\begin{align}
    &E^{\perp}(z)=\sum_{m\in\mathbb{Z}}\mathsf{E}^{\perp}_{m}z^{-m},\quad F^{\perp}(z)=\sum_{m\in\mathbb{Z}}\mathsf{F}^{\perp}_{m}z^{-m},\quad K^{\pm\perp}(z)=K^{\pm\perp}\exp\left(\sum_{r>0}\mp\frac{\kappa_{r}}{r}\mathsf{H}^{\perp}_{\pm r}z^{\mp r}\right),\label{eq:perpDIMmode}
\end{align}
and the central element and degree operators are $ C^{\perp},d_{1}^{\perp},d_{2}^{\perp}$.  

To study representations, we must specify the central charges of $C,K^{-}$. We introduce two classes of representations called \emph{vertical representations} and \emph{horizontal representations}:

\begin{itemize}
    \item Vertical representations \cite{FFJMM1, feigin2012quantum,Feigin2011}: vector $(C,K^{-})=(1,1)$, Fock $(1,q_{c}^{1/2})\, (c=1,2,3)$, MacMahon $(1,K)\, (K\in\mathbb{C}^{\times})$.
    \item Horizontal representations \cite{Shiraishi:1995rp,Awata:1995zk,FF,FHSSY,miki2007q,Kojima:2020vtc,Kojima2019,Harada:2021xnm} Fock $(q_{c}^{1/2},1)\,(c=1,2,3)$.
\end{itemize}
The vertical subalgebra is commutative for vertical representations (i.e., $C=1$) while the horizontal subalgebra is commutative for horizontal representations (i.e., $K^{-}=1$). The terminology comes from this fact.

Another important feature we can observe from the mode expansions (\ref{eq:DIM-mode}) is that after specifying the central elements, the algebra is generated by $\sfE_{0}, \sfF_{0}, \sfK_{\pm1}(\sfH_{\pm1})$. For example, higher modes $\sfE_{m}$ and $\sfF_{m}$ come from the recursion relation
\begin{align}
    [\sfK_{1},\sfE_{m}]\propto \sfE_{m+1},\quad [\sfK_{-1},\sfF_{m}]\propto \sfF_{m+1}.
\end{align}
Starting from $\sfE_{0},\sfF_{0}$, we can generate all the modes $\sfE_{m},\sfF_{m}$, and then from their commutation relations, we can generate higher $\sfK_{r}$.
\subsection{Vertical representations}\label{sec:vertical-rep}
As mentioned, vertical representations have the central charge $C=1$. We discuss how to derive these representations following the original articles \cite{FFJMM1, feigin2012quantum,Feigin2011}. See them for details.

\subsubsection{Vector representation}\label{sec:vectorrep}
\begin{figure}[ht]
\begin{center}
\includegraphics{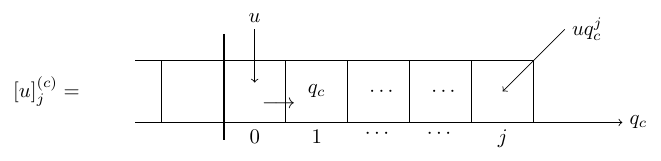}
\end{center}
\caption{Figure of bases of vector representations. The vector $[u]_{j}^{(c)}$ is a sequence of semi-infinite boxes, where there are $j+1$ boxes right to the border. The coordinate of the origin is denoted by $u$ and shifted by $q_{c}$ going in the right direction. }\label{fig:vectorrep}
\end{figure}
Vector representations are representations with central charges $(C,K^{-})=(1,1)$. There are three types depending on $c=1,2,3$, and the actions of the Drinfeld currents are
\bea
        K^{\pm}(z)[u]^{(c)}_{j}=&\bl[\Psi_{[u]^{(c)}_{j}}(z)\br]_{\pm}[u]^{(c)}_{j}\eqqcolon[S_{c}\bl(uq_{c}^{j+1}/z\br)]_{\pm}[u]^{(c)}_{j},\cr
        E(z)[u]^{(c)}_{j}=&\mathcal{E}\delta\bl(uq_{c}^{j+1}/z\br)[u]^{(c)}_{j+1},\cr
        F(z)[u]^{(c)}_{j}=&\mathcal{F}\delta\bl(uq_{c}^{j}/z\br)[u]^{(c)}_{j-1},\quad (c=1,2,3,\quad j\in\mathbb{Z})
    \label{eq:vectorrep}
\eea
where
\bea
    \mathcal{E}\mathcal{F}=&\Tilde{g}\frac{(1-q_{c+1}^{-1})(1-q_{c-1}^{-1})}{(1-q_{c})}.
\eea
Note that the subscript $c$ of $q_{c}$ is taken cyclically modulo $3$. The sign $[f(z)]_{\pm}$ means we are formally expanding the rational function $f(z)$ in $z^{\mp k},\,k\geq0$.
We leave the factors $\mathcal{E},\mathcal{F}$ undetermined in this section. We also set
\bea\label{Sc}
    S_{c}(z)=&\frac{(1-q_{c+1}z)(1-q_{c-1}z)}{(1-z)(1-q_{c}^{-1}z)}=\frac{(1-q_{c+1}^{-1}z^{-1})(1-q_{c-1}^{-1}z^{-1})}{(1-z^{-1})(1-q_{c}z^{-1})}\\=&\begin{dcases}\exp\left(\sum_{r>0}\frac{-\kappa_{r}}{r(1-q_{c}^{r})}z^{r}\right),\quad |z|<1\\
    \exp\left(\sum_{r>0}\frac{\kappa_{r}}{r(1-q_{c}^{-r})}z^{-r}\right),\quad |z|>1
    \end{dcases}
\eea
where the properties and their relation to the structure function \eqref{str-fn-QTA} are given\footnote{Note that this reflection property is only true for $z\neq 1,q_{c}^{-1}$. The poles will give extra delta function contributions and play important roles in obtaining the quadratic relation of deformed $\mathcal{W}$-algebras (see \S\ref{sec:deformedW}).}
\bea
    S_{c}(q_{c}z)=S_{c}(z^{-1}),\quad g(z)=\frac{S_{c}(z)}{S_{c}(q_{c}z)}=\frac{S_{c}(q_{c}z^{-1})}{S_{c}(z^{-1})}.
\eea
One can easily show that (\ref{eq:vectorrep}) obeys the defining relations in (\ref{eq:DIMdef}). The KE and KF relations come from the recursion formula
\bea
\frac{\Psi_{[u]_{j+1}^{(c)}}(z)}{\Psi_{[u]^{(c)}_{j}}(z)}=\frac{S_{c}(uq_{c}^{j+2}/z)}{S_{c}(uq_{c}^{j+1}/z)}=g\left(\frac{z}{uq_{c}^{j+1}}\right).
\eea
The EF relation comes from
\bea
    \relax [E(z),F(w)][u]^{(c)}_{j}=&\mathcal{E}\mathcal{F}\delta\bl(z/w\br)\left(\delta\bl(uq_{c}^{j}/z\br)-\delta\bl(uq_{c}^{j+1}/z\br)\right)[u]^{(c)}_{j},\cr
    \Tilde{g}(K^{+}(z)-K^{-}(z))[u]^{(c)}_{j}=&\Tilde{g}\frac{(1-q_{c+1}^{-1})(1-q_{c-1}^{-1})}{(1-q_{c})}\left(\delta\bl(uq_{c}^{j}/z\br)-\delta\bl(uq_{c}^{j+1}/z\br)\right)[u]^{(c)}_{j}
\eea
where we use\footnote{This formula itself is nothing special, and most readers have met it in the context of CFT, although not mentioned explicitly in textbooks. Consider the $\beta\gamma$ system whose OPE is $\beta(z)\gamma(w)\sim\frac{1}{z-w}$ and $\gamma(z)\beta(w)\sim-\frac{1}{z-w}$ where we are assuming radial ordering $|z|>|w|$. Then, $[\beta(z),\gamma(w)]=\left.\frac{1}{z-w}\right|_{|z|>|w|}-\left.\frac{1}{z-w}\right|_{|z|<|w|}=\frac{1}{z}\delta(w/z)$.}
\bea
    \frac{\prod_{i}(1-\alpha_{i}z)}{\prod_{j}(1-\beta_{j}z)}-z^{|\{i\}|-|\{j\}|}\frac{\prod_{i}(z^{-1}-\alpha_{i})}{\prod_{j}(z^{-1}-\beta_{j})}=\sum_{k}\frac{\prod_{i}(1-\alpha_{i}/\beta_{k})}{\prod_{j\neq k}(1-\beta_{j}/\beta_{k})}\delta(\beta_{k}z).\label{eq:residue_formula}
\eea
The first term is understood as a rational function expanded in the region $|z|\ll 1$ while the second term is understood as a function expanded in the region $|z|\gg 1$ (this resembles the radial ordering of CFT). The difference between these two terms gives the singularities in the poles of the rational function, which arise on the right-hand side as a delta function. For example, in the simplest case where we only have one numerator and one denominator, this can be derived as
\begin{align}
\begin{split}
    \frac{1-\alpha z}{1-\beta z}=&(1-\alpha z)\sum_{n=0}^{\infty}(\beta z)^{n},\\
    \frac{\alpha}{\beta}\frac{1-\alpha^{-1}z^{-1}}{1-\beta^{-1}z^{-1}}=&\frac{\alpha}{\beta}(1-\alpha^{-1}z^{-1})\sum_{n=0}^{\infty}(\beta z)^{-n},\\
    \frac{1-\alpha z}{1-\beta z}- \frac{\alpha}{\beta}\frac{1-\alpha^{-1}z^{-1}}{1-\beta^{-1}z^{-1}}=&(1-\frac{\alpha}{\beta})\delta(\beta z)
\end{split}
\end{align}
where we use (\ref{eq:deltafunctiondef}).
The other EE and FF relations can be checked by direct calculation.

In fact, the vector representation has a nice pictorial interpretation. The vector $[u]^{(c)}_{j}$ is understood as a semi-infinite row of boxes, and the operator $E(z),F(z)$ adds and removes a box from the configuration, respectively (see Figure \ref{fig:vectorrep}).

\paragraph{Action on function space} Rewriting (\ref{eq:vectorrep}), we actually can obtain a representation acting on a function space. Changing the vectors as
\begin{align}
    [u]_{j}^{(c)}\rightarrow f(q_{c}^{j+1}u)\equiv f(y)
\end{align}
we have
\begin{align}
    \begin{split}
        K^{\pm}(z)f(y)=&\left[S_{c}(y/z)\right]_{\pm}f(y),\\
        E(z)f(y)=&\mathcal{E}\delta\left(y/z\right)T_{y,q_{c}}f(y),\\
        F(z)f(y)=&\mathcal{F}\delta\left(q_{c}^{-1}y/z\right)T^{-1}_{q_{c},y}f(y)
    \end{split}\label{eq:vectorrepfunction}
\end{align}
where $T_{q,y}$ is the $q$-shift operator $T_{q,y}f(y)=f(qy)$. Expanding the right-hand side of (\ref{eq:vectorrepfunction}) as
\begin{align}
\begin{split}
&K^{+}(z)=K^{+}\exp\left(\sum_{r>0}-\frac{\kappa_{r}}{r}\sfH_{r}z^{-r}\right)\rightarrow\exp\left(\sum_{r>0}-\frac{\kappa_{r}y^{r}}{r(1-q_{c}^{r})}z^{-r}\right),\\
&K^{-}(z)=K^{-}\exp\left(\sum_{r>0}\frac{\kappa_{r}}{r}\sfH_{-r}z^{r}\right)\rightarrow \exp\left(\sum_{r>0}\frac{\kappa_{r}y^{-r}}{r(1-q_{c}^{-r})}z^{r}\right),\\
&E(z)\rightarrow\mathcal{E}\sum_{m\in\mathbb{Z}}\left(\frac{y}{z}\right)^{m}T_{y,q_{c}},\\
&F(z)\rightarrow\mathcal{F}\sum_{m\in\mathbb{Z}}\left(\frac{q_{c}^{-1}y}{z}\right)^{m}T^{-1}_{y,q_{c}}
\end{split}
\end{align}
we can write down the representation of the modes as
\begin{align}
    K^{\pm},C\rightarrow 1,\quad \sfH_{\pm r}\rightarrow \frac{y^{\pm r}}{1-q_{c}^{\pm r}}\,\,(r> 0),\quad \sfE_{m}\rightarrow \mathcal{E}y^{m}T_{y,q_{c}},\quad \sfF_{m}\rightarrow\mathcal{F} q_{c}^{-m}y^{m}T^{-1}_{y,q_{c}}.
\end{align}

\subsubsection{Fock representation}\label{sec:verticalFockrep}
Fock representations are parameterized by central charges $(C,K^{-})=(1,q_{c}^{1/2}),\,(c=1,2,3)$. We denote them $\mathcal{F}_{c}(u),\, (c=1,2,3)$, respectively. They are constructed by taking tensor products of vector representations. The bases are labeled by a Young diagram, which consists of a sequence of non-negative integers in non-increasing order containing only finitely many non-zero terms:
\bea
    \lambda=(\lambda_{1},\lambda_{2},\ldots),\quad \lambda_{1}\geq\lambda_{2}\geq \cdots,\quad \lambda_{i}\in\mathbb{Z}_{\geq 0}.\label{vertical-fock}
\eea

Let us derive the representation $\mathcal{F}_{c}(u)$. We first consider tensor products of two vector representations
\be
    [u]^{(c-1)}_{j} \otimes [q_{c+1}u]^{(c-1)}_{k},\quad j,k\in\mathbb{Z}.
\ee
The shift of the spectral parameter is chosen for later convenience. The action of the generators $E(z),F(z)$ are defined by the coproduct (\ref{eq:coproduct}):
\bea
    \Delta(E(z))[u]^{(c-1)}_{j}\otimes [q_{c+1}u]^{(c-1)}_{k}=&\mathcal{E}\delta\left(uq_{c-1}^{j+1}/z\right)[u]_{j+1}^{(c-1)}\otimes [q_{c+1}u]_{k}^{(c-1)},\cr
    \Delta(F(z))[u]_{j}^{(c-1)}\otimes [q_{c+1}u]^{(c-1)}_{k}=&\mathcal{F}\delta\left(uq_{c+1}q_{c-1}^{k}\right)[u]_{j}^{(c-1)}\otimes [q_{c+1}u]^{(c-1)}_{k-1}
\eea
where we use
\bea
    (K^{-}(z)\otimes E(z))[u]_{j}^{(c-1)}\otimes [q_{c+1}u]^{(c-1)}_{k}=0,\cr
    (F(z)\otimes K^{+}(z))[u]_{j}^{(c-1)}\otimes [q_{c+1}u]^{(c-1)}_{k}=0.
\eea
Thus, $[u]^{(c-1)}_{j}\otimes [q_{c+1}u]^{(c-1)}_{k}\,(j\geq k)$ form a submodule and gives the decreasing condition of the Young diagram. One can similarly do the analysis for tensor products of $N$ vector representations and consider the following submodule
\bea
    \ket{u,\lambda}^{(c)}_{N}=\bigotimes_{j=1}^{N}[q_{c+1}^{j-1}u]^{(c-1)}_{\lambda_{j}-1},\quad \lambda_{1}\geq \lambda_{2}\geq\cdots\geq \lambda_{N},\quad \lambda_{i}\in\mathbb{Z}.
\eea
\begin{figure}[ht]
    \centering
    \includegraphics[width=8.5cm]{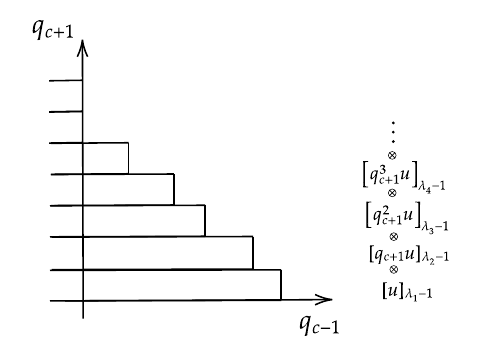}
    \caption{Tensor products of vector representations.}
    \label{fig:Fockrep}
\end{figure}
We illustrate this vector as in Figure \ref{fig:Fockrep} and denote this module $\mathcal{F}_{c}^{(N)}(u)$. The actions of the operators are
\bea
K^{\pm}(z)\ket{u,\lambda}_{N}^{(c)}=&\prod_{i=1}^{N}\left[S_{(c-1)}(uq_{c+1}^{i-1}q_{c-1}^{\lambda_{i}}/z)\right]_{\pm}\ket{u,\lambda}_{N}^{(c)},\\
E(z)\ket{u,\lambda}^{(c)}_{N}=&\mathcal{E}\sum_{i=1}^{N}\prod_{l=1}^{i-1}\left[S_{(c-1)}\bl(uq_{c+1}^{l-1}q_{c-1}^{\lambda_{l}}/z\br)\right]_{-}\delta\bl(uq_{c+1}^{i-1}q_{c-1}^{\lambda_{i}}/z\br)\ket{u,\lambda+\Box_{i}}^{(c)}_{N},\\
F(z)\ket{u,\lambda}^{(c)}_{N}=&\mathcal{F}\sum_{i=1}^{N}\prod_{l=i+1}^{N}\left[S_{(c-1)}\bl(uq_{c+1}^{l-1}q_{c-1}^{\lambda_{l}}/z\br)\right]_{+}\delta\left(uq_{c+1}^{i-1}q_{c-1}^{\lambda_{i}-1}/z\right)\ket{u,\lambda-\Box_{i}}^{(c)}_{N}.\label{eq:Nfockrep}
\eea
To obtain the Fock representation, we need to take the limit $N\rightarrow \infty$ and regularize the actions properly \cite{FFJMM1}. We define the basis
\bea
    \ket{u,\lambda}^{(c)}=\bigotimes_{j=1}^{\infty}[q_{c+1}^{j-1}u]^{(c-1)}_{\lambda_{j}-1}
\eea
where $\lambda_{n}=0$ for $n>\ell(\lambda)$. Note that $\ell(\lambda)$ is the length of the Young diagram. Then, the result is
\bea\label{QTA-Fock}
    K^{\pm}(z)\ket{u,\lambda}^{(c)}=&\bl[\Psi_{\lambda}^{(c)}(z,u)\br]_{\pm}\ket{u,\lambda}^{(c)}\cr
    =&q_{c}^{-1/2}\left[\frac{\mathcal{Y}^{(c)}_{\lambda}(q_{c}^{-1}z,u)}{\mathcal{Y}^{(c)}_{\lambda}(z,u)}\right]_{\pm}\ket{u,\lambda}^{(c)},\cr
    E(z)\ket{u,\lambda}^{(c)}=&\mathcal{E}\sum_{i=1}^{\ell(\lambda)+1}\prod_{l=1}^{i-1}\left[S_{(c-1)}\bl(uq_{c+1}^{l-1}q_{c-1}^{\lambda_{l}}/z\br)\right]_{-}\delta\bl(uq_{c+1}^{i-1}q_{c-1}^{\lambda_{i}}/z\br)\ket{u,\lambda+\Box_{i}}^{(c)},\cr
    F(z)\ket{u,\lambda}^{(c)}=&\mathcal{F}\sum_{i=1}^{\ell(\lambda)}\prod_{l=i+1}^{\infty}\left[S_{(c-1)}\bl(uq_{c+1}^{l-1}q_{c-1}^{\lambda_{l}}/z\br)\right]_{+}\delta\left(uq_{c+1}^{i-1}q_{c-1}^{\lambda_{i}-1}/z\right)\ket{u,\lambda-\Box_{i}}^{(c)}
\eea
where
\bea
    \mathcal{Y}_{\lambda}^{(c)}(z,u)=&(1-u/z)\prod_{x\in\lambda}S_{c}(\chi_{x}/z),\quad \chi_{x}=uq_{c+1}^{i-1}q_{c-1}^{j-1}\,\,(i,j\geq1 ),\label{def-Yc}\cr
    \Psi_{\lambda}^{(c)}(z,u)=&\prod_{i=1}^{\infty}S_{(c-1)}(uq_{c+1}^{i-1}q_{c-1}^{\lambda_{i}}/z).
\eea
The infinite products are regularized in the following way (see Appendix \ref{sec:appendix-combinatorial} for details)
\bea
    \Psi_{\lambda}^{(c)}(z,u)=&\prod_{i=1}^{\infty}S_{(c-1)}(uq_{c+1}^{i-1}q_{c-1}^{\lambda_{i}}/z)\cr
    =&\prod_{i=1}^{\infty}\frac{(1-q_{c+1}q_{c+1}^{i-1}u/z)(1-q_{c}q_{c+1}^{i-1}u/z)}{(1-q_{c+1}^{i-1}u/z)(1-q_{c+1}^{i-1}q_{c-1}^{-1}u/z)}\prod_{x\in\lambda}g\left(\frac{z}{\chi_{x}}\right)\cr
    =&q_{c}^{-1/2}\frac{1-q_{c}u/z}{1-u/z}\prod_{x\in\lambda}g\left(\frac{z}{\chi_{x}}\right)\cr
    =&q_{c}^{-1/2}\frac{1-q_{c}u/z}{1-u/z}\prod_{x\in\lambda}\frac{S_{c}(q_{c}z/\chi_{x})}{S_{c}(z/\chi_{x})}\cr
    =&q_{c}^{-1/2}\frac{\mathcal{Y}^{(c)}_{\lambda}(q_{c}^{-1}z,u)}{\mathcal{Y}_{\lambda}^{(c)}(z,u)}
\label{eq:Fock-eigenvalue}
\eea
which eventually gives only a finite number of products. From the second to the third line, we regularize in a way that gives the extra $q_{c}^{-1/2}$ in front. Similarly, the coefficient in the right-hand side of the action of $F(z)$
\bea
    \prod_{l=i+1}^{\infty}S_{(c-1)}(uq_{c+1}^{l-1}q_{c-1}^{\lambda_{l}}/z)
\eea will be a finite number of products after nontrivial cancellations of poles and zeros.

\begin{figure}[t]
\centering
 \includegraphics{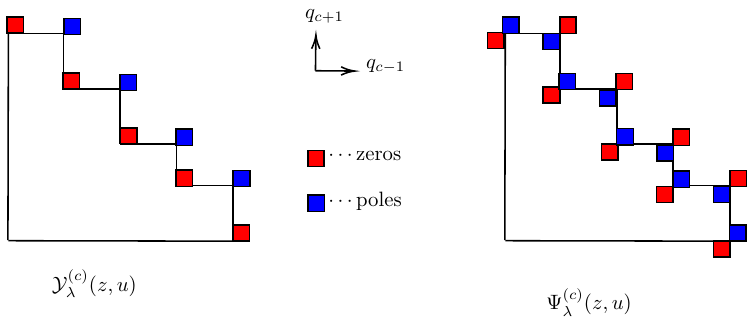}
\caption{The zero and pole structure of $\mathcal{Y}_{\lambda}^{(c)}(z,u)$ and $\Psi^{(c)}_{\lambda}(z,u)$. The red boxes depict the zeros, while the blue boxes depict the poles of the functions.}
\label{fig:pole-zero-Cartan-Fock}
\end{figure}
Further simplification of the expression is also possible (see Figure \ref{fig:pole-zero-Cartan-Fock})
\begin{align}\label{corner_reduction}
   \Psi_{\lambda}^{(c)}(z,u)
   =&q_{c}^{-1/2}\prod_{\sAbox\in \frakA(\lambda)}\frac{1-u q_c \boldsymbol{q}_c(\Abox) /z}{1-u \boldsymbol{q}_c({\Abox})/z}
   \prod_{\sAbox\in  \frakR(\lambda)}\frac{1-u q_c^{-1} \boldsymbol{q}_c({\Abox}) /z}{1-u \boldsymbol{q}_c({\Abox})/z},
\end{align}
where we use 
\begin{align}
    \mathcal{Y}^{(c)}_{\lambda}(z,u)=(1-u/z)\prod_{\sAbox\in\lambda}S_{c}(\chi_{\sAbox}/z)=\frac{\prod\limits_{\sAbox\in \frakA(\lambda)}(1-\chi_{\sAbox}/z)}{\prod\limits_{\sAbox\in \frakR(\lambda)}(1-q_{c}^{-1}\chi_{\sAbox}/z)}=\frac{\prod\limits_{\sAbox\in \frakA(\lambda)}(1-u\boldsymbol{q}_{c}(\Abox)/z)}{\prod\limits_{\sAbox\in \frakR(\lambda)}(1-uq_{c}^{-1}\boldsymbol{q}_{c}(\Abox)/z)}.
\end{align}
This formula is essential to describe the singular vector in the module. Here we introduce a notation $\boldsymbol{q}_c({\Abox}):=q_{c+1}^{x(\sAbox)} q_{c-1}^{y(\sAbox)}$ where $x(\Abox)$ (resp. $y(\Abox)$) is the $x$-coordinate (resp. $y$-coordinate) of a box $\square$ in the Young diagram where the coordinate of the bottom left corner is $(0,0)$.
$\frakA(\lambda)$ (resp. $\frakR(\lambda)$) implies the addable (resp. removable) box of the Young diagram $\lambda$ (see \S\ref{sec:appendix-Youngdiagram} for the notations).

For the generalization in the following, it will be more convenient to change the normalization of the basis, $\ket{u,\lambda}^{(c)}$ to be orthonormal, which will be denoted as $|u,\lambda\rrangle^{(c)}$. (See \eqref{normalization} for the normalization.)
As its name suggests, they satisfy the inner product
\begin{equation}
    \llangle u, \lambda| u, \lambda'\rrangle= \delta_{\lambda,\lambda'}\,.
\end{equation}
In such basis, the defining relations (\ref{QTA-Fock}) are simplified to\footnote{See for example \cite{Prochazka:2015deb} for the discussion of the symmetrized coefficients in the affine Yangian $\mathfrak{gl}_{1}$ case.},
\begin{align}\label{QTA-Fock2-E}
    E(z)|u,\lambda\rrangle^{(c)}\propto &\sum_{\sAbox\in \frakA(\lambda)}  \delta\bl(u \boldsymbol{q}_c({\Abox})/z\br) \sqrt{\underset{z=u\boldsymbol{q}_c({\sAbox})}{\mathrm{Res}} z^{-1}\Psi^{(c)}_\lambda(z,u)}\,|u,\lambda+\Box_{i}\rrangle^{(c)},\\
    F(z)|u,\lambda\rrangle^{(c)}\propto&\sum_{\sAbox\in \frakR(\lambda)}  \delta\bl(u \boldsymbol{q}_c({\Abox})/z\br) \sqrt{\underset{z=u\boldsymbol{q}_c({\sAbox})}{\mathrm{Res}} z^{-1}\Psi^{(c)}_\lambda(z,u)}\,|u,\lambda-\Box_{i}\rrangle^{(c)}\,
    \label{QTA-Fock2-F}
\end{align}
while the action of $K^{\pm}(z)$ remains the same, where $\underset{z=a}{\mathrm{Res}} f(z)=\lim\limits_{z\rightarrow a} (z-a)f(z)$.

\subsubsection{Tensor product of Fock representation}\label{sec:tensor_product}
One obtains a tensor product representation of the Fock modules by taking a coproduct of generators. We introduce $N$-tuples of colors $\boldsymbol{c}=(c_1,c_2,\ldots ,c_N)$, spectral parameters $\boldsymbol{u}=(u_1,u_2,\ldots, u_N)$, and partitions $\boldsymbol{\lambda}=(\lambda^{(1)},\lambda^{(2)},\ldots, \lambda^{(N)})$. The representation space is spanned by $|\boldsymbol{u},\boldsymbol{\lambda}\rangle^{(\boldsymbol{c})}=\otimes_{i=1}^N |u_i,\lambda_i\rangle^{(c_i)}$ for all possible $N$-tuples $\boldsymbol{\lambda}$ of partitions and the generators act on them as in (\ref{QTA-Fock}), where $\lambda$ and $c$ are replaced by $\boldsymbol{\lambda}$ and $\boldsymbol{c}$, respectively.
It is convenient to express the representation through the orthonormal frame, where we need to make obvious replacements
\begin{align}
 |u,\lambda\rrangle^{(c)}   &\rightarrow |\boldsymbol{u},\boldsymbol{\lambda}\rrangle^{(\boldsymbol{c})}\,,\\
 \Psi_{\lambda}^{(c)}(z,u) & \rightarrow \Psi_{\boldsymbol{\lambda}}^{(\boldsymbol{c})}(z,\boldsymbol{u}):=
 \prod_{i=1}^N q_{c_{i}}^{-1/2}\prod_{\sAbox\in \frakA(\lambda^{(i)})}\frac{1-u_i q_{c_i} \boldsymbol{q}_{c_i}({\Abox}) /z}{1-u_i \boldsymbol{q}_{c_i}({\Abox})/z}
   \prod_{\sAbox\in  \frakR(\lambda^{(i)})}\frac{1-u_i q_{c_i}^{-1} \boldsymbol{q}_{c_i}({\Abox}) /z}{1-u_i \boldsymbol{q}_{c_i}({\Abox})/z}\,,\label{PsiTP}
\end{align}
together with
\bea\label{QTA-Fock3-F}
K^{\pm}(z)|\boldsymbol{u},\boldsymbol{\lambda}\rrangle^{(\boldsymbol{c})}=&\bl[\Psi_{\boldsymbol{\lambda}}^{(\boldsymbol{c})}(z,\boldsymbol{u})\br]_{\pm}|\boldsymbol{u},\boldsymbol{\lambda}\rrangle^{(\boldsymbol{c})}\cr
E(z)|\boldsymbol{u},\boldsymbol{\lambda}\rrangle^{(\boldsymbol{c})}\propto &\sum_{i=1}^N\sum_{\sAbox\in \frakA(\lambda^{(i)})}  \delta\bl(u_i \boldsymbol{q}_{c_i}({\Abox})/z\br) \sqrt{\underset{z=u_i\boldsymbol{q}_{c_i}({\sAbox})}{\mathrm{Res}}z^{-1} \Psi^{(c)}_\lambda(z,\boldsymbol{u})}\,|\boldsymbol{u},\boldsymbol{\lambda}+\Box\rrangle^{(\boldsymbol{c})},\cr
F(z)|\boldsymbol{u},\boldsymbol{\lambda}\rrangle^{(\boldsymbol{c})}\propto&\sum_{i=1}^N\sum_{\sAbox\in \frakR(\lambda^{(i)})}  \delta\bl(u_i \boldsymbol{q}_{c_i}({\Abox})/z\br) \sqrt{\underset{z=u_i\boldsymbol{q}_{c_i}({\sAbox})}{\mathrm{Res}}z^{-1} \Psi^{(c)}_\lambda(z,\boldsymbol{u})}\,|\boldsymbol{u},\boldsymbol{\lambda}-\Box\rrangle^{(\boldsymbol{c})}\,.
\eea

The Fock module of ($q$-)$\cW_N$ algebra (with the $\U(1)$ factor) is obtained by taking $c_1=c_2=\cdots =c_N=3$. For the general combination of $\boldsymbol{c}$, one obtains representations of the $q$-deformed corner VOA \cite{Gaiotto:2017euk,Prochazka:2018tlo}. 

\subsubsection{MacMahon representation}\label{sec:MacMahonrep}

The MacMahon representation, denoted by $\mathcal{M}(u,K)$, is obtained by taking tensor products of Fock representations. The central charge of $\mathcal{M}(u,K)$ is $(C,K^{-})=(1,K)$, where $K\in\mathbb{C}^{\times}$ is a generic parameter. The basis of the MacMahon representation is labelled by a plane partition which is defined as a sequence of integers satisfying the following condition:
\bea
    \Lambda=(\Lambda_{i,j})_{i,j\in\mathbb{Z}_{>0}},\quad \Lambda_{i,j}\in\mathbb{Z}_{\geq 0},\quad \Lambda_{i,j}\geq \Lambda_{i+1,j},\quad \Lambda_{i,j}\geq \Lambda_{i,j+1}.
\eea
We can understand it as slices of Young diagrams as
\bea
    \Lambda=(\Lambda^{(1)},\Lambda^{(2)},\Lambda^{(3)},\ldots),\qquad , \Lambda^{(1)}\supseteq\Lambda^{(2)}\supseteq\Lambda^{(3)}\cdots
\eea
To simplify the context, we focus on tensor products of $\mathcal{F}_{3}(u)$ representations:
\bea
    \mathcal{M}(u,K)=\mathcal{F}_{3}(u)\otimes \mathcal{F}_{3}(uq_{3})\otimes \mathcal{F}_{3}(uq_{3}^{2})\otimes \cdots,\cr
    |u,\Lambda\rrangle=|u,\Lambda^{(1)}\rrangle\otimes |u q_3,\Lambda^{(2)}\rrangle\otimes |u q_{3}^{2},\Lambda^{(3)}\rrangle\otimes \cdots.
\eea
The action of the generators on MacMahond representation can be obtained using the coproduct formula in (\ref{eq:coproduct}). For example, the eigenvalue of $K^{\pm}(z)$ is given as the infinite product
\bea
    K^{\pm}(z)|u,\Lambda\rrangle=\left[K_{\Lambda}(z,u)\right]_{\pm}|u,\Lambda\rrangle,\quad K_{\Lambda}(z,u)=\prod_{i=1}^{\infty}\Psi_{\Lambda^{(i)}}(z,uq_{3}^{i-1})\label{eq:MacMahonCartan}
\eea
where we omit the superscript $(3)$ of $\Psi^{(3)}_{\Lambda^{(i)}}(z)$ in what follows. Using (\ref{eq:Fock-eigenvalue}),  we can express $K_{\Lambda}(z,u)$ as
\bea
    K_{\Lambda}(z,u)=&\prod_{k=1}^{\infty}q_{3}^{-1/2}\frac{1-q_{3}^{k}u/z}{1-q_{3}^{k-1}u/z}\prod_{\sAbox\in\Lambda^{(k)}}g\left(\frac{z}{uq_{3}^{k-1}\boldsymbol{q}(\Abox)}\right)\cr
    =&\lim_{N\rightarrow \infty}q_{3}^{-N/2}\frac{1-q_{3}^{N}u/z}{1-u/z}\prod_{\sAbox\in\Lambda}g\left(\frac{z}{u\boldsymbol{q}(\Abox)}\right)~
\eea
where\footnote{Note that when $\Abox=(i,j)\in\Lambda^{(k)}$, we have $\boldsymbol{q}(\Abox)=q_{1}^{i-1}q_{2}^{j-1}$} $\boldsymbol{q}(\Abox)=q_{1}^{i-1}q_{2}^{j-1}q_{3}^{k-1}$ for $\Abox=(i,j,k)\in\Lambda$, and $g$ is the structure function given in \eqref{str-fn-QTA}.
To regularize this expression, we formally replace $q_{3}^{N}$ with an arbitrary parameter $K$ to obtain
\bea
    K_{\Lambda}(z,u)=K^{-1/2}\frac{1-Ku/z}{1-u/z}\prod_{\sAbox\in\Lambda}g\left(\frac{z}{u\boldsymbol{q}(\Abox)}\right).\label{eq:PPCartan}
\eea
\begin{figure}
\centering
\includegraphics[width=0.9\textwidth]{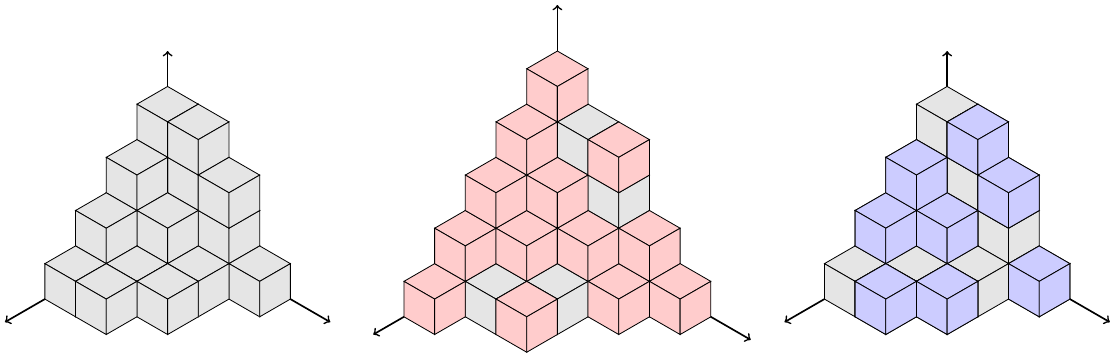}
\caption{Left: Example of a plane partition $\Lambda$. Middle: The red boxes are addable boxes to the plane partition. The set of them is denoted by $\frakA(\Lambda)$.  Right: The blue boxes are removable boxes from the configuration. The set of them is denoted by $\frakR(\Lambda)$. }
\label{fig:planepartition}
\end{figure}

The numerators and denominators of the product of the functions $g(x)$ cancel out with each other, leading to a rewriting in terms of contributions from the \emph{surface} of the plane partition $\Lambda$ \cite{feigin2012quantum}. Eventually, we obtain a formula similar to (\ref{corner_reduction}), where the boxes in the surface correspond to the zeros and poles of the eigenfunction (refer to \cite{feigin2012quantum} for the explicit formula). In fact, the poles of the eigenfunction $K_{\Lambda}(z,u)$ are solely determined by the positions of the addable boxes and removable boxes:
\begin{align}
    K_{\Lambda}(z,u)\propto \prod_{\sAbox\in \frakA(\Lambda)}(1-u\boldsymbol{q}(\Abox)/z)^{-1}\prod_{\sAbox\in \frakR(\Lambda)}(1-u\boldsymbol{q}(\Abox)/z)^{-1}\label{eq:MacMahonpole}
\end{align}
where $\frakA(\Lambda)$ and $\frakR(\Lambda)$ are addable and removable boxes of the plane partition, respectively (see Figure \ref{fig:planepartition}).

The action of the other currents $E(z)$ and $F(z)$ on the MacMahon representation can be obtained similarly using the coproduct formula, albeit with some redefinition of the normalizations. However, the resulting normalization factors are rather intricate, so we present a schematic expression instead of providing the explicit formula:
\bea
    E(z)|u,\Lambda\rrangle\propto &\sum_{\sAbox\in \frakA(\Lambda)}\delta\left(\frac{z}{u\boldsymbol{q}(\Abox)}\right)
    \sqrt{\underset{z=u \boldsymbol{q}(\sAbox)}{\mathrm{Res}}z^{-1} K_{\Lambda}(z,u)}\,
    |u,\Lambda+\Abox\rrangle,\cr
    F(z)|u,\Lambda\rrangle \propto&\sum_{\sAbox\in \frakR(\Lambda)}\delta\left(\frac{z}{u\boldsymbol{q}(\Abox)}\right)
    \sqrt{\underset{z=u \boldsymbol{q}(\sAbox)}{\mathrm{Res}} z^{-1}K_{\Lambda}(z,u)}\,
    |u,\Lambda-\Abox\rrangle.\label{eq:MacMahonEFaction}
\eea
As before, the current $E(z)$ adds a box to the plane partition configuration while $F(z)$ removes one. Moreover, the coefficients are proportional to the residue of the function $z^{-1}K_\Lambda(z,u)$.

\paragraph{MacMahon representation with asymptotic Young diagrams}
\begin{figure}
    \centering
    \begin{minipage}{0.45\linewidth}
    \includegraphics[width=8cm]{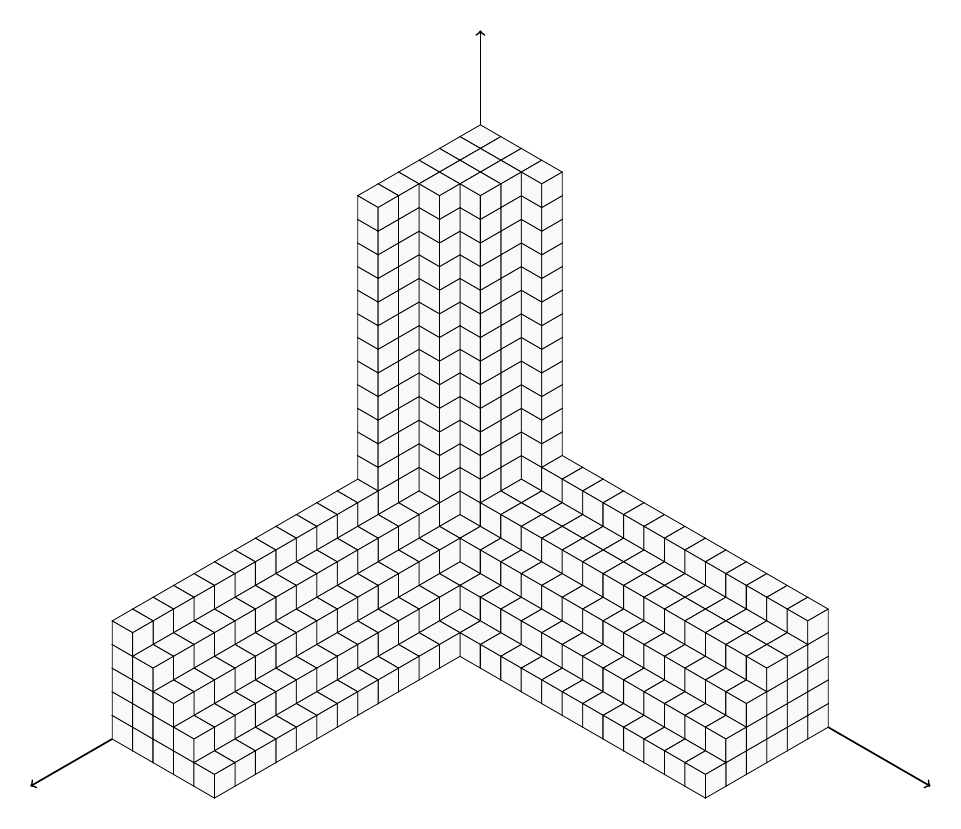}
    \end{minipage}
    \begin{minipage}{0.45\linewidth}
    \includegraphics[width=8cm]{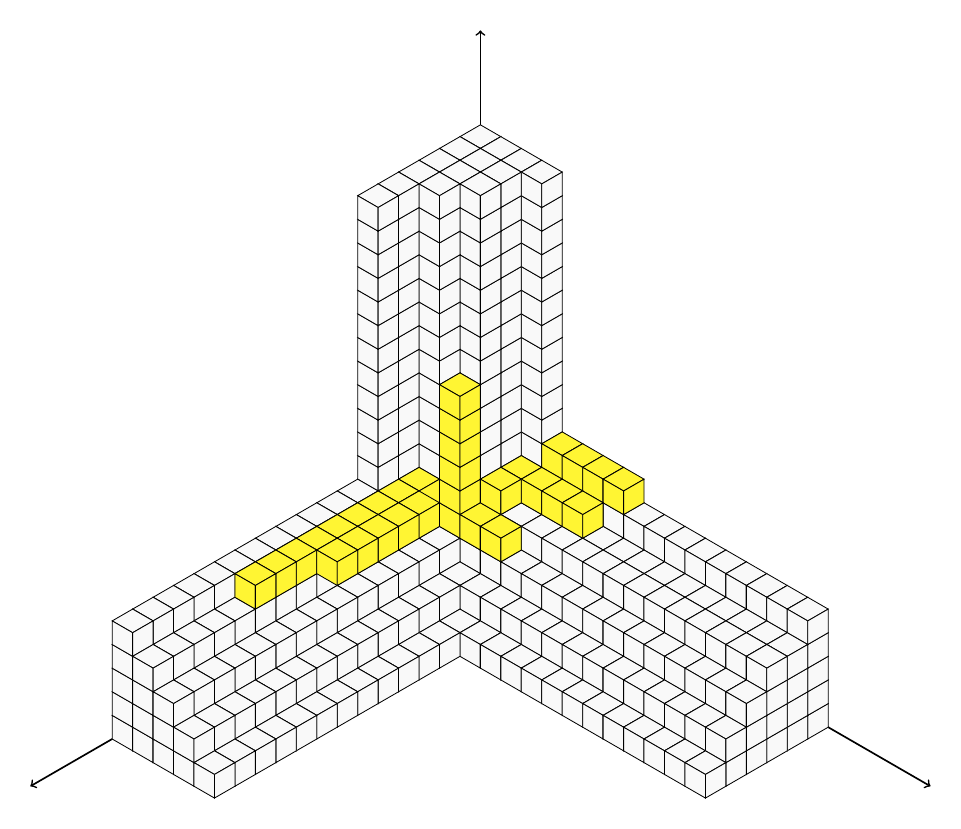}
    \end{minipage}
    \caption{Left: vacuum configuration of a plane partition with asymptotic Young diagrams. Right: an example of a plane partition with asymptotic Young diagrams where the yellow boxes are boxes added to the vacuum configuration.}
    \label{fig:asympplanepartition}
\end{figure}

We can introduce a slightly different version of the plane partition, which involves asymptotic Young diagrams $\lambda_1, \lambda_2,\lambda_3$ appearing at infinity in the $x,y,z$ directions, respectively (see Figure \ref{fig:asympplanepartition}). 

Similar to  (\ref{eq:MacMahonCartan}) and (\ref{eq:MacMahonEFaction}), for plane partitions with asymptotic Young diagrams, the situation remains the same, and the eigenvalue $K_{\Lambda}(z,u)$ obtained from the formula (\ref{eq:PPCartan}) governs the action of the Drinfeld currents. The key difference lies in the vacuum configuration. In the MacMahon representation, where the vacuum configuration contains no boxes, the vacuum charge is simply given by
\begin{align}
    K_{\emptyset}(z,u)=K^{-1/2}\frac{1-Ku/z}{1-u/z}.
\end{align}
On the other hand, the vacuum configuration with asymptotic Young diagrams, denoted by $S_{\lambda_{1}\lambda_{2}\lambda_{3}}$, is not empty, so we count boxes in the vacuum configuration
\begin{align}
K^{\text{vac}}_{\lambda_{1}\lambda_{2}\lambda_{3}}(z,u)=K^{-1/2}\frac{1-Ku/z}{1-u/z}\prod_{\sAbox\in S_{\lambda_{1}\lambda_{2}\lambda_{3}}}g\left(\frac{z}{u\boldsymbol{q}(\Abox)}\right)\label{eq:asympvacfunction}
\end{align}
See the left panel of Figure \ref{fig:asympplanepartition}.

Let us see the simplest example where $(\lambda_{1},\lambda_{2},\lambda_{3})=(\Abox,\emptyset,\emptyset)$:
\begin{align}
\begin{split}
    K^{\text{vac}}_{\Abox,\emptyset,\emptyset}(z,u)&=K^{-1/2}\frac{1-Ku/z}{1-u/z}\prod_{i=1}^{\infty}g\left(\frac{z}{uq_{1}^{i-1}}\right)\\
    &=K^{-1/2}\frac{1-Ku/z}{1-u/z}\prod_{i=0}^{\infty}\frac{(1-uq_{1}^{i-1}/z)(1-uq_{1}^{i}q_{2}^{-1}/z)(1-uq_{1}^{i}q_{3}^{-1}/z)}{(1-uq_{1}^{i+1}/z)(1-uq_{1}^{i}q_{2}/z)(1-uq_{1}^{i}q_{3}/z)}\\
    &=K^{-1/2}\frac{(1-Ku/z)(1-uq_{1}^{-1}/z)}{(1-uq_{2}/z)(1-uq_{3}/z)},\\
    S_{\Abox,\emptyset,\emptyset}&=\{(x_{1},x_{2},x_{3})=(i,1,1)\,|\, i=1,2,\ldots\}
\end{split}
\end{align}
The poles $z=uq_{2},uq_{3}$ are the addable boxes of this vacuum configuration. For other examples, we refer to \cite{feigin2012quantum}.

\paragraph{MacMahon representation with a pit}
\begin{figure}
    \centering
    \includegraphics[width=15cm]{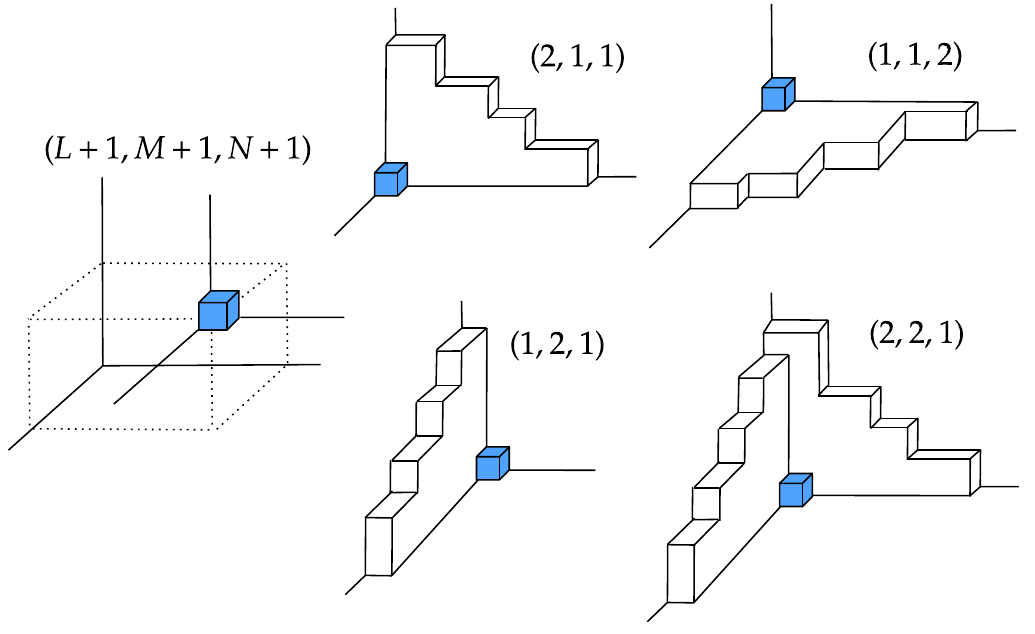}
    \caption{Plane partitions with a pit. The blue box is the pit that restricts the plane partition from growing further. }
    \label{fig:pit}
\end{figure}
The action of current $E$ in (\ref{eq:MacMahonEFaction}) and the pole structure (\ref{eq:MacMahonpole}) show that we can add a box to the configuration if there are poles in the eigenfunction $K_\Lambda(z,u)$. Moreover, the coefficient of the action of $E$ is proportional to the residue of the eigenfunction. When $q_1, q_2, q_3$ satisfies
\begin{equation}\label{pit_condition}
    q_1^L q_2^M q_3^N =K
\end{equation}
for some non-negative integers  $L,M,N$, the numerator $(1-Ku/z)$ cancels out with a pole coming from the denominator $(1-uq_{1}^{L}q_{2}^{M}q_{3}^{N}/z)$. Consequently, there is no pole at this position, and the residue vanishes. Thus, the states $|u,\Lambda\rrangle$ where $\Lambda$ contains a box at $(L+1,M+1,N+1)$ cannot be created by the action of $E(z)$. The position of such a box is called a ``\emph{pit}'' (see Figure \ref{fig:pit}). Therefore, under the condition \eqref{pit_condition}, a plane partition with a box at the pit position does not arise in the representation. For instance, if the location of the pit is $(1,1,N+1)$, the height of the plane partition is limited to $N$. If we decompose each layer as the Fock spaces, one may identify it with the $q$-$\mathcal{W}_N$ module. For the pit condition \eqref{pit_condition}, the algebra which generates such a module is referred to as the corner vertex operator algebra ($q$-deformed version of) $Y_{LMN}$ in \S\ref{sec:corner}.

Because of $q_1 q_2 q_3=1$, once (\ref{pit_condition}) is satisfied, other combination $(L',M',N')= (L+m, M+m, N+m)$ satisfies the same equation. This implies $Y_{LMN}\sim Y_{L+m,M+m, N+m}$, which is the shift symmetry \eqref{shift} of the corner VOA.
In particular, when $L=M=N$, we have the identification $Y_{LLL}\sim Y_{000}$, which implies the algebra is trivial.
As seen in \S\ref{sec:corner}, $Y_{LMN}$ can be constructed via the Miura transformation with $L+M+N$ free bosons, it follows that $Y_{111}$, described by three bosons,
is trivial \cite{Prochazka:2018tlo}, \cite{Harada:2021xnm}.

\subsubsection{Characters of representations and partition functions}\label{sec:characMacMahon}
As discussed in the previous sections, the basis of the vertical representations is spanned by $d$-dimensional partitions ($d=1,2,3$). Hence, the character of a vertical representation amounts to a box-counting problem of such combinatorial configurations and can be defined schematically as
\begin{equation}
    Z(\mathfrak{q})=\sum_{\substack{\Lambda:\,\text{possible}\\\text{ configurations}}}\mathfrak{q}^{|\Lambda|}
\end{equation}
where $|\Lambda|$ is the number of boxes in the configuration. Using the operator $\sfL_{0}=\frac{1}{2}\psi_{2}$ of the affine Yangian $\mathfrak{gl}_{1}$ in (\ref{eq:AYCFTcorr}), the partition can also be written as a trace formula\footnote{This is because $\psi_{2}\ket{\Lambda}=2|\Lambda|\ket{\Lambda}$. See \cite[Eq.(4.76)]{Prochazka:2015deb} for the 2d Young diagram case.}
\begin{equation}
    Z(\mathfrak{q})=\Tr \mathfrak{q}^{|\sfL_{0}|}~.
\end{equation}
In the $q$-deformation language, we can use the degree operator defined in (\ref{eq:gradingop1}), (\ref{eq:gradingop2}) to define the character as
\begin{equation}\label{q-1d-character}
    Z(\mathfrak{q})=\Tr\mathfrak{q}^{d_{1}}~.
\end{equation}
For a basis element $\ket{\Lambda}$, it can be constructed via the vertical representation as
\begin{align}
    \ket{\Lambda}\propto \oint\cdots\oint dz_{|\Lambda|}\cdots dz_{1}E(z_{|\Lambda|})E(z_{|\Lambda|-1})\cdots E(z_{1})\ket{\emptyset},\quad d_{1}\ket{\emptyset}=0~,
\end{align}
where the contour of the multi-integral is properly chosen.
Then, it follows from the property $\mathfrak{q}^{d_{1}}E(z)=\mathfrak{q}E(z)\mathfrak{q}^{d_{1}}$ in (\ref{eq:gradingop2}) that $d_{1}\ket{\Lambda}=|\Lambda|\ket{\Lambda}$.  Therefore, \eqref{q-1d-character} counts the number of boxes of $\Lambda$.

We denote the grand canonical ensemble partition function as
\begin{align}
    Z_{d}(\mathfrak{q})=\sum_{n=0}^{\infty}p_{d}(n)\mathfrak{q}^{n}
\end{align}
where $p_{n}(d)$ is the number of $d$-dimensional partitions with size $n$. Here, we only focus on $d=1,2,3$. For discussions on $d=4$, see for example \cite{Nekrasov:2017cih,Nekrasov:2018xsb}. 

The $d=1$ partition is a sequence of boxes starting from the origin and extending to the right as in (\ref{eq:1dpartition}). Note this $d=1$ partition is different from the vector representation introduced in \S\ref{sec:vectorrep}. The vector representation is not the highest weight representation, but the $d=1$ partition gives the highest weight representation\footnote{Strictly speaking, this representation is not a representation of the quantum toroidal $\mathfrak{gl}_{1}$ but \emph{shifted} quantum toroidal $\mathfrak{gl}_{1}$. We also discuss the 1d partition because it has similar properties with the other 2d and 3d partitions.}.  The $d=2$ partitions correspond to Young diagrams introduced in \S\ref{sec:verticalFockrep} while the $d=3$ partitions are known as plane partitions. Finally, the $d=4$ partitions are called \emph{solid partitions} \cite{Nekrasov:2017cih,Nekrasov:2018xsb}.
\paragraph{One-dimensional partition}
A one-dimensional partition is a configuration like
\begin{align}
\includegraphics{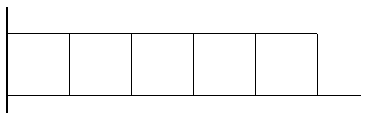}\label{eq:1dpartition}
\end{align}
where the above configuration has only 5 boxes. Since we can only add boxes in a one-dimensional way, the configuration having $k\geq 0$ boxes is unique. Thus, the character is 
\begin{align}
    Z_{1}(\mathfrak{q})=1+\mathfrak{q}+\mathfrak{q}^{2}+\cdots \mathfrak{q}^{k}+\cdots=\frac{1}{1-\mathfrak{q}}.
\end{align}

Actually, this character is solely the partition function of the harmonic oscillator 
\begin{align}
\begin{split}
&\sfJ,\, \sfJ^{\dagger},\quad [\sfJ,\sfJ^{\dagger}]=1,\quad  H=\sfJ^{\dagger}\sfJ,\\
&Z_{1}(\mathfrak{q})=\Tr \mathfrak{q}^{H}.
\end{split}
\end{align}
The eigenstate $\ket{n}\propto (\sfJ^{\dagger})^{n}\ket{0}$ corresponds with the 1d partition with $n$ boxes. 

\paragraph{Young diagram}
A Young diagram $\lambda=(\lambda_{1},\lambda_{2},\ldots)$ can be rewritten as $\lambda=(r^{m_{r}},\ldots,2^{m_{2}},1^{m_{1}}),$  where $m_{i}\geq 0$ indicates the number of times each integer occurs as a part. Thus, the character is obtained as 
\begin{align}\label{inv-Dedekind}
\begin{split}
    Z_{2}(\mathfrak{q})&=\sum_{m_{1},m_{2},m_{3},\ldots=0}^{\infty}\mathfrak{q}^{\sum_{n=1}^{\infty}nm_{n}}=\prod_{n=1}^{\infty}\sum_{m_{n}=1}^{\infty}\mathfrak{q}^{nm_{n}}\\
   &=\prod_{n=1}^{\infty}\frac{1}{1-\mathfrak{q}^{n}}=\frac{1}{(\frakq)_{\infty}}=1+\mathfrak{q}+2\frakq^{2}+3\frakq^{3}+5\frakq^{4}+7\frakq^{5}+11\frakq^{6}+15\frakq^{7}+\mathcal{O}(\frakq^{8}),
\end{split}
\end{align}
where $(z;q)_{\infty}=\prod_{k=0}^{\infty}(1-zq^{k})$ and $(q)_{\infty}=(q;q)_{\infty}$. Similar to the 1d partition case, we can relate this with the partition function of the free bosons:
\begin{align}
\begin{split}
    &[\sfJ_{n},\sfJ_{m}]=n\delta_{n+m,0},\quad H=\sum_{n>0}\sfJ_{-n}\sfJ_{n},\\
    &Z_{2}(\mathfrak{q})=\Tr \mathfrak{q}^{H}.
\end{split}
\end{align}
Explicitly, the Young diagram $\lambda=(\lambda_{1},\lambda_{2},\ldots)$ corresponds to the following element of the Hilbert space (see Appendix \ref{app:Macdonald}): $\sfJ_{-\lambda}\ket{0}=\sfJ_{-\lambda_{1}}\sfJ_{-\lambda_{2}}\cdots\ket{0}$.

\paragraph{Plane partition}
The generating function of plane partitions is the MacMahon function \cite{macmahon1899partitions,macmahon1915combinatory}
\begin{equation}
	Z_{3}(\mathfrak{q})=\prod_{n=1}^\infty \frac{1}{(1-\mathfrak{q}^n)^{n}}=1+\frakq+3\frakq^2+6\frakq^3+13 \frakq^4+24\frakq^5+48\frakq^6+86 \frakq^7+\mathcal{O}\left(q^{8}\right).
\end{equation}
We can check at a low level this is true:
\begin{itemize}
    \item level 0: vacuum (empty configuration)
    \item level 1: 
$$\includegraphics{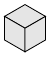}$$
    \item level 2:
$$\includegraphics{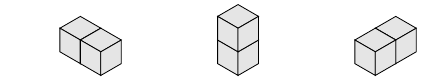}$$
\item level 3:
$$\includegraphics{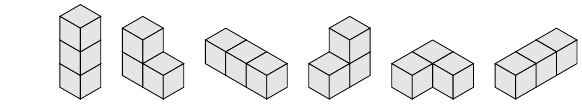}$$
\item level 4:
$$\begin{aligned}
\includegraphics{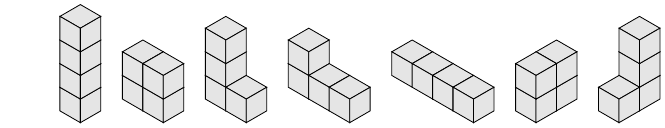}\\ 
\includegraphics{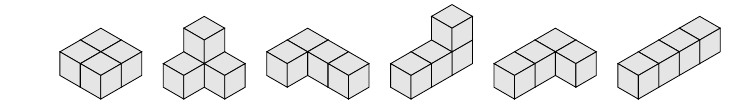}
\end{aligned}$$
\end{itemize}

The plane partition can be realized through a 3d oscillator, denoted as $\sfJ_{r,s}$ where  $r,s\in \mathbb{Z} + \frac{1}{2}$, as discussed in recent works concerning the relationship between 3d free bosons and plane partitions (see for example, \cite{Zenkevich:2017tnb, wang20233d}).
The commutation relation is given by
\begin{equation}
    \left[\sfJ_{r_1,s_1}, \sfJ_{r_2,s_2}\right] = \delta_{r_1 + r_2, 0} \, \delta_{s_1 + s_2, 0},
\end{equation}
Additionally, this is accompanied by a degree operator defined by:
\begin{equation}
    \left[D, \sfJ_{r,s}\right] = -(r + s)\sfJ_{r,s}.
\end{equation}
The vacuum state is defined by
\begin{equation}
    \sfJ_{r,s}|\emptyset\rangle = 0 \quad \text{for} \quad r > 0 \quad \text{or} \quad s > 0~.
\end{equation}
Then, the sequence of non-vanishing states can be represented as follows:
\begin{align*}
    \text{Level 0}:\quad & |\emptyset\rangle, \\
    \text{Level 1}:\quad & \sfJ_{-1/2,-1/2}|\emptyset\rangle, \\
    \text{Level 2}:\quad & (\sfJ_{-1/2,-1/2})^2|\emptyset\rangle,\quad
    \sfJ_{-1/2,-3/2}|\emptyset\rangle,\quad
    \sfJ_{-3/2,-1/2}|\emptyset\rangle,
\end{align*}
Consequently, the generating function becomes the MacMahon function:
\begin{equation}
    \text{Tr}(\mathfrak{q}^D) = \prod_{n=1}^\infty (1-\mathfrak{q}^n)^{-n}.
\end{equation}
Interestingly, the same generating function emerges when substituting $\sfJ_{r,s}$ with $E_{r,s}$ in the $\mathfrak{gl}_\infty$ algebra, given the commutation relation specified in (\ref{glinfty}) and a particular choice for the highest weight condition as in (\ref{HWC_glinfty}). We observe that the constant $C$ must be a non-integer in order to derive a valid plane partition.

\paragraph{Plane partition with asymptotic Young diagram} 
The generating function for plane partitions with asymptotic Young diagrams can be expressed as
\begin{equation}
	Z_{\lambda_{1}\lambda_{2}\lambda_{3}}(\frakq)=\frakq^{-\frac{1}{2}(\|\lambda_{1}\|^{2}+\|\lambda_{2}\|^{2}+\|\lambda_{3}\|^{2})}C_{\lambda_1\lambda_2\lambda_3}(\frakq)Z_{3}(\frakq)
\end{equation}
where $C_{\lambda_1\lambda_2\lambda_3}(\frakq)$ is the unrefined topological vertex \cite{AKMV}. Note that the definition of the refined topological vertex is given in (\ref{exp-top-v}), and we simply take the unrefined limit $t=q$ to obtain it. This formula can be derived by using a vertex operator as a transfer matrix \cite{Okounkov:2003sp} (also refer to Appendix \ref{app:Schur}).

For example when $(\lambda_{1},\lambda_{2},\lambda_{3})=(\Abox,\emptyset,\emptyset)$:
\begin{align}
    &C_{\Abox,\emptyset,\emptyset}(\frakq)=s_{\Abox}(\frakq^{-\rho}),\quad s_{\lambda}(\frakq^{-\rho})=\frac{\frakq^{\frac{||\lambda^{t}||^{2}}{2}}}{\prod\limits_{\sAbox\in\lambda}(1-\frakq^{h_{\lambda}(\sAbox)})}\\
    &Z_{\Abox,\emptyset,\emptyset}(\frakq)=\frac{1}{1-\frakq}\prod_{n=1}^{\infty}\frac{1}{(1-\frakq^{n})^{n}}=1+2\frakq+5\frakq^{2}+11\frakq^{3}+24\frakq^{4}+48\frakq^{5}+\cdots
\end{align}
where $\rho=(-1/2,-3/2,\ldots)$ and $h_{\lambda}(\Abox)$ is (\ref{eq:hooklength}). 

At low level, we have the following configurations:
\begin{itemize}
    \item level 0: asymptotic Young diagram $\Abox$ in the axis 1 extending semi-infinitely
$$\includegraphics{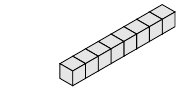}$$
\item level 1: 2 possible configurations with one yellow box
$$\includegraphics{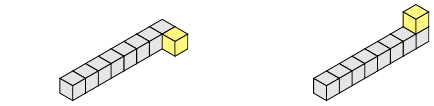}$$
\item level 2:
$$\includegraphics{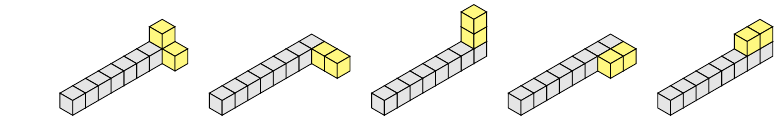}$$
\item level 3:
\begin{align}\nonumber
&\includegraphics{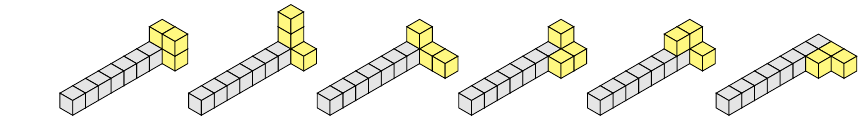}\cr
&\includegraphics{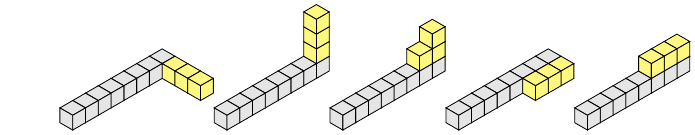}
\end{align}
\end{itemize}

\paragraph{Plane partition with a pit}
Due to the shift symmetry \eqref{shift}, we can set one of the parameters $L,M,N$ to be zero to give a nontrivial pit reduction. For our analysis, we set $N=0$ and focus on the case where a pit is located at $(n+1,m+1)$, assuming without loss of generality that $n\geq m$. To extend our investigation, we can also introduce asymptotic Young diagrams $(\mu,\nu,\lambda)$ to the three axes. We denote the corresponding generating function as  $Z_{\mu,\nu,\lambda}^{n,m}(\frakq)$. 

In \cite{bershtein2018plane}, three equivalent formulas for this generating function are presented. Here, we quote only one of the formulas for reference, and we refer to \cite[Theorem 2, Theorem 3]{bershtein2018plane} for the other two formulas.

\begin{theorem}\cite[Theorem 1]{bershtein2018plane}\\
The generating function $Z_{\mu,\nu,\lambda}^{n,m}(\frakq)$ is 
\begin{align}
    Z_{\mu,\nu,\lambda}^{n,m}(\frakq)=\frac{(-1)^{mn-r}\frakq^{\Delta^{n,m}_{\mu,\nu,\lambda}}}{(\frakq)_{\infty}^{m+n}}\det\left(\begin{array}{cc}
       \left(\sum_{a\geq 0}(-1)^{a}\frakq^{\binom{a+1}{2}}\frakq^{(N_{j}-M_{i})a}\right)_{\substack{1\leq i\leq m\\1\leq j\leq n}}  &  (\frakq^{-M_{i}Q_{j}})_{\substack{1\leq i\leq m\\1\leq j\leq m-r}}\\
        (\frakq^{-N_{j}(P_{i}+1)})_{\substack{1\leq i\leq n-r\\ 1\leq j\leq n}} & 0
    \end{array}\right)\label{eq:onepitcharacter}
\end{align}
where 
\begin{align}
\begin{split}
&r=\min\{s\,|\,\lambda_{n-s}\geq m-s\},\,\,(0\geq r\geq \min\{n,m\}),\\
&\pi_{i}=\lambda_{i}-(m-r),\,\,(i=1,\ldots,n-r),\quad \kappa_{j}=\lambda^{t}_{j}-(n-r), \,\,(j=1,\ldots,m-r),\\
&N_{i}=\nu_{i}+n-i,\quad M_{j}=\mu_{j}+m-j,\quad P_{i}=\pi_{i}+(n-r)-i,\quad Q_{j}=\kappa_{j}+(m-r)-j
\end{split}
\end{align}
and 
\begin{align}
    \Delta^{n,m}_{\mu,\nu,\lambda}=\sum_{j=1}^{m-r}M_{j}Q_{j}=\sum_{i=1}^{n-r}N_{i}(P_{i}+1).
\end{align}
\end{theorem}

The generating function of the plane partitions with a pit, and the trivial asymptotic Young diagrams  can be expressed in a simple form  (see \cite[(2.5)]{bershtein2018plane}) :
\bea
    &Z_{\emptyset,\emptyset,\emptyset}^{n,m}(\frakq)\\
    =&\frac{1}{(\frakq)^{n+m}_{\infty}}\sum_{\alpha_{1}\geq\alpha_{2}\geq\ldots\geq \alpha_{m}\geq 0}(-1)^{\sum_{i=1}^{m}\alpha_{i}}\frakq^{\sum_{i=1}^{m}\frac{1}{2}\alpha_{i}(\alpha_{i}+2i-1)}\prod_{1\leq i<j\leq m}(1-\frakq^{\alpha_{i}-\alpha_{j}-i+j})\prod_{1\leq i<j\leq n}(1-\frakq^{\alpha_{i}-\alpha_{j}-i+j}).
\eea
As an illustration, consider the case when $n=1$ and $m=0$, corresponding to the pit condition $(L,M,N)=(1,0,0)$, which restricts the Young diagrams to span the $(2,3)$-plane.  According to this formula, we have $Z^{1,0}_{\emptyset,\emptyset,\emptyset}(\frakq)=\frac{1}{(\frakq)_{\infty}}$ and it indeed matches with the generating function \eqref{inv-Dedekind} of the Young diagram.

\subsection{Horizontal representations}\label{sec:horizontalrep}
In this section, we introduce the horizontal representations (see \cite{Shiraishi:1995rp,Awata:1995zk,FF,FHSSY,miki2007q,Kojima:2020vtc,Kojima2019,Harada:2021xnm} for details).  For convenience, we adopt the normalization $\tilde{g}=1/\kappa_{1}$ because it is manifestly invariant under the triality. Roughly speaking, horizontal representations are constructed by assigning each of Drinfeld currents to a vertex operator so that the currents obey the defining relations in (\ref{eq:DIMdef}). Since these representations are described in free bosons, they are also called Fock representations. We denote the three Fock spaces $\mathcal{F}_{c}(u),\,(c=1,2,3)$, where $u$ is the spectral parameter.

The central charges of the representations are $(C,K^{-})=(q_{c}^{1/2},1),\,(c=1,2,3)$. Then, \eqref{eq:DIMdef} tells us that the currents $K^{\pm}(z)$ do not commute with each other, and they yield the $q$-Heisenberg relation of their modes in \eqref{Drinfeld} 
\bea
    &[\sfH_{r},\sfH_{s}]=\frac{r}{\kappa_{r}}(q_{c}^{r/2}-q_{c}^{-r/2})\delta_{r+s,0}.
\eea
Thus, $\sfH_{\pm r}(r>0)$ are the modes of the free bosons with $q$ parameters in the normalization.
We changed the normalization for later convenience as 
\bea
 &\sfH_{r}=\frac{\sfa^{(c)}_{r}}{q_{c}^{r/2}-q_{c}^{-r/2}},\quad [\sfa^{(c)}_{r},\sfa^{(c)}_{s}]=-\frac{r}{\kappa_{r}}(q_{c}^{r/2}-q_{c}^{-r/2})^{3}\delta_{r+s,0}.\label{eq:commuteqboson}
 \eea
Now, the horizontal representation maps the Drinfeld currents to vertex operators as\footnote{The construction here resembles the Frenkel-Kac construction of the affine $\mathfrak{sl}_{2}$ Lie algebra: $E(z)=:e^{i\varphi(z)}:$, $F(z)=:e^{-i\varphi(z)}:$, $H(z)=i\varphi'(z)$, where $\varphi(z)$ is the free-boson in (\ref{free_boson_ope}) and $E(z),F(z),H(z)$ correspond to the affine current of $\mathfrak{sl}_{2}$ Lie algebra with level $1$.
Here, we also follow the same path: We assign free bosons to the Cartan part, and then from the commutation relations, we find free-field representations of the other currents.}
\bea
    E(z)\rightarrow u\frac{(1-q_{c})}{\kappa_{1}}\eta_{c}(z),\quad F(z)\rightarrow -u^{-1}\frac{(1-q_{c}^{-1})}{\kappa_{1}}\xi_{c}(z),\quad K^{\pm}(z)\rightarrow \varphi^{\pm}_{c}(z),\label{eq:horizon_rep_level1}
\eea
where the vertex operators are defined by the $q$-Heisenberg modes $\sfa^{(c)}_{r}$
\bea\label{vertex-operator}
    \eta_{c}(z)=&\exp\left(\sum_{r>0}\frac{\kappa_{r}}{r}\frac{\sfa^{(c)}_{-r}}{(q_{c}^{r/2}-q_{c}^{-r/2})^{2}}z^{r}\right)\exp\left(\sum_{r>0}\frac{\kappa_{r}}{r}\frac{q_{c}^{-r/2}\sfa^{(c)}_{r}}{(q_{c}^{r/2}-q_{c}^{-r/2})^{2}}z^{-r}\right),\cr
    \xi_{c}(z)=&\exp\left(-\sum_{r>0}\frac{\kappa_{r}}{r}\frac{q_{c}^{r/2}\sfa^{(c)}_{-r}}{(q_{c}^{r/2}-q_{c}^{-r/2})^{2}}z^{r}\right)\exp\left(-\sum_{r>0}\frac{\kappa_{r}}{r}\frac{\sfa^{(c)}_{r}}{(q_{c}^{r/2}-q_{c}^{-r/2})^{2}}z^{-r}\right),\cr
    \varphi_{c}^{\pm}(z)=&\exp\left(-\sum_{r>0}\frac{\kappa_{r}}{r}\frac{\sfa^{(c)}_{\pm r}}{q_{c}^{r/2}-q_{c}^{-r/2}}z^{\mp r}\right).
\eea
One can show that they indeed give a representation of quantum toroidal $\mathfrak{gl}_{1}$. The normalization of the vertex operators is chosen to obey the EF relation in (\ref{eq:DIMdef}). Let us check only this relation explicitly. Using
\bea
        \eta_{c}(z)\xi_{c}(w)=&\frac{(1-q_{c+1}^{1/2}q_{c-1}^{-1/2}w/z)(1-q_{c+1}^{-1/2}q_{c-1}^{1/2}w/z)}{(1-q_{c}^{1/2}w/z)(1-q_{c}^{-1/2}w/z)}:\eta_{c}(z)\xi_{c}(w):,\cr
        \xi_{c}(z)\eta_{c}(w)=&\frac{(1-q_{c+1}^{1/2}q_{c-1}^{-1/2}w/z)(1-q_{c+1}^{-1/2}q_{c-1}^{1/2}w/z)}{(1-q_{c}^{1/2}w/z)(1-q_{c}^{-1/2}w/z)}:\eta_{c}(w)\xi_{c}(z):
\eea
we have
\be
[\eta_{c}(z),\xi_{c}(w)]=\frac{(1-q_{c+1})(1-q_{c-1})}{1-q_{c}^{-1}}\left(\delta\bl(q_{c}^{1/2}w/z\br)\varphi_{c}^{+}(z)-\delta\bl(q_{c}^{1/2}z/w\br)\varphi_{c}^{-}(w)\right)
\ee
where we use
\bea
\varphi_{c}^{+}(z)=:\eta_{c}(z)\xi_{c}(q_{c}^{-1/2}z):,\quad \varphi_{c}^{-}(w)=:\eta_{c}(q_{c}^{-1/2}w)\xi_{c}(w):
\eea
and (\ref{eq:residue_formula}).
Changing the normalization as (\ref{eq:horizon_rep_level1}), the right-hand side factor will be $\tilde{g}=1/\kappa_{1}$.

By modifying the prefactors in front of the vertex operators, we can also obtain representations with central charges $(C,K^{-})=(q_{c}^{1/2},\gamma^{n})$ as 
\bea
    E(z)\rightarrow u\frac{(1-q_{c})}{\kappa_{1}}\gamma^{n}z^{-n}\eta_{c}(z),\quad F(z)\rightarrow -u^{-1}\frac{(1-q_{c}^{-1})}{\kappa_{1}}\gamma^{-n}z^{n}\xi_{c}(z),\quad K^{\pm}(z)\rightarrow \gamma^{\mp n}\varphi^{\pm}_{c}(z),\label{eq:horizon_rep_general}
\eea
where $\gamma\in\mathbb{C}^{\times}$ is an arbitrary parameter.

\subsection{Connections to Macdonald symmetric functions}\label{sec:QTA-Mac}
Despite their different appearance, the vertical and horizontal Fock representations introduced in the previous two sections both have ties to Macdonald symmetric functions. This section will explore this connection briefly. For special functions, we refer to Appendix~\ref{app:symmetric-functions}.
\paragraph{Vertical representation}
As previously discussed in \S\ref{sec:QTrep}, the actions of the higher modes can be derived from the mode relations (\ref{eq:DIM-mode}) once the actions of $E_{0},F_{0},K_{\pm1}$ are known. The following theorem, from \cite{FFJMM1}, establishes a correspondence between these modes and elements of Macdonald functions:
\begin{theorem} \cite{FFJMM1}\\
We set $q_{1}=t^{-1},q_{2}=q$. Under the isomorphism of the vector spaces
\bea
\mathcal{F}^{(3)}_{N}(u)\xrightarrow{\sim}\mathbb{C}[X_{1}^{\pm1},\ldots,X_{N}^{\pm1}]^{\mathfrak{S}_{N}},\quad \ket{\lambda}^{(3)}_{N}\mapsto P_{\lambda}(X;q,t)
\eea
the operators of quantum toroidal $\mathfrak{gl}_{1}$ are identified as
\bea
\mathsf{E}_{0}&\mapsto \mathcal{E}e_{1}(X),\quad \mathsf{F}_{0}\mapsto \mathcal{F}e_{-1}(X),\\
\mathsf{K}_{\pm 1}&\mapsto u^{\pm1}(1-q_{1}^{\pm 1})(1-q_{3}^{\pm1})q_{1}^{\pm \frac{N-1}{2}}D_{N}^{(\pm1)}(X; q,t)
\eea
where $e_{1}$ is the first elementary symmetric function \eqref{eq:elem-symm}, $D^{(1)}_{N}(X;q,t)$  is the Macdonald difference operator \eqref{Macdonald-diff} and
\bea
e_{-1}(X)\coloneqq e_{1}(X^{-1}),\quad D^{(-1)}_{N}(X;q,t)\coloneqq D^{(1)}_{N}(X;q^{-1},t^{-1}).\label{eq:Macdonald_diff_-1}~
\eea
\end{theorem}
\begin{proof}
From (\ref{eq:Nfockrep}), we have
\begin{align}
    \mathsf{K}_{\pm1}\ket{u,\lambda}^{(3)}_{N}=&u^{\pm1}(1-q_{1}^{\pm1})(1-q_{3}^{\pm1})\sum_{i=1}^{N}q_{1}^{\pm(i-1)}q_{2}^{\pm\lambda_{i}}\ket{u,\lambda}^{(3)}_{N},\label{K}\\
    \mathsf{E}_{0}\ket{u,\lambda}^{(3)}_{N}=&\mathcal{E}\sum_{i=1}^{N}\prod_{j=1}^{i-1}\frac{(1-q_{1}^{j-i+1}q_{2}^{\lambda_{j}-\lambda_{i}})(1-q_{1}^{j-i-1}q_{2}^{\lambda_{j}-\lambda_{i}-1})}{(1-q_{1}^{j-i}q_{2}^{\lambda_{j}-\lambda_{i}})(1-q_{1}^{j-i}q_{2}^{\lambda_{j}-\lambda_{i}-1})}\ket{u,\lambda+\Abox_{i}}^{(3)}_{N},\label{E}\\
\mathsf{F}_{0}\ket{u,\lambda}^{(3)}_{N}=&\mathcal{F}\sum_{i=1}^{N}\prod_{j=i+1}^{N}\frac{(1-q_{1}^{j-i+1}q_{2}^{\lambda_{j}-\lambda_{i}+1})(1-q_{1}^{j-i-1}q_{2}^{\lambda_{j}-\lambda_{i}})}{(1-q_{1}^{j-i}q_{2}^{\lambda_{j}-\lambda_{i}+1})(1-q_{1}^{j-i}q_{2}^{\lambda_{j}-\lambda_{i}})}\ket{u,\lambda-\Abox_{i}}^{(3)}_{N}.\label{F}
\end{align}
The action \eqref{K} of $\mathsf{K}_{\pm1}$ is indeed equivalent to the Macdonald difference operators $D^{(\pm 1)}_{N}$. 
Since the Macdonald functions $P_{\lambda}$ are eigenfunctions of the Macdonald difference operators $D^{(\pm 1)}_{N}(X;q,t)$ in \eqref{Macdonald-diff} with eigenvalues
\bea
e_{\pm1}(t^{\rho}q^{\lambda})=t^{\pm \frac{N-1}{2}}\sum_{i=1}^{N}t^{\mp(i-1)}q^{\pm \lambda_{i}}~,
\eea 
we obtain the correspondence $\mathsf{K}_{\pm1}\leftrightarrow D_{N}^{(\pm1)}(X;q,t)$, $\ket{\lambda}^{(3)}_{N}\leftrightarrow P_{\lambda}(X;q,t)$  after setting $q_{1}=t^{-1},q_{2}=q$.

Furthermore, the action \eqref{E} of $\mathsf{E}_{0}$ can be understood as the Pieri formula \eqref{Pieri} where  $\lambda+\Abox_{i}/\lambda$ ($i=1,\ldots,N$) are vertical 1-strips:
\bea
e_{1}P_{\lambda}=&\sum_{i=1}^{N}\psi'_{\lambda+\Abox_{i}/\lambda}P_{\lambda+\Abox_{i}}~,\cr
\psi'_{\lambda+\Abox_{i}/\lambda}=&\prod_{j=1}^{i-1}\frac{(1-q^{\lambda_{j}-\lambda_{i}}t^{i-j-1})(1-q^{\lambda_{j}-\lambda_{i}-1}t^{i-j+1})}{(1-q^{\lambda_{j}-\lambda_{i}}t^{i-j})(1-q^{\lambda_{j}-\lambda_{i}-1}t^{i-j})}
\eea
Similarly, the action \eqref{E} of $\mathsf{E}_{0}$ can be interpreted as the Pieri formula. To see the analogous formula for Macdonald functions of $N$ variables, let $\lambda$ be a partition of length $N$, $\ell(\lambda)=N$. Then, the action \eqref{E} can be regarded as 
\bea
e_{-1}P_{\lambda}=&\sum_{i=1}^{N}\psi'_{\lambda/\lambda-\Abox_{i}}P_{\lambda-\Abox_{i}}\qquad \textrm{for }\ \ell(\lambda)=N~,\cr
\psi'_{\lambda/\lambda-\Abox_{i}}=&\prod_{j=i+1}^{N}\frac{(1-q^{\lambda_{j}-\lambda_{i}+1}t^{i-j-1})(1-q^{\lambda_{j}-\lambda_{i}}t^{i-j+1})}{(1-q^{\lambda_{j}-\lambda_{i}+1}t^{i-j})(1-q^{\lambda_{j}-\lambda_{i}}t^{i-j})}~.
\eea
Hence, the Pieri coefficients align with the actions of $\mathsf{E}_{0}$ and $\mathsf{F}_{0}$, leading to the correspondence $\mathsf{E}_{0}\leftrightarrow e_{1}$ and  $\mathsf{F}_{0}\leftrightarrow e_{-1}$ .
\end{proof}

\paragraph{Horizontal representation}

Now, let us investigate the connection between Macdonald functions and the horizontal representation. This connection stems from the free field realization of Macdonald functions, as detailed in Appendix \ref{app:Macdonald}. The ring of symmetric functions and the Fock space are isomorphic under the following map:
\begin{align}
\begin{split}
    &\mathsf{a}_{-\lambda}\ket{0}\mapsto p_{\lambda},\quad \bra{0}V_{+}^{q,t}\mathsf{a}_{-\lambda}\ket{0}=p_{\lambda},\\
    &V^{q,t}_{+}=\exp\left(\sum_{n>0}\frac{1-t^{n}}{1-q^{n}}\frac{\mathsf{a}_{n}}{n}p_{n}\right),\quad [\mathsf{a}_{m},\mathsf{a}_{n}]=m\frac{1-q^{|m|}}{1-t^{|m|}}\delta_{m+n,0}
  \end{split} 
\end{align}
where $p_{n}$ is the power sum, and $a_{-\lambda}\ket{0}=a_{-\lambda_{1}}a_{-\lambda_{2}}\cdots\ket{0}$. The dual Fock space is defined similarly (see Appendix \ref{app:Macdonald}). Under this correspondence, the Macdonald difference operator $D^{(1)}_{N}(X;q,t)$ in (\ref{Macdonald-diff}) turns out to be 
\begin{align}
    \widehat{D}(z)=\exp\left(\sum_{n=1}^{\infty}\frac{1-t^{-n}}{n}\mathsf{a}_{-n}z^{n}\right)\exp\left(-\sum_{n=1}^{\infty}\frac{1-t^{n}}{n}\mathsf{a}_{n}z^{n}\right)\eqqcolon\eta(z)
\end{align}
where we use the fact in (\ref{eq:Macdonalop_vertexop}) and denote this vertex operator as $\eta(z)$. The contraction formula is 
\begin{align}
\begin{split}
    \eta(z)\eta(w)=&\exp\left(-\sum_{n=1}^{\infty}\frac{1}{n}(1-t^{-n})(1-q^{n})\left(\frac{w}{z}\right)^{n}\right):\eta(z)\eta(w):\\
    =&\frac{(1-w/z)(1-qt^{-1}w/z)}{(1-t^{-1}w/z)(1-qw/z)}:\eta(z)\eta(w):
\end{split}
\end{align}
We further can introduce another vertex operator corresponding to $D^{(1)}_{N}(X;q^{-1},t^{-1})$ as
\begin{align}
    \xi(z)=\exp\left(-\sum_{n=1}^{\infty}\frac{1-t^{-n}}{n}q^{-n/2}t^{n/2}\mathsf{a}_{-n}z^{n}\right)\exp\left(\sum_{n=1}^{\infty}\frac{1-t^{n}}{n}q^{-n/2}t^{n/2}\mathsf{a}_{n}z^{-n}\right)
\end{align}
whose contraction formula is 
\begin{align}
   \xi(z)\xi(w)=\frac{(1-w/z)(1-tq^{-1}w/z)}{(1-tw/z)(1-q^{-1}w/z)}:\xi(z)\xi(w):.
\end{align}
The commutation relation of these two operators will give two other operators $\varphi^{\pm}(z)$ as
\begin{align}
    [\eta(z),\xi(w)]=\frac{(1-q)(1-t^{-1})}{1-q/t}\left(\delta(\gamma w/z)\varphi^{+}(z)-\delta(\gamma z/w)\varphi^{-}(w)\right),\quad \gamma=(t/q)^{1/2}
\end{align}
where 
\begin{align}
    \varphi^{+}(z)=:\eta(z)\xi(\gamma^{-1}z):,\quad \varphi^{-}(z)=:\eta(\gamma^{-1}z)\xi(z):.
\end{align}
Obviously, we managed to reproduce the horizontal representations with central charge $C=q_{3}^{1/2}$ of the quantum toroidal $\mathfrak{gl}_{1}$ after setting $q_{1}=t^{-1},q_{2}=q,q_{3}=t/q$.

In particular, we find that the Macdonald difference operators correspond to the zero-modes of $\eta(z),\xi(z)$, as we need to integrate these operators, as outlined in (\ref{eq:Macdonalop_vertexop}). As a result, in the horizontal representation, there is a correspondence \begin{align}
    \mathsf{E}_{0}\leftrightarrow D^{(1)}_{N}(X;q,t),\quad \mathsf{F}_{0}\leftrightarrow D^{(-1)}_{N}(X;q,t)=D^{(1)}_{N}(X;q^{-1},t^{-1}).
\end{align}
The reason for this correspondence differing from the vertical representation, where the Macdonald difference operators correspond to $\mathsf{K}_{\pm1}$,  can be attributed to the Miki-automorphism explained in (\ref{eq:DIMsubalgebra})--(\ref{eq:Miki-duality}).

\subsection{Deformed \texorpdfstring{$\mathcal{W}$}{W}-algebra}\label{sec:deformedW}
Under the horizontal representation, the Drinfeld currents are mapped to the vertex operators. Then, it is natural to ask whether the vertex operators will form an interesting algebra. In fact, the free-field realization, \eqref{eq:horizon_rep_level1} and \eqref{vertex-operator}, of the horizontal representation provides the connection to the deformed $\mathcal{W}$-algebras \cite{Shiraishi:1995rp,Awata:1995zk,frenkel1996quantum,frenkel1997deformations,FF,FHSSY,miki2007q,Kojima:2020vtc,Kojima2019,Harada:2021xnm}.
To determine the generating currents of the deformed $\mathcal{W}$-algebra from the horizontal representation, we need to decouple the unnecessary Heisenberg current. We mainly focus on determining the currents using $E(z)$, but the same approach can be applied to $F(z)$. To decouple the Heisenberg current, we define the following current commuting with the $\U(1)$ part
\bea
    t(z)=\alpha(z)E(z)\beta(z)
\eea
where
\bea
    \alpha(z)=\exp\left(\sum_{r=1}^{\infty}\frac{\kappa_{r}}{r}\frac{\sfH_{-r}}{C^{r}-C^{-r}}z^{r}\right),\quad
    \beta(z)=\exp\left(\sum_{r=1}^{\infty}\frac{\kappa_{r}}{r}\frac{\sfH_{r}}{1-C^{2r}}z^{-r}\right).
\eea
Using
\be
    [\sfH_{r},\alpha(z)]=z^{r}\alpha(z),\quad [\sfH_{-r},\beta(z)]=-C^{-r}z^{-r}\beta(z)\quad (r>0)
\ee
and the commutation relation (\ref{eq:DIMdef}), one can show $[t(z),\sfH_{r}]=0\,(r\neq 0)$. Note that we assume $C\neq 1$ to avoid divergence whenever using them.

Let us consider the representation of the current $E(z)$ on the tensor product of the horizontal Fock spaces $\mathcal{F}_{\boldsymbol{c}}(\boldsymbol{u})=\mathcal{F}_{c_{1}}(u_{1})\otimes \cdots\otimes\mathcal{F}_{c_{n}}(u_{n})$:
\bea
    \rho_{\boldsymbol{c}}(E(z))=&\sum_{i=1}^{n}\frac{1-q_{c_{i}}}{\kappa_{1}}\Lambda_{i}(z|u_{i}),\cr
    \Lambda_{i}(z|u_{i})=&\varphi^{-}_{c_{1}}(q_{c_{1}}^{1/2}z)\otimes\cdots\otimes\varphi^{-}_{c_{i-1}}(q_{c_{1}}^{1/2}\cdots q_{c_{i-1}}^{1/2}z)\otimes u_{i} \eta_{c_{i}}(q_{c_{1}}^{1/2}\cdots q_{c_{i-1}}^{1/2}z)\otimes 1\otimes\cdots\otimes 1.\label{eq:lambdadef}
\eea
We modify the normalization parameter $u_{i}$ so that the currents are written as
\bea
    \rho_{\boldsymbol{c}}(E(z))=\sum_{i=1}^{n}y_{i}\Lambda_{i}(z),\quad y_{i}=\frac{q_{c_{i}}^{1/2}-q_{c_{i}}^{-1/2}}{q_{3}^{1/2}-q_{3}^{-1/2}}.\label{eq:CVOAgeneratingcurrent}
\eea
Moving forward, we will omit the spectral parameter in the currents when it is clear from the context. This omission is suitable when we are only interested in the quadratic relations of the currents. Nonetheless, it is appropriate to reintroduce the spectral parameter when considering the screening currents. As we have already obtained the contraction formulas for each of the currents, we can explicitly compute the contraction between $\Lambda_{i}(z)$ and $\Lambda_{j}(w)$:
\bea
\Lambda_{i}(z)\Lambda_{j}(w)=\begin{dcases}:\Lambda_{i}(z)\Lambda_{j}(w):\quad (i>j),\cr \exp\left(\sum_{r=1}^{\infty}\frac{\kappa_{r}}{r(1-q_{c_{i}}^{r})}\left(\frac{w}{z}\right)^{r}\right):\Lambda_{i}(z)\Lambda_{j}(w):\quad (i=j),\cr\exp\left(\sum_{r=1}^{\infty}\frac{\kappa_{r}}{r}\left(\frac{w}{z}\right)^{r}\right):\Lambda_{i}(z)\Lambda_{j}(w):\quad (i<j)\end{dcases}.\label{eq:lambdacontraction}
\eea
To extract the generator with no extra $\mathfrak{gl}_{1}$ factor, we need to use the current $t(z)$ instead of $E(z)$ as
\bea
    \rho_{\boldsymbol{c}}(t(z))=\sum_{i=1}^{n}y_{i}\Tilde{\Lambda}_{i}(z),\label{eq:CVOAgeneratingcurrent-no-heisenberg}
\eea
where $\Tilde{\Lambda}_{i}(z)$ is a vertex operator dressed with $\rho_{\boldsymbol{c}}(\Delta(\alpha(z)))$ and $\rho_{\boldsymbol{c}}(\Delta(\beta(z)))$. We can also derive the contraction formula of these dressed currents:
\bea
    t(z)t(w)=\exp\left(\sum_{r>0}\frac{\kappa_{r}}{r}\frac{1}{C^{2r}-1}\left(\frac{w}{z}\right)^{r}\right)\alpha(z)\alpha(w)E(z)E(w)\beta(z)\beta(w).
\eea
Taking the coproducts with the representation map $\rho_{\boldsymbol{c}}$ of both sides and using the contraction formulas of $\Lambda_{i}(z)$s, we obtain
\bea
    \Tilde{\Lambda}_{i}(z)\Tilde{\Lambda}_{j}(w)=\begin{dcases}\exp\left(-\sum_{r>0}\frac{\kappa_{r}}{r}\frac{1}{1-q_{\boldsymbol{c}}^{r}}\left(\frac{w}{z}\right)^{r}\right):\Tilde{\Lambda}_{i}(z)\Tilde{\Lambda}_{j}(w):\quad (i>j)\cr \exp\left(\sum_{r>0}\frac{\kappa_{r}}{r}\frac{q_{c_{i}}^{r}-q_{\boldsymbol{c}}^{r}}{(1-q_{c_{i}}^{r})(1-q_{\boldsymbol{c}}^{r})}\left(\frac{w}{z}\right)^{r}\right):\Tilde{\Lambda}_{i}(z)\Tilde{\Lambda}_{j}(w):\quad (i=j)\cr
    \exp\left(\sum_{r>0}\frac{\kappa_{r}}{r}\frac{q_{\boldsymbol{c}}^{r}}{q_{\boldsymbol{c}}^{r}-1}\left(\frac{w}{z}\right)^{r}\right):\Tilde{\Lambda}_{i}(z)\Tilde{\Lambda}_{j}(w):\quad (i<j)\end{dcases},\label{eq:tildelambdacontraction}
\eea
where $q_{\boldsymbol{c}}=\prod_{i=1}^{n}q_{c_{i}}$.

From now on, we will omit the representation map $\rho_{\boldsymbol{c}}$ and identify the currents with the representations. For the connection to deformed $\cW$-algebras, we need to set all the central charges are identical, and here we choose all $c_i=3$ without loss of generality. To determine the deformed $\mathcal{W}$ algebra structure, we need to find the so-called \emph{quadratic relations}, which can be thought of as a $q$-deformed version of the OPE in the CFT context. Let us derive the quadratic relations in a few examples. Note that,  precisely speaking, $q$-deformation in the following is indeed a deformation by the two parameters $(q_1,q_2)$.

\paragraph{$q$-Virasoro} In the $q$-Virasoro algebra, there is a single current denoted $T(z)$, which is the $q$-analog of the energy-momentum tensor \eqref{EM} in the 2d CFT.  We take the representation space to be $\mathcal{F}_{3}(u_{1})\otimes \mathcal{F}_{3}(u_{2})$:
\bea
    T(z)=\Tilde{\Lambda}_{1}(z)+\Tilde{\Lambda}_{2}(z).
\eea
Defining
\bea
    f(z)=\exp\left(\sum_{r=1}^{\infty}\frac{\kappa_{r}}{r}\frac{1}{(q_{3}^{r}-q_{3}^{-r})}z^{r}\right)
\eea
we have the following
\bea
    f(w/z)\Tilde{\Lambda}_{i}(z)\Tilde{\Lambda}_{j}(w)=\begin{dcases}\frac{(1-q_{1}w/z)(1-q_{2}w/z)}{(1-w/z)(1-q_{3}^{-1}w/z)}:\Tilde{\Lambda}_{i}(z)\Tilde{\Lambda}_{j}(w):\quad (i>j),\cr:\Tilde{\Lambda}_{i}(z)\Tilde{\Lambda}_{j}(w):\quad (i=j),\cr
    \frac{(1-q_{1}^{-1}w/z)(1-q_{2}^{-1}w/z)}{(1-w/z)(1-q_{3}w/z)}:\Tilde{\Lambda}_{i}(z)\Tilde{\Lambda}_{j}(w):\quad (i<j).
    \end{dcases}
\eea
Then, the quadratic relation of the $q$-Virasoro can be computed as
\bea
    &f(w/z)T(z)T(w)-f(z/w)T(w)T(z)\cr
    =&\frac{(1-q_{1})(1-q_{2})}{(1-q_{3}^{-1})}\left(:\Tilde{\Lambda}_{1}(z)\Tilde{\Lambda}_{2}(q_{3}^{-1}z):\delta\left(\frac{q_{3}w}{z}\right)-:\Tilde{\Lambda}_{2}(z)\Tilde{\Lambda}_{1}(q_{3}z):\delta\left(\frac{w}{q_{3}z}\right)\right)\cr
    =&\frac{(1-q_{1})(1-q_{2})}{(1-q_{3}^{-1})}\left(\delta\left(\frac{q_{3}w}{z}\right)-\delta\left(\frac{w}{q_{3}z}\right)\right)
\eea
where we use $:\Tilde{\Lambda}_{1}(q_{3}z)\Tilde{\Lambda}_{2}(z):=1$. Setting $q_{1}=q,q_{2}=t^{-1}$, we recover the defining relation originally introduced in \cite{Shiraishi:1995rp}.  Since the right-hand side does not contain new currents anymore, the algebra is closed with the single current $T(z)$. If the quadratic relation contains new currents, we need to include them and compute the quadratic relations again until no new currents show up.

While we have derived the free field realization of the $q$-Virasoro algebra starting from the horizontal representation, this algebra was originally derived in \cite{Shiraishi:1995rp,Awata:1995zk} using the trigonometric deformation of the Miura transformation discussed in \S\ref{sec:W}. We introduce the trigonometric Miura operators as:
\bea
R_{i}(z)=1-\Tilde{\Lambda}_{i}(z)q_{3}^{-D_{z}},\quad i=1,2
\eea
where $D_{z}=z\frac{d}{dz}$. Then, the generators are derived as
\bea
&:R_{1}(z)R_{2}(z):=\sum_{i=0}^{2}(-1)^{i}T_{i}(z)q_{3}^{-iD_{z}},\\
&T_{0}(z)\equiv 1,\quad T_{1}(z)=\Tilde{\Lambda}_{1}(z)+\Tilde{\Lambda}_{2}(z),\quad T_{2}(z)=:\Tilde{\Lambda}_{1}(z)\Tilde{\Lambda}_{2}(q_{3}^{-1}z):=1.
\eea

\paragraph{$q$-$\widetilde{\cW}_{N}$ algebra} To generalize it to higher ranks, we consider the representation space $\mathcal{F}_{3}(u_{1})\otimes\cdots\mathcal{F}_{3}(u_{N})$ with the current
\bea
    T_{1}(z)=\sum_{i=1}^{N}\Tilde{\Lambda}_{i}(z)
\eea
and $q_{\boldsymbol{c}}=q_{3}^{N}$. The algebra generated by this current is the deformed version of the $\cW_{N}$-algebra\footnote{Note that we are denoting the $\cW$-algebra without the extra Heisenberg algebra as $\widetilde{\cW}$ (see \S\ref{sec:W}).}  introduced in \S\ref{sec:W}.

First, we define
\bea
    f_{1,1}(z)=\exp\left(\sum_{r>0}\frac{\kappa_{r}}{r}\frac{q_{3}^{Nr}-q_{3}^{r}}{(1-q_{3}^{r})(1-q_{3}^{Nr})}z^{r}\right).
\eea
From the contraction formulas (\ref{eq:tildelambdacontraction}), we have
\bea
    f_{1,1}(w/z)T_{1}(z)T_{1}(w)=&\sum_{i}:\Tilde{\Lambda}_{i}(z)\Tilde{\Lambda}_{i}(w):\cr
    &+\frac{(1-q_{1}w/z)(1-q_{2}w/z)}{(1-w/z)(1-q_{3}^{-1}w/z)}\sum_{i>j}:\Tilde{\Lambda}_{i}(z)\Tilde{\Lambda}_{j}(w):\cr
    &+\frac{(1-q_{1}^{-1}w/z)(1-q_{2}^{-1}w/z)}{(1-w/z)(1-q_{3}w/z)}\sum_{i<j}:\Tilde{\Lambda}_{i}(z)\Tilde{\Lambda}_{j}(w):
\eea
and thus obtain
\bea
         &f_{1,1}(w/z)T_{1}(z)T_{1}(w)- f_{1,1}(z/w)T_{1}(w)T_{1}(z)\cr
         =&\frac{(1-q_{1})(1-q_{2})}{(1-q_{3}^{-1})}\left(T_{2}(z)\delta\left(\frac{z}{q_{3}w}\right)-T_{2}(w)\delta\left( \frac{w}{q_{3}z}\right)\right)
\eea
where
\bea
    T_{2}(z)=\sum_{i<j}:\Tilde{\Lambda}_{i}(z)\Tilde{\Lambda}_{j}(q_{3}^{-1}z):.
\eea
In the $q$-Virasoro case, we do not get more currents after taking the quadratic relation of the generator $T(z)$. However, a new current $T_{2}(z)$ appears on the right-hand side in the $q$-$\widetilde{\cW}_{N}$ case. Because of this, we also need to compute the quadratic relation between $T_{1}(z)$ and $T_{2}(z)$. 
We define
\bea
    f_{1,2}(z)=&\exp\left(\sum_{r>0}\frac{\kappa_{r}}{r}\frac{q_{3}^{(N-1)r}-q_{3}^{r}}{(1-q_{3}^{r})(1-q_{3}^{Nr})}z^{r}\right),\quad
    f_{2,1}(z)=f_{1,2}(q_{3}^{-1}z)
\eea
and obtain
\bea
        &f_{1,2}(w/z)T_{1}(z)T_{2}(w)-f_{2,1}(z/w)T_{2}(w)T_{1}(z)\cr
        =&\frac{(1-q_{1})(1-q_{2})}{(1-q_{3}^{-1})}\left(T_{3}(z)\delta\left(\frac{z}{q_{3}w}\right)-T_{3}(w)\delta\left(\frac{w}{q_{3}^{2}z}\right)\right)
\eea
where
\bea
    T_{3}(z)=\sum_{i<j<k}:\Tilde{\Lambda}_{i}(z)\Tilde{\Lambda}_{j}(q_{3}^{-1}z)\Tilde{\Lambda}_{k}(q_{3}^{-2}z):.
\eea
Thus, another current $T_{3}(z)$ is produced. After doing this process recursively, one will see that the currents are defined as
\bea
    T_{m}(z)=\sum_{1\leq i_{1}<i_{2}<\cdots<i_{m}\leq N}:\Tilde{\Lambda}_{i_{1}}(z)\Tilde{\Lambda}_{i_{2}}(q_{3}^{-1}z)\cdots\Tilde{\Lambda}_{i_{m}}(q_{3}^{-m+1}z):\quad (1\leq m\leq N).
\eea
Actually the final current $T_{N}(z)$ obeys the condition $T_{N}(z)=1$. Quadratic relations between currents $T_{i}(z)$ and $T_{j}(z)$ are rather complicated \cite{Shiraishi:1995rp,Odake:2001ad}:
\bea
    &f_{i,j}(w/z)T_{i}(z)T_{j}(w)-f_{j,i}(z/w)T_{j}(w)T_{i}(z)\\
    =&\frac{(q_{1}^{1/2}-q_{1}^{-1/2})(q_{2}^{1/2}-q_{2}^{-1/2})}{(q_{3}^{1/2}-q_{3}^{-1/2})}\sum_{k=1}^{i}\prod_{l=1}^{k-1}\frac{(1-q_{1}q_{3}^{-l})(1-q_{2}q_{3}^{-l})}{(1-q_{3}^{-l-1})(1-q_{3}^{-l})}\cr
    \times&\left(\delta\left(\frac{q_{3}^{k}w}{z}\right)f_{i-k,j+k}(1)T_{i-k}(q_{3}^{-k}z)T_{j+k}(q_{3}^{k}w)-\delta\left(\frac{q_{3}^{i-j-k}w}{z}\right)f_{i-k,j+k}(q_{3}^{j-i})T_{i-k}(z)T_{j+k}(w)\right)\nonumber
\eea
where we assumed $i\leq j$ and use
\bea
     f_{i,j}(z)=\exp\left(\sum_{r>0}\frac{\kappa_{r}}{r}\frac{(1-q_{3}^{ir})}{(1-q_{3}^{r})^{2}}\frac{(q_{3}^{(N-j+1)r}-q_{3}^{r})}{(1-q_{3}^{rN})}z^{r}\right)\quad f_{j,i}(z)=f_{i,j}(q_{3}^{i-j}z)\quad (i\leq j).
\eea

Similar to the $q$-Virasoro case, we can introduce Miura operators and derive the form of the generators as
\bea
    R_{i}(z)=1-\Tilde{\Lambda}_{i}(z)q_{3}^{-D_{z}},\quad i=1,\ldots,N,\\
    :R_{1}(z)R_{2}(z)\cdots R_{N}(z):=\sum_{i=0}^{N}(-1)^{i}T_{i}(z)q_{3}^{-iD_{z}}~,
\eea
which can be understood as the deformed version of \eqref{Miura_tr}.

\paragraph{$\cW_{q_{1},q_{2}}(\mathfrak{sl}_{2\,|\,1})$ algebra} If we choose one of the central charges different from the others, we obtain another variant of the $\cW$-algebra. Let us consider a nontrivial example where the representation space is $\mathcal{F}_{3}(u_{1})\otimes \mathcal{F}_{3}(u_{2})\otimes\mathcal{F}_{2}(u_{3})$. The central charges are $\boldsymbol{c}=(c_{1},c_{2},c_{3})=(3,3,2)$ and $q_{\boldsymbol{c}}=q_{3}^{2}q_{2}$. We introduce the following currents:
\begin{align}
T_{1}(z)=&\Tilde{\Lambda}_{1}(z)+\Tilde{\Lambda}_{2}(z)+\frac{q_{2}^{1/2}-q_{2}^{-1/2}}{q_{3}^{1/2}-q_{3}^{-1/2}}\Tilde{\Lambda}_{3}(z),\\
T_{2}(z)=&:\Tilde{\Lambda}_{1}(z)\Tilde{\Lambda}_{2}(q_{3}^{-1}z):+\frac{q_{2}^{1/2}-q_{2}^{-1/2}}{q_{3}^{1/2}-q_{3}^{-1/2}}:\Tilde{\Lambda}_{1}(z)\Tilde{\Lambda}_{3}(q_{3}^{-1}z):\nonumber\\
&+\frac{q_{2}^{1/2}-q_{2}^{-1/2}}{q_{3}^{1/2}-q_{3}^{-1/2}}:\Tilde{\Lambda}_{2}(z)\Tilde{\Lambda}_{3}(q_{3}^{-1}z):\nonumber\\
&+\frac{q_{2}^{1/2}-q_{2}^{-1/2}}{q_{3}^{1/2}-q_{3}^{-1/2}}\frac{q_{2}^{1/2}q_{3}^{-1/2}-q_{2}^{-1/2}q_{3}^{1/2}}{q_{3}-q_{3}^{-1}}:\Tilde{\Lambda}_{3}(z)\Tilde{\Lambda}_{3}(q_{3}^{-1}z):,\\
T_{n\geq3}(x)=&:\prod_{j=1}^{n}\frac{q_{2}^{1/2}q_{3}^{-(j+1)/2}-q_{2}^{-1/2}q_{3}^{(j-1)/2}}{q_{3}^{j/2}-q_{3}^{-j/2}}\Tilde{\Lambda}_{3}(q_{3}^{-j+1}z):\nonumber\\
&+:\Tilde{\Lambda}_{1}(z)\prod_{j=1}^{n-1}\frac{q_{2}^{1/2}q_{3}^{-(j+1)/2}-q_{2}^{-1/2}q_{3}^{(j-1)/2}}{q_{3}^{j/2}-q_{3}^{-j/2}}\Tilde{\Lambda}_{3}(q_{3}^{-j}z):\nonumber\\
&+:\Tilde{\Lambda}_{2}(z)\prod_{j=1}^{n-1}\frac{q_{2}^{1/2}q_{3}^{-(j-1)/2}-q_{2}^{-1/2}q_{3}^{(j-1)/2}}{q_{3}^{j/2}-q_{3}^{-j/2}}\Tilde{\Lambda}_{3}(q_{3}^{-j}z):\nonumber\\
&+:\Tilde{\Lambda}_{1}(z)\Tilde{\Lambda}_{2}(q_{3}^{-1}z)\prod_{j=1}^{n-2}\frac{q_{2}^{1/2}q_{3}^{-(j+1)/2}-q_{2}^{-1/2}q_{3}^{(j-1)/2}}{q_{3}^{j/2}-q_{3}^{-j/2}}\Tilde{\Lambda}_{3}(q_{3}^{-j-1}z):
\end{align}
where the $\Tilde{\Lambda}_{i}(z)(i=1,2,3)$ are the vertex operators obeying (\ref{eq:tildelambdacontraction}). Note that the first current $T_{1}(z)$ is defined as $\rho_{\boldsymbol{c}=(3,3,2)}(E(z))=T_{1}(z)$. Other currents are derived recursively as the $q$-$\cW_{N}$ algebra. The quadratic relations with the first generating currents are
\bea
&f_{1,m}(w/z)T_{1}(z)T_{m}(w)-f_{m,1}(z/w)T_{m}(w)T_{1}(z)\\
=&\frac{(1-q_{1})(1-q_{2})}{(1-q_{3}^{-1})}\left(T_{m+1}(z)\delta\left(\frac{z}{q_{3}^mw}\right)-T_{m+1}(w)\delta\left(\frac{w}{q_{3}^{m}z}\right)\right)
\eea
where
\bea
f_{1,m}(z)=\frac{\prod_{k=0}^{m-1}f_{1,1}(q_{3}^{-k}z)}{\prod_{k=0}^{m-2}S_{3}(q_{3}^{-k}z)},\quad f_{m,1}(z)=f_{1,m}(q_{3}^{1-m}z),\\ f_{1,1}(z)=\exp\left(-\sum_{r=1}^{\infty}\frac{\kappa_{r}}{r}\frac{(q_{3}^{r}-q_{3}^{2r}q_{2}^{r})}{(1-q_{3}^{r})(1-q_{3}^{2r}q_{2}^{r})}z^{r}\right).
\eea
We similarly can derive other quadratic relations between currents $T_{i}(z)$ and $T_{j}(z)$; see \cite{Kojima2019} for details. Unlike $q$-Virasoro and $q$-$\cW_{N}$, there are an infinite number of currents in this algebra. Similar to the $q$-Virasoro and $q$-$\cW_{N}$, the algebra admits the realization by the Miura transformation involving the vertex operators $\Tilde{\Lambda}_{i}(z)$. However, the algebra requires an infinite number of the Miura operators, resulting in an infinite number of currents (see \cite{Kojima2019,Kojima:2020vtc,Harada:2021xnm} for details).

Note that taking a different order of tensor products
\begin{align}
    \mathcal{F}_{3}(u_{1})\otimes \mathcal{F}_{2}(u_{2})\otimes \mathcal{F}_{3}(u_{3}),\quad \mathcal{F}_{2}(u_{1})\otimes \mathcal{F}_{3}(u_{2})\otimes \mathcal{F}_{3}(u_{3})~
\end{align}
gives two other realizations of $\cW_{q_{1},q_{2}}(\mathfrak{sl}_{2\,|\,1})$. These choices correspond to other Dynkin diagrams of $\mathfrak{sl}_{2|1}$ (see \eqref{eq:sl21figure1}, \eqref{eq:sl21figure2}, \eqref{eq:sl21figure3}). In these cases, 
the free field realizations of the currents will be different, but it is known that the quadratic relations of the currents are the same \cite{Kojima2019}. 

\paragraph{Deformed corner vertex operator algebra}Generalizations to the tensor product of generic $N$ Fock spaces $\mathcal{F}_{\boldsymbol{c}}(\boldsymbol{u})$ is straightforward now. Suppose that there are $N_1$ 1's, $N_2$ 2's and $N_3$ 3's 's in a sequence $\boldsymbol{c}=\left(c_1, c_2, \cdots, c_N\right)$ of the central charges where $N=N_1+N_2+N_3$. The first generating current will be (\ref{eq:CVOAgeneratingcurrent-no-heisenberg}). One can do the above process and derive higher currents by studying the quadratic relations individually. The resulting algebra is a deformed version of the corner VOA $Y_{N_1 N_2 N_3}$ elucidated in \S\ref{sec:corner}. Actually, the form of the higher currents can be derived by $q$-deformed Miura operators (see \cite{Harada:2021xnm} for details):
\begin{align}
    &R^{(c)}(z)=\sum_{n=0}^{\infty}:\prod_{j=1}^{n}\left(-\frac{q_{c}^{\frac{1}{2}}q_{3}^{-\frac{j-1}{2}}-q_{c}^{-\frac{1}{2}}q_{3}^{\frac{j-1}{2}}}{q_{3}^{\frac{j}{2}}-q_{3}^{-\frac{j}{2}}}\Lambda_{j}(q_{3}^{-j+1}z)\right):q_{3}^{-nD_{z}},\\
    &:R^{(c_{1})}_{1}(z)R^{(c_{2})}_{2}(z)\cdots R^{(c_{N})}_{N}(z):=\sum_{n=0}^{\infty}(-1)^{n}T_{n}(z)q_{3}^{-nD_{z}}~,
\end{align}
which can be regarded as the deformed version of \eqref{Miura-tranformation}.
One can also derive the quadratic relations of these higher currents (see \cite{Kojima:2020vtc} for details).

\paragraph{Superalgebra structure}
Consider the sequence of tensor products $\mathcal{F}_{\boldsymbol{c}}(\boldsymbol{u})=\mathcal{F}_{c_{1}}(u_{1})\otimes\mathcal{F}_{c_{2}}(u_{2})\otimes\cdots\mathcal{F}_{c_{n}}(u_{n})$. We define the root currents as
\begin{align}
  A_{i}(z)&\equiv :\frac{\Lambda_{i}(\prod_{l=1}^{i}q_{c_{l}}^{-1/2}z)}{\Lambda_{i+1}(\prod_{l=1}^{i}q_{c_{l}}^{-1/2}z)}:,\quad i=1,\ldots,n-1
\end{align}
Explicitly, they are
\begin{align}
    :A_{i}(z)^{-1}:=1\otimes\cdots\otimes\xi_{c_{i}}(z)\otimes \eta_{c_{i+1}}(z)\otimes1\otimes\cdots\otimes 1.
\end{align}
Let us consider the contractions of these root currents. We define contractions as
\begin{align}
    A_{i}(z)A_{j}(w)=\exp\left(\sum_{r=1}^{\infty}\frac{B^{[r]}_{i,j}}{r}\left(\frac{w}{z}\right)^{r}\right):A_{i}(z)A_{j}(w):.
\end{align}
Using the contractions of the operators $\Lambda_{i}(z)$ in (\ref{eq:lambdacontraction}), we obtain
\begin{align}
\begin{split}
    B^{[r]}_{i,j}=&\kappa_{r}\left\{-\frac{(q_{c_{i}}^{r/2}q_{c_{i+1}}^{r/2}-q_{c_{i}}^{-r/2}q_{c_{i+1}}^{-r/2})}{(q_{c_{i}}^{r/2}-q_{c_{i}}^{-r/2})(q_{c_{i+1}}^{r/2}-q_{c_{i+1}}^{-r/2})}\delta_{i,j}+\frac{1}{q_{c_{j}}^{r/2}-q_{c_{j}}^{-r/2}}\delta_{i,j-1}+\frac{1}{q_{c_{i}}^{r/2}-q_{c_{i}}^{-r/2}}\delta_{i,j+1}\right\}
\end{split}
\end{align}
Eventually, the nontrivial contractions occur only between the root currents adjacent to each other:
\begin{align}
    B_{i,i}^{[r]}=-\kappa_{r}\frac{(q_{c_{i}}^{r/2}q_{c_{i+1}}^{r/2}-q_{c_{i}}^{-r/2}q_{c_{i+1}}^{-r/2})}{(q_{c_{i}}^{r/2}-q_{c_{i}}^{-r/2})(q_{c_{i+1}}^{r/2}-q_{c_{i+1}}^{-r/2})}\quad B_{i-1,i}^{[r]}=B_{i,i-1}^{[r]}=\kappa_{r}\frac{1}{q_{c_{i}}^{r/2}-q_{c_{i}}^{-r/2}}.
\end{align}
Note that the matrix $B^{[r]}_{i,j}$ is a symmetric matrix.
Considering $B^{[1]}_{i,j}$, and taking the degenerate limit, we have
\begin{align}
    B^{[1]}_{i,j}\rightarrow \begin{dcases}
    \epsilon_{1}\epsilon_{2}\epsilon_{3}\frac{1}{\epsilon_{c_{j}}},\quad j=i+1\\
    -\epsilon_{1}\epsilon_{2}\epsilon_{3}\frac{\epsilon_{c_{i}}+\epsilon_{c_{i+1}}}{\epsilon_{c_{i}}\epsilon_{c_{i+1}}},\quad i=j\\
    \epsilon_{1}\epsilon_{2}\epsilon_{3}\frac{1}{\epsilon_{c_{i}}},\quad j=i-1.
    \end{dcases}
\end{align}

Let us consider the three examples introduced previously: $q$-Virasoro, $q$-$\cW_{N}$, and $\cW_{q_{1},q_{2}}(\mathfrak{sl}_{2\,|\,1})$.
For example, for $q$-Virasoro ($\mathcal{F}_{3}\otimes \mathcal{F}_{3}$), we only have one root current, and the matrix is
\begin{align}
    B^{[1]}\rightarrow\epsilon_{1}\epsilon_{2}\begin{pmatrix}
    -2
    \end{pmatrix}.
\end{align}
For $q$-$\cW_{3}$ ($\mathcal{F}_{3}\otimes\mathcal{F}_{3}\otimes\mathcal{F}_{3}$), we have
\begin{align}
    B^{[1]}\rightarrow \epsilon_{1}\epsilon_{2}\begin{pmatrix}
    -2&1\\
    1&-2
    \end{pmatrix}
\end{align}
and generally, for $q$-$\cW_{N}$, we have an $(N-1)\times (N-1)$ matrix
\begin{align}
    B^{[1]}\rightarrow\epsilon_{1}\epsilon_{2}\begin{pmatrix}
    -2&1&0&0&\cdots&0\\
    1&-2&1&0&\cdots&0\\
    0&1&-2&1&\cdots&0\\
    \vdots&\vdots&\vdots&\ddots&\cdots&\vdots\\
    0&\cdots&0&1&-2&1\\
    0&\cdots&0&0&1&-2
    \end{pmatrix}.
    \end{align}
Namely, when all of the colors of the tensor products are the same, the degenerate limit of the matrix $B^{[1]}$ is proportional to the Cartan matrix of the underlying Lie algebra. We can also illustrate the associated Dynkin diagram as
\begin{align}
\includegraphics[width=7cm]{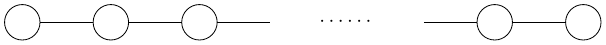}
\end{align}

For $\cW_{q_{1},q_{2}}(\mathfrak{sl}_{2|1})$, let us consider the three possible cases: $\mathcal{F}_{3}\otimes\mathcal{F}_{3}\otimes\mathcal{F}_{2}$, $\mathcal{F}_{3}\otimes\mathcal{F}_{2}\otimes\mathcal{F}_{3}$, and $\mathcal{F}_{2}\otimes\mathcal{F}_{3}\otimes\mathcal{F}_{3}$. We have
\begin{align}\label{eq:sl21figure1}
    B^{[1]}\rightarrow \epsilon_{1}^{2}\begin{pmatrix}
    -\frac{2\epsilon_{2}}{\epsilon_{1}}&\frac{\epsilon_{2}}{\epsilon_{1}}\\
    \frac{\epsilon_{2}}{\epsilon_{1}}&1
    \end{pmatrix},\qquad \includegraphics[scale=0.7]{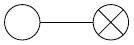}
\end{align}
for $\mathcal{F}_{3}\otimes\mathcal{F}_{3}\otimes\mathcal{F}_{2}$
\begin{align}\label{eq:sl21figure2}
    B^{[1]}\rightarrow \epsilon_{1}^{2}\begin{pmatrix}
    1&\frac{\epsilon_{3}}{\epsilon_{1}}\\
    \frac{\epsilon_{3}}{\epsilon_{1}}&1
    \end{pmatrix},\qquad  \includegraphics[scale=0.7]{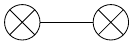}
\end{align}
for $\mathcal{F}_{3}\otimes\mathcal{F}_{2}\otimes\mathcal{F}_{3}$, and
\begin{align}\label{eq:sl21figure3}
    B^{[1]}\rightarrow \epsilon_{1}^{2}\begin{pmatrix}
    1&\frac{\epsilon_{2}}{\epsilon_{1}}\\
    \frac{\epsilon_{2}}{\epsilon_{1}}&-\frac{2\epsilon_{2}}{\epsilon_{1}}
    \end{pmatrix},\qquad  \includegraphics[scale=0.7]{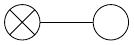}
\end{align}
for $\mathcal{F}_{2}\otimes\mathcal{F}_{3}\otimes\mathcal{F}_{3}$, where we assigned the bosonic node $\bigcirc$ when the adjacent Fock spaces are the same, and the fermionic node $\raisebox{-1.5pt}{\includegraphics[scale=0.6]{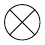}}$ when the adjacent Fock spaces are different. Similarly, we can do the superalgebra generalization for general Dynkin diagrams (see \cite{feigin2021deformations} for details).

\subsection{Screening currents}\label{sec:QTA-screening}

In the preceding section, we derived the generators of the deformed $\cW$-algebra explicitly by analyzing the tensor products of the horizontal representations $\mathcal{F}_{\boldsymbol{c}}(\boldsymbol{u})$. Instead, an alternative approach is to define screening charges and characterize the algebra by the operators that commute with them. Specifically, we can introduce a $q$-deformed version of the screening charges in (\ref{eq:WNscreening}).

To this end, it suffices to study the action of the screening currents on two adjacent Fock spaces of  $\mathcal{F}_{\boldsymbol{c}}(\boldsymbol{u})$. This can be done by examining the screening currents acting on either the same type of Fock spaces, such as $\mathcal{F}_{c}(u_{1})\otimes \mathcal{F}_{c}(u_{2})$
or different types of Fock spaces $\mathcal{F}_{i}(u_{1})\otimes \mathcal{F}_{j}(u_{2})\,(i\neq j)$.

We adopt the following notation for the zero modes:
\bea
    u=\exp\left(-\frac{\epsilon_{1}\epsilon_{2}\epsilon_{3}}{\epsilon_{c}^{2}}\sfa_{0}\right),\quad [\sfa_{0},\hat{Q}]=-\frac{\epsilon_{c}^{3}}{\epsilon_{1}\epsilon_{2}\epsilon_{3}}.
\eea
The normal ordering is defined in the order $\hat{Q},\sfa_{0},\sfa_{-r},\sfa_{r}\,(r>0)$.
\paragraph{Same type}
In this case, we have two screening charges\footnote{The screening charges are Jackson integrals of the screening currents. We denote this schematically as $\int_{0}^{\infty} d_{q}z\,f(z)\coloneqq (1-q)\sum\limits_{n\in\mathbb{Z}}f(q^{n})$, but note that we need to specify this parameter $q$ in explicit computations (see \cite[Lec. 11]{Etingof:1998ru} for example). We can also introduce a parameter dependence as $\int_{0}^{w\infty}d_{q}z\,f(z)=w(1-q)\sum\limits_{n\in\mathbb{Z}}q^{n}f(q^{n}w)$. For showing the commutativity, we do not need the factor $w(1-q)$, so we omit it. Instead of using the Jackson integral, we can also use the ordinary contour integral formula but with some appropriate support function $g(z)$ as $\int d_{q}z\,f(z)=\oint dz g(z)f(z)$. \label{footnote_screening}}
\bea
\mathcal{S}^{\pm}_{cc}=&\int d_{q}z\,\,S_{cc}^{\pm}(z),\cr
S^{+}_{cc}(z)=&e^{\frac{\epsilon_{c+1}}{\epsilon_{c}}\hat{Q}^{-}}z^{\frac{\epsilon_{c+1}}{\epsilon_{c}}\sfa_{0}^{-}+\frac{\epsilon_{c}}{\epsilon_{c-1}}}\exp\left(\sum_{r=1}^{\infty}\frac{-(q_{c+1}^{r/2}-q_{c+1}^{-r/2})}{r(q_{c}^{r/2}-q_{c}^{-r/2})}\sfh_{-r}z^{r}\right)\exp\left(\sum_{r=1}^{\infty}\frac{(q_{c+1}^{r/2}-q_{c+1}^{-r/2})}{r(q_{c}^{r/2}-q_{c}^{-r/2})}\sfh_{r}z^{-r}\right),\cr
S^{-}_{cc}(z)=&e^{\frac{\epsilon_{c-1}}{\epsilon_{c}}\hat{Q}^{-}}z^{\frac{\epsilon_{c-1}}{\epsilon_{c}}\sfa_{0}^{-}+\frac{\epsilon_{c}}{\epsilon_{c+1}}}\exp\left(\sum_{r=1}^{\infty}\frac{-(q_{c-1}^{r/2}-q_{c-1}^{-r/2})}{r(q_{c}^{r/2}-q_{c}^{-r/2})}\sfh_{-r}z^{r}\right)\exp\left(\sum_{r=1}^{\infty}\frac{(q_{c-1}^{r/2}-q_{c-1}^{-r/2})}{r(q_{c}^{r/2}-q_{c}^{-r/2})}\sfh_{r}z^{-r}\right)
\label{eq:same-type-screening}
\eea
where
\bea
    &\sfh_{r}=q_{c}^{-r/2}\sfa_{r,1}-q_{c}^{-r}\sfa_{r,2},\quad \sfh_{-r}=q_{c}^{r}\sfa_{-r,1}-q_{c}^{r/2}\sfa_{-r,2}\quad (r>0),\cr
    &\sfa_{0}^{-}=\sfa_{0,1}-\sfa_{0,2},\quad \hat{Q}^{-}=\hat{Q}_{1}-\hat{Q}_{2},\cr
    &[\sfh_{r},\sfh_{s}]=-\frac{r}{\kappa_{r}}(q_{c}^{r/2}-q_{c}^{-r/2})^{2}(q_{c}^{r}-q_{c}^{-r})\delta_{r+s,0}.
\eea
Note that $\sfa_{r,1}=\sfa_{r}\otimes 1$ and $\sfa_{r,2}=1\otimes \sfa_{r}$. The screening charges commute with the action of the algebra. Let us show this explicitly for $S_{cc}^{+}$, $\sfH_{\pm r}$, and $E(z)$.
Using the commutation relations of $\sfa_{r}$ in (\ref{eq:commuteqboson}), one can immediately show that the $\sfh_{r}$ commutes with $\Delta(\sfH_{r})$
\be
[\Delta(\sfH_{r}),\sfh_{s}]=0
\ee
which means they commute with the Drinfeld currents $\Delta(K^{\pm}(z))$. Using the coproduct once, we obtain
\bea
    \Delta(E(z))=&\frac{1-q_{c}}{\kappa_{1}}\left(\Lambda_{1}(z|u_{1})+\Lambda_{2}(z|u_{2})\right)
\eea
where we use (\ref{eq:lambdadef}). After calculating the contraction formulas and using the residue formula in (\ref{eq:residue_formula}), one will finally get
\bea
    &[\Lambda_{1}(z|u_{1})+\Lambda_{2}(z|u_{2}),S^{+}_{cc}(w)]\cr
    =&-(1-q_{c+1}):\Lambda_{1}(z|u_{1})S^{+}_{cc}(q_{c-1}^{1/2}z):\left(\delta\left(q_{c-1}^{-1/2}w/z\right)-q_{c-1}\delta(q_{c-1}^{1/2}w/z)\right)\\
    =&-(1-q_{c+1}):\Lambda_{1}(z|u_{1})S^{+}_{cc}(q_{c-1}^{1/2}z):\mathcal{D}_{q_{c-1}}^{w}\left[w\delta\left(\frac{w}{q_{c-1}^{1/2}z}\right)\right]
\eea
where we also used
\be
 :\Lambda_{1}(z|u_{1})S^{+}_{cc}(q_{c-1}^{1/2}z):=q_{c}:\Lambda_{2}(z|u_{2})S^{+}_{cc}(q_{c-1}^{-1/2}z):\label{eq:root-screening-correspondence}
\ee
and $\mathcal{D}^{w}_{a}[f(w)]=\frac{f(w)-f(aw)}{w}$. Note that the extra $q_{c}$ comes from the zero modes of the screening current (\ref{eq:same-type-screening}). Note also that equation (\ref{eq:root-screening-correspondence}) gives the relation between the vertex operator part of the screening current and the root current as
\begin{align}
    A(q_{c+1}^{1/2}z)\propto :\frac{S^{+}_{cc}(z)}{S^{+}_{cc}(q_{c-1}z)}:.
\end{align}
To show the commutativity with the screening charges, instead of using the integral form, for example, we can define the screening charge as 
\begin{align}
    \mathcal{S}_{cc}^{+}(z)=\sum_{k\in\mathbb{Z}}q_{c-1}^{k}S^{+}_{cc}(q_{c-1}^{k}z),
\end{align}
where we introduced a spectral parameter dependence in this screening charge (see footnote \ref{footnote_screening}). Then, we will have 
\begin{align}
\begin{split}
    &[\Lambda_{1}(z|u_{1})+\Lambda_{2}(z|u_{2}),\mathcal{S}^{+}_{cc}(w)]\\
    =&-(1-q_{c+1}):\Lambda_{1}(z|u_{1})S^{+}_{cc}(q_{c-1}^{1/2}z):\\
    &\qquad\times\sum_{k\in\mathbb{Z}}\left(q_{c-1}^{k}\delta(q_{c-1}^{k-1/2}w/z)-q_{c-1}^{k+1}\delta(q_{c-1}^{k+1/2}w/z)\right)\\
    =&0
\end{split}
\end{align}
where in the last line, we shifted the second term $k\rightarrow k-1$. This means that $\mathcal{S}^{+}_{cc}(z)$ completely commutes with the currents. 

We have similar relations for the other screening current $S_{cc}^{-}(z)$ such as
\begin{align}
    A(q_{c-1}^{1/2}z)\propto :\frac{S_{cc}^{-}(z)}{S_{cc}^{-}(q_{c+1}z)}:.
\end{align}
The contraction between the two screening currents $S_{cc}^{+}(z),S_{cc}^{-}(w)$ is
\bea
    S_{cc}^{+}(z)S_{cc}^{-}(w)=\frac{1}{(z-q_{c}^{1/2}w)(z-q_{c}^{-1/2}w)}:S_{cc}^{+}(z)S_{cc}^{-}(w):=S_{cc}^{-}(w)S_{cc}^{+}(z)
\eea
which means they formally commute with each other.

\paragraph{Mixed type}
Let us consider only the case $\mathcal{F}_{1}(u_{1})\otimes \mathcal{F}_{2}(u_{2})$. Other types of screening currents are obtained by changing the $q$ parameters: $q_{1}\leftrightarrow q_{2}\leftrightarrow q_{3}$. Similarly to the previous case, we introduce a screening current
\bea
    S_{12}(z)=e^{\frac{\epsilon_{2}}{\epsilon_{1}}\hat{Q}_{1}-\frac{\epsilon_{1}}{\epsilon_{2}}\hat{Q}_{2}}z^{\frac{\epsilon_{2}}{\epsilon_{1}}\sfa_{0,1}-\frac{\epsilon_{1}}{\epsilon_{2}}\sfa_{0,2}+\frac{\epsilon_{2}}{\epsilon_{3}}}\exp\left(\sum_{r=1}^{\infty}-\frac{\sfh_{-r}}{r}z^{r}\right)\exp\left(\sum_{r=1}^{\infty}\frac{\sfh_{r}}{r}z^{-r}\right),\label{eq:mixed-type-screening}
\eea
where
\bea
    \sfh_{r}=&\frac{q_{1}^{-r/2}(q_{2}^{r/2}-q_{2}^{r/2})}{q_{1}^{r/2}-q_{1}^{-r/2}}\sfa_{r,1}-\frac{q_{3}^{r/2}(q_{1}^{r/2}-q_{1}^{-r/2})}{q_{2}^{r/2}-q_{2}^{-r/2}}\sfa_{r,2},\cr
    \sfh_{-r}=&\frac{q_{1}^{r}(q_{2}^{r/2}-q_{2}^{-r/2})}{q_{1}^{r/2}-q_{1}^{-r/2}}\sfa_{-r,1}-\frac{q_{1}^{-r/2}(q_{1}^{r/2}-q_{1}^{-r/2})}{q_{2}^{r/2}-q_{2}^{-r/2}}\sfa_{-r,2},\cr
    [\sfh_{r},\sfh_{s}]=&r\delta_{r+s,0}.
\eea
By direct calculation, one can also see
\be
    [\Delta({\sfH_{r}}),\sfh_{s}]=0
\ee
and show that the modes commute with the Drinfeld currents $\Delta(K^{\pm}(z))$. Using the coproduct, we have
\be
    \Delta(E(z))=\frac{1-q_{1}}{\kappa_{1}}\Lambda_{1}(z|u_{1})+\frac{1-q_{2}}{\kappa_{1}}\Lambda_{2}(z|u_{2}).
\ee
from (\ref{eq:lambdadef}). Using the free field realizations and the residue formula (\ref{eq:residue_formula}), we have
\be
    [\Delta E(z),S_{12}(w)]=-\frac{(1-q_{1})(1-q_{2})}{\kappa_{1}}:\Lambda_{1}(z|u_{1})S_{12}(q_{3}^{1/2}z):\mathcal{D}_{q_{3}}^{w}\left[w\delta\left(q_{3}^{-1/2}\frac{w}{z}\right)\right]
\ee
where we use
\be
    :\Lambda_{1}(z|u_{1})S_{12}(q_{3}^{1/2}z):=q_{2}:\Lambda_{2}(z|u_{2})S_{12}(q_{3}^{-1/2}z):.
\ee
From this relation, it is obvious that the vertex operator part of the screening operator is related to the root current as
\begin{align}
    A(q_{2}^{-1/2}z)\propto:\frac{S_{12}(z)}{S_{12}(q_{3}z)}:.
\end{align}
We also note
\be
    S_{12}(z)S_{12}(w)=(z-w):S_{12}(z)S_{12}(w):=-S_{12}(w)S_{12}(z)
\ee
which means the screening currents anti-commute with each other. This type of screening current is called \emph{fermionic} screening current.

\paragraph{Degenerate limit}
Let us compare the degenerate limit of the screening currents in (\ref{eq:same-type-screening}) and (\ref{eq:mixed-type-screening}) with the screening currents in (\ref{eq:degenerate-screening}).

For the same-type screening currents (\ref{eq:same-type-screening}), the commutation relations of the free bosons in (\ref{eq:commuteqboson}) will transform in the limit $q_{1,2,3}\rightarrow 1$ as
\begin{align}
\sfa_{r}^{(c)}\rightarrow \sfa_{\text{deg},r}^{(c)},\quad
    [\sfa_{\text{deg},r}^{(c)},\sfa_{\text{deg},s}^{(c)}]=-r\frac{\epsilon_{c}^{3}}{\epsilon_{1}\epsilon_{2}\epsilon_{3}}\delta_{r+s,0}.
\end{align}
If we rescale the free bosons as
\begin{align}
    \sfa_{\text{deg},0}=\sqrt{\frac{\epsilon_{c}\epsilon_{c}}{\epsilon_{c+1}\epsilon_{c-1}}}\sfJ_{0},\quad \hat{Q}_{\text{deg}}=-\sqrt{\frac{\epsilon_{c}\epsilon_{c}}{\epsilon_{c+1}\epsilon_{c-1}}}\sfq_{0},\quad \sfa_{\text{deg},\pm r}^{(c)}=\pm\sqrt{\frac{\epsilon_{c}\epsilon_{c}}{\epsilon_{c+1}\epsilon_{c-1}}}\sfJ_{\pm r}\quad (r>0)\,,
\end{align}
the screening currents will transform as
\begin{align}
 S_{cc}^{\pm}(z)\rightarrow e^{-\sqrt{\frac{\epsilon_{c\pm1}}{\epsilon_{c\mp1}}}\sfq^{-}_{0}}z^{\frac{\epsilon_{c}}{\epsilon_{c\mp1}}+\sqrt{\frac{\epsilon_{c\pm1}}{\epsilon_{c\mp1}}} \sfJ^{-}_{0}}\exp\left(\sqrt{\frac{\epsilon_{c\pm1}}{\epsilon_{c\mp1}}}\sum_{r=1}^{\infty}\frac{\sfh_{\text{deg},-r}}{r}z^{r}\right)\exp\left(-\sqrt{\frac{\epsilon_{c\pm1}}{\epsilon_{c\mp1}}}\sum_{r=1}^{\infty}\frac{\sfh_{\text{deg},r}}{r}z^{-r}\right)
\end{align}
where
\begin{align}
\begin{split}
    &\sfq_{0}^{-}=\sfq_{0}\otimes 1-1\otimes\sfq_{0},\quad \sfJ_{0}^{-}=\sfJ_{0}\otimes1-1\otimes\sfJ_{0},\\
    &\sfh_{\text{deg},\pm r}=\sfJ_{\pm r}\otimes 1-1\otimes\sfJ_{\pm r}\,\,(r>0)
\end{split}
\end{align}
and reproduce the results in (\ref{eq:degenerate-screening}).

Similarly, for the mixed-type screening current, we have
\begin{align}
\begin{split}
    S_{12}(z)\rightarrow& e^{-(\sqrt{\frac{\epsilon_{2}}{\epsilon_{3}}}\sfq_{0}\otimes1-\sqrt{\frac{\epsilon_{1}}{\epsilon_{3}}}1\otimes\sfq_{0})}z^{\frac{\epsilon_{2}}{\epsilon_{3}}+\sqrt{\frac{\epsilon_{2}}{\epsilon_{3}}}\sfJ_{0}\otimes 1-\sqrt{\frac{\epsilon_{3}}{\epsilon_{2}}}1\otimes\sfJ_{0}}\\
    &\exp\left(\sum_{r=1}^{\infty}-\frac{\sfh_{\text{deg},-r}}{r}z^{r}\right)\exp\left(\sum_{r=1}^{\infty}\frac{\sfh_{\text{deg},r}}{r}z^{-r}\right)
\end{split}
\end{align}
where
\begin{align}
    \sfh_{\text{deg},\pm r}=\pm\left(\sqrt{\frac{\epsilon_{2}}{\epsilon_{3}}}\sfJ_{\pm r}\otimes 1-\sqrt{\frac{\epsilon_{1}}{\epsilon_{3}}}1\otimes\sfJ_{\pm r}\right)\,\,(r>0),
\end{align}
which meets (\ref{eq:degenerate-screening}) again.

\subsection{Intertwiners}\label{sec:intertwiner}
In this section, we introduce algebraic objects of quantum toroidal $\mathfrak{gl}_{1}$ called \emph{intertwiners}. Originally, these algebraic quantities were introduced in the context of quantum affine algebras to study correlation functions and the integrability of XXZ-models (for a nice textbook on this, see \cite{jimbo1994algebraic}). An interesting property of the intertwiners of the quantum toroidal $\mathfrak{gl}_{1}$ is that they are directly related to not only solvable lattice models but physical observables of supersymmetric gauge theories. We will discuss this property in \S\ref{sec:AFStopvertex}. In this section, we will only discuss the algebraic aspects of them.

From now on, we only consider the Fock representations. As one can see from the vertical representations, we have a degree of freedom to change the vectors in the following form
\be
    \ket{u,\lambda}^{(c)}\rightarrow\mathcal{N}^{(c)}(u,\lambda)\ket{u,\lambda}^{(c)}
\ee
as long as $\mathcal{N}^{(c)}(u,\lambda)$ does not change the pole structure drastically. Note that this process does not change the eigenvalue of $K^{\pm}(z)$ but only changes the factors appearing in the action of $E(z)$ and $F(z)$.
We use this degree of freedom and use the following notation
\bea
    E(z)\ket{u,\lambda}=&\frac{1-q_{c}}{\kappa_{1}}\sum_{x\in \frakA(\lambda)}\delta\left(\chi_{x}/z\right)\underset{z= \chi_{x}}{\Res}\frac{1}{z\mathcal{Y}^{(c)}_{\lambda}(z,u)}\ket{u,\lambda+x},\cr
    F(z)\ket{u,\lambda}=&-\frac{1-q_{c}^{-1}}{\kappa_{1}}q_{c}^{-1/2}\sum_{x\in \frakR(\lambda)}\delta(\chi_{x}/z)\underset{z= \chi_{x}}{\Res}z^{-1}\mathcal{Y}^{(c)}_{\lambda}(q_{c}^{-1}z,u)\ket{u,\lambda-x},\cr
    K^{\pm}(z)\ket{u,\lambda}=&[\Psi^{(c)}_{\lambda}(z,u)]_{\pm}\ket{u,\lambda},\label{DIM-rep}
\eea
where we omit the superscript $(c)$ of the vectors. To define the intertwiners, we also need to use the dual representation. The dual representation is a representation acting on the bra state. It is defined using the following identity
\be
    \bra{u,\mu}\bl(g(z)\ket{u,\lambda}\br)=\bl(\bra{u,\mu}g(z)\br)\ket{u,\lambda},\quad g(z)\in \QTA.
\ee
We further need to assign the normalization of the vectors
\be
   \braket{u,\mu}{u,\lambda}=\Xi_{\lambda}^{-1}\delta_{\lambda,\mu}
\ee
where $\Xi_{\lambda}$ is some factor. For later convenience, we define normalization as
\be\label{normalization}
    \Xi_{\lambda}=\frac{(uq_{c}^{1/2})^{-|\lambda|}\prod_{x\in\lambda}\chi_{x}}{N^{(c)}_{\lambda\lambda}(1;q_{1},q_{2})}
\ee
where
\bea
    N^{(c)}_{\lambda_{1}\lambda_{2}}(u_{1}/u_{2};q_{1},q_{2})=\prod_{\substack{x\in\lambda_{1}\\ y\in\lambda_{2}}}S_{c}\left(\frac{\chi_{x}}{\chi_{y}}\right)\times\prod_{x\in\lambda_{1}}\left(1-\frac{\chi_{x}}{q_{c}u_{2}}\right)\times\prod_{x\in\lambda_{2}}\left(1-\frac{u_{1}}{\chi_{x}}\right)\label{eq:Nekrasovfactor-general}
\eea
is the Nekrasov factor (the physical interpretation of this factor is discussed in \S\ref{sec:AFStopvertex}). See (\ref{def-Yc}) for the definition of the box content $\chi_{x}$. We add an index $(c)$ denoting the central charge of the representation space. Later, we will set this $c=3$ and omit this index.

Using this normalization, one can show that
\be
    \bra{u,\lambda+x}E(z)\ket{u,\lambda}=-q_{c}\bra{u,\lambda}F(z)\ket{u,\lambda+x}
\ee
and the dual representation will be
\bea
    \bra{u,\lambda}E(z)=&-\frac{1-q_{c}}{\kappa_{1}}q_{c}^{-1/2}\sum_{x\in \frakR(\lambda)}\delta\left(\chi_{x}/z\right)\underset{z= \chi_{x}}{\Res}z^{-1}\mathcal{Y}^{(c)}_{\lambda}(q_{c}^{-1}z,u)\bra{u,\lambda-x},\cr
    \bra{u,\lambda}F(z)=&\frac{1-q_{c}^{-1}}{\kappa_{1}}\sum_{x\in \frakA(\lambda)}\delta\left(\chi_{x}/z\right)\underset{z= \chi_{x}}{\Res}\frac{1}{z\mathcal{Y}^{(c)}_{\lambda}(z,u)}\bra{u,\lambda+x},\cr
    \bra{u,\lambda}K^{\pm}(z)=&[\Psi^{(c)}_{\lambda}(z,u)]_{\pm}\bra{u,\lambda}.
\eea

Using the above conventions, let us introduce the intertwiners between different representations. To simplify the discussion, we only consider representations with central charges $(C,K^{-})=(q_{3}^{1/2},q_{3}^{n/2}), (1,q_{3}^{1/2})$. We call them representations with level\footnote{\label{footnote:level}In this subsection, we use the notation $\rho_{u}^{(0,1)}$ and $\rho_{u}^{(1,0)}$ to represent the vertical and horizontal Fock representation, respectively. The specific form of the level $(1,n)_{u}$ representation, denoted by $\rho_{u}^{(1,n)}$, is provided in  (\ref{eq:horizon_rep_general}).} $(1,n)_{u}$ and $(0,1)_{u}$, respectively, and the subscript denotes the spectral parameter. We also set $\gamma\coloneqq q_{3}^{1/2}$ for later convenience.
\paragraph{Intertwiner}
We define and illustrate the intertwiner as a map
\bea
\adjustbox{valign=c}{\includegraphics{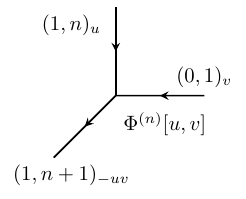}}=&\begin{array}{l}\Phi^{(n)}[u,v]:(0,1)_{v}\otimes (1,n)_{u}\rightarrow (1,n+1)_{u'},\\
    \Phi^{(n)}[u,v](\ket{v,\lambda}\otimes 1)=\Phi^{(n)}_{\lambda}[u,v]:(1,n)_{u}\rightarrow (1,n+1)_{u'},\\
    \Phi^{(n)}[u,v]=\sum\limits_{\lambda}\Xi_{\lambda}\bra{v,\lambda}\otimes \Phi_{\lambda}^{(n)}[u,v].
    \end{array}
\eea
The second and third equations come from the fact that we can take $\ket{v,\lambda}$ as the basis of the representation $(0,1)_{v}$. Namely, the intertwiner can be understood as an infinite dimensional diagonal matrix whose elements are $\Phi^{(n)}[u,v]=(\Phi^{(n)}_{\lambda}[u,v])_{\lambda}$. The following equation characterizes the intertwiner
\bea
    \rho_{u'}^{(1,n+1)}(e)\Phi^{(n)}[u,v]=\Phi^{(n)}[u,v](\rho_{v}^{(0,1)}\otimes \rho_{u}^{(1,n)})\Delta(e),\quad e\in  \QTA ~.\label{eq:intertwinerprop}
\eea
Inserting the Drinfeld currents in this equation, we obtain
\begin{align}
    u'\left(\frac{\gamma}{z}\right)^{n+1}\eta(z)\Phi^{(n)}_{\lambda}[u,v]=
    &u\left(\frac{\gamma}{z}\right)^{n}[\Psi_{\lambda}(z,v)]_{-}\Phi_{\lambda}^{(n)}[u,v]\eta(z) \label{eq:intertwinerE}\\
    &+\sum_{x\in \frakA(\lambda)}\delta(z/\chi_{x})\underset{z=\chi_{x}}{\Res}z^{-1}\mathcal{Y}_{\lambda}(z,v)^{-1}\Phi_{\lambda+x}^{(n)}[u,v]\nonumber\\
     u'^{-1}\left(\frac{z}{\gamma}\right)^{n+1}\xi(z)\Phi^{(n)}_{\lambda}[u,v]=&u^{-1}\left(\frac{z}{\gamma}\right)^{n}\Phi^{(n)}_{\lambda}[u,v]\xi(z),\label{eq:intertwinerF}\\
     &+\gamma^{-(n+1)}\sum_{x\in \frakR(\lambda)}\delta\left(\frac{\gamma z}{\chi_{x}}\right)\underset{z=\chi_{x}}{\Res}z^{-1}\mathcal{Y}_{\lambda}(q_{3}^{-1}z,v)\Phi^{(n)}_{\lambda-x}[u,v]\varphi^{+}(\gamma z)\nonumber\\
\gamma^{-1}\varphi^{+}(z)\Phi^{(n)}_{\lambda}[u,v]=&\left[\Psi_{\lambda}(z,v)\right]_{+}\Phi_{\lambda}^{(n)}[u,v]\varphi^{+}(z),\label{eq:intertwinercartan1}\\
\gamma\varphi^{-}(z)\Phi_{\lambda}^{(n)}[u,v]=&\left[\Psi_{\lambda}(\gamma^{-1}z,v)\right]_{-}\Phi^{(n)}_{\lambda}[u,v]\varphi^{-}(z).\label{eq:intertwinercartan2}
\end{align}
The intertwiner is determined uniquely:
\be
    \Phi^{(n)}_{\lambda}[u,v]=t_{n}(\lambda,u,v):\Phi_{\emptyset}[v]\prod_{x\in \lambda}\eta(\chi_{x}):,\label{inttw-def}
\ee
where
\bea
    \Phi_{\emptyset}[v]=&\exp\left(-\sum_{r=1}^{\infty}\frac{v^{r}}{r}\frac{\sfa_{-r}}{1-q_{3}^{-r}}\right)\exp\left(\sum_{r=1}^{\infty}\frac{v^{-r}}{r}\frac{\gamma^{-r}}{1-q_{3}^{r}}\sfa_{r}\right),\\
    [\sfa_{r},\sfa_{s}]=&-\frac{r}{\kappa_{r}}(q_{3}^{r/2}-q_{3}^{-r/2})^{3}\delta_{r+s,0},\\
    t_{n}(\lambda,u,v)=&(-uv)^{|\lambda|}\prod_{x\in\lambda}(\gamma/\chi_{x})^{n+1},\quad u'=-uv.\label{Phi-int-def}
\eea
The strategy to derive this is to start from the commutation relation with the Cartan part $\varphi^{\pm}(z)$ in (\ref{eq:intertwinercartan1}) and (\ref{eq:intertwinercartan2}). From these equations, we can derive the operator part and show that
\be
    \Phi^{(n)}_{\lambda}[u,v]\propto:\Phi_{\emptyset}[v]\prod_{x\in\lambda}\eta(\chi_{x}):.
\ee
Using the other two equations (\ref{eq:intertwinerE}) and (\ref{eq:intertwinerF}), one can see that the intertwiner exists only when $u'=-uv$. Then, using this condition, we can derive the zero mode factor $t_{n}(\lambda,u,v)$. See appendix \ref{appendix:contraction} for the contraction formulas of these operators.

\paragraph{Dual intertwiner}
We also need the dual intertwiner defined as a map
\bea
\adjustbox{valign=c}{\includegraphics{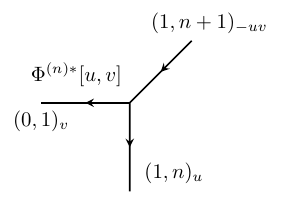}}=&
\begin{array}{l}
    \Phi^{(n)*}[u,v]:(1,n+1)_{u'}\rightarrow (1,n)_{u}\otimes (0,1)_{v},\\
    \Phi^{(n)*}[u,v](1\otimes \ket{v,\lambda})=\Phi^{(n)*}_{\lambda}[u,v]:(1,n+1)_{u'}\rightarrow (1,n)_{u},\\
    \Phi^{(n)*}[u,v]=\sum\limits_{\lambda}\Xi_{\lambda}\Phi^{(n)*}_{\lambda}[u,v]\otimes \ket{v,\lambda}.
\end{array}
\eea
The property defining this dual intertwiner is
\bea
    \Phi^{(n)*}[u,v]\rho_{u'}^{(1,n+1)}(e)=(\rho_{u}^{(1,n)}\otimes \rho_{v}^{(0,1)})\Delta(e)\Phi^{(n)*}[u,v],\quad e\in \QTA \label{eq:dualinterwinerprop}
\eea
and expanding it, we obtain
\begin{align}
        u'\left(\frac{\gamma}{z}\right)^{n+1}\Phi^{(n)*}_{\lambda}[u,v]\eta(z)=&u\left(\frac{\gamma}{z}\right)^{n}\eta(z)\Phi^{(n)*}_{\lambda}[u,v]\label{eq:dintertwinerE}\\
        &-\gamma^{n-1}\sum_{x\in \frakR(\lambda)}\delta\left(\gamma z/\chi_{x}\right)\underset{z=\chi_{x}}{\Res}z^{-1}\mathcal{Y}_{\lambda}(q_{3}^{-1}z,v)\varphi^{-}(\gamma z)\Phi^{(n)*}_{\lambda-x}[u,v],\notag\\
        u'^{-1}\left(\frac{z}{\gamma}\right)^{n+1}\Phi^{(n)*}_{\lambda}[u,v]\xi(z)=&u^{-1}\left(\frac{z}{\gamma}\right)^{n}[\Psi_{\lambda}(z,v)]_{+}\xi(z)\Phi^{(n)*}_{\lambda}[u,v]\notag\\
            &-\sum_{x\in \frakA(\lambda)}\delta\left(z/\chi_{x}\right)\underset{z=\chi_{x}}{\Res}z^{-1}\mathcal{Y}_{\lambda}(z,v)^{-1}\Phi^{(n)*}_{\lambda+x}[u,v],\label{eq:dintertwinerF}\\
        \gamma^{-1}\Phi^{(n)*}_{\lambda}[u,v]\varphi^{+}(z)=&[\Psi_{\lambda}(\gamma^{-1}z,v)]_{+}\varphi^{+}(z)\Phi^{(n)*}_{\lambda}[u,v],\label{eq:dintertwinercartan1}\\
        \gamma\Phi^{(n)*}_{\lambda}[u,v]\varphi^{-}(z)=&[\Psi_{\lambda}(z,v)]_{-}\varphi^{-}(z)\Phi^{(n)*}_{\lambda}[u,v]\label{eq:dintertwinercartan2}.
\end{align}
The solution of these equations is determined uniquely in the following form
\be
    \Phi^{(n)*}_{\lambda}[u,v]=t_{n}^{*}(\lambda,u,v):\Phi_{\emptyset}^{*}[v]\prod_{x\in\lambda}\xi(\chi_{x}):,\label{dual-int-def}
\ee
where
\begin{align}
    &\Phi^{*}_{\emptyset}[v]=\exp\left(\sum_{r=1}^{\infty}\frac{v^{r}}{r}\frac{\gamma^{r}}{1-q_{3}^{-r}}\sfa_{-r}\right)\exp\left(-\sum_{r=1}^{\infty}\frac{v^{-r}}{r}\frac{1}{1-q_{3}^{r}}\sfa_{r}\right),\label{dPhi-int-def}\\
    &t^{*}_{n}(\lambda,u,v)=(u\gamma)^{-|\lambda|}\prod_{x\in\lambda}(\chi_{x}/\gamma)^{n},\quad u'=-uv.
\label{t*}
\end{align}
Similar to the intertwiner, the operator part is determined from the commutation relations with the Cartan part $\varphi^{\pm}(z)$ in (\ref{eq:dintertwinercartan1}) and (\ref{eq:dintertwinercartan2}). The remaining equations (\ref{eq:dintertwinerE}) and (\ref{eq:dintertwinerF}) give the constraint $u'=-uv$ and the zero mode part $t^{*}_{n}(\lambda,u,v)$. See appendix \ref{appendix:contraction} for related contraction formulas of these operators.

\paragraph{Compositions of the intertwiners}Let us study the gluing rules of the intertwiners. We have two different representations (vertical and horizontal), and the gluing rules depend on which representation is glued.

For the gluings in the horizontal representations, we have four patterns:
\begin{align}
\includegraphics[width=13cm]{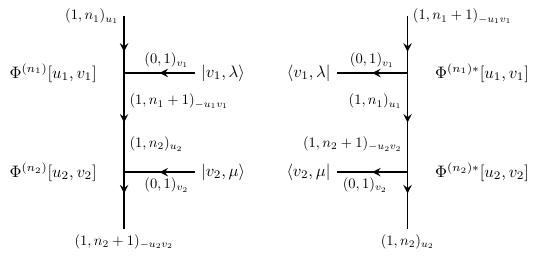}\cr
\includegraphics[width=13cm]{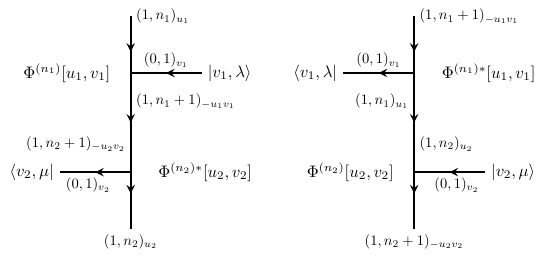}\label{eq:gluing_intertwiner_figure}
\end{align}
where we took the matrix elements in the vertical representations. After this process, the gluings are contractions in the vertex operator representations (horizontal representation):
\bea
\Phi^{(n_{2})}_{\mu}[u_{2},v_{2}]\Phi^{(n_{1})}_{\lambda}[u_{1},v_{1}]=&\mathcal{G}(q_{3}^{-1}v_{1}/v_{2})N_{\lambda\mu}(v_{1}/v_{2};q_{1},q_{2})^{-1}:\Phi^{(n_{2})}_{\mu}[u_{2},v_{2}]\Phi^{(n_{1})}_{\lambda}[u_{1},v_{1}]:,\\
\Phi^{(n_{2})\ast}_{\mu}[u_{2},v_{2}]\Phi^{(n_{1})\ast}_{\lambda}[u_{1},v_{1}]=&\mathcal{G}(v_{1}/v_{2})N_{\lambda\mu}(q_{3}v_{1}/v_{2};q_{1},q_{2})^{-1}:\Phi^{(n_{2})\ast}_{\mu}[u_{2},v_{2}]\Phi^{(n_{1})\ast}_{\lambda}[u_{1},v_{1}]:,\\
\Phi_{\mu}^{(n_{2})\ast}[u_{2},v_{2}]\Phi_{\lambda}^{(n_{1})}[u_{1},v_{1}]=&\mathcal{G}(\gamma^{-1}v_{1}/v_{2})^{-1}N_{\lambda\mu}(\gamma v_{1}/v_{2};q_{1},q_{2}):\Phi_{\mu}^{(n_{2})\ast}\Phi_{\lambda}^{(n_{1})}[u_{1},v_{1}]:,\\
\Phi^{(n_{2})}_{\mu}[u_{2},v_{2}]\Phi^{(n_{1})\ast}_{\lambda}[u_{1},v_{1}]=&\mathcal{G}(\gamma^{-1}v_{1}/v_{2})^{-1}N_{\lambda\mu}(\gamma v_{1}/v_{2};q_{1},q_{2}):\Phi^{(n_{2})}_{\mu}[u_{2},v_{2}]\Phi^{(n_{1})\ast}_{\lambda}[u_{1},v_{1}]:,\label{eq:intertwiner-contraction}
\eea
where\footnote{This function $\mathcal{G}(z)$ will be related to the one-loop part of the partition function which will be introduced in (\ref{eq:one-loop}). Note that this is just the $q$-deformation of the double gamma function.}
\bea
\mathcal{G}(z)=\exp\left(-\sum_{n=1}^{\infty}\frac{z^{n}}{n}\frac{1}{(1-q_{1}^{n})(1-q_{2}^{n})}\right).\label{eq:oneloopintertwiner}
\eea 
Note that when we glue the intertwiners, the levels and the spectral parameters should obey the conservation laws. For example, the gluing rule for the first gluing $\Phi^{(n_{2})}_{\mu}\Phi^{(n_{1})}_{\lambda}$ is
\bea
n_{1}+1=n_{2},\quad u_{2}=-u_{1}v_{1}.
\eea
Other conservation laws can be read from (\ref{eq:gluing_intertwiner_figure}).

For the gluing in the vertical representation, we have
\bea
\adjustbox{valign=c}{\includegraphics{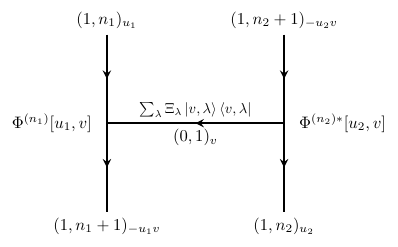}}=\sum_{\lambda}\Xi_{\lambda}\Phi^{(n_{2})\ast}_{\lambda}[u_{2},v]\otimes \Phi^{(n_{1})}_{\lambda}[u_{1},v]\label{eq:verticalglue}
\eea
where we inserted $\mathbb{1}=\sum_{\lambda}\Xi_{\lambda}\ket{v,\lambda}\bra{v,\lambda}$ and draw the intertwiners without bending the arrows. From now on, we draw the lines in straight lines.

\paragraph{Relation with screening currents}
After gluing the vertical representations and using the two defining relations (\ref{eq:intertwinerprop}) and (\ref{eq:dualinterwinerprop}), one can show that the operator $\Phi^{(n_{1})}[u_{1},v]\cdot \Phi^{(n_{2})\ast}[u_{2},v]$ in (\ref{eq:verticalglue}), where the product $\cdot$ means gluing in the vertical representation, obeys the following property
\begin{align}
\begin{split}
   \Phi^{(n_{1})}[u_{1},v]\cdot \Phi^{(n_{2})\ast}[u_{2},v] (\rho_{-u_{2}v}^{(1,n_{2}+1)}&\otimes \rho_{u_{1}}^{(1,n_{1})})\Delta(g)\\
   &=(\rho_{u_{2}}^{(1,n_{2})}\otimes \rho_{-u_{1}v}^{(1,n_{1}+1)})       \Delta(g)\Phi^{(n_{1})}[u_{1},v]\cdot \Phi^{(n_{2})\ast}[u_{2},v]
\end{split}
\end{align}
where $g\in\QTA$. This relation resembles the property of the screening charge, and actually, the intertwiners can be rewritten in terms of the screening currents. By omitting the zero mode parts, let us focus only on the non-zero modes:
\begin{align}
    \Phi^{(n_{2})\ast}_{\lambda}[u_{2},v]\otimes \Phi^{(n_{1})}_{\lambda}[u_{1},v]\leadsto :\Phi_{\emptyset}^{\ast}[v]\otimes \Phi_{\emptyset}[v]\prod_{x\in\lambda}\xi(\chi_{x})\otimes \eta(\chi_{x}):.
\end{align}
Using 
\begin{align}
\begin{split}
    \Phi_{\emptyset}^{\ast}[v]\otimes \Phi_{\emptyset}[v]&=\exp\left(\sum_{r=1}^{\infty}\frac{v^{r}}{r}\frac{\gamma^{r}\mathsf{a}_{-r,1}-\mathsf{a}_{-r,2}}{1-q_{3}^{-r}}\right)\exp\left(\sum_{r=1}^{\infty}\frac{v^{-r}}{r}\frac{-\mathsf{a}_{r,1}+\gamma^{-r}\mathsf{a}_{r,2}}{1-q_{3}^{r}}\right)\\
    \xi(z)\otimes \eta(z)&=\exp\left(-\sum_{r>0}\frac{z^{r}}{r}\frac{(1-q_{1}^{r})(1-q_{2}^{r})}{1-q_{3}^{-r}}(\gamma^{r}\mathsf{a}_{-r,1}-\mathsf{a}_{-r,2})\right)\\
    &\qquad \times \exp\left(-\sum_{r>0}\frac{z^{-r}}{r}\frac{(1-q_{1}^{-r})(1-q_{2}^{-r})}{1-q_{3}^{r}}(-\mathsf{a}_{r,1}+\gamma^{-r}\mathsf{a}_{r,2})\right)
\end{split}
\end{align}
where $\mathsf{a}_{\pm r, 1}=\mathsf{a}_{\pm r}\otimes 1$ and $\mathsf{a}_{\pm r,2}=1\otimes \mathsf{a}_{\pm r}$, we have 
\begin{align}
\begin{split}
    &:\Phi_{\emptyset}^{\ast}[v]\otimes \Phi_{\emptyset}[v]\prod_{x\in\lambda}\xi(\chi_{x})\otimes \eta(\chi_{x}):\\
    =&\exp\left(\sum_{r=1}^{\infty}\frac{1}{r}\frac{\gamma^{r}\mathsf{a}_{-r,1}-\mathsf{a}_{-r,2}}{1-q_{3}^{-r}}\left(v^{r}-(1-q_{1}^{r})(1-q_{2}^{r})\sum_{x\in\lambda}\chi_{x}^{r}\right)\right)\\
    &\times \exp\left(\sum_{r=1}^{\infty}\frac{1}{r}\frac{-\mathsf{a}_{r,1}+\gamma^{-r}\mathsf{a}_{r,2}}{1-q_{3}^{r}}\left(v^{-r}-(1-q_{1}^{-r})(1-q_{2}^{-r})\sum_{x\in\lambda}\chi_{x}^{-r}\right)\right).
\end{split}
\end{align}
The vertex operator part of the screening current $S^{+}(z)\coloneqq S_{33}^{+}(q_{1}^{1/2}z)$, where we shifted the variables for convenience, is 
\begin{align}
    S^{+}_{\text{vert}}(z)=\exp\left(\sum_{r=1}^{\infty}\frac{1}{r}\frac{1-q_{1}^{r}}{1-q_{3}^{-r}}(\gamma^{r}\mathsf{a}_{-r,1}-\mathsf{a}_{-r,2})z^{r}\right)\exp\left(\sum_{r=1}^{\infty}\frac{1}{r}\frac{1-q_{1}^{-r}}{1-q_{3}^{r}}(-\mathsf{a}_{r,1}+\gamma^{-r}\mathsf{a}_{r,2})z^{-r}\right).
\end{align}
Using the identities in (\ref{eq:Xvariables1}) and (\ref{eq:Xvariables2}), we have 
\begin{align}
    :\Phi_{\emptyset}^{\ast}[v]\otimes \Phi_{\emptyset}[v]\prod_{x\in\lambda}\xi(\chi_{x})\otimes \eta(\chi_{x}):=:\prod_{i=1}^{\infty}S^{+}_{\text{vert}}(xq_{1}^{i-1}q_{2}^{\lambda_{i}}):~.\label{eq:interwtiner-screening}
\end{align}
This establishes a connection between the intertwiner and the screening currents. A similar computation can be performed using the other screening current $S^{-}_{33}(z)$, but in this case, the variables inside the screening currents will be determined by the transposed Young diagram $xq_{1}^{\lambda^{t}_{j}}q_{2}^{j-1}$. A more comprehensive analysis, which includes the zero-mode contributions and the coefficient factor $\Xi_{\lambda}$, is possible but will be omitted. 

The introduced intertwiner will be used later to show the algebraic derivation of physical observables such as the instanton partition functions,  and it arises as a consequence of AGT correspondence or BPS/CFT correspondence. Another approach to the BPS/CFT correspondence is through the quiver $\cW$-algebra formalism introduced in \cite{Kimura:2015rgi}. In that formalism, screening currents play a significant role. (See \S\ref{sec:final}.) Interestingly, the equivalence between the quiver $\cW$-algebra formalism and the intertwiner formalism stems from the identity (\ref{eq:interwtiner-screening}).

\paragraph{Generalized intertwiners} The intertwiners defined above are obtained using one vertical representation. Instead, we can consider compositions of the intertwiners and obtain a generalization where we have multiple vertical representations \cite{Bourgine:2017jsi}. We denote them as
\bea
\Phi^{(n,m)}[u,\boldsymbol{v}]&\coloneqq \Phi^{(n_{m})}[u_{m},v_{m}]\cdots\Phi^{(n_{2})}[u_{2},v_{2}]\Phi^{(n_{1})}[u_{1},v_{1}],\\
(0,1)_{v_{m}}&\otimes(0,1)_{v_{m-1}}\cdots\otimes(0,1)_{v_{1}}\otimes (1,n)_{u}\rightarrow (1,n+m)_{u\prod_{i=1}^{m}(-v_{i})},\\
u_{1}=u,&\quad u_{i}=u\prod_{j=1}^{i-1}(-v_{j}),\quad n_{1}=n,\quad n_{i}=n+i-1,\label{eq:generalized-intertwiner1}
\eea
and
\bea
\Phi^{(n^{\ast},m)\ast}[u^{\ast},\boldsymbol{v}]&\coloneqq\Phi^{(n_{m}^{\ast})\ast}[u_{m}^{\ast},v_{m}]\cdots\Phi^{(n_{2}^{\ast})\ast}[u_{2}^{\ast},v_{2}]\Phi^{(n_{1}^{\ast})\ast}[u_{1}^{\ast},v_{1}],\\
(1,n^{\ast}+m)_{u^{\ast}\prod_{j=1}^{m}(-v_{j})}&\rightarrow (1,n^{\ast})_{u^{\ast}}\otimes (0,1)_{v_{m}}\otimes (0,1)_{v_{m-1}}\otimes\cdots\otimes (0,1)_{v_{1}}, \\
u_{m}^{\ast}=u^{\ast},&\quad u_{i}^{\ast}=u^{\ast}\prod_{j=i+1}^{m}(-v_{j}),\quad n_{m}^{\ast}=n^{\ast},\quad n_{i}^{\ast}=n+m-i \label{eq:generalized-intertwiner2}
\eea
where $\boldsymbol{v}=(v_{1},v_{2},\ldots,v_{m})$ and the compositions are
done in the horizontal representations. We illustrate these generalized intertwiners using thick lines for the vertical representations as
\begin{align}
    \includegraphics[width=0.9\textwidth]{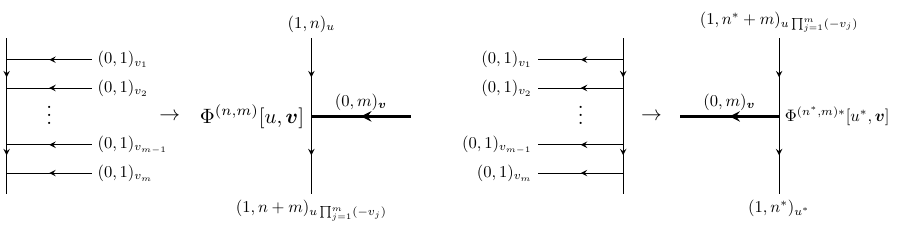}\label{eq:generalized-intertwiner-figure}
\end{align}
where $(0,m)_{\boldsymbol{v}}=(0,1)_{v_{m}}\otimes(0,1)_{v_{m-1}}\cdots\otimes(0,1)_{v_{1}}$.

\subsection{Degenerations of representations}\label{sec:repAY}
Now, let us consider the representation theory of affine Yangian $\AY$. In contrast to the elements of $\QTA$ graded by $\bZ^2$ as in \eqref{eq:DIMsubalgebra}, the elements of $\AY$ are graded by $\bZ\times \bZ_{\ge0}$ (see Figure \ref{fig:SdH-W}) where the generators $\sfe_n,\sff_n,\uppsi_n$ $(n\ge0)$ in \eqref{AY-generators} are located at $(\pm1,n)$, $(0,n)$. Consequently, the $S$-duality (or the Miki automorphism) is lost under the degeneration from $\QTA$ to $\AY$. Therefore, $\AY$ has only vertical representations (no horizontal representations).

Let us briefly look at the vertical representations of $\AY$ \cite{Tsymbaliuk}, which are degenerations of those in \S\ref{sec:vertical-rep}. We define $$\delta^{+}(w):=1+w+w^{2}+\cdots.$$

\paragraph{Vector representation}
 For $u\in \bC$, let $^aV(u)$ be a $\bC$-vector space with the basis $\{[u]_j\}_{j\in \bZ}$.
 The following formulas define a $\AY$-action on $^aV(u)$:
\begin{small}\begin{align}\label{AY-Vector}
  \begin{split}
    \psi(z)[u]_j^{(c)}=&\left[\frac{(z-(j\epsilon_c+\epsilon_{c+1}+u))(z-((j+1)\epsilon_c+\epsilon_{c-1}+u))}{(z-(j\epsilon_c+u))(z-((j+1)\epsilon_c+u))}\right]_+\cdot [u]_j^{(c)}~,\cr
    e(z)[u]_j^{(c)}=&\frac{1}{\epsilon_cz}\delta^+(((j+1)\epsilon_c+u)/z)[u]_{j+1}=\left[\frac{1}{\epsilon_c(z-(j+1)\epsilon_c-u)}\right]_+\cdot[u]^{(c)}_{j+1}~,\cr
    f(z)[u]_j^{(c)}=&-\frac{1}{\epsilon_cz}\delta^+((j\epsilon_c+u)/z)[u]_{j-1}=\left[\frac{-1}{\epsilon_c(z-j\epsilon_c-u)}\right]_+\cdot[u]^{(c)}_{j-1}~.
    \end{split}
\end{align}\end{small}
It is straightforward to check that this representation is the rational degeneration of \eqref{eq:vectorrep}.

\paragraph{Fock representations} Let us analogously define \emph{Fock representations} $\{^a\cF_c(u)\}_{u\in \bC}$ for the case of $\AY$.
For $u\in \bC$, let $^a\cF_c(u)$ be a $\bC$-vector space with the basis $\{|\lambda\rangle\}$.
The following formulas define a $\AY$-action on $^a\cF_c(u)$:
\begin{footnotesize}
\begin{align}
  \psi(z)|\lambda\rangle^{(c)}=&
    \left[\prod_{i=1}^{\infty}\frac{(z-(\lambda_i\epsilon_{c-1}+i\epsilon_{c+1}+u))(z-((\lambda_{i+1}-1)\epsilon_{c-1}+(i-1)\epsilon_{c+1}+u))}
    {(z-(\lambda_{i+1}\epsilon_{c-1}+i\epsilon_{c+1}+u))(z-((\lambda_i-1)\epsilon_{c-1}+(i-1)\epsilon_{c+1}+u))}\right.\times\cr
  &  \left.\frac{z-((\lambda_1-1)\epsilon_{c-1}-\epsilon_{c+1}+u)}{z-(\lambda_1\epsilon_{c-1}+u)}\right]_+\cdot |\lambda\rangle^{(c)}.\cr
   e(z)|\lambda\rangle^{(c)}=&
     \frac{1}{\epsilon_{c-1}z}\sum_{i= 1}^{\ell(\lambda)+1}\prod_{j=1}^{i-1}
     \frac{((\lambda_i-\lambda_j)\epsilon_{c-1}+(i-j-1)\epsilon_{c+1})((\lambda_i-\lambda_j+1)\epsilon_{c-1}+(i-j+1)\epsilon_{c+1})}
          {((\lambda_i-\lambda_j)\epsilon_{c-1}+(i-j)\epsilon_{c+1})((\lambda_i-\lambda_j+1)\epsilon_{c-1}+(i-j)\epsilon_{c+1})}\times\cr
    & \delta^+\left(\frac{\lambda_i\epsilon_{c-1}+(i-1)\epsilon_{c+1}+u}{z}\right)\cdot|\lambda+\Box_i\rangle^{(c)},\cr
   f(z)|\lambda\rangle^{(c)}=&
     -\frac{1}{\epsilon_{c-1}z}\sum_{i= 1}^{\ell(\lambda)}\prod_{j=i+1}^{\infty}
     \frac{((\lambda_j-\lambda_i+1)\epsilon_{c-1}+(j-i+1)\epsilon_{c+1})((\lambda_{j+1}-\lambda_i)\epsilon_{c-1}+(j-i)\epsilon_{c+1})}
          {((\lambda_{j+1}-\lambda_i+1)\epsilon_{c-1}+(j-i+1)\epsilon_{c+1})((\lambda_j-\lambda_i)\epsilon_{c-1}+(j-i)\epsilon_{c+1})}\times\cr
  &   \frac{(\lambda_{i+1}-\lambda_i)\epsilon_{c-1}}{(\lambda_{i+1}-\lambda_i+1)\epsilon_{c-1}+\epsilon_{c+1}}
       \delta^+\left(\frac{(\lambda_i-1)\epsilon_{c-1}+(i-1)\epsilon_{c+1}+u}{z}\right)\cdot|\lambda-\Box_i\rangle^{(c)}\label{deg-Fock}
\end{align}
\end{footnotesize}
It is straightforward to check that this representation is the rational degeneration of \eqref{QTA-Fock}.

\paragraph{MacMahon representation} Let us analogously define \emph{MacMahon representations} $\{^a\cM(u)\}_{u\in \bC}$ for the case of $\AY$.
\bea 
\psi(z)\ket{\Lambda}=&\psi_{\Lambda}(z,u)\ket{\Lambda},\quad \psi_{\Lambda}(z,u)=\frac{z-u+\psi_0\sigma_{3}}{z-u}\left[\prod_{\Abox \in \Lambda} \varphi\left(z-u-\epsilon(\Abox)\right)\right]_{+}\ket{\Lambda}\\
e(z)\ket{\Lambda}=&\frac1z\sum_{\sAbox\in \frakA(\Lambda)}e(\Lambda\rightarrow \Lambda+\Abox)\delta^+\left(\frac{u+\epsilon(\Abox)}{z}\right)\ket{\Lambda+\Abox},\\
f(z)\ket{\Lambda}=&\frac1z\sum_{\sAbox\in \frakR(\Lambda)}f(\Lambda\rightarrow \Lambda-\Abox)\delta^+\left(\frac{u+\epsilon(\Abox)}{z}\right)\ket{\Lambda-\Abox}\label{eq:planeAY}
\eea
 where $\epsilon(\Abox)=\epsilon_{1}({i-1})+\epsilon_{2}({j-1})+\epsilon_{3}({k-1})$ for $\Abox=(i,j,k)\in\Lambda$, and $\varphi$ is the structure function of $\AY$ as in \eqref{str-fn-AY}. The spectral parameter $u$ serves as the coordinate of the origin of a plane partition, and we can set it $u=0$ generally. Similar to the trigonometric MacMahon representation in (\ref{eq:MacMahonEFaction}), the coefficients $e(\Lambda\rightarrow \Lambda+\Abox),\, f(\Lambda\rightarrow \Lambda-\Abox)$ are proportional to the residue of the eigenvalue of $\psi(z)$ as 
 \bea
     e(\Lambda\rightarrow \Lambda+\Abox)\propto \sqrt{\underset{z=u+\epsilon(\Abox)}{\Res}\psi_{\Lambda}(z,u)} \\
     f(\Lambda\rightarrow \Lambda-\Abox)\propto \sqrt{\underset{z=u+\epsilon(\Abox)}{\Res}\psi_{\Lambda}(z,u)}.
 \eea
Note also that we have similar degenerate versions of the discussions in \S\ref{sec:MacMahonrep}. Namely, we also have degenerate versions of MacMahon representations with asymptotic Young diagrams or even with a pit. Similar to the trigonometric case, the eigenvalue $\psi_{\Lambda}(z,u)$ determines the action of the generators on the module. 
\begin{itemize}
    \item Asymptotic Young diagrams: the vacuum function is determined by multiplying contributions of the asymptotic Young diagrams as
    \begin{align}
        \psi^{\text{vac}}_{\lambda_{1}\lambda_{2}\lambda_{3}}(z,u)=\frac{z-u+\psi_{0}\sigma_{3}}{z-u}\prod_{\sAbox\in S_{\lambda_{1}\lambda_{2}\lambda_{3}}}\varphi(z-u-\epsilon(\Abox))
    \end{align}
    where $S_{\lambda_{1}\lambda_{2}\lambda_{3}}$ is the set of boxes in the vacuum configuration as in Figure \ref{fig:asympplanepartition} (see also (\ref{eq:asympvacfunction})). 
    \item The pit reduction occurs when we tune the central element as
    \begin{align}
        \psi_{0}\sigma_{3}=-(L\epsilon_{1}+M\epsilon_{2}+N\epsilon_{3})\quad (L,M,N\geq 0)\label{eq:onepitcondAY}
    \end{align}
   By setting this condition, the pole of $\psi_{\Lambda}(z,u)$ at $z=u+L\epsilon_{1}+M\epsilon_{2}+N\epsilon_{3}$ disappears. As a result, the growth of the plane partition is forbidden at position $(L+1,M+1,N+1)$, and configurations containing this pit are no longer included in the module. This condition is the degenerate analog of (\ref{pit_condition}). 
    
    For later use, we introduce the notation $\mu_{i}=-\frac{\psi_{0}\sigma_{3}}{\epsilon_{i}}$. Then, the condition $\sum_{i=1}^{3}\epsilon_{i}=0$ can be expressed as
    \begin{equation}
        \frac{1}{\mu_{1}}+\frac{1}{\mu_{2}}+\frac{1}{\mu_{3}}=0~.\label{eq:AYpitparameters}
    \end{equation}
The pit reduction condition is rewritten as 
\begin{equation}
     \frac{L}{\mu_{1}}+\frac{M}{\mu_{2}}+\frac{N}{\mu_{3}}=1. \label{eq:pitreductionAY}
\end{equation}
This condition is the same as \eqref{restriction}, which implies that the plane partition subjected to pit reduction is indeed associated with the corner VOA.
    \item In the correspondence (\ref{eq:AYCFTcorr}) between the affine Yangian and the $\cW_{1+\infty}[\mu]$-algebra, we can express the central charge (\ref{eq:AYcentralcharge}) as follows:
\begin{equation}
c=1+(\mu_{1}-1)(\mu_{2}-1)(\mu_{3}-1)\label{eq:CVOAcentralcharge}.
\end{equation}
Moreover, for MacMahon representations (vacuum, with asymptotic Young diagram, pit reduction etc.), we can determine the conformal dimension and U(1) charge by using (\ref{eq:AYCFTcorr})
    \begin{equation}
        \mathsf{J}_{0}=\uppsi_{1},\quad \mathsf{L}_{0}=\frac{1}{2}\uppsi_{2}~.
    \end{equation}
    This can be achieved by extracting the coefficients of the $z^{-2}$ and $z^{-3}$ terms in the eigenvalue of $\psi(z)$, where we set the additional spectral parameter $u=0$ for simplicity. (see \cite{Prochazka:2015deb,Prochazka:2017qum} for general formulas). For example, if we insert the asymptotic Young diagram $\lambda=(\lambda_{1},\lambda_{2},\cdots)$ in the second axis, the U(1) charge and conformal dimension are determined as 
    \bea
    &j=-\frac{1}{\epsilon_{2}}\sum_{i}\lambda_{i}~,\\
    &\Delta =-\frac{\mu_{2}}{2\mu_{3}}\sum_{i}\lambda_{i}^{2}-\frac{\mu_{2}}{2\mu_{1}}\sum_{i}(2j-1)\lambda_{i}+\frac{\mu_{2}}{2}\sum_{i}\lambda_{i}~.\label{eq:AYppconfdim}
    \eea
    Formulas for configurations with one asymptotic Young diagram in other axes are obtained by replacing $\epsilon_{2}$ and $\mu_{2}$ using the triality. When we have two asymptotic Young diagrams $\lambda$ and $\lambda'$ along the first and second axes, respectively, the situation changes slightly, and we need to take into account the overlapping boxes at the intersection of $\lambda$ and $\lambda'$. In this case, the U(1) charge and conformal dimension are determined as
    \bea
    j=j_{1}(\lambda)+j_{2}(\lambda'),\quad \Delta=\Delta_{1}(\lambda)+\Delta_{2}(\lambda')-\#(\lambda\cap\lambda')\label{eq:AYdoubleppconfdim}
    \eea
    where $j_{1}(\lambda)$, $j_{2}(\lambda')$, $\Delta_{1}(\lambda)$,  $\Delta_{2}(\lambda')$ are the U(1) charges and conformal dimensions of the configuration with only one asymptotic Young diagram, and $\#(\lambda\cap\lambda')$ denotes the number of the overlapping boxes.
\end{itemize}

\subsection{\texorpdfstring{$\mathcal{W}_N$}{WN} minimal models from affine Yangian of \texorpdfstring{$\mathfrak{gl}_{1}$}{gl1}}\label{sec:AYminimalmodel}
In \S\ref{sec:vertical-rep} and \S\ref{sec:repAY}, we have discussed the representations of the quantum toroidal $\mathfrak{gl}{1}$ algebra and the affine Yangian, respectively. These representations give rise to $q$-$\cW$ algebras and their degenerations. Now, considering the existence of the $\cW_{N}$ minimal models as demonstrated in \S\ref{sec:minimal}, the question arises whether we can utilize the representation theory of the affine Yangian or the quantum toroidal $\mathfrak{gl}_{1}$ algebra to construct the $\cW_{N}$ minimal models.

In this subsection, we show that the minimal models can be obtained from the MacMahon representation of the affine Yangian $\mathfrak{gl}_{1}$, where the asymptotic Young diagrams include \emph{two} pits \cite{Harada:2018bkb,Harada-doctor,Prochazka:2015deb,Prochazka:2017qum}. Most of the discussion in this subsection follows the presentation in \cite{Harada-doctor}.

\subsubsection{Double constrained plane partition}
In the plane partition configuration, null states of the Hilbert space appear as pit reductions. The pit reductions are introduced by imposing conditions (\ref{eq:pitreductionAY}) on the parameter of the affine Yangian, such as the central charge or the deformation parameters. Let us examine the scenario when two truncation conditions are simultaneously satisfied
\begin{align}
    \frac{L_{1}}{\mu_{1}}+\frac{M_{1}}{\mu_{2}}+\frac{N_{1}}{\mu_{3}}=1,\quad \frac{L_{2}}{\mu_{1}}+\frac{M_{2}}{\mu_{2}}+\frac{N_{2}}{\mu_{3}}=1\label{eq:doubleconstraint}
\end{align}
where we assumed that $(L_{2}-L_{1},M_{2}-M_{1},N_{2}-N_{1})$ is not proportional to $(1,1,1)$ to avoid the trivial case due to (\ref{eq:AYpitparameters}). For the construction of the $\cW_{N}$ minimal model, we set one of the pit conditions as $(L_{2},M_{2},N_{2})=(0,0,N)$. We consider the generic situation when the other pit condition is
\begin{align}
(L_{1},M_{1},N_{1})=(L,M,0),\quad \text{gcd}(L+N,M+N)=1~,  \label{eq:dppgenericcond}
\end{align} 
so that the two conditions are independent.

With these pit conditions, not all plane partitions are allowed, and constraints are imposed on them. We denote $\boxed{1}$ (resp. $\boxed{2}$) as the box at $(L+1,M+1,N+1)$ (resp. $(1,1,N+1)$). Setting the coordinate of the origin to be zero $u=0$, the coordinate for a box $\Abox=(i,j,k)$ is $\epsilon(\Abox)=(i-1)\epsilon_{1}+(j-1)\epsilon_{2}+(k-1)\epsilon_{3}$. The condition (\ref{eq:doubleconstraint}) gives 
\begin{align}
    \epsilon(\boxed{1})=\epsilon(\boxed{2})= -\psi_0 \sigma_3.
\end{align}
When only one pit is present at $(L+1,M+1,N+1)$, as mentioned in (\ref{eq:onepitcondAY}) (also refer to discussions below (\ref{pit_condition})), the term $z+\psi_{0}\sigma_{3}$ eliminates the pole arising from $z=L\epsilon_{1}+M\epsilon_{2}+N\epsilon_{3}$.

However, the situation changes when two pits are present.  When both $\boxed{1}$ and $\boxed{2}$ belong to $\frakA(\Lambda)$, the function\footnote{We denote the eigenvalue when $u=0$ as $\psi_{\Lambda}(z,0)=\psi_{\Lambda}(z)$.} $\psi_{\Lambda}(z)$ takes the form
\begin{align}
    \psi_{\Lambda}(z)\propto \frac{z+\psi_{0}\sigma_{3}}{(z-\epsilon(\boxed{1}))(z-\epsilon(\boxed{2}))}=\frac{1}{z+\psi_0 \sigma_3}~.
\end{align}
In this case, the pole at $z=-\psi_0 \sigma_3$ remains, enabling the inclusion of boxes $\boxed{1}$ and $\boxed{2}$ in the plane partition $\Lambda$. The action of $e(z)$ in this situation is given by
\begin{align}
    e(z)\ket{\Lambda}\propto \frac{1}{z+\psi_0 \sigma_3}\left(\ket{\Lambda+\boxed{1}}+\ket{\Lambda+\boxed{2}}\right)+\cdots
\end{align}
which implies that only the configuration $\ket{\Lambda+\boxed{1}}+\ket{\Lambda+\boxed{2}}$ is created by $e(z)$ while $\ket{\Lambda+\boxed{1}}-\ket{\Lambda+\boxed{2}}$ is not.  Consequently, the state $\ket{\Lambda+\boxed{1}}-\ket{\Lambda+\boxed{2}}$ is identified with a null state, and we must treat them as equivalent, i.e., $\ket{\Lambda+\boxed{1}}\simeq\ket{\Lambda+\boxed{2}}$, in the reduced Hilbert space. The same identification applies to descendant states as well.

As a result of this identification, we have two ways to understand the representation space:
\begin{itemize}
    \item A constraint plane partition with a pit at $(L+1,M+1,1)$.
    \item A constraint plane partition with a pit at $(1,1,N+1)$.
\end{itemize}
The configuration allowed in the reduced Hilbert space must admit both descriptions after the rearrangement of the configuration by $\ket{\Lambda+\boxed{1}}\simeq\ket{\Lambda+\boxed{2}}$. States that do not fit in either of the above descriptions correspond to null states. This is the reason why configurations with two pits reduce the Hilbert space.

Let us illustrate the procedure with an example where we impose the two pit conditions $(L_{1},M_{1},N_{1})=(1,2,0)$ and $(L_{2},M_{2},N_{2})=(0,0,1)$. First, let us impose the pit condition $(L_{2},M_{2},N_{2})=(0,0,1)$. The possible plane partition is a plane partition with only one layer, merely a Young diagram. Next, we impose the condition $(L_{1},M_{1},N_{1})=(1,2,0)$.  As a result of the double pit condition
\bea
    -\psi_{0}\sigma_{3}&=\epsilon_{3}=\epsilon_{1}+2\epsilon_{2}
\eea
a box in $(i,j,k+1)$ and a box in $(i+1,j+2,k)$ are identified: 
\begin{equation}
    (i,j,k)\simeq(i',j',k')\quad \text{when}\quad(i'-i,j'-j,k'-k)\propto(1,2,-1).\label{eq:identificationex1}
\end{equation}
Consequently, the boxes in $(i,j,0)$ in the plane partition with one layer can be decomposed into hook-shaped components and stacked to form a plane partition as follows
\begin{align}
    \adjustbox{valign=c}{\includegraphics[width=6cm]{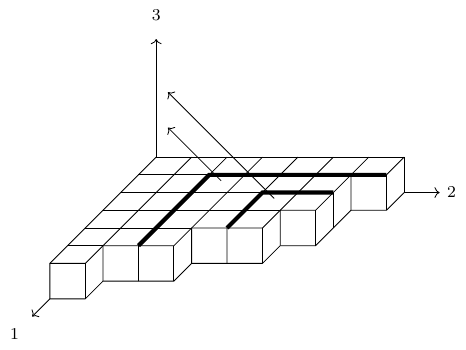}} \quad\Longrightarrow \adjustbox{valign=c}{\includegraphics[width=6cm]{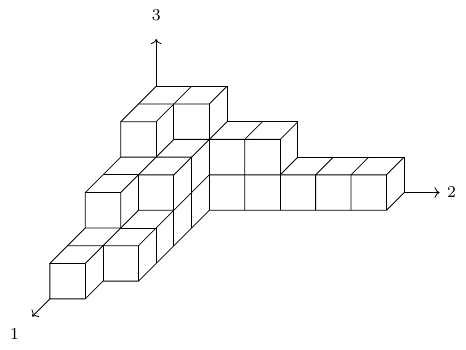}}
\end{align}
For this procedure, there are no constraints on the Young diagram because all possible Young diagrams can be stacked to form a plane partition with a pit at position $(2,3,1)$ without violating the pit condition.

Now, let us consider the opposite direction. We first impose the pit condition coming from $(2,3,1)$. The allowed configurations are plane partitions obeying the pit condition coming from $(2,3,1)$. We then rearrange boxes following the identification in (\ref{eq:identificationex1}) to meet the second pit condition from $(1,1,2)$.  In this case, not all plane partitions will satisfy the Young diagram condition after the rearrangement. For example, the following configuration is not allowed
\begin{align}
    \adjustbox{valign=c}{\includegraphics[width=6cm]{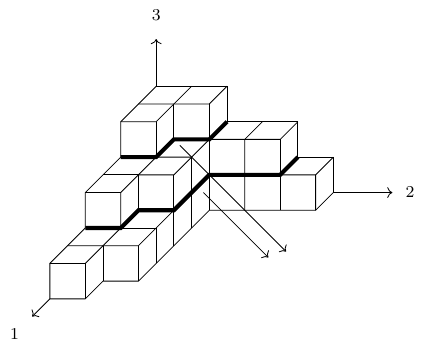}}\Longrightarrow\adjustbox{valign=c}{\includegraphics[width=6cm]{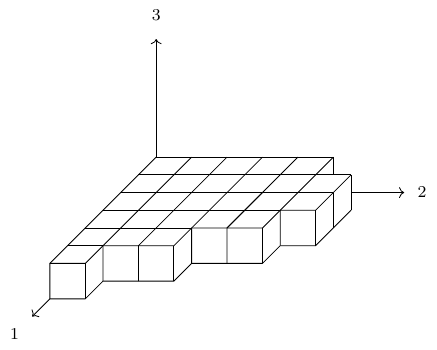}}
\end{align}
Configurations not satisfying the Young diagram condition after identification correspond to null states. This imposes conditions on the initial plane partition with a pit at $(2,3,1)$.

\subsubsection{Plane partitions and \texorpdfstring{$\mathcal{W}_{N}$}{WN} minimal models}
Let us show that there is a one-to-one correspondence between the ${\mathcal{W}}_{N}$ minimal models and the plane partitions with double pit constraint. Placing one of the pits at $(1,1,N+1)$ gives the condition $\mu_{3}=N$. Parametrizing the remaining parameters as $\mu_{1}=N(\beta-1),\,\mu_{2}=N(\beta^{-1}-1)$, the central charge (\ref{eq:AYcentralcharge}) and (\ref{eq:CVOAcentralcharge}) is 
\bea
    c&=1+(N(\beta-1)-1)(N(\beta^{-1}-1)-1)(N-1)\\
    &=(N-1)(1-Q^{2}N(N+1))+1,\quad Q=\sqrt{\beta}-\sqrt{\beta}^{-1}
\eea
which indeed gives the central charge (\ref{cc-Wgln}) of the ${\cW}_{N}$-algebra. To obtain the central charge of the minimal model (\ref{Wminimal}), we need to specify $\beta$ as $\beta=q/p$, where $p,q\geq N$ are coprime integers. This parametrization leads to the condition
\begin{align}
    &\mu_{1}=N(q/p-1),\quad \mu_{2}=N(p/q-1),\quad p\mu_{1}+q\mu_{2}=0.
\end{align}
Combining with $\frac{1}{\mu_{1}}+\frac{1}{\mu_{2}}+\frac{1}{\mu_{3}}=0$ and $\mu_{3}=N$, we have 
\begin{align}
    \frac{q-N}{\mu_{1}}+\frac{p-N}{\mu_{2}}+\frac{0}{\mu_{3}}=1. \label{eq:minimaldoublepitcond}
\end{align}
This means we need another pit at $(q-N+1,p-N+1,1)$, and box identifications of the form $(x,y,z+N)\sim(x+q-N,y+p-N,z)$ are imposed. Such identifications reduce the Hilbert space, as explained above. Note also that the condition that $p,q$ are coprime matches with (\ref{eq:dppgenericcond}).

\begin{figure}
    \centering
    \includegraphics[width=5cm]{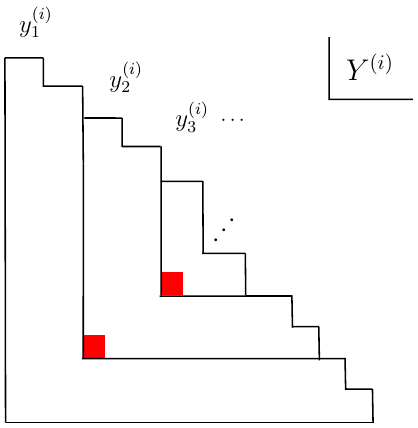}
    \caption{Hook decomposition of each layer of the height $N$ plane partition.}
    \label{fig:hookdecomp}
\end{figure}
Under this condition, let us discuss how the Hilbert space is reduced. As also mentioned in \S\ref{sec:MacMahonrep}, a plane partition can be decomposed into multiple Young diagrams. In this case, if we choose the third direction to be where the height is defined, we will have $N$-tuple Young diagrams $\Lambda=(Y^{(1)},Y^{(2)},\cdots ,Y^{(N)})$ obeying the plane partition condition $Y^{(1)}\supseteq Y^{(2)}\supseteq\cdots\supseteq Y^{(N)}$. Starting from an arbitrary $N$-tuple of Young diagrams (a plane partition with height $N$), we decompose each of them into hook-shaped components. For example, a Young diagram at height $z=i$, $Y^{(i)}$ may be decomposed into hook-shaped Young diagrams $Y^{(i)}=(y^{(i)}_{1},y^{(i)}_{2},\ldots)$ as Figure \ref{fig:hookdecomp}. Using the identification, the hook-shaped Young diagram $y^{(i)}_{k}$ will be identified with a Young diagram at height $z=i+Nk$. Thus, after the identification, the resulting plane partition will be a configuration
\begin{align}
    \Lambda'=(y^{(1)}_{1},y^{(2)}_{1},\cdots y^{(N)}_{1},y^{(1)}_{2},y^{(2)}_{2},\cdots).
\end{align}
\begin{figure}[t]
    \centering
    \includegraphics[width=15cm]{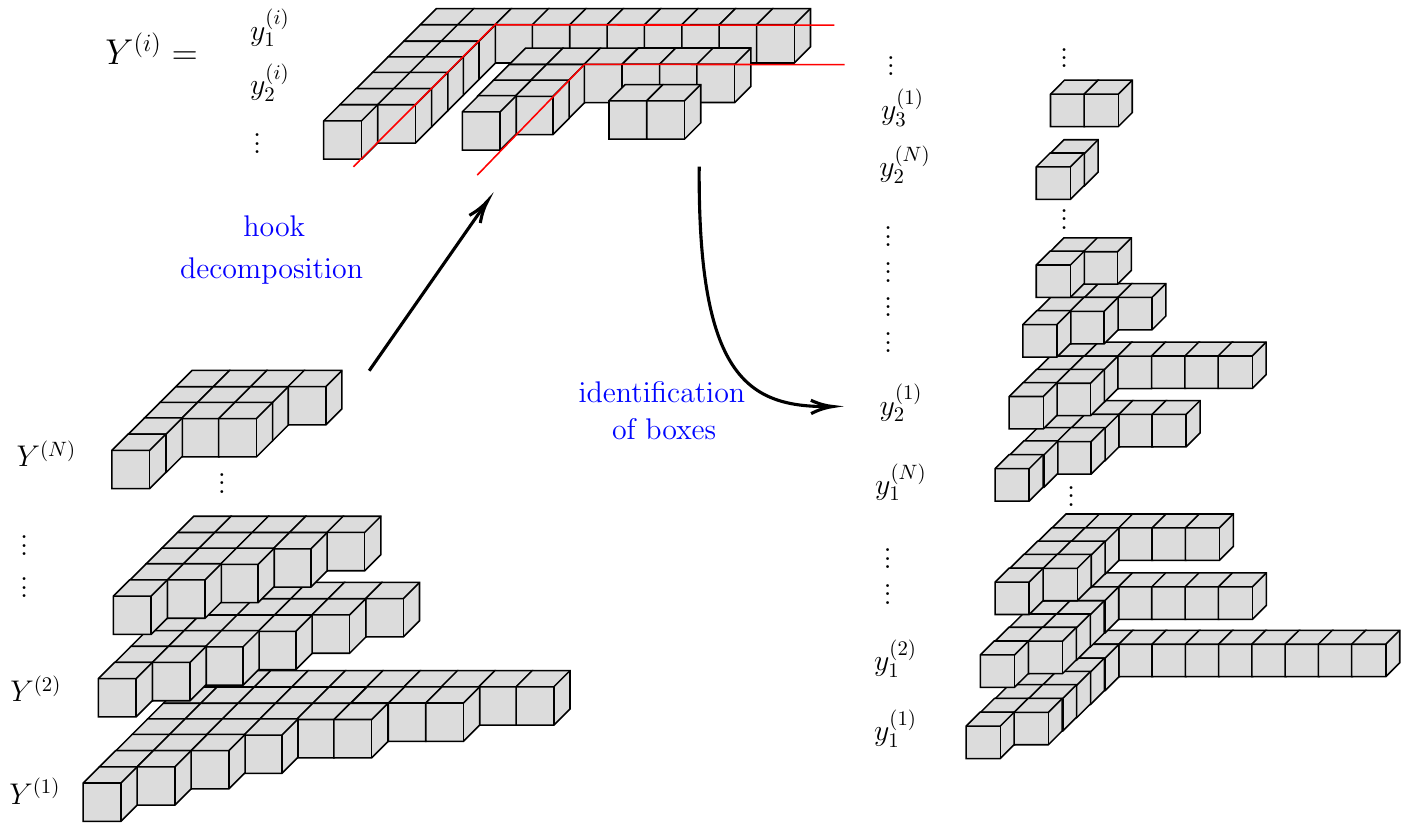}
    \caption{Double constrained plane partition and $\mathcal{W}_{N}$ minimal models. In the right figure, the Young diagrams $y^{(i)}_{k}$ fail to form a valid plane partition. Therefore, it is necessary to introduce specific conditions on $y^{(i)}_{k}$ to form a plane partition, leading to the $\cW_N$ minimal model.}
    \label{fig:WNminimalpit}
\end{figure}
Roughly speaking, we have multiple height $N$ plane partitions $(y^{(1)}_{k},\ldots y^{(N)}_{k})$ stacked on top of each other (see Figure \ref{fig:WNminimalpit}). For the new configuration $\Lambda'$ to be a plane partition, we need the condition $y^{(i)}_{k}\supseteq y^{(i+1)}_{k}$ for $i=1,\ldots, N,\,k\in\mathbb{Z}_{>0}$, as well as $y^{(N)}_{k}\supseteq y^{(1)}_{k+1}$ for $k\in\mathbb{Z}_{>0}$.

Let us determine the explicit condition for this process. Instead of only considering the plane partition with no asymptotic Young diagrams in the axes, let us generalize the story to cases including asymptotic Young diagrams. Since the height in the third direction is restricted, we can impose asymptotic conditions only in the first and second axes. We denote them by $\mu=(\mu_{1},\ldots,\mu_{N})$ (resp. $\nu=(\nu_{1},\ldots,\nu_{N})$) for the first one (resp. second one) (see Figure \ref{fig:WNminimaldef}). We can set $\mu_{N}=0$ and $\nu_{N}=0$ without lost of generality. This is because if $\mu_{N}\neq 0,\nu_{N}\neq 0$, we can just shift the origin of the plane partition. Let us then put the pit in $(q-N+1,p-N+1,1)$, which yields the condition
\be
    \mu_{1}\leq p-N,\quad \nu_{1}\leq q-N. \label{eq:minimalasymptconstr}
\ee
\begin{figure}[t]
    \centering
    \includegraphics[width=12cm]{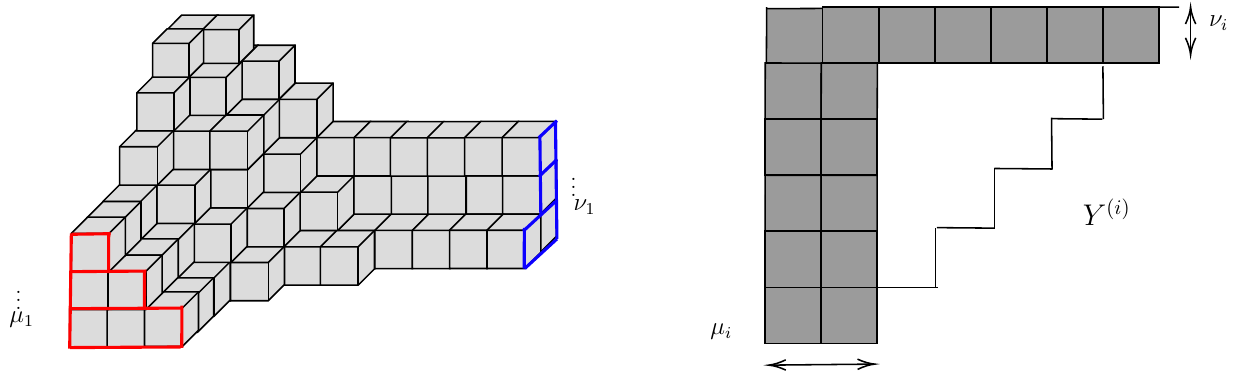}
    \caption{Plane partition with asymptotic Young diagrams $\mu,\nu$. For each layer, we can assign Young diagrams as the right figure where the origin is shifted depending on the asymptotic Young diagrams. }
    \label{fig:WNminimaldef}
\end{figure}
Interestingly, the possible asymptotic Young diagrams satisfying the double pit constraint correspond to the primary fields of the $\widetilde{\cW}_{N}$ minimal model. We define the parameters as
\bea
    \tilde{r}_{i}=\mu_{i}-\mu_{i+1}\geq 0,\quad \tilde{s}_{i}=\nu_{i}-\nu_{i+1}\geq 0,\\
    r_{i}=\tilde{r}_{i}+1>0,\quad s_{i}=\tilde{s}_{i}+1>0
\eea
for $i=1,\ldots,N-1$. The constraint (\ref{eq:minimalasymptconstr}) gives 
\begin{align}
    \sum_{i=1}^{N-1}r_{i}\leq p,\quad \sum_{i=1}^{N-1}s_{i}\leq q~.
\end{align}
The sets of possible positive integers satisfying these conditions give the primary fields condition in (\ref{eq:minimalprimarycond}). By using the formulas in (\ref{eq:AYppconfdim}) and (\ref{eq:AYdoubleppconfdim}), one can compute the conformal dimension of these possible configurations (see \cite{Harada-doctor} and \cite{Harada:2018bkb} for the computations)
\be
    \Delta(\boldsymbol{r}, \boldsymbol{s}) = \frac{1}{24pq} \left(12(\sum_i(q r_i -p s_i)   \boldsymbol{\omega}_i )^2 -N(N^2-1)(p-q)^2\right)\,.
\ee
Under this situation, one can also show that the conditions for the $N$-tuple Young diagrams to make sense in both descriptions are
\be\label{N-Burge}
    Y_{k}^{(i)}-Y^{(i+1)}_{k+r_{i}-1}\geq -(s_{i}-1)\quad (i=1,2,\ldots,N,\,k\geq 1)
\ee
where we identify $Y^{(N+1)}=Y^{(1)}$ and set $r_{N}=p-\sum_{i=1}^{N-1}r_{i}$ and $s_{N}=q-\sum_{i=1}^{N-1}s_{i}$. This condition is known as the $N$-Burge condition, which characterizes the minimal model \cite{BURGE1993210,Belavin:2015ria,Alkalaev:2014sma,Feigin2011,Fukuda:2015ura}.

\paragraph{Cylindrical partition} We have to note that the plane partition which describes the minimal model is related to the cylindrical partition \cite{gessel1997cylindric}, which Foda-Welsh \cite{foda2016cylindric} applied to derive the $N$-Burge condition for $\mathcal{W}_N$-minimal models.
The cylindrical partition is a half-infinite strip with an appropriate periodicity defined with a triple $(\lambda,\mu,d)$ where $\lambda=(\lambda_1,\cdots,\lambda_r)$ and $\mu=(\mu_1,\cdots,\mu_r)$ are Young diagrams ($\mu\preceq \lambda$). We assign each box in $\lambda\setminus\mu$, a non-negative numbers $\pi_{ij}$
($1\leq i\leq r$, $\mu_i\leq j\leq \lambda_i$), such that they satisfy,
\begin{align}
\pi_{i,j}\geq \pi_{i,j+1},\qquad &1\leq i\leq r,\label{cy1}\\
\pi_{i,j}\geq\pi_{i+1,j}\qquad & 1\leq i<r,\, \mbox{max}\left\{\mu_i,\mu_{i+1}\right\}\leq j\leq \mbox{min}\left\{\lambda_i,\lambda_{i+1}\right\}\,.\label{cyl2}
\end{align}
The configuration looks like,
\begin{center}
    \begin{tabular}{ccccccccccc}
         & & & & $\pi_{1,\mu_1+1}$& $\cdots$ &$\cdots$ &$\cdots$ &$\cdots$ &$\cdots$ & $\pi_{1,\lambda_1}$\\
         & & $\pi_{2,\mu_2+1}$ & $\cdots$ & $\pi_{2,\mu_2+1}$ & $\cdots$& $\cdots$& $\cdots$& $\cdots$ & $\pi_{2,\lambda_2}$\\
         & & $\vdots$ & & $\vdots$ & & & $\ddots$ & & & \\
         $\pi_{r,\mu_r+1}$& $\cdots$& $\cdots$& $\cdots$& $\cdots$& $\cdots$& $\pi_{r,\lambda_r}$ & & & &
    \end{tabular}
\end{center}
Furthermore, we extend the range of $i$ to $\mathbb{Z}$ by requiring a periodicity, $\pi_{i+r,j} =\pi_{i,j+d}$. It gives rise to another constraint,
\begin{equation}\label{cyl3}
    \pi_{r,j}\geq \pi_{1,j+d},\qquad \mbox{max}\left\{\mu_r,\mu_{1}-d\right\}\leq j\leq \mbox{min}\left\{\lambda_r,\lambda_{1}-d\right\}\,.
\end{equation}
For the application to the $\cW_N$ minimal model, one sets
$\lambda=(\infty^r)$, $\mu_i-\mu_{i+1}=r_i-1$ and $d=N$.
We need to modify (\ref{cyl2}) and (\ref{cyl3}) by introducing shifts $s_i\geq 1$ ($i=0,\cdots,r$) as,
\begin{align}
    \pi_{i,j}\geq \pi_{i+1,j}-(s_i-1), \qquad
    & 1\leq i\leq r,\, \mbox{max}
    \left\{\mu_i,\mu_{i+1}\right\}\leq j,\label{cp1}\\
    \pi_{r,j}\geq \pi_{1,j+N} -(s_0-1) \qquad & \mbox{max}\left\{\mu_r,\mu_1-N\right\}\leq j\,.\label{cp2}
\end{align}
These constraints are shown in \cite{foda2016cylindric} to produce $N$-Burge condition (\ref{N-Burge}) after the identification $Y^{(i)}_{k}=\pi_{i,\mu_i+k}$.

The cylindrical partition may be regarded as a variant of the plane partition by viewing $\pi_{ij}$ to define the ``height" at the position $(ij)$.
For the open space with the shape of $\mu$, we redefine $\pi_{ij}=\infty$. Thus, $\mu$ may be regarded as the asymptotic Young diagram in $z$-direction. Furthermore, the shifts $s_i$ in (\ref{cp1}) and (\ref{cp2}) gives the asymptotic Young diagram in $x$-direction. The plane partition thus obtained from the cylindrical partition coincides with the diagram we obtained from two pits after the exchange of $y$ and $z$ axes.

\paragraph{Minimal model character}
The minimal model character for $\mathcal{W}_N$-algebra \cite{mizoguchi1991structure,Fateev:1987zh,nakanishi1990non}:
\begin{equation}
    \chi^{N,p,q}_{r,s}(\mathfrak{q})=\prod_{n=1}^\infty (1-\mathfrak{q}^n)^{-N+1}
    \sum_{\alpha\in Q_N}\sum_{\sigma\in \mathfrak{S}_N}
    (-1)^{\ell(\sigma)}\mathfrak{q}^{\frac{pq}{2}\left[\alpha-(r+\rho)/p+\sigma(s+\rho)/q\right]^2}~,
\end{equation}
where $\rho$ is the Weyl vector,
$Q_N$ is a set of $N$-dimensional vector,
$(k_1,\cdots,k_N)$ ($k_i\in \mathbb{Z}$) with $\sum_{i=1}^N k_i=0$, and $r=\sum_{i=0}^{N-1} r_i \boldsymbol{\omega}_i$, $s=\sum_{i=0}^{N-1} s_i \boldsymbol{\omega}_i$ with $\boldsymbol{\omega}_i$ is the fundamental weight of $A^{(1)}_{N-1}$. When $N=2$, this character equals the Rocha-Caridi formula \eqref{Rocha-Caridi}.
This is identical to the following determinant formula \cite{foda2016cylindric}
\begin{equation}
    \chi^{N,p,q}_{r,s}(\mathfrak{q})=\mathfrak{q}^{\Delta(p,q,r,s)}\prod_{n=1}^\infty (1-\mathfrak{q}^n)^{-N+1}
    \sum_{k_1+\cdots+k_N=0} \mathfrak{q}^{q\sum_{i=1}^N k_i(\frac12pk_i-\nu_i+i)}\det_{1\leq t,u<N}\left(\mathfrak{q}^{(\mu_u-u)(pk_t-\nu_t+t+\nu_u-u)}\right)
\end{equation}
where
$\Delta(p,q,r,s)$ is the conformal dimension of the primary field, 
and $\mu_t=\sum_{u=t}^{N-1} s_u$, $\nu_t=\sum_{u=t}^{N-1} r_u$ are the partitions associated with $r,s$.

Summarizing, we have the following correspondence
\begin{table}[ht]
\centering
\renewcommand\arraystretch{1.2}{
\begin{tabular}{|c|c|}\hline
Minimal model & Double constraint plane partition \\
\hline \multirow{2}{*}{number of free bosons $N$} & pit at $(1,1,N+1)$ \\
&$\mu_{3}=N$\\
\hline \multirow{2}{*}{$\beta=q/p$,\hspace{0.5cm} $p,q\geq N$: coprime}& pit at $(q-N+1,p-N+1,1)$\\
& $\mu_{1}=N(q/p-1)$, $\mu_{2}=N(p/q-1)$\\
\hline central charge  &  central charge\\
$c=1+(N-1)\left(1-\frac{(p-q)^{2}}{pq}N(N+1)\right)$& $c=1+(\mu_{1}-1)(\mu_{2}-1)(\mu_{3}-1)$\\
\hline primary fields specified by&  asymptotic Young diagrams $\mu,\nu$\\
 $(r_{i},s_{i}) \,(i=1,\ldots,N-1)$ $r_{i},s_{i}\in\mathbb{Z}_{>0}$& $r_{i}=\mu_{i}-\mu_{i+1}+1$, $s_{i}=\nu_{i}-\nu_{i+1}+1$\\
 \hline \multirow{2}{*}{Hilbert space} & $N$-tuple Young diagrams $(Y^{(i)})_{i=1,\ldots,N}$\\
 &$Y_{k}^{(i)}-Y^{(i+1)}_{k+r_{i}-1}\geq -(s_{i}-1)$ ($N$-Burge cond.)\\
 \hline
\end{tabular}}
\end{table}

\paragraph{Remarks}
In addition to considering the double pit reduction of a single plane partition, it is also possible to extend this technique to the double pit reduction of glued multiple plane partitions. By gluing the two legs with the same asymptotic Young diagrams, we can explore a broader class of $\mathcal{W}$-algebras. These types of algebras are sometimes referred to as the ``web of $\mathcal{W}$-algebra'' \cite{Prochazka:2017qum} (see also \cite{Gaberdiel:2017hcn,Gaberdiel:2018nbs,Harada:2018bkb}). It is expected that by performing the double pit reduction (or multiple pit reductions) in general, we can obtain minimal models of the web of $\mathcal{W}$-algebra. However, at the moment, only a few examples, such as the Bershadsky-Polyakov algebra and the $\mathcal{N}=2$ super Virasoro algebra, have been studied in the literature \cite{Harada:2018bkb,Harada:2020woh}.

\subsubsection{\texorpdfstring{$q$}{q}-\texorpdfstring{$\mathcal{W}_{N}$}{WN} minimal models}
Let us consider the minimal models of the $q$-deformed $\mathcal{W}$-algebra. The analysis for the $\mathcal{W}$-algebra without $q$-deformation is parallel. The procedure explained in the previous sections directly applies to the $q$-deformed case. We just need to replace the MacMahon representation of the affine Yangian $\mathfrak{gl}_{1}$ to the MacMahon representation of quantum toroidal $\mathfrak{gl}_{1}$. Since it is solely a straightforward generalization, we will not reproduce it here but discuss it from a slightly different viewpoint instead. The discussion is essentially equivalent to the previous sections.

The double pit conditions are written as
\begin{equation}\label{double_pit}
	q_1^{L_1} q_2^{M_1} q_3^{N_1}= q_1^{L_2} q_2^{M_2} q_3^{N_2}=K.
\end{equation}
Similar to the previous section, we consider the case 
when $(L_{2},M_{2},N_{2})=(0,0,N)$ and $(L_{1},M_{1},N_{1})=(q-N, p-N,0)$.
These conditions are satisfied if \footnote{We note that a similar condition was studied in \cite{feigin2003symmetric}, in a different context.}
\begin{equation}\label{double_pit2}
    q_1^q q_2^p=1,\quad q_3^N=K\,.
\end{equation}

The pit of $(1,1,N+1)$ reduces the height of the plane partition representation to $N$, giving the $\mathcal{W}_N$-module. Instead of using the MacMahon representation and then reducing, we can also simply start from $N$-tensor products of Fock modules in \S\ref{sec:tensor_product}. We choose $c_1=c_2=\cdots= c_N=3$. For the moment, we still keep the spectral parameters generic (later it will be tuned). Thus, we are now considering $\mathcal{F}_{3}(u_{1})\otimes\cdots\otimes \mathcal{F}_{3}(u_{N})$.

 As explained in \S\ref{sec:minimal}, the strategy to find minimal models is to determine the singular vectors. The singular vector in the Fock module $\ket{\boldsymbol{u},\boldsymbol{\lambda}}^{(\boldsymbol{3})}$ is characterized by the highest weight condition
\begin{equation}\label{HWC}
    F(z)|\chi\rangle =0,
\end{equation}
where $|\chi\rangle$ is expressed as a vector in a module,
\begin{equation}
    |\chi\rangle = \sum_{\boldsymbol{\lambda}}C_{\boldsymbol{\lambda}}
    \ket{\boldsymbol{u},\boldsymbol{\lambda}}^{(\boldsymbol{3})}
\end{equation}
where $C_{\boldsymbol{\lambda}}$ some coefficient.

\begin{equation}
    \boldsymbol{\lambda}=\boldsymbol{\lambda}_A=(\lambda_A^{(1)},\ldots, \lambda_A^{(N)}), \quad
    \boldsymbol{\lambda}_A^{(a)}=\begin{cases}
    (r_A, s_A)\quad & a=A\\
    \emptyset \quad & a\neq A
    \end{cases}\,
\end{equation}
for $A=1,\cdots, N$.

The action of the Drinfeld current $F(z)$ on such states is evaluated using \eqref{QTA-Fock3-F}  and (\ref{corner_reduction}). The coefficient factor $\Psi_{\boldsymbol{\lambda}_A}(z,\boldsymbol{u})$ becomes (we set $c=3$ and omit the upper label $^{(3)}$ in the following computation),
\be
\Psi_{\boldsymbol{\lambda}_A}(z,\boldsymbol{u}) = q_3^{-1/2} 
\frac{(1-u_A q_3 q_1^{s_A}/z)(1-u_A q_3 q_2^{r_A}/z)(1-u_A q_1^{s_A}q_2^{r_A}/z)}{(1-u_A q_1^{s_A}/z)(1-u_A q_2^{r_A}/z)(1-u_A q_3 q_1^{s_A}q_2^{r_A}/z)}
\prod_{b(\neq A)}\frac{1-u_b q_3/z}{1-u_b/z}
\ee
The action of $F(z)$ operator on $|\boldsymbol{u},\boldsymbol{\lambda}_A\rangle$ consists of a single term
with $i=A$ and $(x(\Abox),y(\Abox))=(s_A, r_A)\,\,(\Abox\in \frakR(\boldsymbol{\lambda}_{A}))$:
\be
    F(z)\ket{\boldsymbol{u},\boldsymbol{\lambda}_{A}}\propto   \delta\bl(u_{A} q_{1}^{s_{A}-1}q_{2}^{r_{A}-1}/z\br) \sqrt{\underset{z=u_Aq_{1}^{s_{A}-1}q_{2}^{r_{A}-1}}{\mathrm{Res}}z^{-1} \Psi_\lambda(z,\boldsymbol{u})}\,|\boldsymbol{u},\boldsymbol{\lambda}_{A}-\Box\rrangle
\ee
The factor in the square root contains factors,
\begin{equation}
    \frac{1-\frac{u_b}{u_A}q_{1}^{-s_A} q_{2}^{-r_A}}{1-\frac{u_b}{u_A}q_{1}^{-s_A+1} q_{2}^{-r_A+1}}
\end{equation}
for $b\neq A$.
It implies that tuning
\begin{equation}
    \frac{u_b}{u_A}=q_{1}^{s_A}q_2^{r_A}
\end{equation}
for one of $b\neq A$, the corresponding state $\ket{\boldsymbol{u},\boldsymbol{\lambda}_{A}}$ satisfies the highest weight condition and becomes a singular vector.

To obtain a fully degenerate $q$-$\mathcal{W}_N$ module, one chooses
\begin{equation}
    \frac{u_{a+1}}{u_a} = q_1^{s_a} q_2^{r_a},\,\,\, a=1,\ldots, N,\quad u_{N+1}:=u_1\,.
\end{equation}
We note that this is precisely the condition that the plane partition with fixed asymptotic Young diagrams has a consistent configuration. To be more specific, we construct the plane partition with the background configuration where $a$-th layer is shifted by $(\tilde s_a, \tilde r_a)=(s_a-1, r_a-1)$ with the $(a+1)$-th layer. Inserting an $s_a \times r_a$ rectangle into the $a$-th layer violates the conditions for a valid plane partition.  Therefore, the configuration of this kind is considered to be null.

By taking the product, we obtain a consistency condition
\begin{equation}
    q_1^{\sum_{a=1}^N s_a} q_2^{\sum_{a=1}^N r_a}=1\,.
\end{equation}
This gives a constraint on the deformation parameter $q_1$, $q_2$.
It agrees with the first condition of (\ref{double_pit2}) if we set the condition \eqref{WN-pq}, namely $q=\sum_{a=1}^N s_a, \ p=\sum_{a=1}^N r_a$ .

In the degenerate limit, one obtains
\begin{equation}
    \epsilon_1 q+ \epsilon_2 p=0,
\end{equation}
which gives the constraint (\ref{rationality}) by setting $\beta=-\epsilon_2/\epsilon_1$ (up to sign due to the conventional conflict).

In the free boson representation, the singular vector was constructed by applying the screening currents to the highest weight state (\ref{sing_vectW_n}). The construction in the current section implies the identification.

\section{Connection with integrable models}\label{s:int-model}
Affine Yangian $\frakgl_1$ (or quantum toroidal $\frakgl_1$) plays a fundamental role in an integrable field theory and provides powerful tools for its analysis. As a continuous limit of a discrete integrable model such as a lattice model, an integrable field theory is endowed with an infinite number of mutually commuting conserved observables (integral of motions) \cite{Zamolodchikov:1989hfa}.
It can typically be solved exactly using methods like the Bethe ansatz and Yang-Baxter equations. Starting from conformal field theories, a seminal series of papers \cite{bazhanov1996integrable,bazhanov1997integrable,bazhanov1997quantum,Bazhanov:1998dq} developed a method to construct commuting transfer $T$ matrices and Baxter $Q$-operators in an integrable field theory.

Relevant to affine Yangian $\frakgl_1$ is the Intermediate Long Wave (ILW) system, describing a certain type of wave propagation in one spatial dimension. The ILW equation describes the evolution of waves of intermediate length on a fluid's surface, which interpolates between the Korteweg-de Vries equation for short waves and the Benjamin-Ono equation for long waves. 
The integrability of the model comes from the existence of $\mathcal{R}$-matrix, which satisfies the Yang-Baxter relation. The conserved quantities arise from the Cartan part of the Drinfeld currents of the affine Yangian. This section takes a glimpse at the relationship between the affine Yangian (or the quantum toroidal algebra) and integrable field theory.

In \S\ref{s:MO-Rmatrix}, we give a brief account of the integrability by starting from the definition of the Maulik-Okounkov $\mathcal{R}$-matrix \cite{Maulik:2012wi}, which operates on the tensor product of two Fock spaces of free bosons. It describes the exchange of Miura operators, which define $\mathcal{W}$-algebra.
It has a physical interpretation of the reflection matrix by the Liouville wall. In \S\ref{sec:Monodromy}, we describe the construction of the infinite number of commuting operators from the $\mathcal{R}$-matrix. These operators have a natural interpretation of providing the infinite commuting operators of a (generalized) Calogero-Sutherland model. In \S\ref{s:Bethe}, we describe the introduction of an extra parameter $p$, to describe the ILW model and the Bethe ansatz equation as a consistency condition to have the off-shell Bethe state. Finally, in \S\ref{s:univR}, we explain the $q$-deformed version where Maulick-Okounkov $\mathcal{R}$-matrix is replaced by the universal $\mathcal{R}$-matrix of the quantum toroidal algebra.

\subsection{Maulik-Okounkov \texorpdfstring{$\mathcal{R}$}{R}-matrix}\label{s:MO-Rmatrix}
For the quantum toroidal algebra and the affine Yangian, one may define the integrable field theory associated with them. Let us start from the affine Yangian case, which Maulik-Okounkov \cite[Eq.(14.12)]{Maulik:2012wi} proposed. We follow the summary in \cite{Zhu:2015nha}.

The $\mathcal{R}$-matrix is defined as a linear map from a tensor product of two Fock spaces\footnote{The superscript $(i)$ here represents the tensor components. As mentioned in previous sections, the affine Yangian (or the quantum toroidal $\mathfrak{gl}_{1}$) has the triality symmetry and thus the Fock representation has a degree of freedom referred to as the \textit{color} index. We consistently use the color $c=3$ here, hence the subscript is omitted for brevity. Therefore, $\mathcal{F}^{(i)}(u)$ should be understood as $\mathcal{F}^{(i)}_{3}(u)$, unless stated otherwise. Additionally, we occasionally omit the spectral parameter for simplicity when its presence is clear from the context.} $\mathcal{F}^{(a)}$ ($a=1,2$) onto itself: $\mathcal{R}_{12}(u):\mathcal{F}^{(1)}\otimes \mathcal{F}^{(2)}\rightarrow \mathcal{F}^{(1)}\otimes \mathcal{F}^{(2)}$, where an extra parameter $u$ called the spectral parameter is introduced. Compared with the normal integrable model, in the affine Yangian (or quantum toroidal algebra) case, the representation space will be the infinite-dimensional Fock space $\mathcal{F}$ instead of the finite-dimensional spin degree of freedom. Hence, the $\mathcal{R}$-matrix appearing will be an infinite-dimensional matrix. 

The $\mathcal{R}$-matrix of the affine Yangian is represented in free fields, where we use (\ref{boson-modes}) for the mode expansion. We introduce a combination
$\varphi^-=\frac{1}{\sqrt{2}}(\varphi_1-\varphi_2)$. The eigenvalue of the zero-modes $\sfJ_{0}^{(a)}$ ($a=1,2$) are denoted as $\eta_a$ and the vacuum will be written as $\ket{\eta_{a}},\bra{\eta_{a}}$.
The corresponding Fock space will be denoted as $\mathcal{F}^{(a)}(\eta_{a})$. See Figure \ref{fig:MORmatrix} for a pictorial description of the $\mathcal{R}$-matrix.
\begin{figure}
    \centering
    \includegraphics[width=6cm]{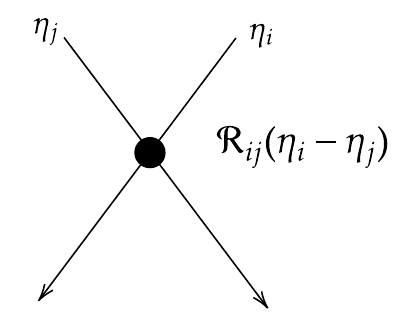}
    \caption{$\mathcal{R}$-matrix}
    \label{fig:MORmatrix}
\end{figure}

The $\mathcal{R}$-matrix is characterized by the following relations\footnote{The motivation of this defining relations will be much clearer in the $q$-deformed language (see \S\ref{sec:universalM0}). }
\begin{align}
      \mathcal{R}_{12}(u) (Q\partial_z -\partial\varphi_1)(Q\partial_z -\partial\varphi_2)&=(Q\partial_z -\partial\varphi_2)(Q\partial_z -\partial\varphi_1)\mathcal{R}_{12}(u)\label{eq:MORmatrix1}\\
     \mathcal{R}_{12}(u)\ket{\eta_{1}}\otimes\ket{\eta_{2}}&=\ket{\eta_{1}}\otimes\ket{\eta_{2}}\,.\label{eq:MORmatrix2}
\end{align}
The argument $u$ of the $\mathcal{R}$-matrix will be identified with the zero mode of $\varphi^-$.
Expanding the first equation (\ref{eq:MORmatrix1}), we have
\begin{align}
    &\mathcal{R}_{12}(u)\left(Q^{2}\partial^{2}-Q(\partial \varphi_{1}+\partial\varphi_{2})\partial-Q\partial^{2}\varphi_{2}+\partial\varphi_1 \partial\varphi_2\right)\nonumber\\
   & ~~~~~~~=(Q^{2}\partial^{2}-Q(\partial\varphi_{1}+\partial\varphi_{2})\partial-Q\partial^{2}\varphi_{1}+\partial\varphi_2\partial\varphi_1)\mathcal{R}_{12}(u)
\end{align}
which gives
\begin{align}
    \mathcal{R}_{12}(u)(\partial\varphi_{1}+\partial\varphi_{2})&=(\partial\varphi_{1}+\partial\varphi_{2})\mathcal{R}_{12}(u),\cr
    \mathcal{R}_{12}(u)(\partial\varphi_1\partial\varphi_2-Q\partial^{2}\varphi_{2})&=(\partial\varphi_1\partial\varphi_2-Q\partial^{2}\varphi_{1})\mathcal{R}_{12}(u).
\end{align}
The first equation means that the $\mathcal{R}$-matrix completely commute with the combination $\varphi^{+}(z)\coloneqq \frac{1}{\sqrt{2}}(\varphi_{1}(z)+\varphi_{2}(z))$, which implies the $\mathcal{R}$-matrix depends only on $\varphi^{-}(z)$. Using this fact the second equation is equivalent to ($\rho:=\frac{Q}{\sqrt{2}}$),
\begin{align}
    &\mathcal{R}_{12}(u)T^-(z;u,\rho)=T^-(z;u,-\rho)\mathcal{R}_{12}(u),\label{RTTR}\\
    &T^-(z; u,\rho) =\left. \frac12(\partial\varphi^{-}(z))^2+\rho\partial^2\varphi^-(z)\right|_{\alpha_0^-\rightarrow u}\,,\quad u=\frac{1}{\sqrt{2}}(\eta_1-\eta_2)\,.
\end{align}
In other words, the $\mathcal{R}$-matrix intertwines two energy-momentum tensors $T^-(z;u, \pm\rho)$ with the same central charge $c=1-6\rho^2$.

We introduce the mode expansion of the energy-momentum tensor as,
\begin{align}
T^-(z;u,\rho) &= \sum_n \sfL_n(u,\rho) z^{-n-2}, \quad \partial\varphi^-(z) = \sum_{n}\sfJ^-_n z^{-n-1}\cr
\sfL_n(u,\rho) &=\frac{1}{2}\sum_{m\neq 0} :\sfJ_{n+m}^-\sfJ_{-m}^- :-\rho n\sfJ_n^--u\sfJ_n^-=:\sfL_n^{(0)}-\rho n \sfJ^-_n -u\sfJ_n^-
\end{align}
and the expansion of $\mathcal{R}(u)$ as,
\begin{equation}
  \mathcal{R}_{12}(u) =\sum_{n=0}^\infty R^{(n)}u^{-n}
\end{equation}
The condition (\ref{RTTR}) implies,
\begin{equation}\label{RTTR1}
    [R^{(m)},\sfJ_n^-]=\left[ R^{(m-1)}, \sfL_n^{(0)}\right]+\rho n \left\{R^{(m-1)},\sfJ^-_n\right\}
\end{equation}
It defines the $\mathcal{R}$-matrix recursively with the initial value condition $R^{(0)}=1$.

We note that there are some options to define the $\mathcal{R}$-matrix for a free boson. Here, we define the operator that switches the sign of $Q=\sqrt{2}\rho$; one may alternatively change the sign of the momentum $u$. The most general choice may be written as (we omit the superscript $-$)
\begin{equation}
    \mathcal{R}(u;s_1,s_2) T(z;u,\rho) = T(z;s_1 u, s_2 \rho)\mathcal{R}(u;s_1,s_2),\quad s_1, s_2=\pm 1\,.
\end{equation}
$\mathcal{R}_{1,-1}$ is the Maulik-Okounkov $\mathcal{R}$-matrix. $\mathcal{R}_{-1,1}$ corresponds to the reflection of the momentum in the Liouville wall, which corresponds to the scattering matrix. We note that $\mathcal{R}_{-1,-1}$ defines the reflection of the free boson itself $\varphi\rightarrow-\varphi$ and the construction is straightforward. The Maulik-Okounkov $\mathcal{R}$-matrix and the scattering matrix are related each other as $\mathcal{R}_{+1,-1}=\mathcal{R}_{-1,-1} \mathcal{R}_{-1,+1}$.

\begin{figure}
    \centering
    \includegraphics[width=14cm]{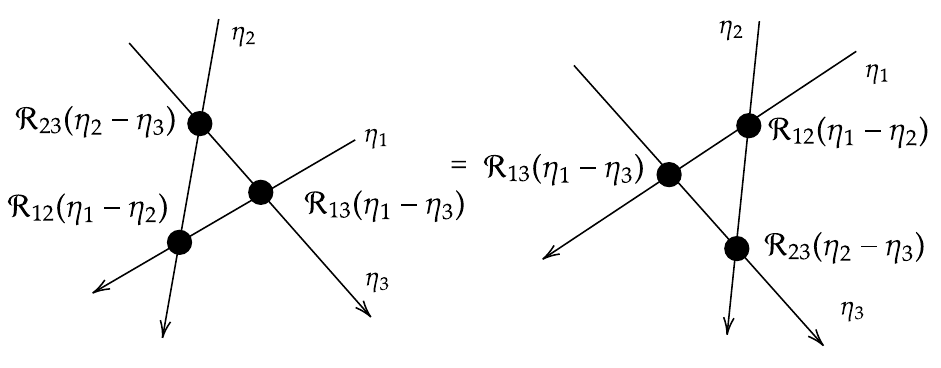}
    \caption{Yang-Baxter equation}
    \label{fig:YBE}
\end{figure}

The above definition shows that the Maulik-Okounkov $\mathcal{R}$-matrix $\mathcal{R}_{ij}(u):\mathcal{F}^{(i)}(\eta_{i})\otimes \mathcal{F}^{(j)}(\eta_{j})\rightarrow \mathcal{F}^{(i)}(\eta_{i})\otimes\mathcal{F}^{(j)}(\eta_{j})$ satisfies the standard Yang-Baxter equation  (see Figure \ref{fig:YBE})
\begin{align}
\begin{split}
    \mathcal{R}_{12}(u_{1}) \mathcal{R}_{13}(u_{2})\mathcal{R}_{23}(u_{3})=
    \mathcal{R}_{23}(u_{3}) \mathcal{R}_{13}(u_{2})\mathcal{R}_{12}(u_{1}),\\
    u_{1}=\frac{\eta_{1}-\eta_{2}}{\sqrt{2}},\quad u_{2}=\frac{\eta_{1}-\eta_{3}}{\sqrt{2}},\quad u_{3}=\frac{\eta_{2}-\eta_{3}}{\sqrt{2}},
    \label{eq:YBE}
\end{split}
\end{align}
which comes from
\begin{align}
    \begin{split}
    &\mathcal{R}_{12}(u_{1})\mathcal{R}_{13}(u_{2})\mathcal{R}_{23}(u_{3})(Q\partial-\partial\varphi_{1})(Q\partial-\partial\varphi_{2})(Q\partial-\partial\varphi_{3})\\
    =&\mathcal{R}_{12}(u_{1})\mathcal{R}_{13}(u_{2})(Q\partial-\partial\varphi_{1})(Q\partial-\partial\varphi_{3})(Q\partial-\partial\varphi_{2})\mathcal{R}_{23}(u_{3})\\
    =&\mathcal{R}_{12}(u_{1})(Q\partial-\partial\varphi_{3})(Q\partial-\partial\varphi_{1})(Q\partial-\partial\varphi_{2})\mathcal{R}_{13}(u_{2})\mathcal{R}_{23}(u_{3})\\
    =&(Q\partial-\partial\varphi_{3})(Q\partial-\partial\varphi_{2})(Q\partial-\partial\varphi_{1})\mathcal{R}_{12}(u_{1})\mathcal{R}_{13}(u_{2})\mathcal{R}_{23}(u_{3}),
    \end{split}\\\notag\\
    \begin{split}
        &\mathcal{R}_{23}(u_{3})\mathcal{R}_{13}(u_{2})\mathcal{R}_{12}(u_{1})(Q\partial-\partial\varphi_{1})(Q\partial-\partial\varphi_{2})(Q\partial-\partial\varphi_{3})\\
        =&\mathcal{R}_{23}(u_{3})\mathcal{R}_{13}(u_{2})(Q\partial-\partial\varphi_{2})(Q\partial-\partial\varphi_{1})(Q\partial-\partial\varphi_{3})\mathcal{R}_{12}(u_{1})\\
        =&\mathcal{R}_{23}(u_{3})(Q\partial-\partial\varphi_{2})(Q\partial-\partial\varphi_{3})(Q\partial-\partial\varphi_{1})\mathcal{R}_{13}(u_{2})\mathcal{R}_{12}(u_{1})\\
        =&(Q\partial-\partial\varphi_{3})(Q\partial-\partial\varphi_{2})(Q\partial-\partial\varphi_{1})\mathcal{R}_{23}(u_{3})\mathcal{R}_{13}(u_{2})\mathcal{R}_{12}(u_{1})\,.
    \end{split}
\end{align}
We note that the $\mathcal{R}$-matrix acts trivially on the Fock space whose index is not included in the $\mathcal{R}$-matrix.

A more explicit form of the $\mathcal{R}$-matrix may be computed order by order in the inverse power of $u$ (we omit $12$ index)
\begin{align}
\begin{split}
    \mathcal{R}(u)&=\sum_{n=0}^\infty R^{(n)}u^{-n}=\exp\left(
    \sum_{n=1}^\infty\frac{r^{(n)}}{u^n}
    \right)\,:\label{R_expansion}\\
    R^{(0)}&=1,\\
    R^{(1)}&=r^{(1)},\\
    R^{(2)}&=r^{(2)}+\frac{1}{2}r^{(1)}r^{(1)},\cdots.
\end{split}
\end{align}
The defining relation (\ref{RTTR1}) implies the relation among $r^{(n)}$,
\begin{align}
    \left[r^{(1)},\sfJ^-_n\right]=& 2\rho \sfJ^-_n\cr
    \left[r^{(2)},\sfJ^{-}_n\right] =& \rho n r^{(1)}\sfJ^-_n+ \rho n \sfJ^-_n r^{(1)}+[r^{(1)},\sfL_n^{(0)}]-\frac12[(r^{(1)})^2,\sfJ_n^-]\cr
    & \cdots
\end{align}
These commutation relation and the normalization condition (\ref{eq:MORmatrix2}) determines $r^{(n)}$ uniquely. For instance,
\begin{align}
\begin{split}
    r^{(1)}=& \sqrt{2} Q\sum_{n=1}^\infty \sfJ^-_{-n}\sfJ^{-}_{n}\,,\cr 
    r^{(2)}=&\sqrt{2}^{-1} Q \sum_{n,m>0}\left(
    \sfJ^-_{-n}\sfJ^-_{-m}\sfJ^-_{n+m}+\sfJ^-_{-n-m}\sfJ^-_{n}\sfJ^-_{m}
    \right)\cr
    & \cdots.
\end{split}\end{align}

For the general Miura operator with the color index, $R^{(c)}$ (\ref{GeneralMiura}), a similar $\mathcal{R}$-matrix was proposed in \cite{Prochazka:2019dvu}.

We note that $\mathcal{R}$-matrix for the $q$-deformed case was derived from bosonic oscillators in \cite{Harada1,Garbali:2020sll,Garbali:2021qko}.

\subsection{Monodromy matrix and mutually commuting operators}\label{sec:Monodromy}
\begin{figure}
    \centering
    \includegraphics[width=10cm]{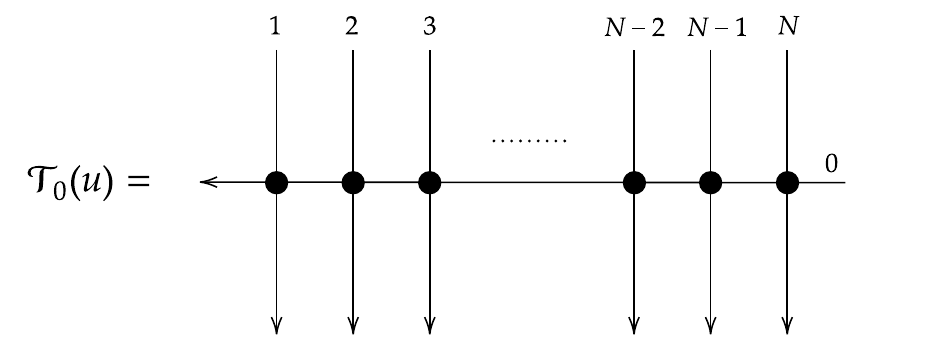}
    \caption{Transfer matrix}
    \label{fig:monodromy}
\end{figure}
From the knowledge of the $\mathcal{R}$-matrix, one may construct an infinite number of commuting operators. For that purpose, we introduce the $T$ matrix (see Figure \ref{fig:monodromy})
\begin{equation}
    \mathcal{T}_0(u) =\mathcal{R}_{01}(u) \mathcal{R}_{02}(u)\cdots \mathcal{R}_{0N}(u)\,.
\end{equation}
Here, the matrix product is taken through the Fock space $\mathcal{F}^{(0)}$, which is referred to as the auxiliary Fock space.  On the other hand, we refer to the tensor product $\mathcal{F}^{(1)}\otimes \cdots \otimes\mathcal{F}^{(N)}$ as the quantum Hilbert space.
It satisfies the so-called RTT relation (see Figure \ref{fig:RTT} and \ref{fig:TTR})
\begin{equation}
    \mathcal{R}_{00'}(u-u') \mathcal{T}_0(u) \mathcal{T}_{0'}(u') =\mathcal{T}_{0'}(u')\mathcal{T}_0(u)\mathcal{R}_{00'}(u-u').\label{eq:RTTrelation}
\end{equation}
\begin{figure}
    \centering
    \includegraphics[width=9cm]{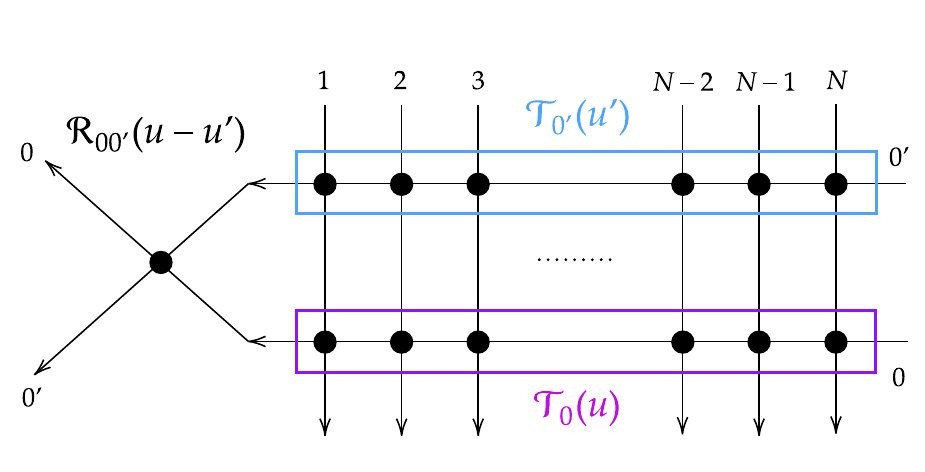}
    \caption{Left hand side of RTT relation (\ref{eq:RTTrelation}).}
    \label{fig:RTT}
\end{figure}
\begin{figure}
    \centering
    \includegraphics[width=9cm]{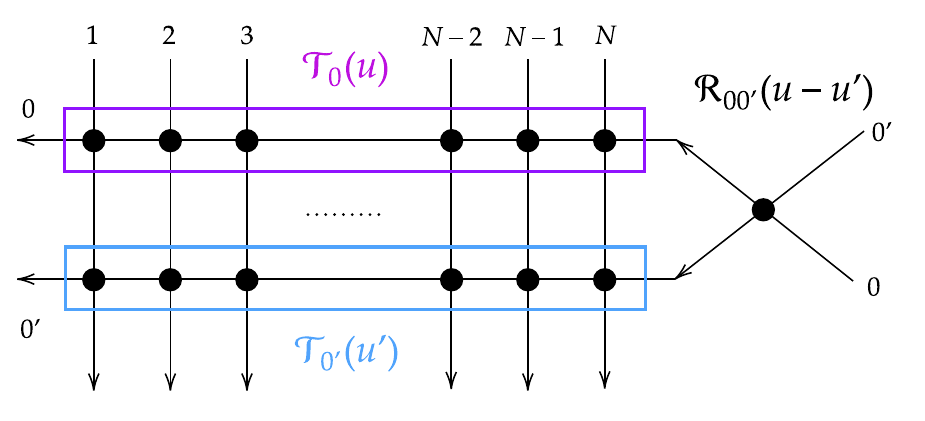}
    \caption{Right hand side of RTT relation (\ref{eq:RTTrelation}).}
    \label{fig:TTR}
\end{figure}
This is obtained as
\begin{align}
\begin{split}
   \mathcal{R}_{00'}(u-u') \mathcal{T}_0(u) \mathcal{T}_{0'}(u')&= \mathcal{R}_{00'}(u-u')\mathcal{R}_{01}(u)\mathcal{R}_{0'1}(u')\cdots \mathcal{R}_{0N}(u)\mathcal{R}_{0'N}(u')\\
   &=\mathcal{R}_{0'1}(u')\mathcal{R}_{01}(u)\cdots \mathcal{R}_{0'N}(u')\mathcal{R}_{0N}(u)\mathcal{R}_{00'}(u-u')\\
   &=\mathcal{T}_{0'}(u')\mathcal{T}_0(u)\mathcal{R}_{00'}(u-u')
\end{split}
\end{align}
where we use the Yang-Baxter equation (\ref{eq:YBE})
\begin{align}
    \mathcal{R}_{00'}(u-u')\mathcal{R}_{0a}(u)\mathcal{R}_{0'a}(u')=\mathcal{R}_{0'a}(u')\mathcal{R}_{0,a}(u)\mathcal{R}_{00'}(u-u')
\end{align}
one by one from the left. This implies that the vacuum expectation value, 
\begin{equation}\label{monodromy_m}
    T(u)=\langle u|\mathcal{T}_{0}(u)|u\rangle
\end{equation}
is mutually commutative
\begin{equation}
    \left[T(u), T(u')\right]=0\,,
\end{equation}
which defines the monodromy matrix \footnote{Usually, in the context of integrable systems, we simply take the trace of the $T$ operator. However, since, in this case, we are considering an infinite dimensional module, the naive trace itself is ill-defined. Thus, we need to add a regulator such as $p^{\sfL^{(0)}_{0}}$ where $\sfL^{(0)}_{0}$ is the Virasoro zero mode and assume $|p|<1$ so that the trace will converge (see for example \cite{Prochazka:2023zdb}). After taking $p\rightarrow0$, the trace will transform into the vacuum expectation value. This discussion is similar to the discussion in finite temperature field theory.}.
This comes from the equations (\ref{eq:RTTrelation})
\begin{align}
\begin{split}
    T(u)T(u')=&{}_{0}\langle u|\mathcal{T}_{0}(u)|u\rangle_{0}{}_{0'}\langle u'|\mathcal{T}_{0'}(u')|u'\rangle_{0'}\\
=&{}_{0}\langle u|\otimes{}_{0'}\langle u'|\mathcal{R}_{00'}(u-u')^{-1}\mathcal{R}_{00'}(u-u')\mathcal{T}_{0}(u)\mathcal{T}_{0'}(u')|u\rangle_{0}\otimes
|u'\rangle_{0'}\\
=&{}_{0}\langle u|\otimes{}_{0'}\langle u'|\mathcal{R}_{00'}(u-u')^{-1}\mathcal{T}_{0'}(u')\mathcal{T}_{0}(u)\mathcal{R}_{00'}(u-u')|u\rangle_{0}\otimes
|u'\rangle_{0'}\\
=&{}_{0}\langle u|\otimes{}_{0'}\langle u'|\mathcal{R}_{00'}(u-u')^{-1}\mathcal{R}_{00'}(u-u')\mathcal{T}_{0}(u)\mathcal{T}_{0'}(u')|u\rangle_{0}\otimes
|u'\rangle_{0'}\\
=&T(u')T(u)
\end{split}
\end{align}
where we use the normalization (\ref{eq:MORmatrix2}).

If we expand, $T(u) =\sum_{n=1}^\infty c^{(n)}u^{-n}$, the coefficients
$c^{(n)}$ are mutually commuting.
The first few components are given as
\begin{align}
    c^{(1)}=&{}_0\langle u|r^{(1)}|u\rangle_{0} =\frac{Q}{\sqrt{2}}\sum_{i=1}^N \sfJ^{(i)}_{-n}\sfJ^{(i)}_{n}\,,\cr
    c^{(2)}-(c^{(1)})^2=& \frac{Q}{4} \sum_{i=1}^N \sum_{n,m=1}^\infty\left(
    \sfJ^{(i)}_{-n}\sfJ^{(i)}_{-m}\sfJ^{(i)}_{n+m}+\sfJ^{(i)}_{-n-m}\sfJ^{(i)}_{n}\sfJ^{(i)}_{m}\right)+\frac{Q^2}{4}\sum_{i<j} \sum_{n>0}n \sfJ_{-n}^{(i)}\sfJ_{n}^{(j)}\label{GeneralizedCS}
\end{align}
For the higher order computation, see, for instance, \cite{Litvinov:2020zeq}. The second formula gives the Hamiltonian for the generalized Calogero-Sutherland, which coincides with 
$\sfD_{0,2}$ in (\ref{D02}, \ref{D02a}) of the affine Yangian $\AY$.
The second term of $c^{(2)}-(c^{(1)})^2$ matches with the extra term coming from the coproduct.
It shows the Yangian structure coming from the Maulik-Okounkov $\mathcal{R}$-matrix is identical to that of $\AY$.

The existence of the second term was also noticed by \cite{Estienne:2011qk} from the study of the conformal block function (see also \cite{Morozov:2013rma} for the discovery in the context of AGT conjecture).

In general, the matrix components
${}_0\langle \mu |\mathcal{T}(u) | \nu\rangle_0$, where $|\mu\rangle_0$ and ${}_0\langle\nu|$ is a basis of the auxiliary Fock space $\mathcal{F}_0$ with the label described by the partitions $\lambda,\mu$, gives the general element of the affine Yangian of $\frakgl_1$. In particular, in \cite{Litvinov:2020zeq}, the authors claims the identifications
\begin{align}
    \psi(u) =& T(u)\cr
    e(u)  \sim &T(u)^{-1}{}_0\langle 0|\mathcal{T}(u)|\Abox\rangle_0\cr
    f(u) \sim & {}_0\langle \Abox|\mathcal{T}(u)|0\rangle_0 T(u)^{-1}
\end{align}
satisfies the algebraic relations similar to the affine Yangian (\ref{eq:AY}).
We note that a similar result for the $q$-deformed case was derived earlier by \cite{Feigin:2015raa}.

The construction of an infinite number of commuting operators in 2d CFT was carried out by a series of celebrated works \cite{bazhanov1996integrable,bazhanov1997integrable,bazhanov1997quantum,Bazhanov:1998dq}. The results obtained by the affine Yangian should be consistent with these results after removing the extra $\U(1)$-factor while the construction is very different. The correspondence between them seems to be an interesting challenge in the future.

Recently, Prochazka and Watanabe \cite{Prochazka:2023zdb} proposed an $\mathcal{R}$-matrix for the whole affine Yangian algebra (or $\mathcal{W}_{1+\infty}$ algebra) by composing the $\mathcal{R}$-matrices in a rectangle form,
\begin{align}
    \boldsymbol{R}&=\mathcal{R}_{1,\bar{N}}\mathcal{R}_{2,\bar{N}}\cdots\mathcal{R}_{N\bar{N}}
    \mathcal{R}_{1\overline{N-1}}\cdots \mathcal{R}_{N\overline{N-1}}\cdots
    \mathcal{R}_{1,\bar{1}}\mathcal{R}_{2,\bar{1}}\cdots\mathcal{R}_{N\bar{1}}
\end{align}
which operates on the tensor product of Fock spaces,
\begin{equation}
    \boldsymbol{R}: \mathcal{F}^{\otimes N}\otimes \mathcal{F}^{\otimes \bar{N}}\quad\longrightarrow\quad \mathcal{F}^{\otimes N}\otimes \mathcal{F}^{\otimes \bar{N}}\,.
\end{equation}
Through the Miura transformation, one may identify $\mathcal{F}^{\otimes N}$ as the Hilbert space of $\mathcal{W}_N$ algebra with an extra $\U(1)$-factor. By taking sufficiently large $N$ and $\bar{N}$, one may identify these Fock spaces to describe $\mathcal{W}_{1+\infty}$ modules. They referred to $\mathcal{F}^{\otimes N}$ (resp. $\mathcal{F}^{\otimes \bar{N}}$) as the quantum (resp. auxiliary) Hilbert space.
$\boldsymbol{R}(u)$ has a similar expansion as (\ref{R_expansion}), where the exponentiated operators $r^{(n)}$ are expressed in terms of $U_n$ (resp. $\bar U_n$) ($n=1,2,3,\cdots$), the spin $n$ fields in the quantum (resp. auxiliary) Hilbert space. The authors of \cite{Prochazka:2023zdb} claim that one may regard $\boldsymbol{R}$ as the universal $\scR$-matrix, which satisfies the equations in Theorem \ref{Th:Universal_R}.

\subsection{Bethe ansatz equation}\label{s:Bethe}
Through the use of the $\cR$-matrix, one may define mutually commuting operators acting of $\mathcal{F}^{(1)}(u_{1})\otimes \cdots \otimes \mathcal{F}^{(n)}(u_{n})$ as
\begin{equation}\label{R:ILW}
    \boldsymbol{T}_p(u)=\mathrm{Tr}_{\mathcal{F}^{(0)}}\left(
    p^{\sfL_0^{(0)}} \mathcal{R}_{0,1}(u-u_1)\mathcal{R}_{0,2}(u-u_2)\cdots \mathcal{R}_{0,n}(u-u_n)
    \right)
\end{equation}
As $T(u)$, one may expand $\boldsymbol{T}_p(u)$ at $u^{-1}=0$ to generate an infinite number of conserved currents $\boldsymbol{I}_k(p)$:
\be\begin{aligned}
    \boldsymbol{I}_1(p)&=\frac{1}{2\pi}\int dx \left(\frac12\sum_{k=1}^n (\partial\varphi_k)^2\right)\cr
    \boldsymbol{I}_2(p)&=\frac{1}{2\pi}\int dx\left(
    \frac13 \sum_{k=1}^n (\partial\varphi_k)^3+Q\left(
\frac{\mathrm{i}}{2}\sum_{i,j}\partial\varphi_i D\partial\varphi_j
+\sum_{i<j}\partial\varphi_i\partial^2\varphi_j
    \right)
    \right) \cr
    & \cdots
\end{aligned}\ee
where $D$ is a non-local operator whose Fourier image is $D(k)=k\frac{1+p^k}{1-p^k}$. 
In \cite{Litvinov:2020zeq}, they are related to the \emph{intermediate long wave} (ILW) conserved charges (see also earlier works \cite{Litvinov:2013zda,Alfimov:2014qua}.)

In the formula (\ref{R:ILW}), we used a single Fock space $\mathcal{F}^{(0)}$ as the auxiliary space. One may generalize it to the tensor product\footnote{Each of the tensor products is the same representation as $\mathcal{F}_{3}(u)$. We note that the colors of the Fock tensor products in the auxiliary space are all set to be $c=3$ and the subscripts are suppressed. } $\tilde{\mathcal{F}}^{(1)}(x_1)\otimes\cdots \otimes \tilde{\mathcal{F}}^{(N)}(x_N)$, where $\tilde{\mathcal{F}}^{(i)}(x)$ implies the Fock space used for the auxiliary space. By combining them, one can construct an operator $\mathcal{T}(\boldsymbol{x}; \boldsymbol{u})$ in the form of a rectangle multiplication of the $\mathcal{R}$-matrices, which gives an operator acting on $\tilde{\mathcal{F}}^{(1)}(x_1)\otimes\cdots\otimes\tilde{\mathcal{F}}^{(N)}(x_N)\otimes \mathcal{F}^{(1)}(u_1) \otimes\cdots\otimes \mathcal{F}^{(n)}(u_n)$. We introduce the off-shell Bethe state as
\begin{equation}
    |B(\boldsymbol{x})\rangle_{\boldsymbol{u}} = \langle \emptyset |\mathcal{T}(\boldsymbol{u},\boldsymbol{x})|\chi\rangle_{\boldsymbol{x}} \otimes |\emptyset\rangle_{\boldsymbol{u}}
\end{equation}
where
\begin{equation}\label{eq:Bethe}
|\chi\rangle_{\boldsymbol{x}}=|\Abox,\cdots,\Abox\rangle_{\boldsymbol{x}}\sim \oint dz_N\cdots \oint dz_1 e(z_N)\cdots e(z_1)|\emptyset\rangle_{\boldsymbol{x}}\,.
\end{equation}
It is claimed that $|B(\boldsymbol{x})\rangle_{\boldsymbol{u}}$ is the simultaneous eigenvector of $\boldsymbol{I}_s(p)$
when the parameters $\boldsymbol{x}$ satisfies
\begin{equation}
    p\prod_{j\neq i}\varphi(x_j-x_i) \prod_{k=1}^n\frac{x_i-u_k -Q/2}{x_i-u_k+Q/2} =1
\end{equation}
where $\varphi(u)$ is the structure function in (\ref{str-fn-AY}).
The eigenvalues of $\boldsymbol{I}_s(p)$ is expressed as,
\begin{equation}
    \boldsymbol{I}_1(p)\sim -\frac12 \sum_{k=1}^n u_k^2+N~,\qquad
    \boldsymbol{I}_2(p)\sim \frac13 \sum_{k=1}^n u_k^3-2\mathrm{i}\sum_{j=1}^N x_j~,\qquad\cdots\,.
\end{equation}

For the $q$-deformed case, an immediate analog of the Bethe ansatz equation can be found in \cite{Feigin:2015raa,Feigin:2016pld}.

A motivation to study the Bethe equation (\ref{eq:Bethe}) is to explore the relation with the integrability of the deformed CFT \cite{bazhanov1996integrable,bazhanov1997integrable,Bazhanov:1998dq}, where an infinite number of commuting operators, associated with the Virasoro generators, and their spectrum was studied. The mutually commuting operators are written in the form
\bea\label{mutually-commuting}
 I_1 &= \sfL_{0}-\frac{c}{24}\\
 I_2 &= \sfL_0^2 +2\sum_{m>0} \sfL_{-m}\sfL_{m}-\frac{c+2}{12}\sfL_0 +\frac{c(5c+22)}{2880}\\
 I_3&=\sum_{m_1+m_2+m_3=0}:\sfL_{m_1}\sfL_{m_2}\sfL_{m_3}:+\cdots,\\
 & \cdots
\eea
These mutually commuting operators correspond to the Cartan generators $\uppsi_n$ in the affine Yangian ($\mathsf{K}^+_n$ (or $\mathsf{K}^-_n$) for $q$-deformation). Note that, at the level of affine Yangian, the diagonal basis corresponds to the Fock representation \eqref{deg-Fock}. However, since the U(1) part is omitted in \eqref{mutually-commuting}, finding the eigenstates of the operators \eqref{mutually-commuting} requires considering this omission. 

The Bethe ansatz equation (\ref{eq:Bethe}) serves as the intertwining equation by incorporating the extra parameter $p$. In the limit $p\to 0$ (resp. $p\to\infty$), a striking simplification occurs in (\ref{eq:Bethe}), leading to the vanishing of the numerator (resp. the denominator) of the left-hand side. The solution $x_i$ to such a simplified system describes the position of a box in the Young diagram (or plane partition) as in the vertical representation. As for the commuting operators $\boldsymbol{I}_k(p)$, the nonlocality disappears as $D(k)\sim \pm k$, and one can see that they became identical to the commuting operators $c^{(k)}$ (\ref{GeneralizedCS}) in the previous subsection. It is referred to as the affine Yangian or the Benjamin-Ono limit. 

In contrast, in the $p\to 1$ limit, the operator $D(k)$ becomes singular, which implies that the $\U(1)$ factor becomes singular and decouples from the affine Yangian. The system is described by the local operators such as $\mathcal{W}_\infty[\mu]$, and the commuting operators $\boldsymbol{I}_k(p)$ should be equivalent to $I_k$ in (\ref{mutually-commuting}) after the proper reduction, such as $\cW_\infty[\mu]$ to Virasoro algebra. We also note that the $p=1$ specialization of \eqref{eq:Bethe} describes the Bethe equation for the ODE/IM correspondence (for instance, \cite{Dorey:1997rb,Bazhanov:2003ni}). 

Therefore, the Bethe equation (\ref{eq:Bethe}) establishes a profound intertwining relationship between these two distinct limits and encapsulates the ILW system.  Further investigation on the Bethe equation can be found in \cite{Bonelli:2014iza,Koroteev:2015dja,Koroteev:2016znb,Prochazka:2023zdb}.

\subsection{Relation with universal \label{s:univR}\texorpdfstring{$\mathcal{R}$}{R}-matrix of quantum toroidal algebra}\label{sec:universalM0}
Although it is unclear for the affine Yangian case, it was shown and proven that the universal $\mathcal{R}$-matrix exists for the \QTA \cite{negut2014shuffle} (see also \cite{miki2007q}). Moreover, as usual quantum groups, \QTA\, can also be realized as a quantum double \cite{Feigin:2015raa,burban2012hall} (see Appendix \ref{app:QG} for a review of quantum groups). 

We introduce four sets of subalgebras of \QTA depending on the principal and homogeneous degrees (see (\ref{eq:phdegreee}),(\ref{eq:DIMsubalgebra})) as
\begin{align}\label{triangular}
\begin{split}
    \mathcal{E}_{\geq}&=\langle \sfE_{n} \,(n\in\mathbb{Z}),\mathsf{H}_{r}\,(r>0), C,C^{\perp}, D, D^{\perp} \rangle,\\
    \mathcal{E}_{\leq}&=\langle \sfF_{n} \,(n\in\mathbb{Z}),\mathsf{H}_{-r}\,(r>0), C,C^{\perp}, D, D^{\perp} \rangle,\\
    \mathcal{E}^{\perp}_{\geq}&=\langle \sfE^{\perp}_{n} \,(n\in\mathbb{Z}),\mathsf{H}^{\perp}_{r}\,(r>0), C,C^{\perp}, D, D^{\perp} \rangle,\\
    \mathcal{E}^{\perp}_{\leq}&=\langle \sfF^{\perp}_{n} \,(n\in\mathbb{Z}),\mathsf{H}^{\perp}_{-r}\,(r>0), C,C^{\perp}, D, D^{\perp} \rangle,
\end{split}
\end{align}
where shortly we denoted $\mathcal{E}=\QTA$. Generally, to construct a quantum double, one starts from a Hopf algebra $H$, then introduces a symmetric non-degenerate bilinear form $\langle \ , \ \rangle:H\times H\rightarrow \mathbb{C}$ obeying properties such as (\ref{eq:qdpairing}). With a pair $(H,\langle \ , \ \rangle)$ we can obtain the quantum double $D(H)=H\otimes H^{*\text{cop}}$ of $H$, where $H^{*\text{cop}}$ is the dual Hopf algebra with the opposite coalgebra structure (see Appendix \ref{app:QG} for details). There is a canonical element of this quantum double called the universal $\mathcal{R}$-matrix obeying the properties (see also \ref{Th:Universal_R}):
\begin{align}
\begin{split}
    \mathcal{R}\Delta(x)&=\Delta^{\text{op}}(x)\mathcal{R},\\
    (\Delta\otimes \text{id})\mathcal{R}=\mathcal{R}_{13}\mathcal{R}_{23},&\quad (\text{id}\otimes \Delta)\mathcal{R}=\mathcal{R}_{12}\mathcal{R}_{12},\\
    \mathcal{R}_{12}\mathcal{R}_{13}\mathcal{R}_{23}&=\mathcal{R}_{23}\mathcal{R}_{13}\mathcal{R}_{12}.
\end{split}
\end{align}
Using the bialgebra $H=\mathcal{E}^{\perp}_{\geq}$ and its dual $H^{\ast\text{cop}}=\mathcal{E}^{\perp}_{\leq}$ and imposing non-trivial pairings, the quantum double of $\mathcal{E}^{\perp}_{\geq}$ is identified with $\mathcal{E}^{\perp}_{\geq}\otimes \mathcal{E}^{\perp}_{\leq}$ which is isomorphic to $\QTA$ after taking the quotient of extra central elements (see \cite{Feigin:2015raa} and references there for details).

Under this situation, the universal $\mathcal{R}$-matrix of $\QTA$ is an element of a certain completion of $\mathcal{E}^{\perp}_{\geq}\otimes \mathcal{E}^{\perp}_{\leq}$ with the structure
\begin{align}
&\mathcal{R}=\mathcal{R}^{(0)}\mathcal{R}^{(1)}\mathcal{R}^{(2)},\label{eq:universalRgl1}
\end{align}
where $\mathcal{R}^{(0)}$ is the part expressed with modes $C,C^{\perp},d,d^{\perp}$, $\mathcal{R}^{(1)}$ is expressed with $\mathsf{H}_{r}^{\perp}\otimes \mathsf{H}_{-r}^{\perp}$, and $\mathcal{R}^{(2)}$ is the part expressed with $\mathsf{E}_{n}^{\perp},\mathsf{F}_{n}^{\perp}$.

Practically, to determine the $\mathcal{R}$-matrix, one needs to specify the representation of the algebra. In the literature, only $\mathcal{R}$-matrices of Fock representations are known. Focusing on Fock representations, one starts from the vertex operator Fock representation in \S\ref{sec:QTrep} and inserts it into
\begin{align}
    \mathcal{R}(u_{1}/u_{2})\rho_{\mathcal{F}_{c_{1}}(u_{1})\otimes \mathcal{F}_{c_{2}}(u_{2})}(\Delta(x))&=\rho_{\mathcal{F}_{c_{2}}(u_{2})\otimes \mathcal{F}_{c_{1}}(u_{1})}(\Delta^{\text{op}}(x))\mathcal{R}(u_{1}/u_{2}),\label{eq:Rmatrixdef}\\
    \mathcal{R}(u_{1}/u_{2})=\rho_{\mathcal{F}_{c_{1}}(u_{1})\otimes \mathcal{F}_{c_{2}}(u_{2})}(\mathcal{R}),&\quad x\in\QTA
\end{align}
where $\rho_{\mathcal{F}_{c_{1}}(u_{1})\otimes \mathcal{F}_{c_{2}}(u_{2})}$ is the representation map. Note here that generally, we can consider $\mathcal{R}$-matrices between Fock representations with different color indices $c_{i}\,(i=1,2,)$ (see for example \cite{Prochazka:2019dvu}). There will be also an extra parameter that comes from the spectral parameter of the representation, and the $\mathcal{R}$-matrix depends on it: $\mathcal{R}(u)$. Since the Fock representation is an infinite dimensional representation, the $\mathcal{R}$-matrix will have infinite terms expanded by the spectral parameter, and the coefficients are obtained recursively (similar to the discussion in \S\ref{s:MO-Rmatrix}). This procedure was first done in \cite{Harada1}, and it was shown that the $\mathcal{R}$-matrix obtained in this way would match with the Maulik-Okounkov $\mathcal{R}$-matrix in (\ref{eq:MORmatrix1}) after taking the degenerate limit. Actually, the key equation (\ref{eq:Rmatrixdef}) already gives us an intuition that the $\mathcal{R}$-matrix in the degenerate limit should be obtained by using Miura operators as in (\ref{eq:MORmatrix1}). This is because, in the degenerate limit, the arising algebra is spherical degenerate DAHA/$\cW$-algebra/affine Yangian and the coproduct structure is given in (\ref{SdH-coproduct}, \ref{D02}, \ref{D02a}) which are equivalent with composition of Miura operators.

\section{Algebraic approach to topological vertex}\label{sec:algebraic_topvertex}
After an exhaustive exploration of the representations of quantum toroidal $\frakgl_1$ in \S\ref{sec:QTrep}, we now turn our attention to their applications to supersymmetric theories. Central to this section is the study of BPS states in 5d $\cN=1$ supersymmetric gauge theories as the space of BPS states receives an action of quantum toroidal $\frakgl_1$. The anti-self-dual connections, known as instantons, satisfy the BPS equations in 5d $\cN=1$ supersymmetric gauge theories. Consequently, an instanton partition function introduced by Nekrasov as a generation function of BPS states can be constructed from quantum toroidal $\frakgl_1$. The relationship between BPS states and quantum toroidal $\frakgl_1$ stands out as a concrete realization of the BPS/CFT correspondence.

The celebrated Alday-Gaiotto-Tachikawa (AGT) duality is a particular and remarkable manifestation of the BPS/CFT correspondence. It posits a striking connection between BPS states in 4d $\mathcal{N}=2$ theories and representations of the $\cW$-algebras. 
Given that the instanton partition functions encapsulate non-perturbative dynamics, it provides a novel way for evaluating non-perturbative effects in 4d $\mathcal{N}=2$ theories using 2d CFT methods. This is the foundational principle of the BPS/CFT correspondence. Various extensions and elaborations of this duality will be the focus of this section.

A 5d $\cN=1$ supersymmetric gauge theory can be constructed from the M-theory (or type IIB string theory) compactified on Calabi-Yau 3-fold, and the instanton partition function can therefore be computed from the topological string approach, particularly employing the topological vertex. In this section, we provide a concise overview of the role of the topological vertex in the derivation of the instanton partition function. Furthermore, we delve into the intriguing link, as unveiled in \cite{Awata:2011ce}, between the topological vertex and the intertwiner of
$\QTA$ discussed in \S\ref{sec:intertwiner}.  This naturally embeds instanton partition function into the algebraic framework of $\QTA$, and through its intimate relation with the (q-deformed) $\cW$-algebras discussed in this article from different aspects in \S\ref{sec:AY}, \S\ref{sec:SHc-AY-W} and \S\ref{sec:deformedW}, the AGT duality can be explained in a universal way for both the 4d and 5d supersymmetric gauge theories. 

\subsection{4d \texorpdfstring{$\cN=2$}{N=2} and 5d \texorpdfstring{$\cN=1$}{N=1} instanton partition functions}\label{sec:AFStopvertex}
Nekrasov instanton partition functions \cite{Nekrasov:2002qd} play a central role in this article. It is the partition function of a 5d $\cN=1$ gauge theory put on the $\Omega$-background, i.e. a flat space $\mathbb{C}^2\simeq\mathbb{R}^4$ equipped with a U(1)$\times$U(1) isometry
\begin{equation}
    (z_1,z_2)\in\mathbb{C}^2\rightarrow (q_1z_1,q_2z_2)
\end{equation}
when we go around the $S^1$ circle. It is called the instanton partition function because it receives non-perturbative contributions only from instantons in the gauge theory.

We consider gauge theories with quiver structure $\Gamma$, also known as quiver gauge theories. Each circle node $V_i\in\Gamma$ of the quiver is associated with a gauge field $G_i$, referred to as a gauge node. There is a bifundamental hypermultiplet corresponding to each edge $e_{i\rightarrow j}$ in the quiver, transforming under the fundamental representation of $G_i$ and the anti-fundamental representation of $G_j$. (See Figure \ref{f:quiver}) Generally, the quiver is endowed with hypermultiplets with flavor symmetry in the fundamental representation of a gauge group $G_i$, which are represented graphically by square nodes. As we only consider gauge theories with U($N$)-type gauge groups and flavor symmetry, it is enough to specify the rank $N$ of U($N$) group attached to each vertex and box in the quiver. We denote this information by an integer number inside each vertex and box as in Figure \ref{f:quiver}.

The full partition function of 5d supersymmetric gauge theories on $\mathbb{C}^2_{\epsilon_1,\epsilon_2}\times S^1$ is factorized into the classical piece, the perturbative part (one-loop factor) and the non-perturbative contributions (also known as the instanton partition function).
\begin{equation}
    \cZ=\cZ_{\textrm{cl}}\cZ_{\textrm{1-loop}}\cZ_{\textrm{inst}}.
\end{equation}
The classical part is given by \cite{Nekrasov:2002qd,Kim:2019uqw}
\begin{equation}
    \cZ_{\textrm{cl}}=\exp\lt(\frac{1}{\epsilon_1\epsilon_2}\lt(\frac{m_0}{2}h_{ij}\fraka_i\fraka_j+\frac{\kappa}{6}d_{ijk}\fraka_i\fraka_j\fraka_k\rt)\rt)\label{eq:classical}
\end{equation}
where $h_{ij}=\Tr(T_iT_j)$, $d_{ijk}=\frac{1}{2}\Tr(T_i\{T_j,T_k\})$, $\mathfrak{q}=e^{-Rm_0}$, and $\kappa$ is the Chern-Simons level.
The one-loop factor is further decomposed to a Cartan part and a root part
\bea
    \cZ^G_{\text{Cartan}}={\rm P.E.}\lt[-\frac{{\rm rank}(G)}{2}\lt(\frac{1}{(1-q_1)(1-q_2)}+\frac{1}{(1-q^{-1}_1)(1-q^{-1}_2)}\rt)\rt],\cr
    \cZ^G_{\text{root}}={\rm P.E.}\lt[-\lt(\frac{1}{(1-q_1)(1-q_2)}+\frac{1}{(1-q^{-1}_1)(1-q^{-1}_2)}\rt)\sum_{\alpha\in\Delta_+}e^{-R\alpha\cdot \fraka}\rt],\label{eq:one-loop}
\eea
where $\Delta_+$ represents the set of positive roots, the plethystic exponential ${\rm P.E.}$ is defined as
\begin{equation}
    {\rm P.E.}\lt(f(x_1,x_2,\dots,x_n)\rt):=\exp\lt(\sum_{k=1}^\infty\frac{1}{k}f(x^k_1,x^k_2,\dots,x^k_n)\rt),
\end{equation}
applied to the variables $q_1$, $q_2$, $u_n:=e^{R\fraka_n}$ in \eqref{eq:one-loop}, and $\fraka=(\fraka_n)_{n=1}^N=(\fraka_1,\fraka_2,\dots,\fraka_N)$ is the set of Coulomb branch parameters.

The instanton partition function can be written purely in terms of the so-called Nekrasov factor, which is usually defined as \eqref{def-Nekra}. In this article, an equivalent expression of the Nekrasov factor is more useful for us\footnote{This is the factor introduced in (\ref{eq:Nekrasovfactor-general}) after setting $c=3$.}
\bea
    N_{\lambda^{(1)}\lambda^{(2)}}(u_1/u_2;q_1,q_2)=\prod_{\substack{x\in \lambda^{(1)}\\ y\in\lambda^{(2)}}}S\left(\frac{\chi_{x}}{\chi_{y}}\right)\times \prod_{x\in\lambda^{(1)}}\left(1-\frac{\chi_{x}}{q_{3}u_{2}}\right)\times \prod_{y\in\lambda^{(2)}}\left(1-\frac{u_{1}}{\chi_{y}}\right),\label{Nekra-S}
\eea
where
\bea
    \chi_{x}=u_{l}q_{1}^{i-1}q_{2}^{j-1},\quad x=(i,j)\in\lambda^{(l)}
\eea
is the box in the $i$-th row and $j$-th column in the Young diagram $\lambda^{(l)}$ associated to the Coulomb branch parameter $u_{l}$
and we write $S_{c=3}$ in \eqref{Sc} as
\begin{equation}
    S(z)=\frac{(1-q_1z)(1-q_2z)}{(1-z)(1-q_1q_2z)}.\label{def-S}
\end{equation}
A derivation of \eqref{Nekra-S} from \eqref{def-Nekra} is presented in Appendix \ref{a:Nekra}.

Given the field contents of a gauge theory, its instanton partition function can be spelled out by multiplying the corresponding contribution from each multiplet.
\begin{itemize}
    \item Vector multiplet of $\U(N)$ gauge theory:
    \bea
        \mathcal{Z}_{\text{vect.}}(\boldsymbol{u},\boldsymbol{\lambda})=\prod_{m,n=1}^{N}\frac{1}{N_{\lambda^{(m)}\lambda^{(n)}}(u_{m}/u_{n};q_1,q_2)},\quad \boldsymbol{u}=(u_{1},\ldots,u_{N}),\quad \boldsymbol{\lambda}=(\lambda^{(1)},\ldots,\lambda^{(N)}).\label{Z_vect}
    \eea
    \item Bifundamental hypermultiplet under $\U(N)\times \U(N')$ with mass $\mu$:
    \bea
        \mathcal{Z}_{\text{bfd.}}(\boldsymbol{u},\boldsymbol{\lambda}\,|\,\boldsymbol{u}',\boldsymbol{\lambda}'\,|\,\mu)=\prod_{n=1}^{N}\prod_{m=1}^{N'}N_{\lambda^{(n)}\lambda^{\prime(m)}}(\mu v_n/v'_m;q_1,q_2).
    \eea
    \item Fundamental and anti-fundamental matter multiplet with global symmetry $\U(f)$ and $\U(\tilde{f})$ and mass $\boldsymbol{\mu}^{(\text{f})}$ and $\boldsymbol{\mu}^{(\text{af})}$:
    \bea
        \mathcal{Z}_{\text{fund.}}(\boldsymbol{\mu}^{(\text{f})},\boldsymbol{\lambda})=&\mathcal{Z}_{\text{bfd.}}(\boldsymbol{u},\boldsymbol{\lambda}\,|\,\boldsymbol{\mu}^{(\text{f})},\vec{\emptyset}\,|\,1)=\prod_{x\in\boldsymbol{\lambda}}\prod_{j=1}^{f}\left(1-q_{3}^{-1}\chi_{x}(\mu^{(\text{f})}_{j})^{-1}\right),\quad \boldsymbol{\mu}^{(\text{f})}=(\mu^{(\text{f})}_{1},\ldots,\mu^{(\text{f})}_{f}),\cr
        \mathcal{Z}_{\text{a.f.}}(\boldsymbol{\mu}^{(\text{af})},\boldsymbol{\lambda})=&\mathcal{Z}_{\text{bfd.}}(\boldsymbol{\mu}^{(\text{af})},\vec{\emptyset}\,|\,\boldsymbol{v},\boldsymbol{\lambda}\,|\,1)=\prod_{x\in\lambda}\prod_{j=1}^{\tilde{f}}\left(1-\mu_{j}^{(\text{af})}\chi_{x}^{-1}\right),\quad \boldsymbol{\mu}^{\text{af}}=(\mu^{(\text{af})}_{1},\ldots,\mu_{\tilde{f}}^{(\text{af})}).
    \eea
    \item Chern-Simons term of level $\kappa$:
    \bea
        \mathcal{Z}_{\text{CS}}(\kappa,\boldsymbol{\lambda})=\prod_{x\in\boldsymbol{\lambda}}(\chi_{x})^{\kappa}.
    \eea
    \item Topological term for $\U(N)_{i}$:
    \bea
        \mathcal{Z}_{i,\text{top}}=\mathfrak{q}_{i}^{|\boldsymbol{\lambda}_{i}|}
    \eea
    \item Total partition function of a quiver $\Gamma$:
    \bea
        \mathcal{Z}_{\text{inst}}[\Gamma]=\sum_{\{\boldsymbol{\lambda}_{i}\}}&\prod_{i\in\Gamma}\mathfrak{q}_{i}^{|\boldsymbol{\lambda}_{i}|}\mathcal{Z}_{\text{vect.}}(\boldsymbol{v}_{i},\boldsymbol{\lambda}_{i})\mathcal{Z}_{\text{CS}}(\kappa_{i},\boldsymbol{\lambda}_{i})\mathcal{Z}_{\text{fund.}}(\boldsymbol{\mu}^{(\text{f})}_{i},\boldsymbol{\lambda}_{i})\mathcal{Z}_{\text{a.f.}}(\boldsymbol{\mu}^{(\text{af})}_{i},\boldsymbol{\lambda}_{i})\cr
        &\times\prod_{\langle ij\rangle\in\Gamma}\mathcal{Z}_{\text{bfd.}}(\boldsymbol{v}_{i},\boldsymbol{\lambda}_{i}|\boldsymbol{v}_{j},\boldsymbol{\lambda}_{j}|\mu_{ij}).\label{eq:instantonpartitionfunction}
    \eea

\end{itemize}

\paragraph{4d limit}

One can take the radius $R$ of $S^1$ to zero to obtain the partition function of 4d $\cN=2$ gauge theories. To do so, we re-parameterize the Coulomb branch parameters, the $\Omega$-deformation parameters and the mass parameters as
\bea
    u_n=e^{R\fraka_n},\quad {\rm for}\ n=1,\dots,N,\cr
    q_{1,2}=e^{R\epsilon_{1,2}},\quad \mu_j=e^{Rm_j}.\label{4d-lim}
\eea
In the limit $R\rightarrow 0$
\begin{equation}
    1-e^{R\phi}=-2e^{\frac{R\phi}{2}}\sinh\lt(\frac{R\phi}{2}\rt)\rightarrow -R\phi
\end{equation}
and thus the trigonometric expression of the 5d partition function reduces to its rational version in 4d. In practice, one only needs to replace the Nekrasov factor with its 4d version:
\begin{align}
    N^{\text{4d}}_{\lambda\nu}(\fraka_m,\fraka_n;\epsilon_1,\epsilon_2)=&\prod_{(i,j)\in\lambda}\lt(\fraka_n-\fraka_m-\epsilon_1(\lambda^t_j-i+1)+\epsilon_2(\nu_i-j)\rt)\cr
   & \times\prod_{(i,j)\in\nu}\lt(\fraka_n-\fraka_m+\epsilon_1(\nu^t_j-i)-\epsilon_2(\lambda_i-j+1)\rt).
\end{align}

The instanton counting parameter $\mathfrak{q}$ is rescaled in this limit according to the field contents in the gauge theory. For a theory with $f$ fundamental and $\tilde{f}$ anti-fundamental hypermultiplets, $\mathfrak{q}$ is rescaled to 
\begin{equation}
    \mathfrak{q}\rightarrow \mathfrak{q}_{4d}R^{2N-f-\tilde{f}},\label{rescal-q}
\end{equation}
to obtain a finite result of the instanton partition function. 

The perturbative part of the partition function can also be reduced to 4d in a more implicit manner. By putting
\begin{equation}
    \cZ_{\text{pert}}=\exp\lt(\sum_{\alpha\in\Delta}\gamma_{\epsilon_1,\epsilon_2}(\mathfrak{a}_\alpha|R,\Lambda)\rt),
\end{equation}
with the sum running over the set of all roots $\Delta$, one can write the perturbative partition function in terms of the function \cite{Nekrasov:2003rj} (up to some terms irrelevant in this article), 
\begin{equation}
    \gamma_{\epsilon_1,\epsilon_2}(x|R,\Lambda)=\frac{1}{2\epsilon_1\epsilon_2}\lt(-\frac{R}{6}\lt(x+\frac{\epsilon_1+\epsilon_2}{2}\rt)^3+x^2\log(R\Lambda)\rt)+\sum_{n=1}^\infty \frac{1}{n}\frac{e^{-R nx}}{(e^{R n\epsilon_1}-1)(e^{R n\epsilon_2}-1)},
\end{equation}
and its 4d limit $R\rightarrow 0$ is given by \cite{Nekrasov:2003rj}
\begin{equation}
    \gamma^{\text{4d}}_{\epsilon_1,\epsilon_2}(x;\Lambda)=\lt.\frac{{\rm d}}{{\rm d}s}\lt(\frac{\Lambda^s}{\Gamma(s)}\int_0^\infty {\rm d}t\ t^{s-1}\frac{e^{-tx}}{(e^{\epsilon_1 t}-1)(e^{\epsilon_2 t}-1)}\rt)\rt|_{s=0}.
\end{equation}
One can see that the integral is coming from the infinite sum by identifying $t\equiv nR$, 
\begin{equation}
    \sum_{n=1}^\infty \frac{1}{n}\frac{e^{-R nx}}{(e^{R n\epsilon_1}-1)(e^{R n\epsilon_2}-1)}\rightarrow \int_0^\infty {\rm d}t\ t^{-1}\frac{e^{-tx}}{(e^{\epsilon_1 t}-1)(e^{\epsilon_2 t}-1)},
\end{equation}
and also follows from the fact that $\lt.\frac{{\rm d}}{{\rm d}s}\frac{1}{\Gamma(s)}\rt|_{s=0}=1$.

\subsection{Alday-Gaiotto-Tachikawa duality}\label{sec:AGT}

A large class of 4d $\cN=2$ supersymmetric field theories, called class $\cS$, has been constructed by compactifying the 6d $\cN=(2,0)$ theory of type $ADE$ $\frakg$ on Riemann surfaces with (both regular and irregular) punctures \cite{Gaiotto:2009we,Gaiotto:2009hg,Gaiotto:2009ma,Xie:2012hs}. More precisely, the basic information of a theory of class $\mathcal{S}$ is encoded in a Hitchin system \cite{0887284}
\begin{equation}
\label{eq:spectral}
\begin{tikzcd}
\Sigma \arrow[r,hook]
& T^*C_{g,n} \arrow[d] \\
& C_{g,n}
\end{tikzcd}\qquad
\end{equation}
where the base curve $C_{g,n}$ is called a Gaiotto (UV) curve, and
the spectral curve $\Sigma$ is a Seiberg-Witten (IR) curve. Often, the theory is denoted by $T[C_{g,n},\frakg]$ where $\frakg$ is of type $ADE$.\footnote{It is worth noting that the theories are sensitive to not only gauge Lie algebra but also the global structure of the gauge group although this subtlety is not included here. For more information, we refer to \cite{Aharony:2013hda,Tachikawa:2013hya}.}
Motivated by this construction, a tantalizing duality, called AGT duality, between a 4d $\cN=2$ theory $T[C_{g,n},\frakg]$ and a 2d $\cW(\frakg)$ Toda CFT on $C_{g,n}$ was discovered in  \cite{Alday:2009aq,Wyllard:2009hg}, and the duality predicts the equality of the partition functions. A generic class $\cS$ theory does not admit Lagrangian description, which makes it challenging to compute its partition function. However, for theories that do have Lagrangian descriptions, the equivalence between their partition functions has been largely verified. In this equivalence, the instanton partition function maps to the conformal block while the one-loop determinant corresponds to (a product of ) the three-point function. As an example, we consider the case of $T[C_{0,4},\fraksu_2]$, the SU(2) gauge theory with $N_f=4$ hypermultiplets, which is dual to the dual Liouville theory with four primaries on a sphere.

The Liouville theory, or the $A_1$ Toda theory, is an interacting 2d CFT where the action is given by
\begin{equation}
S=\frac{1}{4 \pi} \int d^2 z \sqrt{g}\left(g^{a b} \partial_a \phi \partial_b \phi+iQ R \phi+4 \pi \mu e^{2 i\sqrt{\beta} \phi}\right) ~.
\end{equation}
The central charge is independent of $\mu$ and is given by $c=1-6Q^2$ as in \eqref{cc-Wsln}. Under the AGT duality, the parameter is mapped on the gauge theory side as
\be 
Q=\sqrt{\beta}-\frac{1}{\sqrt{\beta}}=-i\left(\sqrt{\frac{\e_2}{\e_1}}+\sqrt{\frac{\e_1}{\e_2}}\right)~.
\ee
The vertex operator $V_\a=e^{2i\a \phi}$ is a primary field of conformal dimension $\Delta=\alpha(\alpha-Q)$, which correspond to a hypermultiplet in the gauge theory side. 

The conformal block is fully determined by the Virasoro algebra. By using a conformal transformation, the primary fields can be placed at $(0,1,\frakq,\infty)$ on a sphere coordinate $S^2=\bC\cup\{\infty\}$ with a cross-ratio $\frakq$. The four-point conformal block is then expressed as a sum of inner products
\begin{equation}\label{conf-block}
    F_\Delta\begin{bmatrix}\Delta_1\Delta_2\\ \Delta_3\Delta_4\end{bmatrix}=\sum_{N=0}^\infty \frakq^N{}_{43}^{\ \Delta}\langle N\ket{N}^\Delta_{21}
\end{equation}
where an $N$-th descendant state $\ket{N}^\Delta_{21}$ satisfies a recursive relation \cite{Belavin:1984vu}
\begin{equation}
    \sfL_k\ket{N}^\Delta_{21}=(\Delta+k\Delta_2-\Delta_1+N-k)\ket{N-k}^\Delta_{21},\ {\rm for}\ k>0,\label{conf-ket}
\end{equation}
and $\Delta_i$ for $i=1,2,3,4$ are the conformal dimensions of external primary fields, $\Delta$ is the conformal dimension of the primary field propagating along the internal line (See Figure \ref{f:4-pt}).
As the Coulomb branch parameter $\fraka$, masses of hypermultiplets $\frakm_{1,2,3,4}$ and $\Omega$-background parameters $\epsilon_{1,2}$ all have mass dimension 1 in the gauge theory side, one can define dimensionless parameters by using $\sqrt{\epsilon_1\epsilon_2}$, which are mapped to the Liouville side under the duality as
\begin{equation}
    -i\frac{\mathfrak{a}}{\sqrt{\epsilon_1\epsilon_2}}=\a-\frac{Q}{2} ~,\qquad  -i\frac{\mathfrak{m}_i}{\sqrt{\epsilon_1\epsilon_2}}=\a_i-\frac{Q}{2}~.
\end{equation}
With this parameter matching, \eqref{conf-block} can be identified with the instanton partition function of the SU(2) gauge theory with $N_f=4$ hypermultiplets \cite{Alday:2009aq}, which can be computed using the prescription outlined in the previous section.

\begin{figure}
\begin{center}
\includegraphics{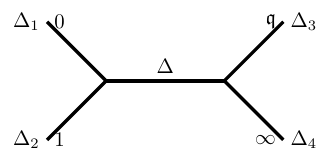}
\end{center}
\caption{Four-point function in 2d CFT.}
\label{f:4-pt}
\end{figure}

The 4d $\SU(N)$ pure Yang-Mills theory is constructed by a sphere with two irregular punctures \cite{Gaiotto:2009ma} in the sense of class $\cS$. The state corresponding to the irregular puncture in the 2d CFT is a coherent state of the $\widetilde\cW_N$-algebra, rather than a primary state. Let us construct this coherent state.
The mode expansions of the $j$-th current of $\widetilde\cW_N$-algebra
\begin{equation}
 \widetilde W^{(j)}(z)=\sum_{n\in \bZ} \frac{\widetilde\sfW^{(j)}_n}{z^{n+j}}~, \qquad (2\le j \le N)~.
\end{equation}
Note that $\widetilde\sfW^{(2)}_n=\sfL_n$ are the Virasoro generators. A primary state $| \boldsymbol{\Delta}\rangle$ of the $\widetilde\cW_N$-algebra is given by
\begin{equation}
  \widetilde\sfW^{(j)}_0| \boldsymbol{\Delta}\rangle = \Delta^{(j)} | \boldsymbol{\Delta}\rangle
\end{equation}
where $\boldsymbol{\Delta}:=(\Delta^{(2)},\ldots,\Delta^{(N)})$.  Let us define the Kac-Shapovalov matrix for the $\widetilde\cW_N$-algebra
\begin{equation}
Q_{\boldsymbol{\Delta}}(\boldsymbol\mu ;\boldsymbol \lambda):=\left\langle\boldsymbol{\Delta}\left|\widetilde\sfW^{(N)}_{\mu^{(N)}} \cdots \widetilde\sfW^{(2)}_{\mu^{(2)}}\cdot \widetilde\sfW^{(2)}_{-\lambda^{(2)}} \cdots \widetilde\sfW^{(N)}_{-\lambda^{(N)}} \right| \boldsymbol{\Delta}\right\rangle
\end{equation}
where $\boldsymbol\mu:=(\mu^{(2)},\ldots,\mu^{(N)})$, $\boldsymbol\lambda:=(\lambda^{(2)},\ldots,\lambda^{(N)})$ are collections of partitions.

We define a coherent state, a.k.a. Gaiotto-Whittaker state, of the $\cW_N$-algebra for the irregular puncture of the pure Yang-Mills
\begin{equation}\label{coh_W}
|\frakG\rangle := \sum_{\boldsymbol{\lambda}}  \frakq^{k/2} Q_{\boldsymbol{\Delta}}^{-1}(\emptyset,\ldots,\emptyset,[1^k];\boldsymbol \lambda) \widetilde{\boldsymbol{\sfW}}_{-\boldsymbol \lambda} \left| \boldsymbol{\Delta}\right\rangle
\end{equation}
where $\boldsymbol\mu:=(\mu^{(2)},\ldots,\mu^{(N)})=(\emptyset,\ldots,\emptyset,[1^k])$, and
\be
\widetilde{\boldsymbol{\sfW}}_{-\boldsymbol \lambda} \left| \boldsymbol{\Delta}\right\rangle:=\widetilde\sfW_{-\lambda^{(2)}}^{(2)} \cdots \widetilde\sfW_{-\lambda^{(N)}}^{(N)}  \left| \boldsymbol{\Delta}\right\rangle~.
\ee
Then, $\widetilde\sfW^{(N)}_{1}$ acts on the coherent state as the non-trivial constant, and higher modes annihilate the coherent state
\bea
\widetilde\sfW^{(N)}_{1}|\frakG\rangle  =&\frakq^{\frac12}|\frakG\rangle ~,\cr
\widetilde\sfW^{(N)}_{i}|\frakG\rangle  =&0 \quad \textrm{for} \ \ i>1~.
\eea

The instanton partition function of the $\SU(N)$ $\cN=2$ pure  Yang-Mills  is given by
\be
\cZ_{\mathrm{inst}}=\langle \frakG|\frakG\rangle~.\label{eq:partitionGaiottoSU2}
\ee
with a suitable change of variables.  This identity is indeed mathematically proven in \cite{Schiffmann:2012aa,Maulik:2012wi,Braverman:2014xca,Borot:2021btb}. The Seiberg-Witten curve  is given by
\begin{equation}\label{SW}
x^{N}+\sum_{k=0}^{N-2}\phi_{N-k}(z) x^{k}=0~
\end{equation}
where
\begin{equation}
  \phi_{N-k}(z) =\lim_{\e_1,\e_2\to 0}\langle \frakG| \widetilde W^{(N-k)}(z)|\frakG\rangle~.
\end{equation}
The quantum Miura transformation for $\widetilde\cW_N$-algebra is given in \eqref{Miura_tr}, and the Seiberg-Witten curve \eqref{SW} can be written as
\begin{equation}
\lim_{\e_1,\e_2\to 0}\frac{\langle \frakG| \prod_{i=1}^{N}\left(x-\boldsymbol{\nu}_{i} \cdot \partial \boldsymbol{\phi}(z)\right)|\frakG\rangle}{\langle \frakG|\frakG\rangle}=0~.
\end{equation}
Therefore, the quantum Miura transformation can be interpreted as the quantization of the Hitchin system \eqref{eq:spectral} where the free boson quantizes the Higgs field.

\subsubsection*{Coherent state by affine Yangian}
In \S\ref{sec:dDAHA}, we established that affine Yangian $\AY$ is isomorphic to the central extension $\SdH^{\boldsymbol{c}}$ of the large $N$ limit of the spherical degenerate DAHA. Also, we have seen that this algebra ($\AY$ or $\SdH^{\boldsymbol{c}}$) encapsulates the $\cW_N$-algebras as a subalgebra. As a direct consequence, the coherent state in (\ref{coh_W}) can be derived from the framework of the affine Yangian.

We note that the coherent state (\ref{coh_W}) in terms of $\widetilde\sfW$-modes is written in terms of the inverses of the Kac-Shapovalov matrices.
While it can be determined order by order, determining these matrices explicitly remains challenging.
One of the primary advantages of leveraging the affine Yangian (or the quantum toroidal algebra) over the ($q$-)$\cW$-algebra lies in the presence of an orthogonal basis for the Fock representation, which is a suitable change of normalization from \eqref{deg-Fock}. This basis facilitates the derivation of the coherent state in a more succinct closed form \cite{Yanagida:2011,Schiffmann:2012aa}:
\begin{equation}\label{eq:Gaiotto_state}
|\frakG,\boldsymbol{\fraka}\rangle^{(N)}=\sum_{\boldsymbol\lambda}\sqrt{\mathcal{Z}_{\text{vect.}}(\boldsymbol{\fraka},\boldsymbol{\lambda})}\frakq^{\frac{|\boldsymbol{\lambda}|}{2}}|\boldsymbol{\fraka}, \boldsymbol{\lambda}\rrangle\,
\end{equation}
where we use the degenerate representation of \eqref{QTA-Fock2-E} and \eqref{QTA-Fock2-F}
\begin{align}
&\sfD_{1,n}|\boldsymbol{\fraka}, \boldsymbol{\lambda}\rrangle=\sum_{x\in \frakA(\boldsymbol{\lambda})}\phi_x^n\Psi^{\frac{1}{2}}_x(\boldsymbol{\lambda},\boldsymbol{a})|\boldsymbol{\fraka},\boldsymbol{\lambda}+x\rrangle~,\cr
&\sfD_{-1,n}|\boldsymbol{\fraka}, \boldsymbol{\lambda}\rrangle=\sum_{x\in \frakR(\boldsymbol{\lambda})}\phi_x^n\Psi^{\frac{1}{2}}_x(\boldsymbol{\lambda},\boldsymbol{a})|\boldsymbol{\fraka},\boldsymbol{\lambda}-x\rrangle~,\cr
&\sfD_{0,n}|\boldsymbol{\fraka}, \boldsymbol{\lambda}\rrangle=\sum_{x\in\boldsymbol{\lambda}}\phi^n_x|\boldsymbol{\fraka}, \boldsymbol{\lambda}\rrangle~,
\end{align}
with $\{\phi_x=\mathfrak{a}_\alpha+(i-1)\epsilon_1+(j-1)\epsilon_2\}_{(i,j)\in\boldsymbol{\lambda}},$ and
\begin{equation}
    \Psi_{\boldsymbol{\lambda}}(z,\boldsymbol{a})=\prod_{x\in \frakA(\boldsymbol{\lambda})}\frac{z-\phi_x-\epsilon_3}{z-\phi_x}\prod_{x\in \frakR(\boldsymbol{\lambda})}\frac{z-\phi_x+\epsilon_3}{z-\phi_x},\quad \Psi_x(\boldsymbol{\lambda},\boldsymbol{a})=\underset{z\rightarrow \phi_x}{\Res}\Psi_{\boldsymbol{\lambda}}(z,\boldsymbol{a}).
\end{equation}
For the $N$-th tensor product representation, one can show
\begin{equation}
\sfD_{-1,\ell} |\frakG, \boldsymbol{\fraka}\rrangle^{(N)} = \gamma_\ell |\frakG,\boldsymbol{\fraka}\rrangle^{(N)}\,
\end{equation}
where 
\begin{equation}
    \gamma_\ell =\left\{
    \begin{array}{ll}
    0     & \ell<N-1 \\
    \frac{1}{\sqrt{-\epsilon_1\epsilon_2}} \frakq^{\frac12}  & \ell=N-1\\
    \frac{1}{\sqrt{-\epsilon_1\epsilon_2}}\frakq^{\frac12}\sum_{i=1}^N a_i  & \ell=N
    \end{array}
    \right.
\end{equation}
by using the 4d version of the recursive relation \eqref{Nekra-rec2}. 

After rewriting in terms of $\mathcal{W}$-generators (see \S\ref{sec:SHc-AY-W} for detail), one obtains
\begin{align}
    \widetilde{\sfW}^{(d)}_1 |\frakG,\boldsymbol{\fraka}\rrangle =& \lambda_1^{(d)}|\frakG,\boldsymbol{\fraka}\rrangle
\end{align}
with
\begin{equation}
    \lambda_1^{(d)}=\left\{
    \begin{array}{ll}
    0 \quad & d<N\\
    (-1)^N\frac{\beta^{-N/2}}{\sqrt{-\epsilon_1\epsilon_2}}\frakq^{\frac12} \quad & d=N\,.
    \end{array}
    \right.
\end{equation}
which agrees with the coherent state condition for $\mathcal{W}$-algebra.

\subsection{Topological vertex and intertwiners of \texorpdfstring{$U_{q_{1},q_{2},q_{3}}(\ddot{\frakgl}_1)$}{quantum toroidal gl1} }\label{sec:TV-intertwiners}

The topological string theory is a string theory based on 2d $\cN=(2,2)$ sigma model with topological twist as its worldsheet description. When we have a U(1) R-symmetry, it is possible to shift the spin current $J$ with the U(1) current
\begin{equation}
    \tilde{J}=J-F_{\U(1)}.
\end{equation}
This shift changes the definition of the spin, and splits the supercharges into scalar ones and vector charges. In 2d $\cN=(2,2)$ theory, we have two copies of U(1) R-symmetries, usually denoted as the vector and the axial U(1) currents, $F_V$ and $F_A$, so there are two ways to perform the topological twist: the A-twist uses the vector current
\begin{equation}
    \tilde{J}_A=J-F_V
\end{equation}
and the B-twist is realized with the axial current
\begin{equation}
    \tilde{J}_B=J+F_A.
\end{equation}
The topological string theory was originally studied as a toy model of the superstring theory \cite{Witten:1991zz,DVV,Witten-CS,BCOV}. However, subsequent research revealed a profound connection between this theory and 4d and 5d supersymmetric theories featuring eight supercharges. Intriguingly, topological string theory also facilitates the computation of non-perturbative effects in supersymmetric theories.

$$
\begin{tikzpicture}[background rectangle/.style=
{draw=red!40,fill=red!10,rounded corners=1ex},
   show background rectangle]
\node[rounded rectangle,fill=blue!15] at (0,0) {5d $\mathcal{N}=1$ theory};
\node[rounded rectangle,fill=blue!15] at (-5,2.5) {M-theory on non-compact toric CY3 };
\node[rounded rectangle,fill=blue!15] at (4,2.5) {Fivebrane web in Type IIB};
\draw[<->] (-1.2,2.5) to node [above=0.1] {Duality} (1.2,2.5);
\draw[<-] (1.5,.5) to node [right=0.25] {worldvolume theory} (4,2);
\draw[<-] (-1.5,.5) to node [left=0.25]  {compactification} (-4,2);
\node[red,scale=1.25] at (0,1.3) {$\QTA$};
\node[red,scale=.8] at (-7,0) {on $\Omega$-background $\bC^2_{\e_1,\e_2}$};
\end{tikzpicture}$$

Let us briefly present how such an interesting connection is established. It was argued in \cite{Gopakumar:1998ii,Gopakumar:1998jq} that the A-model topological string on Calabi-Yau $X$ captures the BPS spectrum of M-theory on the same Calabi-Yau $X$. When $X$ is a toric Calabi-Yau, M-theory on $X$ can be T-dualized to type IIB string theory on a Taub-NUT space, as described in \cite{Leung-Vafa}.
Furthermore, this theory is dual to a web of $(p,q)$ 5-branes compactified on $S^1$, as introduced in \cite{AHK}. The configuration of these branes is described in Table \ref{t:brane-web}. In the low-energy limit, the theory of the 5-brane web reduces to a 5d $\cN=1$ theory on $S^1$, providing a duality between M-theory on a toric Calabi-Yau and a 5d $\cN=1$ theory. This duality is called the \emph{geometric engineering} \cite{Katz:1996fh,Katz:1997eq}. By taking the limit where the radius of $S^1$ approaches zero, one can obtain 4d $\cN=2$ theories.

In what follows, we will delve into the concept of a non-compact toric Calabi-Yau three-fold and examine the powerful computational tool for topological string partition functions, known as the topological vertex \cite{AKMV,Awata:2005fa,IKV}. Remarkably, the BPS sectors on general $\Omega$-backgrounds in this physical system are governed by $\QTA$. In this section, we will demonstrate, as first uncovered in \cite{Awata:2005fa}, that toric diagrams, fivebrane webs and the topological vertex admit an algebraic interpretation as $\QTA$ intertwiners in \S\ref{sec:intertwiner}.

\begin{table}[ht]
\begin{center}
\begin{tabular}{|c|c|c|c|c|c|c|c|c|c|c|}
\hline
& 0 & 1 & 2 & 3 & 4& 5 & 6 & 7 & 8 & 9 \\
\hline
D5 & $\bullet$ & $\bullet$ & $\bullet$ & $\bullet$ & $\bullet$ & $\bullet$ & $-$ & $-$ & $-$ & $-$ \\
\hline
NS5 & $\bullet$ & $\bullet$ & $\bullet$ & $\bullet$ & $\bullet$ & $-$ & $\bullet$ & $-$ & $-$ & $-$ \\
\hline
7-brane & $\bullet$ & $\bullet$ & $\bullet$ & $\bullet$ & $\bullet$ & $-$ & $-$ & $\bullet$ & $\bullet$ & $\bullet$ \\
\hline
\end{tabular}
\caption{Configuration of branes in the brane web construction. The bullet $\bullet$ represents the direction the D-brane is extending while the $-$ represents the direction where the D-brane is point like. }
\label{t:brane-web}
\end{center}
\end{table}

\subsubsection{Basics of toric Calabi-Yau}

A $d$-dimensional toric variety is specified by a set of defining equations on a $(d+r)$-dim complex manifold with coordinates $x_i$ carrying $r$ types of U(1) charges, $Q^a_i$ for $a=1,\dots,r$. Under these U(1) transformations with transforming parameters $\alpha_a$, $x_i\rightarrow e^{i\sum_a Q^a_i\alpha_a}x_i$. The toric variety is given by the defining equations under the identification under this U(1) transformation
\begin{equation}
    \lt\{\sum_{i=1}^{d+r}Q^a_i|x_i|^2=t^a\rt\}/\U(1)^r
\end{equation}
where $t^a$ are K\"ahler parameters. The Calabi-Yau condition is known to be equivalent to
\begin{equation}
    \sum_i Q^a_i=0.\label{CY-cond}
\end{equation}

When we consider (complex) three-dimensional toric Calabi-Yau, it is always possible to choose a set of local coordinates $(z_1,z_2,z_3)$ transforming as
\begin{equation}
    (z_1,z_2,z_3)\rightarrow (e^{i\alpha}z_1,e^{i\beta}z_2,e^{-i\alpha-i\beta}z_3),\label{2U(1)}
\end{equation}
under two U(1) actions. The geometry always looks like $T^2\times\mathbb{R}$ fibered over $\mathbb{R}^3$ locally, and one can draw a projected diagram to mark all the loci where the $T^2\simeq \U(1)\times\U(1)$ degenerates. This diagram is called the dual toric diagram, and it can be drawn from the toric data $\{Q^a_i\}$ as follows. We first define $d+r$ ($d$-dim) vectors satisfying
\begin{equation}
    \sum_iQ^a_i\boldsymbol{v}_i=0.
\end{equation}
The Calabi-Yau condition \eqref{CY-cond} implies that we can choose the vectors $\{\boldsymbol{v}_i\}$ to lie on a single (2-dimensional) plane. These vectors span a diagram called the toric diagram. For example for $d=3$, $r=2$, $(Q^1_i)=(-2,1,0,1,0)$ and $(Q^2_i)=(-2,0,1,0,1)$, we may take $\boldsymbol{v}_1=(0,0,1)$, $\boldsymbol{v}_2=(1,0,1)$, $\boldsymbol{v}_3=(0,1,1)$, $\boldsymbol{v}_4=(-1,0,1)$ and $\boldsymbol{v}_5=(0,-1,1)$. One can represent such a set of vectors as dots on the plane $z=1$, i.e. expressing $\boldsymbol{v}=(\boldsymbol{w},1)$, the points $\boldsymbol{w}$ and lines connecting them give the corresponding toric diagram of the geometry under consideration. In the current example, the toric diagram is given by 
$$\includegraphics{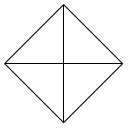}$$
Further connecting the perpendicular bisectors of each internal line gives the dual toric diagram (see Figure \ref{fig:toric-dia}). In the duality described in \S\ref{sec:TV-intertwiners}, the toric diagram is mapped to the web diagram in $(p,q)$-brane web construction, and there each internal line parallel to the vector $(p,q)$ has a physical meaning of a $(p,q)$ 5-brane in type IIB brane construction. We can see that the example we just saw corresponds to the brane web construction of 5d $\cN=1$ pure SU(2) gauge theory, and it will be an important example in this article.

\begin{figure}
    \centering
\includegraphics{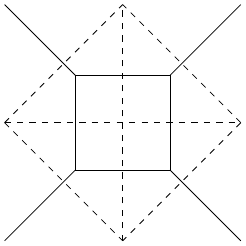}
    \caption{Dual toric diagram (solid line) obtained from the toric diagram (dashed line).}
    \label{fig:toric-dia}
\end{figure}

\subsubsection{Topological vertex}

The partition function of the A-model topological string on toric Calabi-Yau can be computed with the gluing prescription \cite{AKMV}. We assign a trivalent vertex to each vertex in the dual toric diagram, called the topological vertex, with each of its three legs labeled by a Young diagram, $C_{\mu\nu\lambda}$. To each internal line in the dual toric diagram, we assign a propagator
\begin{equation}
    \sum_\lambda(-Q)^{|\lambda|}
\end{equation}
to glue the vertices on the two ends of the line, where $\lambda$ runs over all possible configurations of Young diagrams and $Q$ is the exponentiated K\"ahler parameter associated.

The refined topological vertex \cite{IKV} is expressed by Macdonald $P_\lambda$ and skew-Schur $s_{\mu/\nu}$ functions as
\begin{equation}
    C_{\mu\nu\lambda}(t,q)=q^{\frac{\|\mu^t\|^2}{2}}t^{-\frac{\|\mu\|^2}{2}}P_\lambda(t^{-\rho};q,t)\sum_\eta \lt(\frac{q}{t}\rt)^{\frac{|\eta|+|\mu|-|\nu|}{2}}s_{\mu^t/\eta}(q^{-\lambda}t^{-\rho})s_{\nu/\eta}(t^{-\lambda^t}q^{-\rho})~,\label{exp-top-v}
\end{equation}
where
\begin{equation}
    P_\lambda(t^{-\rho};q,t)=t^{\frac{\|\lambda^t\|^2}{2}}\prod_{(i,j)\in\lambda}\frac{1}{1-q^{a(i,j)}t^{\ell(i,j)+1}}~.
\end{equation}
We refer the reader to Appendix \ref{app:symmetric-functions}
 for the definitions and properties of Macdonald and skew-Schur functions.
In our convention, the preferred direction (corresponding to the third leg associated with the Young diagram $\lambda$) is taken to be the horizontal direction and the other two legs are aligned in a clockwise manner:
$$\includegraphics{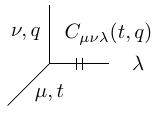}$$
An intuitive explanation of how the above expression is obtained (in the unrefined limit $q=t$) can be found in for example \cite{Marino-book}, and the refined version is defined in the melting crystal picture in \cite{IKV} or by replacing the Schur functions with the corresponding Macdonald functions \cite{Awata:2005fa}. There are two types of inequivalent topological vertices, namely $C_{\mu\nu\lambda}(t,q)$ and $C_{\mu\nu\lambda}(q,t)$. To distinguish them, we shall also assign the labels $t$ and $q$ to the non-preferred legs. The rules for such assignment are given as follows: 1) only $t$ leg can be glued with a $q$ leg, and vice versa; 2) when gluing preferred legs, $t$ legs (resp. $q$ legs) of two vertices should be aligned on the same side of the preferred lines.

Another important ingredient in the gluing prescription is the assignment of the framing factor. When we glue two vertices in the following manner
$$\includegraphics{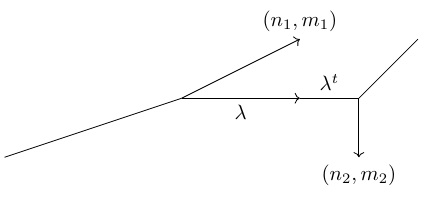}$$
the internal line associated to the Young diagram $\lambda$ is sandwiched by the left above line pointing to the direction $(n_1,m_1)$ and the right bottom line pointing to $(n_2,m_2)$, and we shall assign a framing factor \cite{Taki07}
\begin{equation}
    f^{n_1m_2-n_2m_1}_\lambda
\end{equation}
where depending on the $q$-, $t$- or preferred nature of the leg labeled by $\lambda$, the framing factor is given by
\begin{align}
    f_\lambda=\lt\{\begin{array}{cc}
    f_{\lambda}(q,t)=(-1)^{|\lambda|}q^{n(\lambda)}t^{-n(\lambda^t)}& t\textrm{-leg},\\
    f_{\lambda}(t,q)=(-1)^{|\lambda|}t^{n(\lambda)}q^{-n(\lambda^t)}& q\textrm{-leg},\\
    \tilde{f}_\lambda(q,t)=(-1)^{|\lambda|}q^{-\frac{\|\lambda^t\|^2}{2}}t^{\frac{\|\lambda\|^2}{2}} & \textrm{preferred-direction}.\\
    \end{array}\rt.
    \label{def-framing}
\end{align}

Let us give an example involving a non-trivial framing factor:
\begin{equation}
\includegraphics{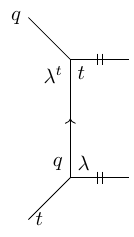}\label{U1CS1}
\end{equation}
All external lines are set to $\emptyset$. The topological string partition function is given by
\begin{align}
    \cZ=&\sum_\lambda (-Q)^{|\lambda|}f_\lambda(t,q) C_{\lambda\emptyset\emptyset}(t,q)C_{\emptyset\lambda^t\emptyset}(t,q)\cr
    =&\sum_\lambda (-Q)^{|\lambda|}q^{\frac{\|\lambda^t\|^2}{2}}t^{-\frac{\|\lambda\|^2}{2}}s_{\lambda^t}(t^{-\rho})s_{\lambda^t}(q^{-\rho})f_\lambda(t,q)\cr
    =&\sum_\lambda (-Q\sqrt{q/t})^{|\lambda|}s_{\lambda^t}(t^{-\rho})s_{\lambda^t}(q^{-\rho})\cr
    =&\prod_{i,j=1}^\infty\frac{1}{1-Qt^{i-1}q^j},\label{Z-hori}
\end{align}
where we use the Cauchy identity \eqref{Schur-nor-id} to perform the summation over all Young diagrams.

A very important fact (more precisely observation) is that when the preferred direction can be consistently chosen, meaning every vertex has one edge with a preferred direction, the final result of the partition function remains unchanged regardless of the chosen direction. This can be demonstrated through the example in \eqref{U1CS1}, where we chose the preferred direction to be horizontal, but it could also have been chosen vertically or by rotating it $90^\circ$ and still maintaining the horizontal direction as the preferred one.
$$\includegraphics{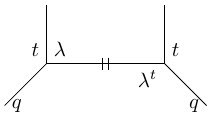}$$
The partition function can be computed as
\begin{align}
    \cZ=&\sum_\lambda(-Q)^{|\lambda|}\tilde{f}_\lambda(q,t)C_{\emptyset\emptyset\lambda}(q,t)C_{\emptyset\emptyset\lambda^t}(t,q)\cr
    =&\sum_\lambda (-Q)^{|\lambda|}\tilde{f}_\lambda(q,t)P_\lambda(q^{-\rho};t,q)P_{\lambda^t}(t^{-\rho};q,t).\label{Z-vert}
\end{align}
One can compare the expansion of \eqref{Z-hori} and \eqref{Z-vert} to confirm that they agree with each other. Furthermore, by using the identities \eqref{n-l} and \eqref{n-a}, we obtain
\begin{equation}
    P_{\lambda^t}(t^{-\rho};q,t)P_{\lambda}(q^{-\rho};t,q)=(-1)^{|\lambda|}(q/t)^{\frac{|\lambda|}{2}}N^{-1}_{\lambda\lambda}(1;q,t^{-1}).\label{P-Nekra}
\end{equation}
Rewriting the partition function $\cZ$ with this identity, we have
\begin{align}
    \cZ=\sum_\lambda (-Q)^{|\lambda|}N^{-1}_{\lambda\lambda}(1;q,t^{-1})\prod_{(i,j)\in\lambda}q_1^{-i}q_2^{-j}
\end{align}
and it can be identified as the instanton partition function of U(1) gauge theory with Chern-Simons level $-1$.

In general, the refinement of topological vertex was done in such a way \cite{Awata:2005fa} that the partition function of refined topological string (on toric Calabi-Yau) is equal to the partition function of the dual 5d gauge theory on an $\Omega$-background with parameters $q_1=q$ and $q_2=t^{-1}$, which has been computed in \cite{Nekrasov:2002qd}. This refines the duality described at the beginning of \S\ref{sec:TV-intertwiners} to gauge theories on general $\Omega$-background.

\paragraph{Remark:} In fact $P_\lambda(t^{-\rho};q,t)$ is a Macdonald function with the principal specialization. Therefore the dual Macdonald function $Q_\lambda$ in \eqref{dual-Macdonald} appears in the computation via
\begin{equation}
    Q_\lambda(t^{-\rho};q,t)=t^{\frac{\|\lambda^t\|^2}{2}}q^{-\frac{\|\lambda\|^2}{2}}P_{\lambda^t}(q^{-\rho};t,q)=(-1)^{|\lambda|}\tilde{f}_\lambda(t,q)P_{\lambda^t}(q^{-\rho};t,q)
\end{equation}
and from the Cauchy formula \eqref{Cauchy-1}
the partition function \eqref{Z-vert} can be computed to
\begin{align}
    \cZ=\sum_\lambda Q^{|\lambda|}P_\lambda(q^{-\rho};t,q)Q_\lambda(q^{-\rho};t,q)=\prod_{i,j}\frac{1}{1-Qt^{i-1}q^j}
\end{align}
which matches with \eqref{Z-hori} exactly.

\paragraph{Remark:} Not all dual toric diagram admits a consistent choice of the preferred direction, for example, local $\mathbb{P}^2$ (see Figure \ref{f:local-P2}) might be the most famous example, so one cannot evaluate its partition function with the refined topological vertex or has to introduce a new vertex \cite{Iqbal:2012mt}. In this article, we only focus on dual toric diagrams in which all vertices have a horizontal leg so that we can take the preferred direction horizontally.

\begin{figure}
    \centering
\includegraphics{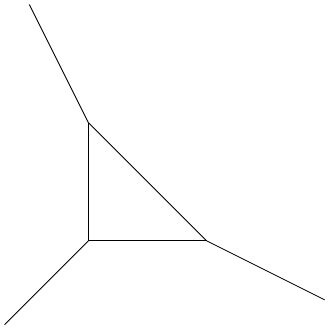}
    \caption{Dual toric diagram of local $\mathbb{P}^2$.}
    \label{f:local-P2}
\end{figure}

\subsubsection{Rewriting topological vertex by intertwiners}

As remarked before, we only consider dual toric diagrams with horizontal legs attached to each vertex. This means that there are only two cases one needs to consider in the vertical gluing. The first one involves two vertices whose preferred legs point to the same direction:
\begin{equation}
\includegraphics{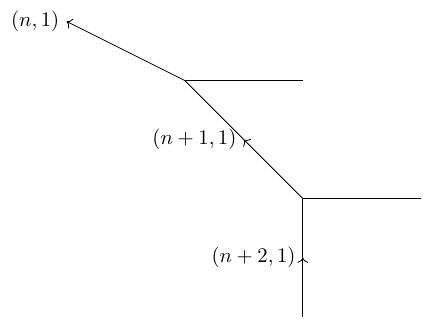}\label{v-glue-1}
\end{equation}
Such a gluing involves a framing factor $f_\lambda^{\pm 1}$ defined in \eqref{def-framing}, and as in \eqref{Z-hori}, the summation over all Young diagrams can always be done exactly with the Cauchy identity of Schur functions (that is the framing factor multiplied cancels with the prefactors in the topological vertex, so there is no net framing factor remaining at the end of the computation). Another case has two vertices pointing in the opposite direction
\begin{equation}
\includegraphics{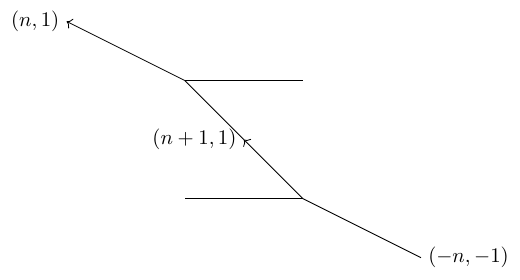}\label{v-glue-2}
\end{equation}
In this gluing, no framing factor is needed, but the computation can also be done with the Cauchy identity. It is generally true that for the theories we are considering, all the vertical gluing (i.e. gluing not in the preferred direction) can be performed exactly. This fact is related to the vertex operator form of the topological vertex, and as shown in \cite{Awata:2011ce} such vertex operators also play an important role in the context of the quantum toroidal algebra as the intertwiner. In the remaining of this section, we review the core part of \cite{Awata:2011ce} and rewrite the topological vertex into the intertwiners, which have already been defined in this article in \eqref{inttw-def} and \eqref{dual-int-def}.

In the expression \eqref{exp-top-v}, the prefactor $q^{\frac{\|\mu^t\|^2}{2}}t^{-\frac{\|\mu\|^2}{2}}$ cancels with the framing factor or with the same prefactor with $q\leftrightarrow t$ and $\lambda$ replaced by $\lambda^t$ appearing in the gluing, two Macdonald functions $P_\lambda$ combine into the Nekrasov factor as shown in \eqref{P-Nekra}, the only remaining part we need to take care of is
\begin{align}
    \widetilde{C}_{\mu\nu\lambda}(q,t):=&\sum_\eta \lt(\frac{t}{q}\rt)^{\frac{|\eta|+|\mu|-|\nu|}{2}}s_{\mu^t/\eta}(t^{-\lambda}q^{-\rho})s_{\nu/\eta}(q^{-\lambda^t}t^{-\rho})\cr
    =&t^{\frac{|\mu|}{2}}q^{-\frac{|\nu|}{2}}\sum_\eta s_{\mu^t/\eta}(t^{-\lambda}q^{-\rho+1/2})s_{\nu/\eta}(q^{-\lambda^t}t^{-\rho-1/2}),\label{rewriting-1}
\end{align}
where we use
\begin{equation}
   s_{\mu/\nu}(Qx)=s_{\mu/\nu}(x)Q^{|\mu|-|\nu|}.
\end{equation}
As the $t$-leg is always glued with a $q$-leg, the prefactor $t^{\frac{|\mu|}{2}}q^{-\frac{|\nu|}{2}}$ in \eqref{rewriting-1} will just cancel out or shift the K\"ahler parameter by $(t/q)^{\frac{1}{2}}$ (depending on the gluing of type \eqref{v-glue-1} or \eqref{v-glue-2}). We further note that in the vertical gluing, the K\"ahler parameter is assigned as
\begin{equation}
    \sum_{\mu,\nu,\dots}\dots \widetilde{C}_{\dots\mu\dots}(-Q_1)^{|\mu|}\widetilde{C}_{\mu\nu\lambda}(q,t)(-Q_2)^{|\nu|}(-Q_2)^{|\nu|}\widetilde{C}_{\dots\nu\dots}
\end{equation}
and we can split each K\"ahler parameter into two\footnote{In the context of brane web construction, it is easy to see that $u$ represents the exponentiated position of a D5-brane (in the 6-th direction), i.e. $u=e^{-x_6}$. See Table \ref{t:brane-web}.}, e.g. $Q_1:=u/u_1$ and $Q_2:=u_2/u$, and each of them can be absorbed into one $\widetilde{C}$. More precisely, we are interested in the combination
\begin{equation}
    \widetilde{T}_{\mu\nu\lambda}(q,t):=u^{|\mu|+|\nu|}\sum_\eta s_{\mu^t/\eta}(t^{-\lambda}q^{-\rho+1/2})s_{\nu/\eta}(q^{-\lambda^t}t^{-\rho-1/2}).
\end{equation}
Writing the skew Schur functions by the vertex operator as in  \eqref{skew-Schur-fermion}, we obtain the vertex-operator expression for the topological vertex
\begin{align}
    \widetilde{T}_{\mu\nu\lambda}(q,t)=&\sum_\eta \bra{\mu^t}V_-(ut^{-\lambda}q^{-\rho+1/2})\ket{\eta}\bra{\eta}V_+(u^{-1}q^{-\lambda^t}t^{-\rho-1/2})\ket{\nu}\cr
    =&\bra{\mu^t}V_-(ut^{-\lambda}q^{-\rho+1/2})V_+(u^{-1}q^{-\lambda^t}t^{-\rho-1/2})\ket{\nu}
\end{align}
where we use the completeness of the fermionic basis, $\mathbb{1}=\sum_\eta\ket{\eta}\bra{\eta}$.
(See Appendix \ref{app:Schur} for more details.)

Let us further perform the summation in the above vertex operator, firstly for $\lambda=\emptyset$
\begin{align}
    V_-(uq^{-\rho+1/2})V_+(u^{-1}t^{-\rho-1/2})=\exp\lt(\sum_{n=1}^\infty \frac{1}{n}\frac{q^n}{1-q^n}\sfJ_{-n}u^n\rt)\exp\lt(\sum_{n=1}^\infty \frac{1}{n}\frac{1}{1-t^n}\sfJ_{n}u^{-n}\rt).\label{vert-vac}
\end{align}
Recall from \eqref{Phi-int-def} that the intertwiner at $\lambda=\emptyset$ is given by
\begin{align}
    \Phi_{\emptyset}[u]=&\exp\left(-\sum_{r=1}^{\infty}\frac{u^{r}}{r}\frac{\sfa_{-r}}{1-q_{3}^{-r}}\right)\exp\left(\sum_{r=1}^{\infty}\frac{u^{-r}}{r}\frac{\gamma^{-r}}{1-q_{3}^{r}}\sfa_{r}\right),\label{Phi-vac}\\
    [\sfa_{r},\sfa_{s}]=&-\frac{r}{\kappa_{r}}(q_{3}^{r/2}-q_{3}^{-r/2})^{3}\delta_{r+s,0}
\end{align}
where $\kappa_r=(q_1^{\frac{r}{2}}-q_1^{-\frac{r}{2}})(q_2^{\frac{r}{2}}-q_2^{-\frac{r}{2}})(q_3^{\frac{r}{2}}-q_3^{-\frac{r}{2}})=-(1-q_1^r)(1-q_2^r)(1-q_3^r)$. One can identify
\begin{align}
\begin{split}
    &\sfa_r=\gamma^{r}\frac{1-q_3^r}{1-q_2^{-r}}\sfJ_r,\quad r>0,\\
    &\sfa_{-s}=-\frac{q_1^{s}(1-q_3^{-s})}{1-q_1^{s}}\sfJ_{-s},\quad s>0
\end{split}\label{a-dict}
\end{align}
for $q_1=q$, $q_2=t^{-1}$, $\gamma=q_3^{\frac{1}{2}}$, and then it is easy to see that \eqref{vert-vac} matches with the intertwiner \eqref{Phi-vac}.

Adding a box $(i=\lambda^t_j+1,j=\lambda_i+1)$ to $\lambda$ multiplies a vertex operator
\begin{align}
    &\exp\lt(\sum_{n=1}^\infty\frac{1}{n}t^{-n\lambda_i}(t^{-n}-1)q^{in}\sfJ_{-n}u^n\rt)\exp\lt(\sum_{n=1}^\infty\frac{1}{n}q^{-n\lambda^t_j}(q^{-n}-1)t^{n(j-1)}\sfJ_nu^{-n}\rt)\cr
    =&\exp\lt(\sum_{n=1}^\infty\frac{1}{n}(1-t^n)\sfJ_{-n}(q/t)^{n}\chi_x^{n}\rt)\exp\lt(-\sum_{n=1}^\infty\frac{1}{n}(1-q^{-n})\sfJ_n\chi_x^{-n}\rt),\label{add-V}
\end{align}
to the vertex operator \eqref{vert-vac} or equivalently  $\Phi_\emptyset[u]$, where for $(i,j)\in\lambda$
\begin{equation}
    \chi_x=ut^{-j+1}q^{i-1}
\end{equation}
In the intertwiner \eqref{inttw-def}, the corresponding vertex operator is
\begin{align}
    \eta(z)=\eta_{3}(z)=&\exp\left(\sum_{r>0}\frac{\kappa_{r}}{r}\frac{\sfa_{-r}}{(q_{3}^{r/2}-q_{3}^{-r/2})^{2}}z^{r}\right)\exp\left(\sum_{r>0}\frac{\kappa_{r}}{r}\frac{q_{3}^{-r/2}\sfa_{r}}{(q_{3}^{r/2}-q_{3}^{-r/2})^{2}}z^{-r}\right)\cr
    =&\exp\lt(\sum_{r>0}\frac{1-q_2^{-r}}{r}\sfJ_{-r}q_3^{-r}z^r\rt)\exp\lt(-\sum_{r>0}\frac{1-q_1^{-r}}{r}\sfJ_rz^{-r}\rt).\label{eta-J}
\end{align}
We see that for $q_1=q$ and $q_2=t^{-1}$, \eqref{eta-J} matches with \eqref{add-V}, i.e. by denoting
\begin{equation}
    \widetilde{T}_{\mu\nu\lambda}(q,t)=\bra{\mu^t}V_\lambda(q,t)\ket{\nu}
\end{equation}
we have
\begin{equation}
    V_\lambda(q,t)=:\Phi_{\emptyset}[u]\prod_{x\in\lambda}\eta(\chi_x):.
\end{equation}

Exchanging $q$ and $t$, let us instead define
\begin{align}
    \widetilde{T}_{\mu\nu\lambda}(t,q)=&:\bra{\nu}U_\lambda(t,q)\ket{\mu^t}\cr
    =&\bra{\nu}V_-(ut^{-\lambda}q^{-\rho-1/2})V_+(u^{-1}q^{-\lambda^t}t^{-\rho+1/2})\ket{\mu^t}
\end{align}
then\footnote{Note that the preferred directions in $\widetilde{T}_{\mu\nu\lambda}(q,t)$ and $\widetilde{T}_{\mu\nu\lambda}(t,q)$ are in the opposite direction while the Young diagrams $\mu$, $\nu$ and $\lambda$ are always assigned in the clockwise manner. That is why we need to modify the definition of the vertex operator $U$ when exchanging $q$ and $t$ to keep the physical meaning of $u$ as the exponentiated position of D5 branes. } for $\lambda=\emptyset$
\begin{equation}
    U_\emptyset(t,q)=\exp\lt(\sum_{n=1}^\infty \frac{1}{n}\frac{1}{1-q^n}\sfJ_{-n}u^{n}\rt)\exp\lt(\sum_{n=1}^\infty \frac{1}{n}\frac{t^n}{1-t^n}\sfJ_{n}u^{-n}\rt).
\end{equation}
In the gluing of $C_{\mu\nu\lambda}(t,q)$ and $C_{\sigma\mu\tau}(q,t)$, we cannot ignore the prefactor $t^{\frac{|\mu|}{2}}q^{-\frac{|\nu|}{2}}$ presented in \eqref{rewriting-1} anymore, and we need to absorb it into $U_\lambda(t,q)$. For the first leg (labeled by $\mu$ here), we need to absorb the factor $(tq)^{\frac{|\mu|}{2}}$, and for the second leg (labeled by $\nu$), a factor $(tq)^{-\frac{|\nu|}{2}}$ should be absorbed. We also note that another prefactor $q^{\frac{\|\mu^t\|^2}{2}}t^{-\frac{\|\mu\|^2}{2}}$ in \eqref{exp-top-v} ignored before can be converted to a framing factor by multiplying $(-1)^{|\mu|}$ and we shall further modify the vertex operator accordingly. It gives rise to the vertex operator
\begin{equation}
    \widetilde{U}_\emptyset(t,q)=\exp\lt(-\sum_{n=1}^\infty \frac{1}{n}\frac{(tq)^{\frac{n}{2}}}{1-q^n}\sfJ_{-n}u^{n}\rt)\exp\lt(-\sum_{n=1}^\infty \frac{1}{n}\frac{(t/q)^{\frac{n}{2}}}{1-t^n}\sfJ_{n}u^{-n}\rt).
\end{equation}
It is easy to check that it matches exactly with  \eqref{dPhi-int-def}
\begin{equation}
 \Phi^{*}_{\emptyset}[v]=\exp\left(\sum_{r=1}^{\infty}\frac{v^{r}}{r}\frac{\gamma^{r}}{1-q_{3}^{-r}}\sfa_{-r}\right)\exp\left(-\sum_{r=1}^{\infty}\frac{v^{-r}}{r}\frac{1}{1-q_{3}^{r}}\sfa_{r}\right)
\end{equation}
by using the dictionary \eqref{a-dict}. Adding a box $(i=\lambda^t_j+1,j=\lambda_i+1)$ to $\lambda$ multiplies a vertex operator
\begin{align}
    &\exp\lt(-\sum_{n=1}^\infty\frac{1}{n}t^{-n\lambda_i}(t^{-n}-1)q^{(i-1)n}(tq)^{\frac{n}{2}}\sfJ_{-n}u^n\rt)\exp\lt(-\sum_{n=1}^\infty\frac{1}{n}q^{-n\lambda^t_j}(q^{-n}-1)t^{nj}(tq)^{-\frac{n}{2}}\sfJ_nu^{-n}\rt)\cr
    =&\exp\lt(-\sum_{n=1}^\infty\frac{1}{n}(1-t^n)\sfJ_{-n}(q/t)^{\frac{n}{2}}\chi_x^{n}\rt)\exp\lt(\sum_{n=1}^\infty\frac{1}{n}(1-q^{-n})(t/q)^{\frac{n}{2}}\sfJ_n\chi_x^{-n}\rt),\label{add-W}
\end{align}
to $\widetilde{U}_{\lambda}$, and it is equivalent to multiplying the vertex operator $\xi(\chi_x)$ with
\begin{align}
    \xi(z)=&\exp\left(-\sum_{r>0}\frac{\kappa_{r}}{r}\frac{q_3^{r/2}\sfa_{-r}}{(q_{3}^{r/2}-q_{3}^{-r/2})^{2}}z^{r}\right)\exp\left(-\sum_{r>0}\frac{\kappa_{r}}{r}\frac{\sfa_{r}}{(q_{3}^{r/2}-q_{3}^{-r/2})^{2}}z^{-r}\right)\cr
    =&\exp\lt(-\sum_{r>0}\frac{1-q_2^{-r}}{r}\sfJ_{-r}q_3^{-r/2}z^r\rt)\exp\lt(\sum_{r>0}\frac{1-q_1^{-r}}{r}q_3^{r/2}\sfJ_rz^{-r}\rt).\label{xi-J}
\end{align}
In summary, we have
\begin{equation}
\widetilde{U}_\lambda(t,q)=:\Phi^\ast_\emptyset[v]\prod_{x\in\lambda}\xi(\chi_x):.
\end{equation}

It can be readily verified that the remaining prefactors, $t_{n}(\lambda,u,v)=(-uv)^{|\lambda|}\prod_{x\in\lambda}(\gamma/\chi_{x})^{n+1}$ and $t^{*}_{n}(\lambda,u,v)=(u\gamma)^{-|\lambda|}\prod_{x\in\lambda}(\chi_{x}/\gamma)^{n}$, present in the intertwiners, account for the necessary framing factor in the topological vertex computation:
\begin{align}
    \tilde{f}_\lambda(q,t)=(-1)^{|\lambda|}q^{-\frac{\|\lambda^t\|^2}{2}}t^{\frac{\|\lambda\|^2}{2}}=(-1)^{|\lambda|}(t/q)^{\frac{|\lambda|}{2}}q^{-\sum_j\frac{\lambda^t_j(\lambda^t_j-1)}{2}}t^{\sum_i\frac{\lambda_i(\lambda_i-1)}{2}}\cr
    =(-\gamma)^{|\lambda|}\prod_{(i,j)\in\lambda}q^{-(i-1)}t^{j-1}=(-\gamma)^{|\lambda|}\prod_{(i,j)\in\lambda}(\chi_x/u)^{-1}.
\end{align}
As argued before, there is also no net framing factor (in addition to the skew Schur functions) involved in the gluing in the non-preferred direction, so in this way, we conclude that the intertwiners \eqref{inttw-def} and \eqref{dual-int-def} provide an alternative but equivalent way to perform the topological vertex computation for the instanton partition function. In \S\ref{sec:qq}, we will further describe how a family of new physical observables, namely the $qq$-characters, can be derived/computed from the intertwiner approach.

At last, we remark that the superscript $^{(n)}$ in the intertwiner $\Phi^{(n)}_\lambda$ already encapsulates the information regarding the axio-dilaton charge $(n,1)$ within the 5-brane web. Consequently, it becomes possible to introduce a simplified representation of the brane diagram, wherein all lines not aligned with the preferred direction are depicted as vertical lines while preferred legs are represented by horizontal lines. Further elaboration and illustrative examples of the computational procedures involving intertwiners will be provided in the subsequent subsection.

The correspondence of the intertwiner and the refined topological vertex is summarized in the following diagram:
\begin{align}
  \includegraphics{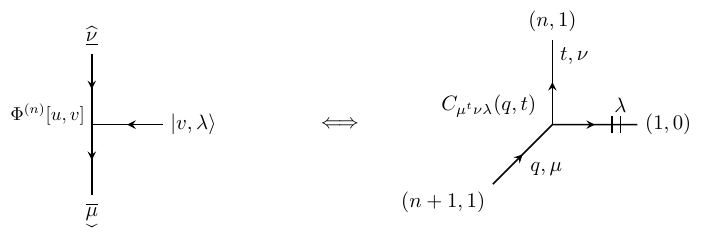}\cr
  \includegraphics{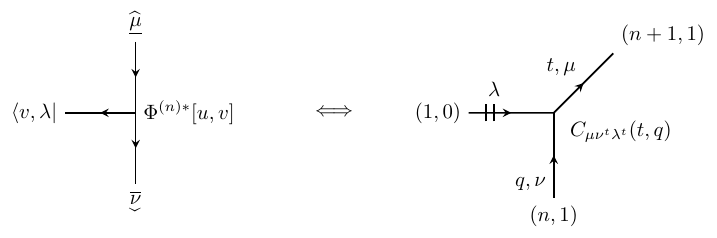}\label{eq:vertexintertwiner}
\end{align}
The correspondences of the parameters are summarized as\footnote{See footnote \ref{footnote:level} for notations. Due to historical reasons, the horizontal representation corresponds to the 5-branes in the vertical direction while the vertical representation corresponds to the 5-branes in the preferred (horizontal) direction. } 
\begin{align}
        \begin{tabular}{|c|c|}
            \hline brane charge& level of representation$(\ell_{1},\ell_{2})$\\
           \hline D5\,\,\, (1,0) & (0,1)\\
           \hline NS5\,\,\, $(n,1)$&$(1,n)$\\\hline
        \end{tabular}\label{eq:chargelevelcorres}
\end{align}

\subsection{Instanton partition functions and intertwiners}\label{sec:int-partition}

Let us now provide a summary of the key insights we have so far. The compactification of M-theory on a non-compact toric Calabi-Yau three-fold $X$ yields a 5d $\mathcal{N}=1$ theory, which is commonly referred to as geometric engineering \cite{Katz:1996fh,Katz:1997eq} (see \S\ref{sec:TV-intertwiners}). This compactification establishes an equivalence between the topological partition function of $X$ and the partition function of the 5d $\mathcal{N}=1$ theory. Furthermore, string duality connects it to fivebrane webs in Type IIB theory where the 5d $\cN=1$ theory is realized as a worldvolume theory of the fivebranes.  Consequently, by assigning the topological vertex to the junctions of the brane web, we can obtain the Nekrasov instanton partition function of the 5d $\mathcal{N}=1$ theory.

In \S\ref{sec:intertwiner}, we introduced the algebraic intertwiners purely from the representation theory of $\QTA$.  However, in the previous section, we established a relation between the topological vertex and the $\QTA$ intertwiner (see (\ref{eq:vertexintertwiner}) and (\ref{eq:chargelevelcorres})). This relation reveals that the BPS sectors of 5d $\cN=1$ theories are controlled by $\QTA$ \cite{Awata:2011ce}. Consequently, we ultimately establish a profound correspondence between instanton BPS states and representations of quantum algebras\footnote{Another approach to derive the correspondence between instanton partition functions and quantum algebras involves the utilization of the quiver $\cW$-algebra introduced in \cite{Kimura:2015rgi}. In this review, we will not discuss from this viewpoint, and we refer to \cite{Kimura-review}. The construction there is essentially equivalent to the discussion using the intertwiners.}. Indeed, the use of the intertwiners is one way to derive the AGT correspondence in \S\ref{sec:AGT}. Moreover, it is widely believed that such a correspondence between BPS states and representations of quantum algebras exists more generally and is referred to as the \emph{BPS/CFT correspondence} \cite{Nekrasov:2015wsu,Nekrasov:2016gud,Nekrasov:2016qym,Nekrasov:2016ydq,Nekrasov:2017gzb,Nekrasov:2017rqy}.

In this subsection, we present two illustrative examples that demonstrate the computation of instanton partition functions utilizing the intertwiner of $\QTA$: the pure super-Yang-Mills theory and linear quiver gauge theories. These examples not only demonstrate the use of the intertwiner but also unveil the underlying algebraic structure inherent in instanton partition functions.

\subsubsection{Pure super Yang-Mills theory}\label{sec:PSY}
 As will be shown explicitly, the instanton partition function coincides with the correlation function of the intertwiners of $\QTA$ \cite{Awata:2011ce}. Before the discovery of \cite{Awata:2011ce}, the 5d AGT relation between the instanton partition functions and the $q$-$\cW_N$ algebras had been investigated in  \cite{Awata:2009ur,Awata:2010yy,Taki:2014fva} by directly employing the 5d counterpart of the procedure outlined in \S\ref{sec:AGT}. Specifically, by defining Gaiotto states of $q$-$\cW_N$, their norm is identified as the 5d $\mathcal{N}=1$ $\U(N)$ partition function. This direct approach neatly fits into the broader framework of the intertwiner approach, thereby the intertwiners offer a simpler way to derive the partition function. Additionally, it becomes apparent that the Gaiotto state itself emerges as a particular type of expectation value of the intertwiners.
 
The pure super Yang-Mills with gauge group $\U(N)$ is constructed from a stack of $N$ D5-branes suspended by two NS5-branes. The charges of the 5-branes actually have physical constraints so that lines in the toric diagram do not intersect. The constraints can be derived from the viewpoint of representation of $\QTA$, but we will not go into them in this discussion.  Specifically, we consider two of the four external NS5-branes with levels $(1,n)_{u}$ and $(1,n^{\ast})_{u^{\ast}}$ as illustrated in the figure below. The charges of the D5-branes in the preferred horizontal direction are assigned $(0,1)_{v_{i}}\,\,(i=1,\ldots,N)$.  Following the rules outlined in \S\ref{sec:intertwiner}, we can then assign the intertwiners, and we obtain
\begin{align}
\begin{split}
\mathcal{Z}[\U(N)]=&\adjustbox{valign=c}{\includegraphics[width=6cm]{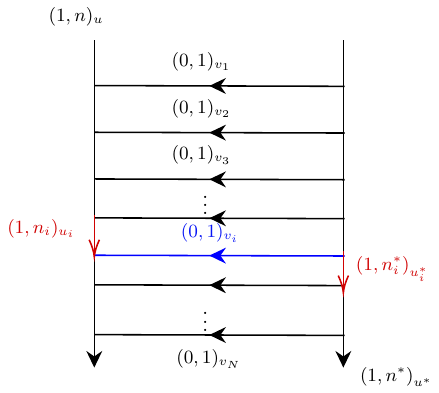}}=\begin{array}{ccc}
\widehat{\,\underline{0}\,}&&\vspace{1mm}\widehat{\,\underline{0}\,}\\
\Phi^{(n_{1})}[u_{1},v_{1}]&\cdot&\Phi^{(n_{1}^{*})*}[u_{1}^{*},v_{1}]\\
\Phi^{(n_{2})}[u_{2},v_{2}]&\cdot&\Phi^{(n_{2}^{*})*}[u_{2}^{*},v_{2}]\\
\vdots&&\vdots\\
\Phi^{(n_{N})}[u_{N},v_{N}]&\cdot&\Phi^{(n_{N}^{*})*}[u_{N}^{*},v_{N}]\vspace{1mm}\\
\uwidehat{\overline{\,0}\,}&&\uwidehat{\overline{\,0}\,}\\
\end{array}
\end{split}
\end{align}
After expanding the intertwiners in the vertical representation, we have
\begin{align}
\begin{split}
\mathcal{Z}[\U(N)]=&\sum_{\boldsymbol{\lambda}}\prod_{i=1}^{N}\Xi_{\lambda^{(i)}}\bra{0}\overleftarrow{\prod_{i=1}^{N}}\Phi^{(n_{i})}_{\lambda^{(i)}}[u_{i},v_{i}]\ket{0}\bra{0}\overleftarrow{\prod_{i=1}^{N}}\Phi^{(n^{*}_{i})*}_{\lambda^{(i)}}[u_{i}^{*},v_{i}]\ket{0}
\end{split}
\end{align}
where the ordered product is defined as
\begin{align}
    \overleftarrow{\prod_{i=1}^{N}}f_{i}(z):= f_{N}(z)f_{N-1}(z)\cdots f_{2}(z)f_{1}(z).
\end{align}
Physically the conservation law of the charges of the 5-branes, or mathematically the existence condition of the intertwiners imposes  the following conditions on the spectral parameters:
\begin{align}
\begin{split}
   u_{1}=u,\quad u^{\ast}_{N}=u^{\ast},&\quad n_{i}=n+i-1,\quad n_{i}^{*}=n^{*}+N-i,\\
    u_{i}=u\prod_{l=1}^{i-1}(-v_{l}),&\quad u_{i}^{*}=u^{*}\prod_{j=i+1}^{N}(-v_{j})
\end{split}\label{eq:pureSYMconservation}
\end{align}
for $i=1,\ldots, N$. 
Using \eqref{Phi-int-def}, \eqref{t*}, and \eqref{eq:oneloopintertwiner}, we evaluate  the expectation values of the intertwiners
\begin{align}
    \bra{0}\overleftarrow{\prod_{i=1}^{N}}\Phi^{(n_{i})}_{\lambda^{(i)}}[u_{i},v_{i}]\ket{0}=&\prod_{i=1}^{N}t_{n_{i}}(\lambda^{(i)},u_{i},v_{i})\prod_{i<j}\mathcal{G}\left(q_{3}^{-1}\frac{v_{i}}{v_{j}}\right)N_{\lambda^{(i)}\lambda^{(j)}}(v_{i}/v_{j};q_{1},q_{2})^{-1},\\
    \begin{split}
    \bra{0}\overleftarrow{\prod_{i=1}^{N}}\Phi^{(n^{*}_{i})*}_{\lambda^{(i)}}[u_{i}^{*},v_{i}]\ket{0}=&\prod_{i=1}^{N}t^{*}_{n^{*}_{i}}(\lambda^{(i)},u^{*}_{i},v_{i})\prod_{i<j}\mathcal{G}\left(\frac{v_{i}}{v_{j}}\right)N_{\lambda^{(i)},\lambda^{(j)}}(q_{3}v_{i}/v_{j};q_{1},q_{2})^{-1}\\
    =&\prod_{x\in\lambda^{(i)}}\left(-\frac{v_{j}}{\chi_{x}}\right)\prod_{x\in\lambda^{(j)}}\left(-\frac{\chi_{x}}{q_{3}v_{i}}\right)\prod_{i=1}^{N}t^{*}_{n^{*}_{i}}(\lambda^{(i)},u^{*}_{i},v_{i})\\
    &\qquad \times\prod_{i<j}\mathcal{G}\left(\frac{v_{i}}{v_{j}}\right)N_{\lambda^{(j)},\lambda^{(i)}}(v_{j}/v_{i};q_{1},q_{2})^{-1}~.
    \end{split}
\end{align}
Hence, the partition function can be decomposed into the perturbative part and the instanton part, given by:
\begin{align}
    \mathcal{Z}=&\mathcal{Z}_{\text{root}}\mathcal{Z}_{\text{inst}},\\
    \mathcal{Z}_{\text{root}}=&\prod_{i<j}\mathcal{G}\left(q_{3}^{-1}\frac{v_{i}}{v_{j}}\right)\mathcal{G}\left(\frac{v_{i}}{v_{j}}\right),\\
\begin{split}
    \mathcal{Z}_{\text{inst}}=&\sum_{\boldsymbol{\lambda}}\prod_{l=1}^{N}\left(\Xi_{\lambda^{(l)}}t_{n_{l}}(\lambda^{(l)},u_{l},v_{l})t^{\ast}_{n_{l}^{\ast}}(\lambda^{(l)},u_{l}^{\ast},v_{l})\right)\\
    &\times \prod_{i<j}\left\{\prod_{x\in\lambda^{(i)}}\left(-\frac{v_{j}}{\chi_{x}}\right)\prod_{x\in\lambda^{(j)}}\left(-\frac{\chi_{x}}{q_{3}v_{i}}\right)\prod_{i\neq j}N_{\lambda^{(i)}\lambda^{(j)}}(v_{i}/v_{j};q_{1},q_{2})^{-1}\right\}.
\end{split}
\end{align}
This result matches with (\ref{eq:one-loop}) and (\ref{eq:instantonpartitionfunction}). Inserting
\begin{align}
    \prod_{l=1}^{N}\Xi_{\lambda^{(l)}}=\prod_{l=1}^{N}\left((\gamma v_{l})^{-|\lambda^{(l)}|}\prod_{x\in\lambda^{(l)}}\chi_{x}\right)\prod_{l=1}^{N}\frac{1}{N_{\lambda^{(l)}\lambda^{(l)}}(1;q_{1},q_{2})}
\end{align}
and computing the zero-modes part in front of the Nekrasov factors, we eventually obtain
\begin{align}
\mathcal{Z}_{\text{inst.}}\left[\U(N);\mathfrak{q},\kappa\right]=\sum_{\boldsymbol{\lambda}}\mathfrak{q}^{\sum_{i=1}^{N}|\lambda^{(i)}|}\left(\prod_{i=1}^{N}\prod_{x\in\lambda^{(i)}}\left(\chi_{x}\right)^{\kappa}\right)\prod_{i,j}N_{\lambda^{(i)}\lambda^{(j)}}(v_{i}/v_{j};q_{1},q_{2})^{-1}
\end{align}
where
\begin{align}
    \mathfrak{q}=-\frac{u}{u^{\ast}}\gamma^{n-n^{*}-N},\quad \kappa=n^{*}-n.\label{eq:SYMintertwinerparameter}
\end{align}
For later use, we denoted the partition function as $\mathcal{Z}[\U(N);\mathfrak{q},\kappa]$, where $\mathfrak{q}$ is the topological term and $\kappa$ is the 5d Chern-Simons level.

\paragraph{Partition function from generalized intertwiner} Using the generalized intertwiners in (\ref{eq:generalized-intertwiner1}), (\ref{eq:generalized-intertwiner2}), and (\ref{eq:generalized-intertwiner-figure}), the partition function is rewritten as
\begin{align}
    \begin{split}
        \mathcal{Z}=\adjustbox{valign=c}{\includegraphics{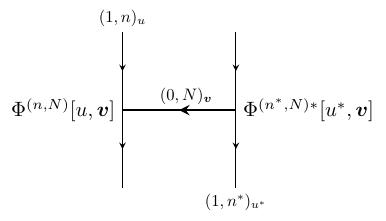}}=\begin{array}{ccc}
\widehat{\,\underline{0}\,}&&\vspace{1mm}\widehat{\,\underline{0}\,}\\
\Phi^{(n,N)}[u,\boldsymbol{v}]&\boldsymbol{\cdot}&\Phi^{(n^{*},N)*}[u^{*},\boldsymbol{v}]\\
\vspace{0.4cm}\uwidehat{\overline{\,0}\,}&&\uwidehat{\overline{\,0}\,}\\
\end{array}
    \end{split}
\end{align}where $\boldsymbol{v}=(v_{1},\ldots,v_{N})$. The generalized intertwiner gives the partition function of a stack of $N$ D5-branes suspended by two NS5-branes.

\paragraph{Gaiotto state}
Moreover, we can introduce the deformed version of the Gaiotto state introduced in (\ref{coh_W}) and (\ref{eq:Gaiotto_state}) using the generalized intertwiners as 
\begin{align}
    \begin{split}
      |\frakG,\boldsymbol{v}\rangle^{(N)}=&\frac{\bra{0}\Phi^{(n^{\ast},N)\ast}[u^{\ast},\boldsymbol{v}]\ket{0}}{\bra{0}\Phi^{(n^{\ast},N)\ast}_{\boldsymbol{\emptyset}}[u^{\ast},\boldsymbol{v}]\ket{0}}=\sum_{\lambda^{(1)},\ldots,\lambda^{(N)}}\prod_{i=1}^{N}\Xi_{\lambda^{(i)}}\frac{\bra{0}\Phi^{(n^{\ast},N)\ast}_{\boldsymbol{\lambda}}[u^{\ast},\boldsymbol{v}]\ket{0}}{\bra{0}\Phi^{(n^{\ast},N)\ast}_{\boldsymbol{\emptyset}}[u^{\ast},\boldsymbol{v}]\ket{0}}\ket{\boldsymbol{v},\boldsymbol{\lambda}}\\
      =&\sum_{\lambda^{(1)},\ldots,\lambda^{(N)}}\prod_{i=1}^{N}\Xi_{\lambda^{(i)}}\prod_{i=1}^{N}t^{\ast}_{n^{\ast}_{i}}(\lambda^{(i)},u^{\ast}_{i},v_{i})\prod_{i<j}N_{\lambda^{(i)}\lambda^{(j)}}(q_{3}v_{i}/v_{j};q_{1},q_{2})^{-1}\ket{\boldsymbol{v},\boldsymbol{\lambda}},\\
      \ket{\boldsymbol{v},\boldsymbol{\lambda}}=&\ket{v_{N},\lambda^{(N)}}\otimes \ket{v_{N-1},\lambda^{(N-1)}}\otimes \cdots\otimes\ket{v_{2},\lambda^{(2)}}\otimes \ket{v_{1},\lambda^{(1)}}
    \end{split}\label{eq:qdef-Gaiottostate}
\end{align}
where $\Phi^{(n^{\ast},N)\ast}_{\boldsymbol{\emptyset}}[u^{\ast},\boldsymbol{v}]$ is the composition of intertwiners with empty Young diagrams. The dual Gaiotto state is defined similarly as
\begin{align}
\begin{split}
    {\small^{(N)}}\langle\frakG,\boldsymbol{v}|=&\frac{\bra{0}\Phi^{(n,N)}[u,\boldsymbol{v}]\ket{0}}{\bra{0}\Phi^{(n,N)}_{\boldsymbol{\emptyset}}[u,\boldsymbol{v}]\ket{0}}=\sum_{\lambda^{(1)},\ldots,\lambda^{(N)}}\prod_{i=1}^{N}\Xi_{\lambda^{(i)}}\frac{\bra{0}\Phi^{(n,N)}_{\boldsymbol{\lambda}}[u,\boldsymbol{v}]\ket{0}}{\bra{0}\Phi^{(n,N)}_{\boldsymbol{\emptyset}}[u,\boldsymbol{v}]\ket{0}}\bra{\boldsymbol{v},\boldsymbol{\lambda}},\\
    \bra{\boldsymbol{v},\boldsymbol{\lambda}}=&\bra{v_{N},\lambda^{(N)}}\otimes \bra{v_{N-1},\lambda^{(N-1)}}\otimes \cdots\otimes \bra{v_{2},\lambda^{(2)}}\otimes \bra{v_{1},\lambda^{(1)}}.
\end{split}\label{eq:qdef-dualGaiottostate}
\end{align}
Using these Gaiotto states the instanton partition function is written as 
\begin{align}
    \mathcal{Z}_{\text{inst.}}=^{(N)}\hspace{-0.15cm}\langle\frakG,\boldsymbol{v}|\frakG,\boldsymbol{v}\rangle^{(N)}\label{eq:5dGaiottonorm}
\end{align}
which gives the 5d version of (\ref{eq:partitionGaiottoSU2}).

For later reference, let us describe the action of the Drinfeld currents on these Gaiotto states:
\begin{align}
\begin{split}
    E(z)|\frakG,\boldsymbol{v}\rangle^{(N)}=&-u^{\ast}\left(\frac{\gamma}{z}\right)^{n^{\ast}}\left(\widehat{\mathcal{Y}}^{+}(q_{3}^{-1}z)-\widehat{\mathcal{Y}}^{-}(q_{3}^{-1}z)\right)|\frakG,\boldsymbol{v}\rangle^{(N)},\\
    F(z)|\frakG,\boldsymbol{v}\rangle^{(N)}=&-\frac{\gamma^{-N}}{u^{\ast}}\left(\frac{z}{\gamma}\right)^{n^{\ast}}\left(\widehat{\mathcal{Y}}^{+}(z)^{-1}-\widehat{\mathcal{Y}}^{-}(z)^{-1}\right)|\frakG,\boldsymbol{v}\rangle^{(N)},\\
    K^{\pm}(z)|\frakG,\boldsymbol{v}\rangle^{(N)}=&\gamma^{-N}\frac{\widehat{\mathcal{Y}}^{\pm}(q_{3}^{-1}z)}{\widehat{\mathcal{Y}}^{\pm}(z)}|\frakG,\boldsymbol{v}\rangle^{(N)}.
\end{split}\label{eq:Gaiottost-action1}
\end{align}
where we defined a formal operator $\widehat{\mathcal{Y}}^{\pm}(z)$ acting diagonally on the vertical representation as 
\begin{align}
    \widehat{\mathcal{Y}}^{\pm}(z)\ket{\boldsymbol{v},\boldsymbol{\lambda}}=\left[\mathcal{Y}_{\boldsymbol{\lambda}}(z,\boldsymbol{v})\right]_{\pm}\ket{\boldsymbol{v},\boldsymbol{\lambda}}=&\left[\prod_{i=1}^{N}\mathcal{Y}_{\lambda^{(i)}}(z,v_{i})\right]_{\pm}\ket{\boldsymbol{v},\boldsymbol{\lambda}}.\label{eq:YopQTA}
\end{align}
Similarly, we can derive the action of the currents on the dual Gaiotto state:
\begin{align}
\begin{split}
^{(N)}\hspace{-0.cm}\langle\frakG,\boldsymbol{v}|E(z)=&u\prod_{j=1}^{N}(-v_{j})\left(\frac{\gamma}{z}\right)^{n+N}\,^{(N)}\hspace{-0cm}\langle\frakG,\boldsymbol{v}|\left(\widehat{\mathcal{Y}}^{+}(z)^{-1}-\widehat{\mathcal{Y}}^{-}(z)^{-1}\right),\\
^{(N)}\hspace{-0cm}\langle\frakG,\boldsymbol{v}|F(z)=&u^{-1}\prod_{j=1}^{N}(-\gamma v_{j})^{-1}\left(\frac{z}{\gamma}\right)^{n+N}\,^{(N)}\hspace{-0cm}\langle\frakG,\boldsymbol{v}|\left(\widehat{\mathcal{Y}}^{+}(q_{3}^{-1}z)-\widehat{\mathcal{Y}}^{-}(q_{3}^{-1}z)\right),\\
    ^{(N)}\hspace{-0cm}\langle\frakG,\boldsymbol{v}|K^{\pm}(z)=&^{(N)}\hspace{-0.05cm}\langle\frakG,\boldsymbol{v}|\gamma^{-N}\frac{
\widehat{\mathcal{Y}}^{\pm}(q_{3}^{-1}z)}{\widehat{\mathcal{Y}}^{\pm}(z)}.
\end{split}\label{eq:Gaiottost-action2}
\end{align}

\subsubsection{Linear quiver gauge theory}

\begin{figure}[ht]
\begin{center}
\includegraphics[width=7cm]{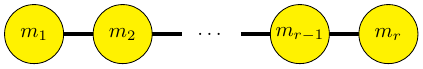}
\end{center}
\caption{A typical example of a linear quiver, where each node represents a vector multiplet of gauge group U($m_i$), and the lines stand for bifundamental hypermultiplets. }
\label{f:quiver}
\end{figure}

It is straightforward to apply this intertwiner method to a linear quiver gauge theory with gauge group $\prod_{\alpha=1}^{r}\U(m_{\alpha})$. Instead of using the original intertwiner one by one, it is simpler to use the generalized intertwiner. For the linear quiver given in Figure \ref{f:quiver}, the brane web is schematically drawn as
$$
   \includegraphics{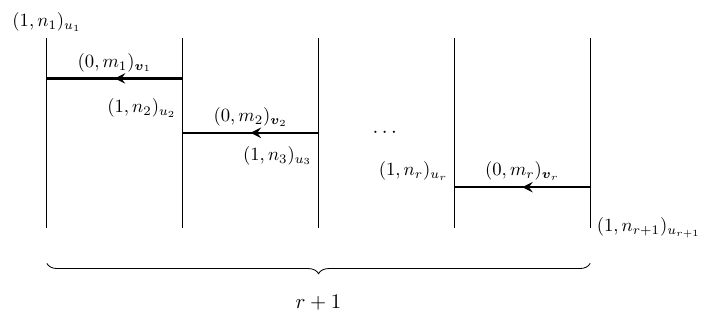}
$$
where $\boldsymbol{v}_{\alpha}=(v^{(1)}_{\alpha},v^{(2)}_{\alpha},\ldots, v^{(m_{\alpha})}_{\alpha})$, $(\alpha=1,\ldots,r)$.
There are a total of $r+1$ vertical 5-branes, and each of them suspends $m_{\alpha}$ D5-branes $(\alpha=1,\ldots,r)$. Each stack of D5-branes corresponds to a $\U(m_{\alpha})$ gauge theory. By employing the generalized intertwiner for each stack of D-branes, we obtain
\begin{align}
    \begin{split}&\mathcal{Z}\left[\prod_{\alpha=1}^{r}\U(m_{\alpha})\right]=\\=&
\begin{array}{cccccccccc}
\widehat{\,\underline{0}\,}&&\vspace{1mm}\widehat{\,\underline{0}\,}&&\widehat{\,\underline{0}\,}&&\widehat{\,\underline{0}\,}\\
\Phi^{(n_{1},m_{1})}[u_{1},\boldsymbol{v}_{1}]&\boldsymbol{\cdot}&\Phi^{(n_{2},m_{1})*}[u_{2},\boldsymbol{v}_{1}]&&&&\\
&&\Phi^{(n_{2},m_{2})}[u_{2},\boldsymbol{v}_{2}]&\ddots&&&\\
&&&\ddots&&&\\
&&&&&&\\
&&&\ddots&\Phi^{(n_{r},m_{r-1})\ast}[u_{r},\boldsymbol{v}_{r-1}]&&\\
&&&&\Phi^{(n_{r},m_{r})}[u_{r},\boldsymbol{v}_{r}]&\boldsymbol{\cdot}&\Phi^{(n_{r+1},m_{r})\ast}[u_{r+1},\boldsymbol{v}_{r}]\vspace{1.3mm}\\
\uwidehat{\overline{\,0}\,}&&\uwidehat{\overline{\,0}\,}&&\uwidehat{\overline{\,0}\,}&&\uwidehat{\overline{\,0}\,}
\end{array}\\
=&\sum_{\boldsymbol{\lambda}_{1},\ldots,\boldsymbol{\lambda}_{r}}\bra{0}\Phi^{(n_{1},m_{1})}_{\boldsymbol{\lambda}_{1}}[u_{1},\boldsymbol{v}_{1}]\ket{0}\prod_{\alpha=2}^{r}\bra{0}\Phi^{(n_{\alpha},m_{\alpha})}_{\boldsymbol{\lambda}_{\alpha}}[u_{\alpha},\boldsymbol{v}_{\alpha}]\Phi^{(n_{\alpha},m_{\alpha-1})*}_{\boldsymbol{\lambda}_{\alpha-1}}[u_{\alpha},\boldsymbol{v}_{\alpha-1}]\ket{0}\\
&\qquad\quad\times \bra{0}\Phi^{(n_{r+1},m_{r})*}_{\boldsymbol{\lambda}_{r}}[u_{r+1},\boldsymbol{v}_{r}]\ket{0}
\end{split}
\end{align}
where $\boldsymbol{\lambda}_{\alpha}=(\lambda^{(1)}_{\alpha},\lambda^{(2)}_{\alpha},\ldots,\lambda^{(m_{\alpha})}_{\alpha})$. The nontrivial part comes from the contractions \eqref{eq:intertwiner-contraction} in the horizontal representations:
\begin{align}
    \frac{\bra{0}\Phi_{\boldsymbol{\lambda}_{\alpha}}^{(n_{\alpha},m_{\alpha})}[u_{\alpha},\boldsymbol{v}_{\alpha}]\Phi_{\boldsymbol{\lambda}_{\alpha-1}}^{(n_{\alpha},m_{\alpha-1})*}[u_{\alpha},\boldsymbol{v}_{\alpha-1}]\ket{0}}{\bra{0}\Phi_{\boldsymbol{\lambda}_{\alpha}}^{(n_{\alpha},m_{\alpha})}[u_{\alpha},\boldsymbol{v}_{\alpha}]\ket{0}\bra{0}\Phi_{\boldsymbol{\lambda}_{\alpha-1}}^{(n_{\alpha},m_{\alpha-1})*}[u_{\alpha},\boldsymbol{v}_{\alpha-1}]\ket{0}}=\mathcal{Z}_{\text{root}}\mathcal{Z}_{\text{bfd.}}(\boldsymbol{v}_{\alpha-1},\boldsymbol{\lambda}_{\alpha-1}\,|\,\boldsymbol{v}_{\alpha},\boldsymbol{\lambda}_{\alpha}\,|\,\gamma^{-1})
\end{align}
for $2\leq \alpha\leq r$.
For the instanton part, we finally obtain
\begin{align}
    \mathcal{Z}_{\text{inst.}}\left[\prod_{\alpha=1}^{r}\U(m_{\alpha});\boldsymbol{\mathfrak{q}},\boldsymbol{\kappa}\right]=&\prod_{\alpha=1}^{r}\mathcal{Z}_{\text{inst.}}[\U(m_{\alpha});\mathfrak{q}_{\alpha},\kappa_{\alpha}]\prod_{\alpha=2}^{r}\mathcal{Z}_{\text{bfd.}}(\boldsymbol{v}_{\alpha-1},\boldsymbol{\lambda}_{\alpha-1}\,|\,\boldsymbol{v}_{\alpha},\boldsymbol{\lambda}_{\alpha}\,|\,\gamma^{-1})
\end{align}
where
\begin{align}
\begin{split}
&\boldsymbol{\mathfrak{q}}=(\mathfrak{q}_{1},\ldots,\mathfrak{q}_{r}),\quad\boldsymbol{\kappa}=(\kappa_{1},\ldots,\kappa_{r}) \\
    &\mathfrak{q}_{\alpha}=-\frac{u_{\alpha}}{u_{\alpha+1}}\gamma^{n_{\alpha}-n_{\alpha+1}-m_{\alpha}},\quad \kappa_{\alpha}=n_{\alpha+1}-n_{\alpha},\quad \alpha=1,\ldots,r.
\end{split}
\end{align}

\section{\texorpdfstring{$qq$}{qq}-characters}\label{sec:qq}

In this section, we spotlight the concept of $qq$-characters, showcasing yet another manifestation of the link between supersymmetric theories and $\QTA$ (or $\AY$).
Originally, the $qq$-character was introduced in \cite{Nekrasov:2015wsu} as an expectation value of 1/2-BPS codimension-four defect in 4d or 5d supersymmetric theories with eight supercharges. From an algebraic perspective, the $qq$-character can be constructed from the vertical representations of the Drinfeld current of $\QTA$ (or $\AY$) and in the intertwiners. Notably, the $qq$-characters in supersymmetric gauge theories offer a deeper connection between quantum integrable systems. Hence, the $qq$-characters tie together symmetries, integrabilities, and BPS defects. The connection between the $qq$-characters and $\QTA$ (or $\AY$) is essential for understanding this example of the BPS/CFT correspondence.

To elucidate these concepts, we initiate our discussion by detailing the physical framework of the $qq$-characters.  Following this, we turn our attention to the relationship between $qq$-characters and the representations of $\QTA$ (or $\AY$) constructed in \S\ref{sec:QTrep}. Additionally, we highlight its connections to integrable models that are rather different from those discussed in \S\ref{s:int-model}.

\subsection{\texorpdfstring{$qq$}{qq}-character from defect partition function}\label{sec:defect}

The $qq$-characters can be constructed from the gauge theory partition function with co-dimension four defects \cite{Kim-qq}. To evaluate these partition functions, one can use the Jeffery-Kirwan (JK) prescription and perform a contour integral. The integrand can be read from the ADHM construction with codimension four defects, which can further be obtained from the brane construction shown in Table \ref{ADHM-brane}. In 5d gauge theory set-up, the defect branes introduce Wilson loops along the $S^1$-circle, and when reduced to 4d, they become point-like defects. We remark that the brane construction presented in Table \ref{ADHM-brane} directly gives a 5d $\cN=1^\ast$ gauge theory with an adjoint hypermultiplet of mass $m$. We take $m\rightarrow \infty$ here to obtain the defect partition function of $\cN=1$ pure Yang-Mills theory.

\begin{table}[ht]
\begin{center}
\begin{tabular}{|c|c|c|c|c|c|c|c|c|c|c|}
\hline
& 0 & 1 & 2 & 3 & 4& 5 & 6 & 7 & 8 & 9 \cr
\hline
D4 & $\bullet$ & $\bullet$ & $\bullet$ & $\bullet$ & $\bullet$ & $-$ & $-$ & $-$ & $-$ & $-$ \cr
\hline
D0 & $\bullet$ & $-$ & $-$ & $-$ & $-$ & $-$ & $-$ & $-$ & $-$ & $-$ \cr
\hline
D$4'$ & $\bullet$ & $-$ & $-$ & $-$ & $-$ & $\bullet$  & $\bullet$ & $\bullet$ & $\bullet$ & $-$ \cr
\hline
\end{tabular}
\caption{Configuration of branes in the brane construction of ADHM construction with Wilson lines (introduced by D4' branes) \cite{Kim-qq}. }
\label{ADHM-brane}
\end{center}
\end{table}

The explicit expression of the integrand reads
\begin{equation}
\cZ(\boldsymbol{z})=\sum_{k=0}^\infty \frac{\mathfrak{q}^k}{k!}\oint_{\textrm{JK}}\prod_{I=1}^k\frac{{\rm d}\phi_I}{2\pi i}\cZ^{(k)}_{\textrm{vec}}\cZ^{(k)}_{\textrm{def}}(\boldsymbol{z})
\end{equation}
Here, $\cZ^{(k)}_{\textrm{vec}}$ and $\cZ^{(k)}_{\textrm{def}}(\boldsymbol{z})$ represent the contributions to the partition function from the vector multiplet and the defect, respectively. The former is given by
\begin{equation}
\cZ^{(k)}_{\textrm{vec}}=\lt(\frac{\llbracket-\e_3\rrbracket}{\llbracket\epsilon_1\rrbracket\llbracket\epsilon_2\rrbracket}\rt)^k\prod_{i=1}^k\prod_{\alpha=1}^N\frac{1}{\llbracket\mathfrak{a}_\alpha-\phi_i\rrbracket\llbracket\phi_i-\mathfrak{a}_\alpha-\e_3\rrbracket}\prod_{\substack{i,j=1\\i\neq j}}^k S^{-1}(\chi_i/\chi_j)
\end{equation}
while the latter is given by
\begin{equation}
\cZ^{(k)}_{\textrm{def}}(\boldsymbol{z})=\prod_{j=1}^{w}\prod_{\alpha=1}^N\llbracket\zeta_j-\mathfrak{a}_\alpha-\e_3\rrbracket\prod_{i=1}^k S(z_j/\chi_i)~, \label{def-cont}
\end{equation}
In the above equations, we use the convention $\llbracket x\rrbracket:=2\sinh(x/2)$, $\chi_i=\exp\lt(\phi_i\rt)$, $z_j=\exp\lt(\zeta_j\rt)$, $u_\alpha=\exp\lt(\mathfrak{a}_\alpha\rt)$ and $S$-function defined in \eqref{def-S} can be rewritten into
\begin{equation}\label{S}
S(z)=\frac{(1-q_1z)(1-q_2z)}{(1-z)(1-zq_1q_2)}=\frac{\llbracket\zeta+\epsilon_1\rrbracket\llbracket\zeta+\epsilon_2\rrbracket}{\llbracket\zeta\rrbracket\llbracket\zeta-\epsilon_3\rrbracket}~.
\end{equation}
Here $w$ denotes the number of defect D$4'$-branes or Wilson loops in the gauge theory picture, and $\zeta_j$'s have the physical meaning of the positions of the defect branes along the 9-direction in the ADHM brane web (see Table \ref{ADHM-brane}). Following the JK prescription, one can pick up the pole at $\phi_i=\zeta_j$, but since there exist zeros at $\phi_{i+1}=\zeta_j+\epsilon_{1,2}$, only the pole at $\phi_i=\zeta_j$ associated with the parameter $\zeta_j$ is allowed in the JK prescription. Therefore, the JK residues fall into two categories: the usual residues when there is no defect, given by  $\{\phi_x=\mathfrak{a}_\alpha+(i-1)\epsilon_1+(j-1)\epsilon_2\}_{(i,j)\in\boldsymbol{\lambda}}$, and the residues associated with the defect, given by $\{\zeta_j\}_{j=1,\ldots,w}$.  The defect contribution evaluated at a set of usual residues is given by
\begin{equation}
\tilde{\cY}_{\boldsymbol{\lambda}}(z,\boldsymbol{u}):=\prod_{\alpha=1}^N\llbracket\zeta-\mathfrak{a}_\alpha-\e_3\rrbracket\prod_{x\in\lambda^{(\alpha)}}S(z/\chi_x)
=\prod_{\alpha=1}^N\llbracket\zeta-\mathfrak{a}_\alpha-\e_3\rrbracket\prod_{x\in\lambda^{(\alpha)}}S\lt(\chi_x/(zq_1q_2)\rt)
\end{equation}
where we use the identity
\begin{equation}
S(z^{-1})=S(zq_3).
\end{equation}
It matches exactly with the $\cY$-function introduced in \eqref{def-Yc} for $c=3$, and let us denote it as
\bea
&\cY_{\boldsymbol{\lambda}}(z,\boldsymbol{u}):=\prod_{\alpha=1}^N\cY_{\lambda^{(\alpha)}}(z,u_\alpha),\\
&\cY_\lambda(z,u):=(1-u/z)\prod_{x\in\lambda}S(\chi_x/z)=\frac{\prod_{x\in \frakA(\lambda)}1-\chi_x/z}{\prod_{y\in \frakR(\lambda)}1-\chi_y q_3^{-1}/z},\label{def-Y}
\eea
then we have
\begin{equation}
\tilde{\cY}_{\boldsymbol{\lambda}}(z,\boldsymbol{u})=\lt(\prod_{\alpha=1}^Nzq_1q_2/u_\alpha\rt)^{\frac{1}{2}}\cY_{\boldsymbol{\lambda}}(zq_1q_2,\boldsymbol{u}).
\end{equation}
Picking one pole at $\phi_i=\zeta$, the defect contribution becomes
\begin{align}
\prod_{\alpha=1}^N\frac{1}{\llbracket\mathfrak{a}_\alpha-\zeta\rrbracket\llbracket\zeta-\mathfrak{a}_\alpha-\e_3\rrbracket}\prod_{x\in\boldsymbol{\lambda}}S^{-1}\lt(z/\chi_x\rt)S^{-1}\lt(\chi_x/z\rt)\times\prod_{\alpha=1}^N\llbracket\zeta-\mathfrak{a}_\alpha-\e_3\rrbracket\prod_{x\in\boldsymbol{\lambda}} S(z/\chi_x)\cr
=\prod_{\alpha=1}^N\frac{1}{\llbracket\mathfrak{a}_\alpha-\zeta\rrbracket}\prod_{x\in\boldsymbol{\lambda}}S^{-1}\lt(\chi_x/z\rt)=(-1)^N\tilde{\cY}_{\boldsymbol{\lambda}}^{-1}(zq_3,\boldsymbol{u}),
\end{align}
where the Young diagrams here have a total size of $\sum_{\alpha=1}^N|\lambda^{(\alpha)}|=k-1$. In summary, for $w=1$, the defect partition function reads
\begin{equation}
\chi_{w=1}(z):=\cZ(z)=\sum_{\boldsymbol{\lambda}}\mathfrak{q}^{|\boldsymbol{\lambda}|} \tilde{\cY}_{\boldsymbol{\lambda}}(z,\boldsymbol{u})\cZ_{\text{vect.}}(\boldsymbol{u},\boldsymbol{\lambda})+(-1)^Nq_3^\frac{N}{2}\mathfrak{q}\sum_{\boldsymbol{\lambda}}\mathfrak{q}^{|\boldsymbol{\lambda}|} \tilde{\cY}^{-1}_{\boldsymbol{\lambda}}(zq_3,\boldsymbol{u})\cZ_{\text{vect.}}(\boldsymbol{u},\boldsymbol{\lambda}).
\end{equation}
For an operator $\cO$, whose eigenvalues $\cO_{\boldsymbol{\lambda}}$ are labeled by a tuple of Young diagrams $\boldsymbol{\lambda}$, we can define its expectation value as
\begin{equation}
\langle \cO\rangle:=\sum_{\boldsymbol{\lambda}} \mathfrak{q}^{|\boldsymbol{\lambda}|}\cO_{\boldsymbol{\lambda}}\cZ_{\text{vect.}}(\boldsymbol{u},\boldsymbol{\lambda})\label{eq:instvev}
\end{equation}
then the $w=1$ defect partition is rewritten into
\begin{equation}
\chi_{w=1}(z)=\langle \tilde{\cY}(z,\boldsymbol{u})\rangle +(-1)^N\mathfrak{q}\langle \tilde{\cY}^{-1}(zq_3,\boldsymbol{u})\rangle.\label{qq-1}
\end{equation}
An interesting fact to notice here is that the $qq$-character, as its name suggests, indeed looks like a fundamental character of $\mathfrak{su}_2$ (with weight $w=1$),
\begin{equation}
\chi^{\mathfrak{su}_2}_{w=1}(y)=y+y^{-1}
\end{equation}
where the $\mathfrak{su}_2$ Lie-algebraic structure comes from the $A_1$ quiver structure of the pure gauge theory. The character nature is due to the invariance under the so-called iWeyl reflection (short for instanton-Weyl reflection), first proposed in \cite{Nekrasov:2012xe} as a tool to compute the Seiberg-Witten curve for (higher rank) quiver gauge theories and then generalized to $qq$-characters in \cite{Nekrasov:2015wsu,Kimura:2015rgi}.

In the next section, we will re-derive the above expression of $\chi_{w=1}$ from a completely different viewpoint, a purely algebraic approach. For later convenience, let us express $\chi_{w=1}$ in terms of $\cY_\lambda$, 
\begin{align}
    \chi_{w=1}(z)=\lt(\prod_{\alpha=1}^Nzq_1q_2/u_\alpha\rt)^{\frac{1}{2}}\lt(\langle \cY(zq_1q_2,\boldsymbol{u})\rangle +\frac{(-1)^Nq_3^{N}\mathfrak{q}\prod_{\alpha=1}^N u_\alpha}{z^N}\langle \cY^{-1}(z,\boldsymbol{u})\rangle\rt)\cr
    =:\lt(\prod_{\alpha=1}^zz^{-1}q_1q_2/u_\alpha\rt)^{\frac{1}{2}}\bar{\chi}_{w=1}(z).
\end{align}
In $\chi_{w=1}$, there are apparent poles in $\cY_\lambda$-function at
\begin{equation}
    \zeta=\phi_y,\quad y\in \frakR(\boldsymbol{\lambda})
\end{equation}
and apparent poles $\cY^{-1}_\lambda$-function at
\begin{equation}
    \zeta=\phi_x,\quad x\in \frakA(\boldsymbol{\lambda}).
\end{equation}
Such apparent poles are in fact cancelled between two terms in the $qq$-character, e.g. \eqref{qq-1}. For example at one-instanton level, i.e. at the order $\cO(\mathfrak{q})$, of the SU(2) $qq$-character, we have
\begin{align}
    &\tilde{\cY}_{\{(1),\emptyset\}}(z,\boldsymbol{u})\cZ_{\text{vect.}}(\{u,u^{-1}\},\{(1),\emptyset\})\cr
    =&-\frac{u(zq_2-u)(zq_1-u)(1-uzq_1q_2)}{q_1q_2z(z-u)(1-q_1)(1-q_2)(1-u^2)(1-u^2q_1q_2)},
\end{align}
which clearly has a pole at $z=u$ with residue
\begin{equation}
    \frac{u^2q_3}{(1-u^2)}.\label{res-1}
\end{equation}
The contribution from the second term in \eqref{qq-1}, $(-1)^Nq_3^{\frac{N}{2}}\mathfrak{q}\langle \tilde{\cY}^{-1}(zq_3)\rangle$, at $\cO(\mathfrak{q})$ order is given by
\begin{equation}
    \frac{zq_3}{(z-u)(z-u^{-1})}
\end{equation}
with residue at $z=u$
\begin{equation}
    -\frac{u^2q_3}{1-u^2},
\end{equation}
which cancels with \eqref{res-1}. Similar cancellation happens for the pole $z=u^{-1}$ in $\tilde{\cY}_{\{\emptyset,(1,0)\}}(z)$. A rough proof of this statement of pole cancellations can be argued by taking the limit to sit on such apparent poles, in the above example $z\rightarrow u$, and then evaluating the contour integral. If finite results are obtained at such poles, it will imply the cancellation of apparent poles among different terms in the $qq$-character. A more precise and convenient way to prove the pole cancellation will be provided in \S\ref{sec:qq-alg}, by using the affine Yangian/quantum toroidal algebra. 

An important consequence of the cancellation of the apparent poles is that the $qq$-characters, as a function of $z$, can only be a Laurent polynomial of $z$. The asymptotic behavior of the $qq$-character of weight $w=1$, $\chi_{w=1}(z)$, at $z\sim 0$ and $z\sim\infty$, are respectively given by
\begin{align}
    &z^{\frac{N}{2}}\chi_{w=1}(z)\sim 1,\quad z\sim 0,\\
    &z^{\frac{N}{2}}\chi_{w=1}(z)\sim z^N,\quad z\sim \infty
\end{align}
which determines $z^{\frac{N}{2}}\chi_{w=1}(z)$, proportional to $\bar{\chi}_{w=1}(z)$, as a polynomial, $T_N(z)$, of $z$ with degree $N$. We thus have 
\begin{equation}
    \bar{\chi}_{w=1}(z)=z^N\langle \cY(zq_1q_2,\boldsymbol{u})\rangle +(-1)^Nq_3^{\frac{N}{2}}\mathfrak{q}\prod_{\alpha=1}^Nu_\alpha\langle \cY^{-1}(z,\boldsymbol{u})\rangle=T_N(z).
\end{equation}

\paragraph{4d limit} It is very straightforward to take the 4d limit of the $qq$-character in parallel to the same limit shown in \eqref{4d-lim} and \eqref{rescal-q}. In this limit, the trigonometric function $\llbracket x\rrbracket$ is reduced to the rational version, 
\begin{equation}
    \llbracket x\rrbracket\rightarrow x,
\end{equation}
and 
\begin{equation}
    \tilde{\cY}_\lambda(z,\boldsymbol{u})\rightarrow R^N \tilde{\cY}^{\text{4d}}_\lambda(\zeta,\boldsymbol{\mathfrak{a}}),
\end{equation}
where 
\begin{equation}
    \tilde{\cY}^{\text{4d}}_\lambda(\zeta,\boldsymbol{\mathfrak{a}})=(\zeta-\mathfrak{a}_\alpha-\e_3)\prod_{x\in\lambda}S^{\text{4d}}(\zeta-\phi_x),
\end{equation}
and 
\begin{equation}
    S^{\text{4d}}(\zeta)=\frac{(\zeta+\epsilon_1)(\zeta+\epsilon_2)}{\zeta(\zeta-\e_3)}.
\end{equation}
It is then clear that a finite $qq$-character can be obtained in the 4d limit as 
\begin{equation}
    \chi_{w=1}(z)\rightarrow R^N\chi^{\text{4d}}_{w=1}(\zeta),
\end{equation}
with the explicit expression given by 
\begin{equation}
    \chi^{\text{4d}}_{w=1}(\zeta)=\langle \tilde{\cY}^{\text{4d}}(\zeta,\boldsymbol{\mathfrak{a}})\rangle +(-1)^N\mathfrak{q}_{\text{4d}}\lt\langle \frac{1}{\tilde{\cY}^{\text{
4d}}(\zeta+\e_3,\boldsymbol{\mathfrak{a}})}\rt\rangle.
\end{equation}

\subsection{\texorpdfstring{$qq$}{qq}-character from affine Yangian}\label{sec:qq-alg}

In the previous subsection, an operator $\cY(M)$ was introduced in a very abstract way, such that its expectation value is given by the function $\cY_{\boldsymbol{\lambda}}(M)$. When considered within the framework of the quantum toroidal algebra or the affine Yangian $\AY$, it is possible to explicitly construct a diagonal operator for $ \cY(M) $. Furthermore, the $ qq $-characters can be deduced from both the action of this foundational algebra and the pole-cancellation condition. In this subsection, we focus on the 4d case following the argument presented in \cite{Bourgine:2015szm}, and a parallel derivation in 5d theories can be found in \cite{Bourgine:2016vsq,Bourgine:2017jsi}.

Using the operator $\Phi(z)$ and the central element $c(z)$ defined in \eqref{def-cPhi}, we define the $\cY$-operator in the spherical degenerate DAHA $\SdH^{\boldsymbol{c}}$ as 
\begin{equation}
    \widehat\cY^{4\text{d}}(z):=\exp\left(c(z)-\Phi(z-1)-\Phi(z+\beta) +\Phi(z)+\Phi(z+\beta-1)\right),\label{eq:YopAY}
\end{equation}
which is related to $\Psi(z)$ defined in \eqref{def-Psi} as
\begin{equation}
\Psi(z)=\frac{ \widehat\cY^{4\text
{d}}(z+1-\beta)}{ \widehat\cY^{4d}(z)} \,.
\end{equation}
From \eqref{deg-Fock}, it is easy to extract out the diagonal action of $ \widehat\cY^{\text{4d}}$ on $|\boldsymbol{\mathfrak{a}},\boldsymbol{\lambda}\rangle^{(\boldsymbol{c})}$
\begin{align}
    \widehat\cY^{\text{4d}}(z)|\boldsymbol{\mathfrak{a}},\boldsymbol{\lambda}\rangle^{(\boldsymbol{c})}=\prod_{j=1}^N\lt(z-(\lambda^{(j)}_1\epsilon_{c-1}+\mathfrak{a}_j)\rt)\prod_{i=1}^\infty\frac{z-(\lambda^{(j)}_{i+1}\epsilon_{c-1}+i\epsilon_{c+1}+\mathfrak{a}_j)}{z-(\lambda^{(j)}_i\epsilon_{c-1}+i\epsilon_{c+1}+\mathfrak{a}_j)}|\boldsymbol{\mathfrak{a}},\boldsymbol{\lambda}\rangle^{(\boldsymbol{c})}\,,\label{eq:YopAYaction}
\end{align}
where we identified $\epsilon_{c+1}=1$ and $\epsilon_{c-1}=-\beta$, and in the remaining of this section, we choose $c=3$ for simplicity. It can be rewritten into a more familiar form 
\begin{equation}
    \widehat\cY^{\text{4d}}(z)|\boldsymbol{\mathfrak{a}},\boldsymbol{\lambda}\rangle^{(\boldsymbol{c})}=\prod_{j=1}^N\frac{\prod_{x\in \frakA(\lambda^{(j)})}(z-\phi_x)}{\prod_{x\in \frakR(\lambda^{(j)})}(z-\phi_x-\epsilon_1-\epsilon_2)}|\boldsymbol{\mathfrak{a}},\boldsymbol{\lambda}\rangle^{(\boldsymbol{c})}.
\end{equation}
The above formula matches with the 4d limit of \eqref{def-Y}. 
One may prove from \eqref{SHD-aY}, \eqref{deg-Fock} and \eqref{DIM-rep} that the Gaiotto (coherent) state (\ref{eq:Gaiotto_state}) satisfies \cite{Bourgine:2015szm},
\bea
D_{-1}(z)|\frakG,\boldsymbol{\fraka}\rangle =& {\beta}^{-1/2} \widehat\cY^{\text{4d}}(z)^{-1}|\frakG,\boldsymbol{\fraka}\rangle\cr
D_1(z) |\frakG,\boldsymbol{\fraka}\rangle =& -{\beta}^{-1/2}{P}_z^- \widehat\cY^{\text{4d}}(z+1-\beta)|\frakG,\boldsymbol{\fraka}\rangle\cr
\langle \frakG,\boldsymbol{\fraka}| D_{-1}(z)=&-{\beta}^{-1/2} \langle \frakG,\boldsymbol{\fraka}|P_z^- \widehat\cY^{\text{4d}}(z+1-\beta)\cr
\langle \frakG,\boldsymbol{\fraka}| D_{1}(z) =&{\beta}^{-1/2} \langle \frakG,\boldsymbol{\fraka}| \widehat\cY^{\text{4d}}(z)^{-1},\label{DYG4}
\eea
where for a given function with the Laurent expansion $f(z)=\sum_{n=-m}^\infty a_nz^{-n}$ (around $z^{-1}\sim 0$), the projection operators $P^{\pm}_z$ are defined as 
\bea
    &P^+_zf(z)=\sum_{n=-m}^0 a_nz^{-n},\cr
    &P^-_z=1-P^+_z,\quad P^-_zf(z)=\sum_{n=1}^\infty a_nz^{-n}.
\eea
We note that the Nekrasov partition function of $\SU(N)$ pure Yang-Mills is written as
\begin{equation}
\cZ^{\mathrm{inst}}(\boldsymbol{\fraka},\mathfrak{q}) = \langle \frakG,\boldsymbol{\fraka}|\mathfrak{q}^{\sfD}|\frakG,\boldsymbol{\fraka}\rangle\,
\end{equation}
where $\sfD=\sfD_{0,1}=\sfL_{0}$ and 
\begin{equation}
    \frakq^\sfD D_{\pm 1}(z)=D_{\pm 1}(z)\frakq^{\sfD\pm 1}.
\end{equation}
This is a derivation of the $\SU(N)$ version of (\ref{eq:partitionGaiottoSU2}) and also the 4d version\footnote{The degree operator $\mathfrak{q}^{\sfD}$ is integrated in the definition of the Gaiotto state in (\ref{eq:partitionGaiottoSU2}) and (\ref{eq:5dGaiottonorm}).} of (\ref{eq:5dGaiottonorm}). 

By evaluating the one-point function in two different ways
\begin{equation}
    \langle \frakG,\boldsymbol{\fraka}|\mathfrak{q}^\sfD \left(D_{1}(z)|\frakG,\boldsymbol{\fraka}\rangle\right) =
    \left(\langle \frakG,\boldsymbol{\fraka}|\mathfrak{q}^\sfD D_{1}(z)\right)|\frakG,\boldsymbol{\fraka}\rangle\,
\end{equation}
one obtains an identity
\begin{equation}
\langle \frakG,\boldsymbol{\fraka}|\mathfrak{q}^\sfD P_z^- \widehat\cY^{\text{4d}}(z+1-\beta)|\frakG,\boldsymbol{\fraka}\rangle
+\langle \frakG,\boldsymbol{\fraka}|\mathfrak{q}^\sfD \mathfrak{q} \widehat\cY^{\text{4d}}(z)^{-1}|\frakG,\boldsymbol{\fraka}\rangle=0
\end{equation}
or, if we write $\langle \bullet \rangle =\langle \frakG,\boldsymbol{\fraka}|\mathfrak{q}^\sfD\bullet|\frakG,\boldsymbol{\fraka}\rangle/\langle \frakG,\boldsymbol{\fraka}|\mathfrak{q}^\sfD|\frakG,\boldsymbol{\fraka}\rangle$
\begin{equation}
    P_z^-\langle \widehat\cY^{\text{4d}}(z+1-\beta) + \mathfrak{q} \widehat\cY^{\text{4d}}(z)^{-1}\rangle =0,\label{pole-can}
\end{equation}
where we use the property $P_z^+\langle \widehat\cY^{\text{4d}}(z)^{-1} \rangle=0$ (and thus $P^-\langle \widehat\cY^{\text{4d}}(z)^{-1} \rangle=\langle \widehat\cY^{\text{4d}}(z)^{-1} \rangle$) as $\lim_{z\rightarrow \infty}\langle \widehat\cY^{\text{4d}}(z)^{-1} \rangle= 0$. 

The equation \eqref{pole-can} suggests that although there are poles respectively in $\bra{\boldsymbol{a},\boldsymbol{\lambda}} \widehat\cY^{\text{4d}}(z)\ket{\boldsymbol{a},\boldsymbol{\lambda}}$ and $\bra{\boldsymbol{a},\boldsymbol{\lambda}} \widehat\cY^{\text{4d}}(z)^{-1}\ket{\boldsymbol{a},\boldsymbol{\lambda}}$, the residues at these apparent poles are cancelled in the $qq$-character $\chi(z)$, 
\begin{equation}
    \chi(z):=\langle \widehat\cY^{\text{4d}}(z+1-\beta) + \mathfrak{q} \widehat\cY^{\text{4d}}(z)^{-1}\rangle=a(z),
\end{equation}
and $a(z)$ is thus a rank $n$ polynomial of $z$, following from the pole-cancellation and the asymptotic behavior 
\begin{equation}
    \lim_{z\rightarrow \infty}\langle z^{-n} \widehat\cY^{\text{4d}}(z) \rangle=1.
\end{equation}

\paragraph{Gaiotto state and $qq$-characters in terms of intertwiners} In 5d gauge theories, the Gaiotto state is alternatively described in terms of the intertwiner as defined in (\ref{eq:qdef-Gaiottostate}). For the pure $\U(N)$ super Yang-Mills theory, we have  $|\frakG,\boldsymbol{v}\rangle^{(N)}$ and ${\small^{(N)}}\langle\frakG,\boldsymbol{v}|$ and the action of the Drinfeld currents are given in (\ref{eq:Gaiottost-action1}) and (\ref{eq:Gaiottost-action2}). Similar to the Gaiotto state in the 4d version, the actions of the Drinfeld currents are described using the operators $\widehat{\mathcal{Y}}^{+}(z),\widehat{\mathcal{Y}}^{-}(z)$ in (\ref{eq:YopQTA}) which are the $q$-analogs of (\ref{eq:YopAY}). The different part is that in the 5d case, we need to introduce two operators expanded both in $z^{-1}$ and $z$ respectively. This is related to the fact that the quantum toroidal $\mathfrak{gl}_{1}$ can be understood as two copies of the affine Yangian $\mathfrak{gl}_{1}$ (or degenerate DAHA) (see the discussions around (\ref{eq:QTAAYmap})).

The $qq$-character of the pure super-Yang-Mills also can be expressed using the generalized intertwiner \cite{Bourgine:2017jsi}. Introduce an operator $\mathcal{T}^{(n,n^{\ast})}_{\U(N)}$ as
\begin{align}
    \mathcal{T}^{(n,n^{\ast})}_{\U(N)}=\Phi^{(n,N)}[u,\boldsymbol{v}]\cdot\Phi^{(n^{\ast},N)\ast}[u^{\ast},\boldsymbol{v}]  ,
\end{align}
then the $qq$-character is defined as 
\begin{align}
    \chi_{\text{5d}}(z)\coloneqq\frac{\nu z^{n^{\ast}+N}}{u^{\ast}q_{3}^{n^{\ast}}}\frac{\left\langle \Delta(E(q_{3}^{-1/2}z))\mathcal{T}^{(n,n^{\ast})}_{\U(N)}\right\rangle}{\left\langle\mathcal{T}^{(n,n^{\ast})}_{\U(N)}\right\rangle},\quad \nu=-\prod_{l=1}^{N}(-q_{3}v_{l})^{-1}.
\end{align}
Using the intertwiner properties (\ref{eq:intertwinerprop}), (\ref{eq:dualinterwinerprop}), and the formulas in appendix \ref{appendix:contraction}, the $qq$-character is rewritten as (see \cite{Bourgine:2017jsi} for the details) 
\begin{align}
    \chi_{\text{5d}}(z)=\left\langle \nu z^{N}\mathcal{Y}(q_{3}^{-1}z,\boldsymbol{v})+\mathfrak{q}\frac{z^{\kappa}}{\mathcal{Y}(z,\boldsymbol{v})}\right\rangle
\end{align}
where $\kappa$ is the Chern-Simons level in (\ref{eq:SYMintertwinerparameter}) and $\langle\bullet\rangle$ is (\ref{eq:instvev}). Up to 5d origin factors, it is exactly the same as the 4d $qq$-character in the degenerate limit. Another way to describe the $qq$-character is to use directly the Gaiotto state and the formal operator $\widehat{\mathcal{Y}}^{\pm}(z)$ as
\begin{equation}
\chi_{\text{5d}}(z)=\frac{\langle \frakG,\boldsymbol{v}|\nu z^{N}\widehat{
\mathcal{Y}}^{\pm}(zq_{3}^{-1})+\mathfrak{q}z^{\kappa}\widehat{\mathcal{Y}}^{\pm}(z)^{-1}|\frakG,\boldsymbol{v}\rangle}{\langle \frakG,\boldsymbol{v}|\frakG,\boldsymbol{v}\rangle}.
\end{equation}
After tuning the Chern-Simons level properly, one can see that this will give a polynomial in $z$ and the analytic properties can be discussed similarly as the 4d case (see \cite{Bourgine:2017jsi} for details).

\subsection{Connection with integrability of gauge theories}\label{sec:NS}

The $qq$-character is deeply related to the integrable aspects of supersymmetric gauge theories. 

As the fundamental $qq$-character of gauge theories of a single gauge node takes the form $\chi\sim \langle \cY\rangle +\langle \cY^{-1}\rangle$, it is expected to reduce to the Seiberg-Witten curve, 
\begin{equation}
    y(z)+\frac{P_1(z)}{y(z)}=P_0(z),\label{SW-curve}
\end{equation}
in the classical limit, $\epsilon_1,\epsilon_2\rightarrow 0$, where $P_1(z)$ is a degree-$N_f$ polynomial of $z$ and $P_0(z)$ is a degree-$N$ polynomial. 

In the classical limit, the instanton partition function is dominated by its saddle-point contribution. At the saddle-point, adding or removing a box from the saddle configuration $\boldsymbol{\lambda}_\ast$ does not change the partition function, i.e. 
\begin{equation}
    \frac{Z_{\rm vect.}(\boldsymbol{u},\boldsymbol{\lambda}_\ast+x)}{Z_{\rm vect.}(\boldsymbol{u},\boldsymbol{\lambda}_\ast)}=1,
\end{equation}
for ${}^\forall x\in \frakA(\boldsymbol{\lambda})$. Then the expectation value of any operator $\cO_{\boldsymbol{\lambda}}$ is given by its value at the saddle-point, 
\begin{equation}
    \langle \cO\rangle=\sum_{\boldsymbol{\lambda}}\frakq^{|\boldsymbol{\lambda}|}\cO_{\boldsymbol{\lambda}}Z_{\rm vect.}(\boldsymbol{u},\boldsymbol{\lambda})\simeq \frakq^{|\boldsymbol{\lambda}_\ast|}\cO_{\boldsymbol{\lambda}_\ast}Z_{\rm vect.}(\boldsymbol{u},\boldsymbol{\lambda}_\ast).
\end{equation}
The $qq$-character reduces to the form 
\begin{equation}
    \chi\rightarrow \frakq^{|\boldsymbol{\lambda}_\ast|}Z_{\rm vect.}(\boldsymbol{u},\boldsymbol{\lambda}_\ast)\lt(\cY^{4d}_{\boldsymbol{\lambda}_\ast}+\frac{(-1)^N\frakq}{\cY^{4d}_{\boldsymbol{\lambda}_\ast}}\rt),
\end{equation}
in the classical limit, and it thus reproduces the expected Seiberg-Witten curve. 

The evaluation of the saddle-point can be done more easily first in the so-called Nekrasov-Shatashvili limit (e.g. refer to \cite{Bourgine:2014tpa}), where one $\Omega$-background parameter is first taken to be zero with the other fixed to be finite. In general, the saddle-point equation is given by 
\begin{align}
\frac{Z^{\fraksu(N)}_{\vec{\lambda}^\ast+x}}{Z^{\fraksu(N)}_{\vec{\lambda}^\ast}}
=&-\frac{\frakq}{\epsilon_1\epsilon_2}\prod_{j=1}^N\frac{\prod_{y\in \frakR(\lambda^{\ast(j)})}(\phi_x-\phi_y-\epsilon_1-\epsilon_2)(\phi_x-\phi_y)}{\prod_{\substack{y\in \frakA(\lambda^{\ast(j)})\\y\neq x}}(\phi_x-\phi_y)(\phi_x-\phi_y+\epsilon_1+\epsilon_2)}.
\end{align}
Let us first set $\epsilon_2\rightarrow 0$, which is also known as the Nekrasov-Shatashvili limit (or NS limit), then the contributions from boxes in the same row are equal and thus cancel with each other (see Figure \ref{f:Young-NS}). One way to avoid the cancellation is to have the difference between the number of boxes in each of two rows of order $\epsilon^{-1}$, i.e. the size of the saddle-point equation tends to be infinite in the NS limit. Instead, we consider Young diagrams with all rows distinct and the number of boxes in each row at the order ${\cal O}(\epsilon_2^{-1})$ (such configurations have been justified in e.g. \cite{Bourgine:2014tpa}). As shown in Figure \ref{f:saddle-NS}, a new type of cancellation occurs, between neighbor addable and removable boxes (due to the fact that each row is distinct and the difference is of order $\epsilon_2^{-1}$). The remaining contributions in the saddle-point equation are 
\begin{equation}
    \frakq\prod_{\substack{j=1\\j\neq i}}^M\frac{u_i-u_j-\epsilon_1}{u_i-u_j+\epsilon_1}\prod_{a=1}^N\frac{1}{(u_i-\xi_a)(u_i-\xi_a+\epsilon_1)}=1,\label{saddle-eq}
\end{equation}
where $\xi_a=\fraka_a+n_a\epsilon_1$ with $n_a$ denoting the number of rows in the $a$-th Young diagram $\lambda^{\ast(a)}$, and the set of variables $\{u_i\}_{i=1}^{M}$ relabeling the set $\{\lim_{\epsilon_2\rightarrow 0}\phi_y\}_{y\in\frakA(\boldsymbol{\lambda}^{\ast})}$ with $M=\sum_{a=1}^Nn_a$. If we take the classical limit $\epsilon_1\rightarrow 0$, the saddle-point configuration $\boldsymbol{\lambda}^\ast$ requires an infinite number of rows. To describe such a configuration, a profile function must be introduced \cite{Nekrasov:2003rj}. The technique for using the profile function is rather sophisticated, and we will refer interested readers to the literature for further details (see, for example, an accessible review and extension to gauge theories with other gauge groups in \cite{Li:2021rqr}).

The connection between the Seiberg-Witten theory and classical integrable models such as the Calogero-Moser and Hitchin systems are well-known \cite{Gorsky:1995zq,Martinec-Warner,Donagi:1995cf} and we refer to the literature for more details on this vast topic. Yet there is another interesting link with {\it quantum} integrable models that appear in the NS limit \cite{Nekrasov:2009uh,Nekrasov:2009ui,Nekrasov:2009rc}. As one can see from the saddle-point equation \eqref{saddle-eq}, it resembles the Bethe ansatz equation of Heisenberg XXX spin chain with twisted periodic boundary condition (parameterized by the twisting angle $\theta$) very much, 
\begin{equation}
    e^{i\theta}\prod_{a=1}^L\frac{\lambda_j+is_a}{\lambda_j-is_a}=-\prod_{k=1}^M\frac{\lambda_j-\lambda_k+i}{\lambda_j-\lambda_k-i},\quad j=1,2,\dots,M,\label{BAE-A1}
\end{equation}
by identifying $\lambda_j\equiv i u_j/\epsilon_1$, and the number of sites $L$ in the spin chain identified to $N$. $s_a$ here denotes the spin of the representation used in the $a$-th site of the spin chain. One can further add matter hypermultiplets to the gauge theory to modify the saddle-point equation to 
\begin{equation}
    \frakq\prod_{\substack{j=1\\j\neq i}}^M\frac{u_i-u_j-\epsilon_1}{u_i-u_j+\epsilon_1}\frac{\prod_{f=1}^{N_f}(u_i-m_f)}{\prod_{a=1}^N(u_i-\xi_a)(u_i-\xi_a+\epsilon_1)}=1.\label{eq:BAE-saddle}
\end{equation}
The Bethe ansatz equation of the spin chain \eqref{BAE-A1} can be realized in this context by properly choosing the mass parameters.

\begin{figure}[ht]\centering
\includegraphics{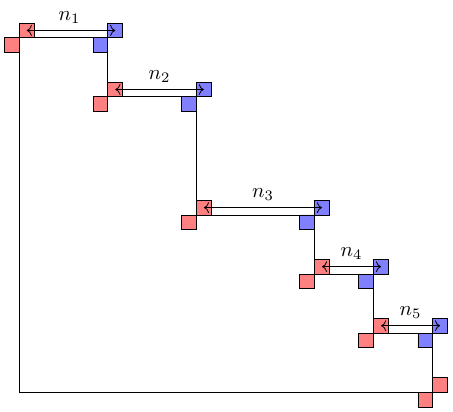}
  \caption{Blue and red boxes respectively give the contributions to the numerator and the denominator in the saddle-point equation. In the NS limit, if $n_i\ll \epsilon_2^{-1}$ then the contributions from the red and blue boxes in the same row (paired by the corresponding arrow) cancel with each other. }
  \label{f:Young-NS}
\end{figure}

\begin{figure}[ht]\centering
\includegraphics[width=0.95\textwidth]{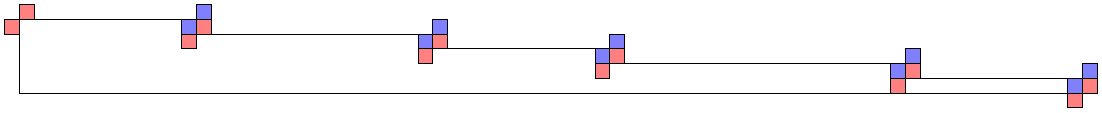}
  \caption{A Young diagram with infinitely long rows in the NS limit. }
  \label{f:saddle-NS}
\end{figure}

The $qq$-character in this context reproduces the $TQ$-relation in the spin chain\footnote{In the NS limit, the $qq$-character also reduces to the $q$-character of Frenkel-Reshetikhin \cite{frenkel1999algebras}.}. The $\cY$-operator at the saddle-point takes the form 
\begin{equation}
    \cY^{\text{4d}}_{\boldsymbol{\lambda}^\ast}(z)=\prod_{a=1}^N(z-\xi_a)\prod_{j=1}^M\frac{z-u_j}{z-u_j-\epsilon_1}.
\end{equation}
The normalized $qq$-character then becomes 
\begin{equation}
    T(z):=\frac{\chi^{\text{4d}}(z)}{\langle 1\rangle}=\prod_{a=1}^N(z-\xi_a+\epsilon_1)\prod_{j=1}^M\frac{z-u_j+\epsilon_1}{z-u_j}+\frakq \frac{\prod_{f=1}^{N_f}(z-m_f)}{\prod_{a=1}^N(z-\xi_a)}\prod_{j=1}^M\frac{z-u_j-\epsilon_1}{z-u_j},
\end{equation}
and by defining the $Q$-function as 
\begin{equation}
    Q(u)=\prod_{j=1}^M(u-u_j),
\end{equation}
we obtain the $TQ$-relation 
\begin{equation}
    T(u)Q(u)=\lt(\prod_{a=1}^N(z-\xi_a+\epsilon_1)\rt)Q(u+\epsilon_1)+\frakq \lt(\frac{\prod_{f=1}^{N_f}(z-m_f)}{\prod_{a=1}^N(z-\xi_a)}\rt)Q(u-\epsilon_1).
\end{equation}
When the mass parameters are chosen appropriately, the $TQ$-relation for the XXX spin chain is reproduced. By evaluating the $TQ$-relation of the zeros of the $Q$-function, which are the Bethe roots $u=u_i$, we obtain the Bethe ansatz equation \eqref{eq:BAE-saddle}. It is worth noting that the spin at each site of the spin chain, denoted by $s_a\sim n_a$, is closely related to the number of magnons $M=\sum_a n_a$ in the current map. To consider spin chains with a more general setup, one can study the 4d theory at the root of its Higgs branch. At this point, there exist vortex strings, which are realized as D2-branes stretching between D4 and NS5 branes \cite{Hanany:2004ea,Eto:2004rz,Fujimori:2008ee}. The effective worldsheet theory on the vortex strings is a 2d $\cN=(2,2)$ gauge theory, and its vortex partition function can be obtained from the 4d instanton partition function at the root of the Higgs branch by shifting the Coulomb branch parameters accordingly \cite{Fujimori:2015zaa}. The vacua equation of such a 2d theory is given by 
\begin{equation}
    (-1)^{N_f+N'_f}e^{\zeta}\prod_{j\neq i}^k\frac{\sigma_i-\sigma_j+m_{adj}}{\sigma_i-\sigma_j-m_{adj}}\frac{\prod_{a=1}^{N'_f}\sigma_i-\bar{m}_a}{\prod_{b=1}^{N_f}\sigma_i+m_b}=1,
\end{equation}
where there exists $N_f$ fundamental and $N'_f$ anti-fundamental chiral multiplets with mass parameters $\{m_a\}_{a=1}^{N_f}$ and $\{\bar{m}_a\}_{a=1}^{N'_f}$ and an adjoint chiral multiplet with mass $m_{adj}$ in the U($k$) gauge theory with FI parameter $\zeta$ under consideration. For 2d $\cN=(2,2)^\ast$ theories with $N_f=N'_f$, the vacua equation is mapped to the Bethe ansatz equation \eqref{BAE-A1} under the identification $L=N_f=N'_f$, $i\theta=\zeta$, and $M=k$. This states the famous Bethe/Gauge correspondence first proposed in \cite{Nekrasov:2009ui,Nekrasov:2009rc}. 

The $TQ$-relation interestingly can further be rewritten into a difference equation form, schematically the $TQ$-relation becomes 
\begin{equation}
   \lt(\tilde{P}_1(z)e^{\epsilon_1\partial_z}-T(z)+\frakq \tilde{P}_2(z)e^{-\epsilon_1\partial_z}\rt)Q(z)=0.\label{QDE-TQ}
\end{equation}
By denoting $\hat{y}:=e^{-\epsilon_1\partial_z}$, the commutation relation $\lt[z,\log\hat{y}\rt]=\epsilon_1$ is just the canonical quantization condition and $Q(z)$ plays the role of the wavefunction. One can alternatively realize the quantization by replacing $z$ as a differential operator over $y:=\frakq \tilde{P}_2(z)\hat{y}$ (the commutation relation between $z$ and $y$ remains to be the same, $\lt[z,\log y\rt]=\epsilon_1$) and viewing $y$ as the variable, 
\begin{equation}
    z=\epsilon_1 y\partial_{y}.
\end{equation}
This replacement makes the difference equation \eqref{QDE-TQ} an ordinary differential equation. Note that $y^{-1}=\hat{y}^{-1}\frakq^{-1}\tilde{P}(z)^{-1}=\frakq^{-1}\tilde{P}(z+\epsilon_1)^{-1}\hat{y}^{-1}$, so in terms of $y$, \eqref{QDE-TQ} becomes 
\begin{equation}
    \lt(y-T(z)+\frakq \tilde{P}_1(z)\tilde{P}(z+\epsilon_1)y^{-1}\rt)Q=0.\label{oper-eq}
\end{equation}
As an example we consider the pure SU(2) gauge theory, in which $\tilde{P}_1(z)\tilde{P}_2(z+\epsilon_1)=1$ and $T(z)=z^2-u$. One can further put $y=-e^{-ix}$, $z=-i\epsilon_1\partial_x$ to rewrite it to the well-known Mathieu equation 
\begin{equation}
   \lt[\epsilon_1^2\partial_x^2-\lt(e^{-ix}-u+\mathfrak{q}e^{ix}\rt)\rt]Q=0.
\end{equation}
We refer to \cite{Tai:2010ps,Maruyoshi:2010iu,Zenkevich:2011zx} for the study of such differential equations and the reproduction of the prepotential in the NS limit starting from these differential equations.

In the classical limit $\epsilon_1\rightarrow 0$, the differential operator $\hat{y}$ becomes a pure function, and we again rescale $\hat{y}$ to define $\hat{y}\frakq \tilde{P}_2(z)\rightarrow y(z)$, then we obtain 
\begin{equation}
    y+\frac{\frakq \tilde{P}_1(z)\tilde{P}_2(z)}{y}=T(z),
\end{equation}
which reproduces the Seiberg-Witten curve \eqref{SW-curve} with the identification $P_1(z)=\frakq \tilde{P}_1(z)\tilde{P}_2(z)$. 

We briefly remark that the relation is not yet completely clear between integrable systems mentioned in this section originated from the $qq$-characters and those described in \S\ref{s:int-model} formulated directly from the universal $\scR$-matrix of the underlying quantum toroidal algebra. This remains to be an interesting open question to work on at the stage of writing this review.

\section{Discussion and outlook}\label{sec:final}

While this note provides an overview of certain aspects of recent developments in quantum toroidal algebras, it is important to acknowledge that this note is by no means exhaustive, and many approaches have been left untouched. 
In the main text, we have briefly discussed some related topics and future directions. However, various subjects still deserve exploration and were not addressed in this note.

\paragraph{Hilbert schemes of points}

Nakajima constructed a geometric action of the $q$-Heisenberg algebra on the equivariant K-theory of the Hilbert schemes of points on a plane \cite{nakajima1997heisenberg}. Inspired by this work, the quantum toroidal $\frakgl_1$ is geometrically constructed by convolution algebra in equivariant  K-theory of the Hilbert schemes of points on a plane \cite{feigin2011equivariant,Schiffmann:2009aa}. Since \QTA~ is realized by the vertex operators written in terms of the $q$-Heisenberg algebra through the horizontal representation as in \S\ref{sec:horizontalrep}, it is natural to expect that \QTA~ admit such a geometric construction. Here we will provide a concise overview of this geometric construction.

Let us consider the Hilbert scheme of $n$-point on a plane $\bC^2$
\begin{equation}
\Hilb_n:=\left\{I \subset \mathbb{C}[x, y] \mid \dim _{\mathbb{C}} \mathbb{C}[x, y] / I=n\right\}
\end{equation}
This can be interpreted as the U(1) $n$-instanton moduli space. 
It receives $T=(\bC^\times)^2$ action as
\be 
T \ \rotatebox[origin=c]{-90}{$\circlearrowright$} \ \Hilb_n(\bC^2); \qquad (z_1,z_2)\cdot I =\{f(z_1^{-1}x,z_2^{-1}y)\mid f\in I\}
\ee 
The character of the torus action $T=(\bC^\times)^2$ is $(q_1,q_2)$. 

Let $\Theta_n \subset \operatorname{Hilb}_n \times \mathbb{C}^2$ be the universal family:
\begin{equation}
\Theta_n:=\left\{\begin{array}{l|l}
(I_Z,Z)\in \operatorname{Hilb}_n \times \mathbb{C}^2 & \begin{array}{ll}
Z \subset \mathbb{C}^2 & \text{: finite subscheme of length } n \\
I_z & \text{: ideal of } Z
\end{array}
\end{array}\right\}
\end{equation}
The tautological bundle of $\Hilb_n$ is the push-forward $\tau_n=p_*(\cO_{\Theta_n})$ of the structure sheaf on $\Theta_n$ under  the projection $p$ : $\Hilb_n \times \bC^2 \rightarrow \Hilb_n$. The fiber of $\tau_n$ at $I$ is $\mathbb{C}[x,y]/I$.

Consider the nested Hilbert schemes
\begin{align}
\Hilb_{n, n+k}:=\{(I, J) \in \Hilb_n \times \Hilb_{n+k} \mid J \subset I \subset \mathbb{C}[x, y]\} ~, \cr 
\Hilb_{n+k,n}:=\{(I, J) \in \Hilb_{n+k} \times \Hilb_{n} \mid I \subset J \subset \mathbb{C}[x, y]\} ~.
\end{align}
In fact, $\Hilb_{n, n\pm1}$ is smooth. 
Now let us consider the following diagram
\be 
\begin{tikzpicture}[x=0.75pt,y=0.75pt,yscale=-1,xscale=1]
\draw    (259.5,92) -- (187.23,134) ;
\draw [shift={(185.5,135)}, rotate = 329.84] [color={rgb, 255:red, 0; green, 0; blue, 0 }  ][line width=0.75]    (10.93,-3.29) .. controls (6.95,-1.4) and (3.31,-0.3) .. (0,0) .. controls (3.31,0.3) and (6.95,1.4) .. (10.93,3.29)   ;
\draw    (328.5,93) -- (395.75,130.04) ;
\draw [shift={(397.5,131)}, rotate = 208.84] [color={rgb, 255:red, 0; green, 0; blue, 0 }  ][line width=0.75]    (10.93,-3.29) .. controls (6.95,-1.4) and (3.31,-0.3) .. (0,0) .. controls (3.31,0.3) and (6.95,1.4) .. (10.93,3.29)   ;

\draw    (160.5,112) -- (160.5,132) ;
\draw [shift={(160.5,134)}, rotate = 270] [color={rgb, 255:red, 0; green, 0; blue, 0 }  ][line width=0.75]    (10.93,-3.29) .. controls (6.95,-1.4) and (3.31,-0.3) .. (0,0) .. controls (3.31,0.3) and (6.95,1.4) .. (10.93,3.29)   ;
\draw    (428.5,109) -- (428.5,129) ;
\draw [shift={(428.5,131)}, rotate = 270] [color={rgb, 255:red, 0; green, 0; blue, 0 }  ][line width=0.75]    (10.93,-3.29) .. controls (6.95,-1.4) and (3.31,-0.3) .. (0,0) .. controls (3.31,0.3) and (6.95,1.4) .. (10.93,3.29)   ;
\draw (262,72.4) node [anchor=north west][inner sep=0.75pt]    {$\Hilb_{n,n+1}$};
\draw (143,140.4) node [anchor=north west][inner sep=0.75pt]    {$\Hilb_{n}$};
\draw (407,136.4) node [anchor=north west][inner sep=0.75pt]    {$\Hilb_{n+1}$};
\draw (152,95) node [anchor=north west][inner sep=0.75pt]    {$\tau _{n}$};
\draw (417,95) node [anchor=north west][inner sep=0.75pt]    {$\tau _{n+1}$};
\draw (207,97) node [anchor=north west][inner sep=0.75pt]    {$\pi _{1}$};
\draw (359,97) node [anchor=north west][inner sep=0.75pt]    {$\pi _{2}$};
\end{tikzpicture}\ee
Then, we can construct the surjection map 
\begin{equation}
\phi: \pi_2^*\left(\tau_{n+1}\right) \twoheadrightarrow \pi_1^*\left(\tau_n\right); \qquad \phi|_{(I,J)}:\mathbb{C}[x,y]/J   \twoheadrightarrow \mathbb{C}[x,y]/I~.
\end{equation}
We define the line bundle of $\Hilb_{n,n+1}$ by $\tau_{n,n+1}:=\ker \phi$, and the definition of $\tau_{n+1,n}$ is given in a similar manner. We can also define the tautological bundle $\tau_{n,n}$ of $\Hilb_{n,n}$, which is the diagonal of $\Hilb_n \times \Hilb_n$, by $\tau_{n,n}=\pi_1^*(\tau_{n})=\pi_2^*(\tau_{n})$.

These sheaves $\tau_{n\pm1,n}$, $\tau_{n,n}$ can be considered as objects in the category  of $T$-equivariant sheaves on $\Hilb_{m}\times \Hilb_{n}$ ($m=n\pm1$, or $m=n$) under the inclusion. Taking the Grothendieck group 
\be \Coh^T(X)\to K^T(X);\ \cF \mapsto [\cF]~, \ee 
one can introduce the convolution product
\bea 
*:K^T(\Hilb_{\ell}\times \Hilb_{m})\otimes K^T(\Hilb_{m}\times \Hilb_{n}) &\to K^T(\Hilb_{\ell}\times \Hilb_{n})\cr 
(x,y) &\mapsto p_{\ell n*}(p_{\ell m}^*(x) \otimes p_{mn}^*(y))
\eea 
where $p_{ij}:\Hilb_{\ell}\times \Hilb_{m}\times \Hilb_{n}\to \Hilb_{i}\times \Hilb_{j}$ is the projection. Let $\tau_{m,n}^{-1}$ be the dual sheaf of $\tau_{m,n}$. We consider an algebra generated by the following elements
\begin{equation}
\begin{aligned}
&\sfE_k=q_2^{-1/2}\prod_n [({\tau}_{n+1,n})^{ {k+1}}]~,  \\
&\sfF_k=-q_1^{-1/2}\prod_n [({\tau}_{n+1,n})^{{k}}]~, \qquad k\in \bZ ~,\\
&\sfK_{0, l}=\prod_n \wedge^l [\tau_{n, n}], \quad \sfK_{0,-l}=\prod_n \wedge^l [\tau_{n, n}^{-1}], \quad l \in \mathbb{Z}_{>0} ~.
\end{aligned}
\end{equation}
Then, it was shown in \cite{feigin2011equivariant,Schiffmann:2009aa} that the algebra is isomorphic to $\QTA$ and these generators are the modes of the Drinfeld currents in \eqref{Drinfeld}. Consequently, the algebra naturally acts on the equivariant $K$-group on $\oplus_{n\ge0} \Hilb_{n}$ through
\be
*:K^T(\Hilb_{n}\times \Hilb_{m})\otimes  K^T(\Hilb_{m}) \to K^T(\Hilb_{n})~.
\ee

\paragraph{Elliptic Hall algebras}
Another approach for the geometric construction of quantum toroidal algebras is through the framework of the Hall algebras
\cite{burban2012hall,schiffmann2012hall,schiffmann2012drinfeld,Kapranov:2012aa,yanagida2015quantum}. 
This framework offers a systematic way to understand these algebras in relation to coherent sheaves over algebraic curves. This perspective provides a bridge between representation theory and algebraic geometry. Importantly, it introduces a strategy for the categorification of quantum groups.  For a more in-depth exploration of this intricate topic, we recommend  \cite{schiffmann2006lectures} as a comprehensive review.

Let $\scA$ be an (essentially small) Abelian category where, for any pair of objects $A, B\in \mathrm{Ob}(\scA)$, the sets $\Hom_{\scA}(A, B)$ and $\Ext_{\scA}^1(A, B)$ are finite. ($\scA$ is finitary.) We represent the cardinality of these sets as $|\Hom_{\scA}(A, B)|<\infty$ and $|\Ext_{\scA}^1(A, B)|<\infty$, respectively. We also assume that the second extension $\operatorname{Ext}_{\scA}^2$ always vanishes. ($\scA$ is hereditary.)

Let us consider a set $\operatorname{Iso}(\scA)$ of the isomorphism classes of $\scA$, and we write the isomorphism class of $A\in \mathrm{Ob}(\scA)$ by $[A]$. One can introduce an algebra structure to $\operatorname{Iso}(\scA)$ by defining a multiplication as follows \cite{ringel1990hall} :
\begin{equation}
[A] * [B]:=\langle A, B\rangle\sum_{[C] \in \operatorname{Iso}(\scA)} g_{A, B}^C[C]
\end{equation}
Here, the structure coefficient $g_{A, B}^C$ is defined by the cardinality of the extensions up to isomorphism \be\Ext_{\scA}^1(A, B)_C:=\{0 \to B \to C \to A \to 0 \}/\sim\ee
divided by the orders of the automorphism groups $\operatorname{Aut}(A)$ and $\operatorname{Aut}(B)$  in $\scA$:
\begin{equation}
g_{A, B}^C=\frac{|\Ext_{\scA}^1(A, B)_C|}{|\operatorname{Aut}(A)| \cdot|\operatorname{Aut}(B)|}~.
\end{equation}
The prefactor is the multiplicative Euler form defined by
\begin{equation}
\langle A, B\rangle:=q^{\frac{1}{2}(|\Ext_{\scA}^0(A, B)|-|\Ext_{\scA}^1(A, B)|)}~.
\end{equation}
The resulting algebra is called the \emph{Hall algebra}, denoted by $H_\scA$. The Hall algebra $H_\scA$ is indeed endowed with a Hopf algebra structure as follows: the coproduct $\Delta: H_\scA \to H_\scA \otimes H_\scA$ and the counit $\varepsilon: H_\scA \to \mathbb{C}$ are defined \cite{green1995hall} by
\begin{equation}
\Delta([A]):=\sum_{[B],[C] \in \operatorname{Iso}(\scA)}\langle B, C\rangle \frac{|\Ext_{\scA}^1(B,C)_A|}{|\operatorname{Aut}(A)|}[B] \otimes[C], \qquad \varepsilon([A])=\delta_{A, 0}~.
\end{equation}
 The antipode map $S: H_\scA \to H_\scA$ is given by the following formula \cite{xiao1997drinfeld}
\begin{align}
S([A]):=&|\operatorname{Aut}(A)|^{-1} \sum_{r \geq 1}(-1)^r \sum_{A \cdot \in \mathcal{F}(A ; r)}(\prod_{i=1}^r\langle A_i / A_{i+1}, A_{i+1}\rangle|\operatorname{Aut}(A_i / A_{i+1})|)\cr 
&\qquad \qquad \cdot\left[A_1 / A_2\right] *\left[A_2 / A_3\right] * \cdots *\left[A_r\right] ~,
\end{align}
where $\mathcal{F}(A ; r)$ is the set of genuine filtrations of $A$ with length $r$
$$
\mathcal{F}(A ; r):=\{A=A_1 \supsetneq A_2 \supsetneq \cdots \supsetneq A_r \supsetneq 0\} .
$$

Now, consider the category $\Coh(X)$ of coherent sheaves on an elliptic curve $X$ over a finite field $\bF_q$.
If we choose $\Coh(X)$ as $\scA$, then the Hall algebra $H_{\Coh(X)}$ is generally too big as a vector space, and it is difficult to compute all of its structure constants. Therefore, we focus on the full subcategory $\Coh^{ss}(X)$ consisting of all semi-stable sheaves in $\Coh(X)$. In particular, we consider the isomorphism classes in $\Coh^{ss}_{(r,d)}(X)$ consisting of all semi-stable sheaves of rank $r$ and degree $d$, and we define elements
\begin{equation}
1_{(r, d)}^{\mathrm{ss}}:=\sum_{[\mathcal{F}] \in \operatorname{Iso}(\Coh_{(r, d)}^{ss}(X))}[\mathcal{F}] .
\end{equation}
Then, the subalgebra of $H_{\Coh(X)}$ generated by 
\be \{1_{(r, d)}^{\mathrm{ss}} \mid (r=1, d\in \bZ) \textrm{ or }  (r=0, d\in \bN)\} \ee
is called the spherical Hall algebra of $ X $. The spherical Hall algebra of $X$ indeed corresponds to the upper part $\cE_\ge$ (or $\cE_\ge^\perp$ in the $S$-dual frame) in \eqref{triangular} of the quantum toroidal $\frakgl_1$. The central extension of the quantum double of the spherical Hall algebra is called the \emph{elliptic Hall algebra}, and it is isomorphic to the quantum toroidal  $\frakgl_1$. Roughly speaking, the lattice in \eqref{eq:DIMsubalgebra} can be understood as the rank and degree of semi-stable sheaves. The generators and relations are explicitly given in \cite{burban2012hall}. Certainly, the explanation here is very heuristic, and readers are directed to the original paper \cite{burban2012hall} for the details.  Note that the $\SL(2,\bZ)$ duality \eqref{SL2Z} of the quantum toroidal  $\frakgl_1$ is realized as the Fourier-Mukai transformation in $\Coh(X)$ in this context.

\paragraph{Shuffle algebras} Quantum toroidal algebras can be formulated within the framework of the shuffle algebras \cite{feigin1998vector}. Notably, the shuffle approach proves to be particularly effective in the natural geometric correspondences seen above. In \cite{feigin2009commutative}, the elliptic deformation of quantum toroidal algebras and mutually commuting Hamiltonians are discussed from the viewpoint of shuffle algebras. The study of quantum toroidal algebras from the perspective of shuffle algebra and its applications in geometric contexts are deepened in a series of works by Negu\c{t} (for instance, see \cite{negut2014shuffle,Negut:2013cz,Negut:2016dxr,Negut:2020npc}).

A comprehensive review of quantum toroidal algebras, approached from the perspective of shuffle algebras, is provided by Tsymbaliuk \cite{Tsymbaliuk:2022bqx}. The present note serves as a complementary viewpoint to the aforementioned review \cite{Tsymbaliuk:2022bqx}. Therefore, it could be insightful to examine and contrast both sets of the notes.

We would like to emphasize that there remains substantial scope for further exploration into the above geometric constructions of quantum toroidal algebras from a physics perspective. It is imperative to understand these aspects from a physics standpoint in a comprehensive manner.

\paragraph{Quiver $\cW$-algebra and BPS/CFT correspondence}
In \S\ref{sec:intertwiner} and \S\ref{sec:algebraic_topvertex}, we introduced the algebraic intertwiners and showed that compositions of them give the instanton partition functions of linear quiver gauge theories. Moreover, we also showed that the intertwiners are related to the screening currents in (\ref{eq:interwtiner-screening}). This correspondence is a consequence of the AGT or BPS/CFT correspondence. Actually, as mentioned below \eqref{eq:interwtiner-screening}, there is another way to derive the instanton partition functions using directly the screening currents of deformed $\mathcal{W}$-algebras associated with a quiver structure \cite{Kimura:2015rgi,Kimura:2016dys,Kimura:2017hez,Kimura:2019xzj,Kimura:2019hnw,Kimura:2022zsm} (see \cite{Kimura-review} for a nice review). Consider a quiver $Q=(Q_{0},Q_{1})$ whose quiver Cartan matrix is given as 
\begin{align}
c_{ij}=(1+q_{1}^{-1}q_{2}^{-1})\delta_{ij}-\sum_{e:i\rightarrow j}\mu_{e}^{-1}-\sum_{e:j\rightarrow i}\mu_{e}q_{1}^{-1}q_{2}^{-1}.
\end{align}
Note $Q_{0}$ is the set of nodes and $Q_{1}$ is the set of edges. This Cartan matrix corresponds to the quiver structure of the associated quiver theory with eight supercharges. The parameters $\mu_{e}$ are identified with the bifundamental masses of the theory. Denoting the degree-$n$ Cartan matrix as $c_{ij}^{[n]}$ which is defined as 
\begin{align}
c_{ij}^{[n]}=c_{ij}|_{(q_{1},q_{2},\mu_{e})\rightarrow (q_{1}^{n},q_{2}^{n},\mu_{e}^{n})},
\end{align}
we introduce the screening currents as 
\begin{align}
    \mathsf{S}_{i}(x)=:x^{\mathsf{s}_{i,0}}e^{\tilde{\mathsf{s}}_{i,0}}\exp\left(\sum_{n\in\mathbb{Z}_{\neq 0}}\mathsf{s}_{i,n}x^{-n}\right):
\end{align}
where 
\begin{align}
\begin{split}
    &[\mathsf{s}_{i,n},\mathsf{s}_{j,m}]=-\frac{1}{n}\frac{1-q_{1}^{n}}{1-q_{2}^{-n}}c_{ji}^{[n]}\delta_{n+m,0}\quad (n\geq 1),\\ &[\tilde{\mathsf{s}}_{i,0},\mathsf{s}_{j,n}]=-\beta^{-1} c_{ji}^{[0]}\delta_{n,0},\quad \beta=-\frac{\epsilon_{2}}{\epsilon_{1}}.
\end{split}
\end{align}
Then, the partition function is described as 
\begin{align}
    \mathcal{Z}=\bra{0}\prod_{x\in\mathcal{X}}^{\succ}\mathsf{S}_{\text{i}(x)}(x)\ket{0}
\end{align}
where 
\begin{align}
    \mathcal{X}=\bigsqcup\limits_{i\in Q_{0}}\mathcal{X}_{i},\quad \mathcal{X}_{i}=\left\{x_{i,\alpha,k}=u_{i,\alpha}q_{1}^{k-1}q_{2}^{\lambda_{i,\alpha,k}}\right\}_{\substack{i\in Q_{0}\\\alpha=1,\ldots,n_{i}\\k=1,\ldots,\infty}}
\end{align}
and $\succ$ is some ordering in the $x$-variables. The $n_{i}$ here is a positive integer and $\lambda_{i,\alpha}$ is a Young diagram $\lambda_{i,\alpha}=\{\lambda_{i,\alpha,1},\ldots,\lambda_{i,\alpha,k},\ldots,\}$.

The map $\text{i}:\mathcal{X}\rightarrow Q_{0}$ is defined as 
\begin{align}
    x\in\mathcal{X}_{i}\Longleftrightarrow \text{i}(x)=i.
\end{align}
The partition function will match with the partition function of a quiver gauge theory $\prod_{i\in Q_{0}}\U(n_{i})$ whose quiver structure is given by $Q$.

Given the screening currents, one can define a deformed $\mathcal{W}$-algebra called the quiver $\mathcal{W}$-algebra. To do this, we need to introduce two operators called the $\mathsf{Y}$-operator and the $\mathsf{A}$-operators (root currents):
\begin{align}
\begin{split}
    &\mathsf{Y}_{i}(x)=q_{1}^{\tilde{\rho}_{i}}:e^{\mathsf{y}_{i,0}}\exp\left(\sum_{n\neq 0}\mathsf{y}_{i,n}x^{-n}\right):,\quad [\mathsf{y}_{i,n},\mathsf{y}_{j,m}]=-\frac{1}{n}(1-q_{1}^{n})(1-q_{2}^{n})\tilde{c}_{ji}^{[-n]}\delta_{n+m,0},\\
    &\mathsf{A}_{i}(x)=q_{1}:e^{\mathsf{a}_{i,0}}\exp\left(\sum_{n\neq 0}\mathsf{a}_{i,n}x^{-n}\right):,\quad \mathsf{a}_{i,n}=\sum_{j\in Q_{0}}\mathsf{y}_{j,n}c_{ji}^{[n]},\quad \tilde{\rho}_{i}=\sum_{j\in Q_{0}}\tilde{c}_{ji}^{[0]}
\end{split}
\end{align}
where the relation with the screening currents are
\begin{align}
[\mathsf{y}_{i,n},\mathsf{s}_{j,m}]=-\frac{1}{n}(1-q_{1}^{n})\delta_{ij}\delta_{m+n,0},\quad [\tilde{\mathsf{s}}_{i,0},\mathsf{y}_{j,n}]=-\log q_{1}\delta_{i,j}\delta_{n,0}
\end{align}
and $\tilde{c}_{ij}^{[n]}$ is the inverse quiver Cartan matrix. Using these operators we then can construct the $\mathsf{T}$-operators which are the generators of the quiver $\mathcal{W}$-algebra:
\begin{align}
    \mathsf{T}_{i}(x)=\mathsf{Y}_{i}(x)+:\mathsf{Y}_{i}(x)\mathsf{A}_{i}(q_{1}^{-1}q_{2}^{-1}x)^{-1}:+\cdots
\end{align}
and actually, they commute with the screening charges:
\begin{align}
[\mathsf{T}_{i}(x),\mathcal{S}_{j}(x')]=0,\quad \mathcal{S}_{i}(x)=\sum_{k\in\mathbb{Z}}\mathsf{S}_{i}(q_{2}^{k}x).
\end{align}
The quiver $\mathcal{W}$-algebra formalism gives a gauge theoretic interpretation of the deformed $\mathcal{W}$-algebras introduced by Frenkel and Reshetekhin \cite{frenkel1996quantum,frenkel1997deformations,frenkel1999algebras}. The $\mathsf{Y}$-operator introduced here is the operator corresponding to the $\mathcal{Y}$-observable introduced in (\ref{def-Y}). The $\mathsf{T}$-operators are the operator version of the $qq$-characters of the supersymmetric gauge theories. Generalizations of these quiver $\mathcal{W}$-algebras were done in \cite{Kimura:2016dys,Kimura:2017hez,Kimura:2019xzj,Kimura:2022zsm} but we still expect there are broader extensions of them. 

As discussed in \S\ref{sec:deformedW}, deformed $\mathcal{W}$-algebras associated with linear quivers ($A_{n}$-type and $A_{m|n}$-type) are obtained using the coproduct structure of $\QTA$. Thus, one would like to find a quantum toroidal algebra reproducing the quiver $\mathcal{
W}$-algebras generally. Such attempts were done in \cite{feigin2021deformations}, where the authors introduced an algebra called $\mathcal{K}$-algebra. The $\mathcal{K}$-algebra admits a comodule structure on which the $\QTA$ acts and reproduces the quiver $\mathcal{W}$-algebra of $B,C,D$-types. Generalizations in this direction are left for future work.

\paragraph{Web of $\mathcal{W}$-algebra}
In \S\ref{sec:corner} and \S\ref{sec:repAY}, we introduced the corner vertex operator algebra $Y_{N_{1},N_{2},N_{3}}$ as a truncation of the $\mathcal{W}_{1+\infty}$ algebra/affine Yangian $\mathfrak{gl}_{1}$. As seen in \eqref{eq:CVOAfigure}, originally it was introduced as an algebra of BPS operators at the corner of three intersecting 5-branes with $N_{1},N_{2},N_{3}$ D3-branes attached to each face between the two 5-branes \cite{Gaiotto:2017euk} (see also \cite{Creutzig:2020zaj,Creutzig:2021dda,Al-Ali:2022wjg} for mathematical papers related). Later, it was shown in \cite{Prochazka:2017qum} that by gluing the trivalent vertices as the topological vertex, we can obtain a larger class of $\mathcal{W}$-algebras called the web of $\mathcal{W}$-algebras. For example, we have the following examples: $\mathcal{N}=2$ SCA and Bershadsky-Polyakov algebra $\mathcal{W}_{3}^{(2)}$
\begin{align}
    \adjustbox{valign=c}{\begin{tikzpicture}[thick]
		\begin{scope}[]
		\node[above,scale=1] at (0,1) {};
		\node[right,scale=1] at (1,0) {};
		  \node[scale=1] at (0.5,0.5){2};
            \node[scale=1] at (-0.1,-0.5){1};
            \node[scale=1] at (-0.5,0.2){};
		\draw[] (1,0) -- (0,0) -- (-0.7,-0.7);
		\draw[] (0,1) -- (0,0);
            \draw[] (-0.7,-0.7)--(-0.7,-1.7);
            \draw[] (-0.7,-0.7)--(-1.7,-0.7);
            \node at (-0.1,-2.5) {$\mathcal{N}=2$ SCA};
		\end{scope}
    \end{tikzpicture}}\hspace{2cm}\adjustbox{valign=c}{\begin{tikzpicture}[thick]
		\begin{scope}[]
		\node[above,scale=1] at (0,1) {};
		\node[right,scale=1] at (1,0) {};
		  \node[scale=1] at (0.5,0.5){3};
            \node[scale=1] at (0.5,-0.35){1};
            \node[scale=1] at (-0.5,0.2){};
		\draw[] (1,0) -- (0,0) -- (-0.7,-0.7);
		\draw[] (0,1) -- (0,0);
            \draw[] (-0.7,-0.7)--(1,-0.7);
            \draw[] (-0.7,-0.7)--(-1.7,-1.2);
            \node at (-0.1,-2.5) {Bershadsky-Polyakov algebra $\mathcal{W}^{(2)}_{3}$};
		\end{scope}
    \end{tikzpicture}}
\end{align}
Such kinds of algebras are obtained by gluing the corner VOAs and adding extra modules connecting them \cite{Prochazka:2018tlo}. From the algebraic viewpoint, since the corner VOAs are obtained from the plane partition representation of the affine Yangian $\mathfrak{gl}_{1}$, we can understand the arising algebra as a truncation of two plane partitions with boundary Young diagrams glued together. Research on this direction was done in \cite{Gaberdiel:2017hcn,Gaberdiel:2018nbs,Li:2019nna}.

The study of the trigonometric deformation of the web of $\mathcal{W}$-algebras can be approached from the perspective of quantum toroidal algebra. Within this framework, both horizontal and vertical representations are prominent. Notably, while the vertical representation offers an intuitive extension of the affine Yangian, it is the horizontal representations that present a new perspective. Each trivalent vertex will be understood as the tensor products of Fock representations and we additionally need to introduce operators connecting the vertices. The pioneering work in this domain is documented in \cite{Harada:2020woh}. Instead of starting from two $\QTA$s and gluing them, one can start from a larger quantum toroidal algebra such as quantum toroidal $\mathfrak{gl}_{n}$ and then decompose the operators using techniques in \cite{feigin2016branching}. This strategic approach enabled the derivation of the free field realizations of the web of deformed $\mathcal{W}$-algebra. The reason behind starting from a larger quantum toroidal algebra is underscored by the dual relationship of brane webs to toric Calabi-Yau three-folds. Starting from quiver quantum toroidal algebra and delving into its branching rules paves the path to a broader class of deformed $\mathcal{W}$-algebras.

\paragraph{Variants of intertwiners}
In  \S\ref{sec:intertwiner}, we introduced the intertwiners of quantum toroidal $\mathfrak{gl}_{1}$ and demonstrated in  \S\ref{sec:AFStopvertex} that their compositions result in physical observables such as instanton partition functions, which eventually lead to the AGT or BPS/CFT correspondence. In the construction of these intertwiners, we primarily used the Fock representation for the vertical representation. However, as outlined in \S\ref{sec:vertical-rep}, there are also vector and MacMahon representations available.

Furthermore, as mentioned in \S\ref{sec:quiver-QTA}, we have other types of quantum toroidal algebras (generally, quiver quantum toroidal algebras) beyond the quantum toroidal $\mathfrak{gl}_{1}$. Thus, it is natural to ask if we can generalize the discussion and introduce intertwiners to such cases.

From a gauge-theoretic perspective, it is necessary to examine generalizations in other dimensions (2d--7d), different quiver structures (e.g., $ABCDEFG$), different gauge groups (e.g., $\SU,\SO$, supergroups), and different space-time structures (e.g., $\mathbb{C}^{2}/\mathbb{Z}_{n}$).

Research has been conducted in this direction, extending the map of BPS/CFT correspondence. We will outline some of these studies.

\begin{itemize}
\item Generalizations to other dimensions:
Intertwiners using vector representation (vector intertwiners) were introduced in \cite{Awata:2018svb, Zenkevich:2018fzl} and shown that it reproduces the holomorphic blocks of 3d quiver gauge theories. Moreover, it was shown in \cite{Zenkevich:2020ufs} that we can construct multiple intertwiners with colors and reproduce intersecting 3d gauge theories in the sense of Nekrasov's gauge origami system \cite{Nekrasov:2016ydq}. The properties of these intertwiners were further discussed in \cite{Zenkevich:2022dju}. Also, we can replicate 3d quiver gauge theories using shifted quantum affine algebras \cite{Bourgine:2021nyw,Bourgine:2022scz}. Generalization to the MacMahon representation was done in \cite{Awata:2018svb,Cheewaphutthisakun:2021cud} and conjectured to be dual to a mysterious 7d gauge theory where the instantons are classified by plane partitions.

\item Quiver structures: The intertwiner formalism discussed in this review is limited to linear quiver gauge theories. Extensions to ABCDEFG quivers have been explored in \cite{D-type,Kimura:2019gon}, where orientifold planes are present. A reflection state was introduced in \cite{D-type} to realize the D-type quiver structure. Alternatively, \cite{Kimura:2019gon} proposed a distinctive type of vertex operator to represent the partition functions of ABCDEFG and affine quiver gauge theories. Although vertex operators were introduced successfully, the algebraic origin of them remains unclear. It is expected  they could be related to the $\mathcal{K}$-algebra introduced in \cite{feigin2021deformations}, but this requires further study.  

\item Gauge groups: This review only considers the relationship between quantum toroidal algebra and instanton partition functions with unitary groups. As highlighted in \S\ref{sec:PSY} and \S\ref{sec:qq}, this relationship becomes evident because the fixed points of the equivariant actions on instanton moduli spaces, or the JK poles of instanton partition functions, are classified by tuples of Young diagrams at the refined level for unitary groups \cite{Nekrasov:2002qd}. In contrast, for other classical gauge groups, although the integral formula was given in \cite{Nekrasov:2004vw,Marino:2004cn}, the JK poles of the ADHM integrals can be classified by tuples of Young diagrams solely at the unrefined level \cite{Nawata:2021dlk} (see also \cite{Nakamura:2014nha,Nakamura:2015zsa,Hollands:2010xa} for attempts to classify the poles in the refined case). Consequently, their $qq$-characters and the corresponding underlying algebraic structures admit natural interpretations only at the unrefined level \cite{Nawata:2023wnk}. This is very similar to the fact that the quantum toroidal algebras of type $A$ can have two deformation parameters $q,t$ whereas those of other types allow for just one deformation parameter $q$ as explained in \S\ref{sec:quiver-QTA}.

Although the quantum toroidal algebra-like structures reproducing the instanton partition functions for $\SO,\Sp$ gauge groups are not well known yet, we already have evidence showing the AGT correspondence of them \cite{Keller:2011ek,Song:2012kgc}. Since the essence of BPS/CFT correspondence is to find vertex operators reproducing the partition function, finding vertex operators reproducing the integral formula \cite{Nekrasov:2004vw,Marino:2004cn} using the method in \cite{Kimura:2019hnw} will be one strategy. Another strategy would be to focus on the unrefined limit and study the relation with topological vertices as we did in \S\ref{sec:algebraic_topvertex}. In the unrefined limit $\epsilon_1+\epsilon_2=0$, generalization of the topological vertex formalism has been achieved to realize gauge groups beyond $A$-type based on the brane construction with orientifolds in \cite{Kim-Yagi,Hayashi:2018bkd,Cheng:2018wll,Hayashi:2019yxj,Hayashi:2020hhb,Kim:2021cua,Kim:2022dbr} and novel gluing methods of the brane webs known as trivalent/quadrivalent gluing in \cite{Hayashi:2017jze,Hayashi:2021pcj,Wei:2022hjx}. Finding vertex operators that reproduce these topological vertices should lead us to new types of intertwiners and thus give us a hint to find new quantum toroidal algebras.

We can also generalize the group structures to supergroups \cite{Dijkgraaf:2016lym,Okuda:2006fb,Kimura:2021ngu,Nieri:2021xpe,Kimura:2019msw,Kimura:2020lmc}. 
BPS/CFT correspondence of superunitary groups was discussed in \cite{Noshita:2022dxv,Kimura:2023iup} and they manage to reproduce the instanton partition function proposed in \cite{Kimura:2019msw}. The generalization of the partition function to orthosymplectic super instanton counting was done in \cite{Kimura:2023ndz} and a contour integral formula was proposed. Studying the quantum algebraic aspects of it is also left for future work.

\item Space-time: Using other quantum toroidal algebras such as $\mathfrak{gl}_{n}$ \cite{Feigin2012gln}, $\mathfrak{gl}_{m|n}$ \cite{Bezerra2019BraidAO,Bezerra2019QuantumTA,Bezerra2021RepresentationsOQ}, $D(2,1;\alpha)$ \cite{Feigin2021CombinatoricsOV}, and quiver quantum toroidal algebras \cite{Li:2020rij,Galakhov:2021vbo,Noshita:2021ldl,Noshita:2022dxv}, we expect we obtain BPS/CFT correspondence in complicated geometry of the space-time. Previous work for quantum toroidal $\mathfrak{gl}_{n}$ was done in \cite{Awata:2017lqa} and shown to reproduce the instanton partition function on $\mathbb{C}^{2}/\mathbb{Z}_{n}\times S^{1}$, where the orbifold action in (\ref{eq:orbifoldaction}) is $\nu_{1}=+1,\nu_{2}=-1,\nu_{3}=0$. For gauge theories on $\mathbb{C}^{2}/\mathbb{Z}_{n}\times S^{1}$ with general orbifold action, the intertwiner formalism was performed in \cite{Bourgine:2019phm}. Generalizations of these discussions in a unified way to other quantum toroidal algebras are one of the studies that are necessary to be done.

\item Rational and elliptic generalizations: 
We also have rational and elliptic generalizations of the intertwiners \cite{Bourgine:2018uod,Zhu-elliptic,Foda-Zhu,Ghoneim:2020sqi,Cheewaphutthisakun:2021bdi,Saito,Konno:2021zvl}. 
\end{itemize}

\paragraph{Representation theory of quantum toroidal algebras}
In the main text, we focus on the representation theory and its applications specific to quantum toroidal $\mathfrak{gl}_{1}$. We have not discussed the representation theory of other types of quantum toroidal algebras. Nonetheless, the representation theory of quantum toroidal algebras and affine Yangians of various types has been explored in numerous works. While it is too vast to cite them all, some notable examples include \cite{varagnolo1996schur,saito1998quantum,saito1998toroidal,miki2001quantum,nakajima2001quiver,hernandez2009quantum,Feigin2012gln,feigin2016branching,tsymbaliuk2019several,kodera2019affine}. Delving into these works exceeds the scope of this note. Instead, we will provide a concise overview of the recent advances in representation theory from a physics standpoint.

In \S\ref{sec:quiver-QTA}, we simply gave the definition of the quiver quantum toroidal algebra which is expected to give the general structure of quantum toroidal algebras associated with quiver structures. Similar to $\QTA$, it is expected that we have two classes of representations (vertical and horizontal). For the vertical representations when the central element is $C=1$, thanks to the innovative work \cite{Li:2020rij}, we now know that the vertical representations for quiver quantum algebras associated with toric CY3-folds are related to BPS crystals. Note that before \cite{Li:2020rij}, representations for quantum toroidal $\mathfrak{gl}_{n}$ were already discovered in \cite{Feigin2012gln}. For example, we have the following 3d BPS crystals \cite{Ooguri:2009ijd}: $\mathbb{C}^{2}/\mathbb{Z}_{2}\times \mathbb{C}$/quantum toroidal $\mathfrak{gl}_{2}$ and the conifold/quantum toroidal $\mathfrak{gl}_{1|1}$
\begin{align}
   \mathbb{C}^{2}/\mathbb{Z}_{2}\times \mathbb{C}: \quad\adjustbox{valign=c}{\includegraphics[width=3cm]{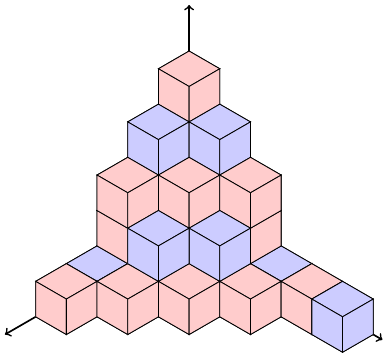}}\hspace{2cm}\text{conifold:}\quad \adjustbox{valign=c}{\includegraphics[width=2.5cm]{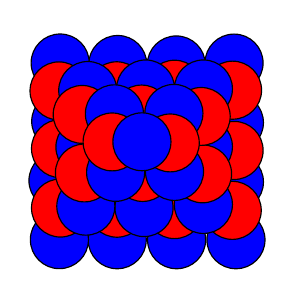}}
\end{align}
The general structure is as follows. The \textit{atoms} of the 3d BPS crystals are \textit{colored} atoms determined from the nodes of the quivers. The above $\mathfrak{gl}_{2}$ and $\mathfrak{gl}_{1|1}$ have only two nodes and therefore there are only two colors appearing in the crystals. The crystals will form the bases of the quantum algebras and the Drinfeld currents $K_{i}^{\pm}(z)$ acts diagonally on them. The other Drinfeld currents $E_{i}(z),F_{i}(z)$ add or remove the atoms from the configurations. In this way, we obtain the MacMahon-like vertical representations. Other vertical representations are obtained by truncation of the representations\footnote{To be specific, under this truncation, we need to slightly modify the algebra to shifted quantum algebras. We will not discuss this so see \cite{Galakhov:2021xum,Noshita:2021dgj}}.  

On the other hand, for the horizontal representations where vertex operator representations are expected to arise, we do not know how to derive them systematically. For quantum toroidal $\mathfrak{gl}_{n}$ \cite{feigin2021evaluation}, quantum toroidal $\mathfrak{gl}_{m|n}$ \cite{Bezerra2019QuantumTA}, quantum toroidal algebras associated with general orbifolds \cite{Bourgine:2019phm} and DE-type quantum toroidal algebras \cite{ginzburg1995langlands}, they are known. Generalizations to other cases are left for future work.

\paragraph{Integrable systems}
One can generalize the discussion of \S\ref{s:int-model} to other infinite-dimensional algebras. In \S\ref{s:int-model}, starting from the Miura transformation of the $\mathcal{W}$ algebra, we constructed the Maulik-Okounkov $\cR$-matrix and the monodromy matrix. Using the monodromy matrix, we managed to derive the integral of motions and the generators of the affine Yangian $\mathfrak{gl}_{1}$. In other words, this construction gives a map from the $\mathcal{W}$-algebra to the affine Yangian. Following this approach, generalizations to affine Yangian $\mathfrak{gl}_{2}$ and affine Yangian $\mathfrak{gl}_{1|1}$ were discussed in \cite{Chistyakova:2021yyd,Kolyaskin:2022tqi}. Generalizations to $BCD$ conformal field theory were also discussed in
\cite{Litvinov:2021phc}. This RLL formalism of general quiver Yangians \cite{Li:2020rij} and quiver quantum toroidal algebras \cite{Galakhov:2021vbo,Noshita:2021dgj} were introduced in \cite{Bao:2022fpk,Bao:2023kkh,Galakhov:2022uyu}. 

Another direction is to find a closed formula of the $\mathcal{R}$-matrix. As mentioned in \S\ref{s:int-model}, since the representation space is infinite-dimensional the $\mathcal{R}$-matrix has also an infinite number of matrix elements. Recently, a closed recursion formula was derived in the context of Macdonald functions in \cite{Negut:2020npc,Garbali:2020sll,Garbali:2021qko}. Generalizations to other quiver quantum toroidal algebras are left for future work.

Along the line of \S\ref{sec:NS}, the $TQ$-relation of spin-chain-like models can be derived from the saddle-point equation in the $qq$-character. One can rewrite this saddle-point as the differential equation \eqref{oper-eq}, which may be regarded as an oper equation.  Such an approach, purely based on the properties of defects in the gauge theory and the non-perturbative Dyson-Schwinger equation \cite{Nekrasov:2017gzb}, provides a clear clue on how the quantum version of the Hitchin system arises in the Nekrasov-Shatashvili limit. Further explorations of gauge theories with matters and higher-rank theories can be found in \cite{Jeong:2018qpc,Koroteev:2020mxs,Jeong:2023qdr}. The relations with the isomonodromic deformation in the context of 2d CFTs and Heisenberg spin chains were also explored in  \cite{Jeong:2020uxz,Jeong:2021rll}. It would be desirable to see its connection with the quantum toroidal algebras, especially with the Maulik-Okounkov $\cR$-matrix in the full $\Omega$-background region in the future, which may clarify the relation among all the different integrable systems mentioned above in a systematic manner.

\paragraph{Knizhnik-Zamolodchikov equations}
One of the topics we did not discuss in this review is the relation with Knizhnik-Zamolodchikov (KZ) equations \cite{Knizhnik:1984nr}. KZ equation is a set of constraints that the correlation functions of CFTs associated with affine Lie-algebras should satisfy. Let $\hat{\mathfrak{g}}_{k}$ be an affine Lie algebra whose level is $k$ and the Coxeter number is $h$. Let $\Phi_{i}(z)$ denote the primary fields transforming in a suitable representation labeled by $i$ and $t^{a}$ denote the basis of the Lie-algebra $\mathfrak{g}$. The KZ equation is then given as 
\begin{align}
    \left((k+h)\partial_{z_{i}}+\sum_{j\neq i}\frac{\sum_{a,b}\eta_{ab}\rho_{i}(t^{a})\otimes \rho_{j}(t^{b})}{z_{i}-z_{j}}\right)\langle\Phi(z_{N})\cdots\Phi(z_{1})\rangle=0,
\end{align}
where $\eta_{ab}$ is the Killing form and $\rho_{i}$ the representation map for each primary field $\Phi_{i}(z_{i})$. The KZ equation means that the correlation functions of the WZW models are characterized by a set of differential equations. The version of quantum affine algebras, called $q$KZ equations, was introduced in \cite{FR-KZ}, and the subsequent generalizations are explored from the geometric perspective in \cite{Maulik:2012wi}. Inspired by these works, the $(q,t)$-deformed version of the KZ equations for the quantum toroidal algebras was first studied in \cite{Awata:2017cnz}. Further insights and complementary explorations on this topic can be found in a collection of works, including \cite{Awata:2017lqa, Awata:2016bdm, Awata:2016mxc, Cheewaphutthisakun:2021cud, Nekrasov:2017gzb, Nekrasov:2021tik}.

As mentioned in \S\ref{sec:AGT}, the AGT correspondence claims that the conformal block of CFTs is dual to the instanton partition function. A corollary of this observation is that when the instanton partition function obeys KZ-like equations, there might exist an underlying algebraic structure. This viewpoint provides compelling motivation to probe deeper into the KZ equations. Let us focus on the $(q,t)$-equation for quantum toroidal $\mathfrak{gl}_{1}$. The $(q,t)$-KZ equation will be a difference equation and is given schematically as 
\begin{align}
    \left(\frac{q}{t}\right)^{z_{k}\partial_{z_{k}}}\langle V_{1}(z_{1})\cdots V_{N}(z_{N})\rangle=\prod_{i\neq k}\mathcal{R}_{i,k}\langle V_{1}(z_{1})\cdots V_{N}(z_{N})\rangle.
\end{align}
The fields $V_{i}(z)$ are, in this case, intertwiners and dual intertwiners labeled by Young diagrams. The factor $\mathcal{R}_{ik}$ is also labeled by Young diagrams and actually related to the diagonal part of the universal $\mathcal{R}$-matrix ($\mathcal{R}^{(0)}\mathcal{R}^{(1)}$ in (\ref{eq:universalRgl1})). Given that the correlation functions of the intertwiners (\ref{eq:intertwiner-contraction}) give the Nekrasov factors  (\ref{Nekra-S}), this implies that the Nekrasov factors inherently satisfy the $(q,t)$-KZ equations. In light of this, a deeper exploration of the generalized KZ equations and their consequent solutions could unveil novel partition functions, potentially suggesting new BPS/CFT correspondences.

\appendix
\section{Notations}\label{app:notations}

\subsection{Glossary of symbols}
The central elements or deformation parameters of the quantum toroidal algebra and affine Yangian of $\frakgl_1$ are written as 
\bea 
q_c=e^{\epsilon_c}~(c=1,2,3),\qquad q_1=q~,\qquad q_2=t^{-1}~,\qquad t=q^\beta~,\qquad \beta=-\frac{\epsilon_2}{\epsilon_1}~.
 \eea 
 These parameters also represent the $\Omega$-deformation parameters in supersymmetric theories. 
As a general rule, single symbols in sans-serif type are used to denote modes for infinite-dimensional algebras. For instance, the $\U(1)$-current is expressed as:
\be 
J(z)=\sum_{n\in \bZ} \frac{\sfJ_n}{z^{n+1}}~. 
\ee 
The $\cW$-algebra of type $\frakgl_N$ is denoted by $\cW_N$-algebra with its currents and generators given by
     \[
     W^{(k)}(z) = \sum \sfW^{(k)}_n z^{-n-k}.
     \]
On the other hand, the $\cW$-algebra of type $\fraksl_N$ is denoted by $\widetilde\cW_N$-algebra with its currents and generators expressed as
     \[
     \widetilde W^{(k)}(z) = \sum \widetilde \sfW^{(k)}_n z^{-n-k}.
     \]
The expansion of a rational function $f(z)$ in $z^{\mp 1}$ is represented by $[f(z)]_{\pm}$. The multiplicated delta function is $\delta(z)=\sum_{m\in\mathbb{Z}}z^{m}$ and has the property $f(z)\delta(z/a)=f(z)\delta(z/a)$. Due to limited notation choices, we sometimes use the same symbol for distinct concepts. For example, the symbol $ T $ may represent an energy-momentum tensor, a DAHA generator, a transfer matrix, or a torus action. Nonetheless, the context should clearly distinguish the intended meaning.

In the following, we list notations of the paper.
\begin{description}\setlength\itemsep{.1em}
  \item [$U_{q_1,q_2,q_3}(\ddt{\frakg})$] quantum toroidal algebra of $\frakg$
  \item [$Y_{\epsilon_{1},\epsilon_{2},\epsilon_{3}}(\dt{\frakg})$] affine Yangian of $\frakg$
  \item  [$\HH_N$] Double affine Hecke algebra of $\GL(N,\bC)$
  \item  [$\SH_N$] spherical subalgebra of double affine Hecke algebra of $\GL(N,\bC)$
  \item  [$\dH_N$] Degenerate double affine Hecke algebra of $\GL(N,\bC)$
  \item  [$\SdH_N$] spherical subalgebra of degenerate double affine Hecke algebra of $\GL(N,\bC)$
   \item  [$\SdH^{\boldsymbol{c}}$] the central extension of the large $N$ limit of $\SdH_N$
  \item   [$\mathtt{a}_{ij}$] Cartan matrix
  \item [$\dt{\mathtt{a}}_{ij}$]  affine Cartan matrix
  \item [$P_\lambda$] Macdonald functions
    \item [$J_\lambda$] Jack functions
  \item [$s_\lambda$] Schur functions
\item [$u$] 5d Coulomb branch parameters (weights of vector representations of \QTA)
\item [$\fraka$] 4d Coulomb branch parameters
\item [$\sfa_r$] modes of $q$-Heisenberg algebra
\item [$\sfJ_r$] modes of Heisenberg algebra
\item [$\sfL_r$] modes of Virasoro algebra
\end{description}

\subsection{Young diagrams}\label{sec:appendix-Youngdiagram}
Here, we present a concise overview of the notations and definitions related to a partition or a Young diagram.
Consider a Young diagram denoted by $\lambda = (\lambda_1, \lambda_2, \ldots)$, characterized by non-negative integers that satisfy the condition $\lambda_i \geq \lambda_{i+1}$. The length of the Young diagram, $\ell(\lambda)$, corresponds to the count of non-zero $\lambda_i$. The transposition of $\lambda$ is represented as $\lambda^t$.
The size and norm of the diagram are defined respectively by
\begin{align}
    |\lambda|=\sum_{i=1}^{\ell(\lambda)}\lambda_{i}~,\qquad \|\lambda\|^{2}=\sum_{i=1}^{\ell(\lambda)}\lambda_{i}^{2}~.
\end{align}
The arm length $a_{\lambda}(i,j)$ of a box at $(i,j)$ is the number of boxes to the right of the box in the diagram $\lambda$, and the leg length $l_{\lambda}(i,j)$ is the number of boxes above the box.
\be
 a_{\lambda}(i,j)=\lambda_{i}-j,\qquad l_{\lambda}(i,j)=\lambda_{j}^{t}-i~.
 \ee
The hook length of a box $(i,j)$ is defined as
\begin{align}
    h_{\lambda}(i,j)=\lambda_{i}-j+\lambda_{j}^{t}-i+1=a_{\lambda}(i,j)+l_{\lambda}(i,j)+1.\label{eq:hooklength}
\end{align}
We denote a set of addable/removable boxes in a Young diagram $\lambda$ as
\begin{itemize}[nosep]
\item   $\frakA(\lambda)$ denotes the set of boxes that can be added to the Young diagram $\lambda$.
\item  $\frakR(\lambda)$ denotes the set of boxes that can be removed from the Young diagram $\lambda$.
\end{itemize}

\begin{figure}[ht]\centering
\includegraphics{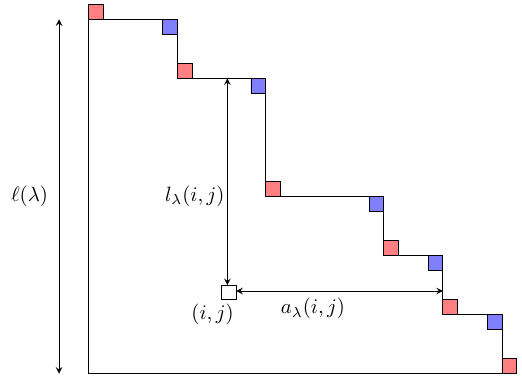}
  \caption{Arm and leg length of a box $(i,j)$ in a Young diagram $\lambda$. The red (resp. blue) boxes can be added to (resp. removed from) the Young diagram, and they form the set denoted by $\frakA(\lambda)$ (resp. $\frakR(\lambda)$).}
\end{figure}

In the context of topological vertex, we frequently use the following quantity for a framing factor
\begin{equation}
    n(\lambda)=\sum_{i=1}^{\ell(\lambda)}(i-1)\lambda_{i}~,
\end{equation}
which satisfies the following properties:
\begin{align}
    n(\lambda)=&\frac{1}{2}\sum_{j=1}^{\lambda_{1}}\lambda_{j}^{t}(\lambda_{j}^{t}-1)=\sum_{x\in\lambda}l'_{\lambda}(x)=\sum_{x\in\lambda}l_{\lambda}(x),\label{n-l}\\
    n(\lambda^{t})=&\frac{1}{2}\sum_{i=1}^{\ell(\lambda)}\lambda_{i}(\lambda_{i}-1)=\sum_{x\in\lambda}a'_{\lambda}(x)=\sum_{x\in\lambda}a_{\lambda}(x),\label{n-a}\\
    \sum_{x\in\lambda}h_{\lambda}(x)=&n(\lambda)+n(\lambda^{t})+|\lambda|=\frac{1}{2}(\|\lambda\|^{2}+\|\lambda^{t}\|^{2})
\end{align}

For Young diagrams $\lambda$ and $\mu$, we write $\lambda \supset \mu$ if  the Young diagram $\lambda$ contains the Young diagram $\mu$, i.e. that $\lambda_{i} \ge \mu_{i}$ for all $i \ge 1$. The set-theoretic difference $\theta=\lambda-\mu$ is called a \emph{skew diagram}. We define $\theta^t=\lambda^t-\mu^t$ in a similar way.

\begin{equation}\label{skew-diagram}
\includegraphics{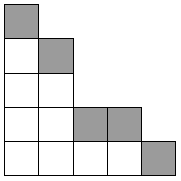}
\qquad {\raisebox{1.5cm}{$\lambda=(5,4,2,2,1)\supset \mu=(4,2,2,1)~.$}}
\end{equation}

A skew diagram $\theta$ is a horizontal $m$-strip (resp. a vertical $m$-strip) if $|\theta|=m$ and $\theta_{i}^{t} \le 1$ (resp. $\theta_{i} \le 1$ ) for each $i \ge 1$. In the example above, the skew diagram $\theta$ is a vertical $5$-strip whereas it is not a horizontal strip.

Another notation for partitions we will use is the Frobenius presentation. Suppose that there are $r$ diagonal boxes $(i, i)$ $(1 \le i \le r)$ in the Young diagram $\lambda$. Let $\alpha_{i}=\lambda_{i}-i+\frac12$ be the length in the $i$-th row of $\lambda$ to the right of $(i, i)$ and let $\beta_{i}=\lambda_{i}^{t}-i+\frac12$ be the length in the $i$-th column of $\lambda$ above $(i, i)$. Then, we have $\alpha_{1}>\alpha_{2}>\ldots>\alpha_{r} \ge 0$ and $\beta_{1}>\beta_{2}>\ldots>$ $\beta_{r} \ge 0$, and we denote the partition $\lambda$ by
$$
\lambda=\left(\alpha_{1}, \ldots, \alpha_{r} | \beta_{1}, \ldots, \beta_{r}\right)=(\alpha | \beta) .
$$
It is easy to see that the transpose of $\lambda=(\alpha | \beta)$ is $\lambda^t=(\beta | \alpha)$.
\bea
\includegraphics{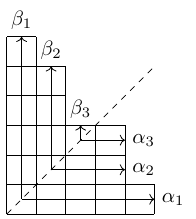}
\label{Frobenius}
\eea
For example, we show the Young diagram $\lambda=(5,4,4,2,2,1)$ above, and its Frobenius presentation are given by $(\alpha_1,\alpha_2,\alpha_3|\beta_1,\beta_2,\beta_3)=\left(\frac{9}{2},\frac{5}{2},\frac{3}{2}|\frac{11}{2},\frac{7}{2},\frac{1}{2}\right)$.

\section{Lightning review of quantum groups}\label{app:QG}

The section gives an overview of the theory of quantum groups developed by Drinfeld \cite{drinfeld1985hopf,drinfeld1986degenerate,drinfeld1987new,Drinfeld1987,drinfeld1990,drinfeld1990quasi,drinfeld1991quasitriangular}. Quantum groups arise from the study of the Yang-Baxter equation, which has been extensively researched as the master equation in integrable models in statistical mechanics and quantum field theory. Around 1980, Faddeev's Leningrad School proposed the quantum inverse scattering method (a.k.a. the algebraic Bethe ansatz), which used the Yang-Baxter equation as the basic commutation relation of operators. This led to the idea of introducing quantum deformations of groups or Lie algebras \cite{kulish1981yang,sklyanin1982some,QG-Jimbo1,QG-Jimbo2}, and Drinfeld provided the rigorous formulation of quantum groups as Hopf algebras.  Here we will not present the full details of the theory, but only highlight some key concepts and results. For more comprehensive treatments, we refer the reader to \cite{chari1995guide,etingof2009lectures,kassel2012quantum}. 

\subsection{Hopf algebras and universal \texorpdfstring{$\scR$}{R}-matrices}\label{app:Hopf}

\begin{definition}\label{Hopf}
 A Hopf algebra defined over a field $\bF$  is a $\bF$-vector space $H$ with the following additional structures:
\begin{equation}
\begin{array}{ll}
\text { multiplication : } & m: H \otimes H \rightarrow H, \\
\text { unit : } & i: \bF \rightarrow H, \\
\text { coproduct : } & \Delta: H \rightarrow H \otimes H, \\
\text { counit : } & \varepsilon: H \rightarrow \bF, \\
\text { antipode : } & S: H \rightarrow H,
\end{array}
\end{equation}
satisfying the following axioms:
\begin{enumerate}
    \item  the multiplication $m$ defines a structure of an associative algebra on $H$ with unit $i(1)$.
    \item  the coproduct $\Delta$ and the counit $\varepsilon$ define a structure of a coassociative coalgebra on $H$, which implies 
        \begin{equation}
        m(1\otimes \varepsilon)\Delta=m(\varepsilon\otimes1)\Delta=i,\quad (\text{id}\otimes \Delta)\Delta=(\Delta\otimes\text{id})\Delta.\label{eq:Hopf1}
    \end{equation}
    \item  the map $\Delta: H \rightarrow H \otimes H$ and $\varepsilon: H \rightarrow \bF$ are homomorphisms of algebras.
    \item  the antipode $S$ is a linear anti-isomorphism (i.e. $S(xy)=S(y)S(x),\,x,y\in H$) and satisfies the following two relations:
\be\label{eq:Hopf2}
m(1 \otimes S) \Delta=m(S \otimes 1) \Delta=i\circ \varepsilon ~.
\ee
\end{enumerate}
\end{definition}

\begin{figure}[ht]\centering
\begin{tikzcd}
H \otimes H \otimes H \arrow[r, "1\otimes m"] \arrow[d, "m\otimes 1"']\arrow[dr, phantom, "\circlearrowright"]  & H\otimes H \arrow[d, "m"]\\
 H\otimes H \arrow[r, "m"'] &H
 \end{tikzcd}\qquad
\begin{tikzcd}
H \arrow[r, "\Delta"] \arrow[d, "\Delta"'] \arrow[dr, phantom, "\circlearrowright"]& H\otimes H \arrow[d, "1\otimes \Delta"]\\
 H\otimes H \arrow[r, "\Delta\otimes 1"'] &H\otimes H \otimes H 
 \end{tikzcd}
 \caption{Associativity in algebra (left) and coassociativity in coalgebra (right).}
 \end{figure}

We often use the shorthand notation $H$ for a Hopf algebra $(H,m,i,\Delta,\varepsilon,S)$. 
Using a map that exchanges two elements
\be 
\sigma:H\otimes H \to H\otimes H ;~ h_1\otimes h_2\mapsto h_2\otimes h_1~, 
\ee 
we can define the opposite coproduct $\Delta^{\textrm{op}}=\sigma\circ \Delta$. Then, using the opposite coproduct, we have a new Hopf algebra $H^{\textrm{cop}}=(H,m,i,\Delta^{\textrm{op}},\varepsilon,S)$.

Given a representation $\rho_i: H \to \textrm{End}_{\bF} V_i$  of a Hopf algebra $H$, the coproduct gives rise to a tensor product of the representations
\begin{equation}
 H \xrightarrow{\Delta} H \otimes H\xrightarrow{\rho_1\otimes\rho_2 }\textrm{End}_{\bF} V_1\otimes V_2, 
\end{equation}
and the antipode provides a contravariant representation 
\begin{equation}
 H \xrightarrow{S} H \xrightarrow{\rho^* }\textrm{End}_{\bF} V^* ~.
\end{equation}

\subsubsection*{Quantum double}

Given a Hopf algebra $H$, Drinfeld's quantum double is a method to construct a new Hopf algebra from it.  For the sake of simplicity, let us assume that $H$ is finite-dimensional $\dim H<\infty$ while we can obtain the same conclusion for an infinite-dimensional vector space $H$. Let $H^*$ be the dual vector space of $H$, and $\langle \ ,\ \rangle$ be the natural pairing of $H$ and $H^*$. Then, we can introduce a Hopf algebra structure to $H^*$ as follows:
\be
\begin{aligned}
m_{H^*}:& \qquad\langle a, m_{H^*}(b_1 \otimes b_2)\rangle=\langle\Delta_{H}(a), b_1 \otimes b_2\rangle \cr 
i_{H^*}:& \qquad\langle a, i_{H^*}(1)\rangle=\varepsilon_{H}(a)\cr
\Delta_{H^*}^{\textrm{op}}:&\qquad\langle a_1 a_2, b\rangle=\langle a_2 \otimes a_1, \Delta_{H^*}^{\textrm{op}}(b)\rangle\cr
\varepsilon_{H^*}:& \qquad\langle i_{H}(1), b\rangle= \varepsilon_{H^*}(b)\cr 
 S_{H^*}:& \qquad\langle S_{H}(a), S_{H^*}(b)\rangle=\langle a, b\rangle\label{eq:qdpairing}
\end{aligned}\ee
The resulting algebra $H^{*\textrm{cop}}=(H^*,m_{H^*},i_{H^*},\Delta_{H^*}^{\textrm{op}},\varepsilon_{H^*}, S_{H^*})$ is called the dual Hopf algebra. 

The dual pair of $H$, $H^{*\textrm{cop}}$, with  $\langle \ ,\ \rangle$ is called a Hopf pairing. With this Hopf pairing, there is a unique Hopf algebra $D(H)$ defined as follows:

\begin{theorem}\label{quantum-double}
Let  ($H$, $H^{*\textrm{cop}}$,  $\langle \ ,\ \rangle$)  be a Hopf pairing. We define $D(H)=D(H,H^{*\textrm{cop}})$ as
\begin{enumerate}
\item As a vector space, $D(H)=H\otimes H^*$.

\item $H \simeq H \otimes 1 \hookrightarrow D(H)$, and  $H^{*\textrm{cop}}\simeq 1 \otimes H^{*\textrm{cop}} \hookrightarrow D(H)$ are $\bF$-algebra homomorphism, which define Hopf subalgebras. 

\item The multiplication on $D(H)\otimes D(H) \to D(H)$ is defined using the Hopf algebra structure on ${H}$ and $H^{*\textrm{cop}}$ as follows:
\be \nonumber
(a \otimes 1) \cdot (1 \otimes b):=a \otimes b
\ee 
\begin{equation} \nonumber
 (1 \otimes b)  \cdot (a \otimes 1) :=\sum \langle  S^{-1}(a_{(1)}), b_{(1)}\rangle\langle a_{(3)}, b_{(3)}\rangle a_{(2)}\otimes b_{(2)}
\end{equation}
where 
\bea \nonumber
\Delta_H^2(a)=&\sum a_{(1)} \otimes a_{(2)} \otimes a_{(3)} \in H^{\otimes 3} \cr 
(\Delta_{H^*}^{\textrm{op}})^2(b)=&\sum b_{(1)} \otimes b_{(2)} \otimes b_{(3)} \in (H^{*\textrm{cop}})^{\otimes 3}
\eea

\item The unit $i:\bF\to D(H)$ is defined by $i=i_{H}\otimes i_{H^*}$.

\item The coproduct  $\Delta:D(H)\to D(H) \otimes D(H)$ is defined as follows: \be\nonumber\Delta(a \otimes b)=\sum (a_{(1)} \otimes b_{(1)}) \otimes(a_{(2)} \otimes b_{(2)})~,\ee where $\Delta_H (a)=\sum a_{(1)} \otimes a_{(2)}$ and $\Delta_{H^*}^{\textrm{op}}(b)=\sum b_{(1)} \otimes b_{(2)}$

\item The counit  $\varepsilon:D(H)\to \bF$ is given by $\varepsilon(a \otimes b)=\varepsilon_H(a) \varepsilon_{H^*}(b)$.

\item  The antipode on $S:D(H)\to D(H)$ is defined as follows: 
\begin{equation} \nonumber
S(a \otimes b)=S(1 \otimes b) S(a \otimes 1)=(1 \otimes S_{H^*}(b))\cdot (S_H(a) \otimes 1)
\end{equation}
\end{enumerate}
Then, $D(H)=(H\otimes H^*,m,i,\Delta,\varepsilon,S)$ becomes a Hopf algebra, called the \emph{quantum double} of $H$. 
\end{theorem}

\subsubsection*{Universal $\mathscr{R}$-matrix}

The universal $\mathscr{R}$-matrix is a key concept in the theory of quantum groups and in the study of integrable systems, which is closely related to the quantum double construction. Let $\{a_i\}$ be a basis $H$, and $\{b_i\}$ be the dual basis of $H^{*\textrm{cop}}$. Then, we can define the canonical element $\scR=\sum_i a_i\otimes b_i\in H\otimes H^{*\textrm{cop}}$. By using the inclusions of the Hopf algebras
\bea \nonumber 
H\hookrightarrow D(H);& \ a\mapsto a\otimes1\cr 
H^{*\textrm{cop}}\hookrightarrow D(H);& \ b\mapsto 1\otimes b
\eea
we can regard $\scR$ as an element of $ D(H)\otimes D(H)\ni \scR$, which is called the \emph{universal $\scR$-matrix}. If $H$ is finite-dimensional, straightforward, though tedious, manipulations using Definition \ref{Hopf} and Theorem \ref{quantum-double} will lead to the following theorem:
\begin{theorem}\label{Th:Universal_R}
    The universal $\mathscr{R}$-matrix satisfy the following properties:
    \begin{description}
        \item[(i)] $
\mathscr{R} \Delta(x) =\Delta^{\textrm{op}}(x)\mathscr{R} \quad \textrm{for } \quad {}^\forall x \in D(H)$
      \item[(ii)] Hexagon relations are satisfied \bea 
(\textrm{id} \otimes \Delta) \mathscr{R}=\mathscr{R}_{13} \mathscr{R}_{12}~,\cr
(\Delta \otimes \textrm{id}) \mathscr{R}=\mathscr{R}_{13} \mathscr{R}_{23}~,
\eea
where $\mathscr{R}_{12}=\sum a_i \otimes b_i \otimes 1, \mathscr{R}_{13}=\sum a_i \otimes$ $1 \otimes b_i, \mathscr{R}_{23}=\sum 1 \otimes a_i \otimes b_i \in D(H)^{\otimes 3}$.
     \item[(iii)] $\mathscr{R}$ is invertible:
     \bea 
 (\varepsilon \otimes \textrm{id}) \scR&=1=(\textrm{id}\otimes \varepsilon) \scR \cr
 (S \otimes \textrm{id}) \scR&=\scR^{-1}=(\textrm{id} \otimes S^{-1}) \scR
     \eea 
    \end{description}
\end{theorem}
It follows that 
$$
\begin{aligned}
\mathscr{R}_{12} \mathscr{R}_{13} \mathscr{R}_{23} & \stackrel{(ii)}{=}\scR_{12}(\Delta \otimes \textrm{id}) \scR  =\sum(\scR \Delta(a_i)) \otimes b_i \\
&  \stackrel{(i)}{=}\sum(\Delta^{\textrm{op}}(a_i) \scR) \otimes b_i  =(\sigma \otimes \textrm{id}) \cdot(\Delta \otimes \textrm{id}) \scR \cdot \mathscr{R}_{12}\\
& \stackrel{(ii)}{=}(\sigma \otimes \textrm{id})(\scR_{13} \scR_{23})\mathscr{R}_{12}  =\mathscr{R}_{23} \mathscr{R}_{13} \mathscr{R}_{12}~,
\end{aligned}
$$
so that $\mathscr{R}$ solves the Yang-Baxter equation
$$
\mathscr{R}_{12} \mathscr{R}_{13} \mathscr{R}_{23}=\mathscr{R}_{23} \mathscr{R}_{13} \mathscr{R}_{12}~.
$$

\begin{definition}
    Let $H$ be a Hopf algebra. An invertible element $\scR \in H \otimes H$ is called a \emph{quasi-triangular structure} on $H$ if
    
(i) $\Delta^{\textrm{op}}=\scR \Delta \scR^{-1}$,

(ii) $(1 \otimes \Delta) \scR=\scR_{13} \scR_{12}~, \quad(\Delta \otimes 1) \scR=\scR_{13} \scR_{23}~.$
\end{definition}

\subsection{Quantum groups}

Let $\frakg$ be a semisimple Lie algebra. We denote, by $(\cdot, \cdot )$, the Killing form and the inner product on a root lattice $\sfQ$ induced from the Killing form. For a simple root $\a_i$, we denote $d_i=\frac{(\alpha_i, \alpha_i)}{2}$. (Therefore, for type $ADE$, $d_i=1$ ${}^\forall i$.) Since the Cartan matrix is defined by $\mathtt{a}_{ij}=\frac{2(\alpha_i, \alpha_j)}{(\alpha_i, \alpha_i)} $, we then have $ d_i \mathtt{a}_{ij}=d_{j} \mathtt{a}_{ji}$.

We write Chevalley generators of $\frakg$ by $\{h_i,e_i,f_i\}_{i\in I}$ for $I=\{1,\ldots,\textrm{rank}\,\frakg\}$ satisfying
\begin{equation}
[h_i, h_j]=0,\quad [h_i, e_j]=\mathtt{a}_{ij} e_j,\quad [h_i, f_j]=-\mathtt{a}_{ij} f_j,\quad [e_i, f_j]=\delta_{ij} h_i~.
\end{equation}
For any $i \neq j \in I$, they also obey the Serre relations 
\be
(\operatorname{ad}(e_i))^{1-\mathtt{a}_{ij}} e_j=0=(\operatorname{ad}(f_i))^{1-\mathtt{a}_{ij}} f_j~.
\ee 
Let $\frakh$, $\frakn^+$, $\frakn^-$ be Lie subalgebras generated by $\{h_i\}$, $\{e_i\}$, $\{f_i\}$, respectively. Then, the Lie algebra admits the triangular decomposition as
\be 
\frakg=\frakn^+\oplus\frakh\oplus\frakn^-~.
\ee
To write relations of quantum groups, we introduce the deformation parameter $q=e^\hbar$, and we define a quantum number and quantum binomial by
\begin{equation}
[m]_q:=\frac{q^m-q^{-m}}{q-q^{-1}}~,\qquad \begin{bmatrix}
m \\
n
\end{bmatrix}_q:=\frac{[m]_q !}{[n]_q![m-n]_q!}~.
\end{equation}

\begin{definition}
Quantum group $U_{q}(\frakg)$ is a Hopf algebra over the ring $\bC\llbracket \hbar\rrbracket$ of the formal power series generated by $\{h_i,e_i,f_i\}_{i\in I}$ with relations
\begin{align}
&[h_i,h_j]=0\label{QG-1} \\ 
&[h_j,e_i]=\mathtt{a}_{ij}e_j   \\
&[h_j,f_i]=-\mathtt{a}_{ij}f_j  \\
&[e_i,f_j]=\delta_{ij}\frac{q^{d_ih_i}-q^{-d_ih_i}}{q^{d_i}-q^{-d_i}}  \label{QG-4}
\end{align}
They also obey Serre relations for $(i \neq j)$
\begin{align}
& \sum_{k=0}^{1-\mathtt{a}_{ij}}(-1)^k\begin{bmatrix}
1-\mathtt{a}_{ij} \\
k
\end{bmatrix}_{q^{d_i}} e_i^k e_j e_i^{1-\mathtt{a}_{ij}-k}=0 ~,\\
& \sum_{k=0}^{1-\mathtt{a}_{ij}}(-1)^k \begin{bmatrix}
1-\mathtt{a}_{ij} \\
k
\end{bmatrix}_{q^{d_i}} f_i^k f_j f_i^{1-\mathtt{a}_{ij}+k}=0 ~.
\end{align}
\end{definition}

Instead of the generators $h_i$ of the Cartan subalgebra, we can introduce the notation  $k_\lambda=e^{\hbar h_\lambda}$ for an element $\lambda\in\sfQ^\vee$ of the coroot lattice $\sfQ^\vee$. Then, (\ref{QG-1}--\ref{QG-4}) can be replaced by the following relations
\begin{align}
& k_\lambda k_\mu=k_{\lambda+\mu}, \quad k_0=1~, \\
& k_\lambda e_i k_\lambda^{-1}=q^{(\lambda, \alpha_i)} e_i ~,\\
& k_\lambda f_i k_\lambda^{-1}=q^{-(\lambda, \alpha_i)} f_i ~,\\
& [e_i, f_j]=\delta_{ij} \frac{k_{\a_i}-k_{\a_i}^{-1}}{q^{d_i}-q^{-d_i}}~.
\end{align}
In the limit $\hbar\to0$, we obtain the universal enveloping algebra of $\frakg$ 
\be 
\lim_{\hbar\to0} U_{q}(\frakg)= U(\frakg)~.
\ee 

\paragraph{Hopf algebra structure}
$U_{q}(\frakg)$ is a Hopf algebra over the ring $\bC\llbracket \hbar\rrbracket$ of the formal power series where
\begin{itemize}
    \item $\Delta: U_{q}(\frakg) \longrightarrow U_{q}(\frakg) \otimes U_{q}(\frakg)$ is an algebra homomorphism such that
    \begin{equation}
\begin{aligned}
& \Delta(k_\lambda)=k_\lambda \otimes 1+1 \otimes k_\lambda \\
& \Delta(e_i)=e_i \otimes 1+k_{\alpha_i}\otimes e_i \\
& \Delta(f_i)=f_i \otimes k_{\alpha_i}^{-1}+1 \otimes f_i
\end{aligned}
\end{equation}

\item $\varepsilon: U_{q}(\frakg) \to \bC\llbracket\hbar\rrbracket$ is an algebra homomorphism such that
$$
\varepsilon(h_i)=\varepsilon(e_i)=\varepsilon(f_i)=0
$$

\item $S: U_{q}(\frakg) \longrightarrow U_{q}(\frakg)$  is an algebra anti-isomorphism such that
\begin{equation}
\begin{aligned}
& S(k_\lambda)=k_\lambda^{-1} \\
& S(e_i)=-k_{\alpha_i}^{-1} e_i \\
& S(f_i)=-f_i k_{\alpha_i}
\end{aligned}
\end{equation}

\end{itemize}

\paragraph{Quantum group as quantum double}
\begin{itemize}
   \item Consider a Hopf subalgebra $U^{\ge0}$ generated by $\{e_i\}_{i\in I}$ and $k_{\lambda\in \sfQ^\vee}$. Let $U^{\le0}$ be another Hopf subalgebra generated by $\{e_i\}_{i\in I}$ and $\bar k_{\lambda\in \sfQ^\vee}$ where the generators in the Cartan subalgebra are distinguished from those in $U^{\ge0}$ by adding a bar.
    \item There exist a unique Hopf pairing $U^{\ge0} \times U^{\le0} \to \bC(q)$ such that
\bea 
& \langle x, y_1 y_2\rangle=\langle\Delta(x), y_1 \otimes y_2\rangle\cr 
& \langle x_1 x_2, y\rangle=\langle x_2 \otimes x_1, \Delta(y)\rangle \quad (\textrm{be careful about ordering})\cr 
& \langle k_\lambda, \bar k_\mu\rangle=q^{-(\lambda, \mu)}\cr 
& \left\langle k_\lambda, f_i\right\rangle=0 \cr
& \left\langle e_i, \bar k_\lambda\right\rangle=0 \cr
& \left\langle e_i, f_j\right\rangle=\delta_{ij} \frac{1}{q^{-d_i}-q^{d_i}}
\eea
   \item Then, the quantum group is the quotient of the quantum double $D(U^{\ge0} , U^{\le0})$ by an ideal
\be 
U_{q}(\frakg) \cong  D(U^{\ge0} , U^{\le0})/(k_\lambda-\bar k_\lambda)_{\lambda\in \sfQ^\vee}~.
\ee 
\end{itemize}

\subsection{Yangians}

The Yangian $Y_{\hbar}(\frakg)$ is the canonical deformation of the universal enveloping algebra $U(\frakg[u])$ of the polynomial current Lie algebra $\frakg[u]$ with an indeterminate $u$.

\begin{definition}\label{def:Yangian} Yangian $Y_{\hbar}(\frakg)$ is a Hopf algebra over $\bC\llbracket\hbar\rrbracket$ generated by $\sfh_{i,r},\sfx_{i,r}^\pm$ $(i\in I,r\in \mathbb{Z}_{\ge0})$ with relations
 \begin{align}
&  {[\sfh_{i,r}, \sfh_{j, s}]=0}\label{Y1} \\ 
& {[\sfh_{i,0}, \sfx_{j, s}^{ \pm}]= \pm d_i \mathtt{a}_{ij}\sfx_{j, s}^{ \pm}} \\ 
&  {[\sfx_{i,r}^{+}, \sfx_{j, s}^{-}]=\delta_{i,j} \sfh_{i,r+s}} \\ 
&  [\sfh_{i,r+1}, \sfx_{j, s}^{ \pm}]-[\sfh_{i,r}, \sfx_{j, s+1}^{ \pm}]= \pm \frac{1}{2} d_i \mathtt{a}_{ij}\hbar\{\sfh_{i,r}, \sfx_{j, s}^{ \pm}\} \label{Y2}\\ 
&  {[\sfx_{i,r+1}^{ \pm}, \sfx_{j, s}^{ \pm}]-[\sfx_{i,r}^{ \pm}, \sfx_{j, s+1}^{ \pm}]= \pm \frac{1}{2} d_i \mathtt{a}_{ij}\hbar\{\sfx_{i,r}^{ \pm}, \sfx_{j, s}^{ \pm}\}} \label{Y3}\\ 
& \mathrm{Sym}_{\pi\in \frakS_{m}}[\sfx_{i,r_{\pi(1)}}^{ \pm},[\sfx_{i,r_{\pi(2)}}^{ \pm}, \ldots,[\sfx_{i,r_{\pi(m)}}^{ \pm}, \sfx_{j, s}^{ \pm}] \cdots]]=0~, \quad m:=1-\mathtt{a}_{ij} \label{Y4}
\end{align}
where we use the notation $\{a, b\}=ab+ba$.
\end{definition}
If we define 
\begin{align}
h_k(u) := 1+\sum_{r=0}^{\infty} \sfh_{k, r} \hbar^{-r} u^{-r-1} ~,\qquad
x_k^\pm(u) := \sum_{r=0}^{\infty} \sfx_{k, r}^\pm \hbar^{-r} u^{-r-1} ~,
\end{align}
then \eqref{Y2} and \eqref{Y3} can be expressed as
\begin{align}
\partial_u \partial_v h_k(u) x_{\ell}^{ \pm}(v)\left(2 u-2 v \mp \mathtt{a}_{k \ell}\right)=&-\partial_u \partial_v x_\ell^{ \pm}(v) h_k(u)\left(2 v-2 u \mp \mathtt{a}_{k \ell}\right)~,\\
\partial_u \partial_v x_k^{ \pm}(u) x_\ell^{ \pm}(v)\left(2 u-2 v \mp \mathtt{a}_{k \ell}\right)=&-\partial_u \partial x_\ell^{ \pm}(v) x_k^{ \pm}(u)\left(2 v-2 u \mp \mathtt{a}_{k \ell}\right)~.
\end{align}

As $\hbar \rightarrow 0$, the aforementioned identification $\lim _{\hbar \to 0} Y_{\hbar}(\frakg) \to U(\frakg[u])$ sends
\begin{equation}
\sfx_{i,r}^{+} \longmapsto \tilde{e}_i \otimes u^r, \quad \sfx_{i,r}^- \longmapsto \tilde{f}_i \otimes u^r, \quad \sfh_{i,r} \longmapsto \tilde{h}_i \otimes u^r .
\end{equation}
where 
\begin{equation}\label{Chevalley-tilde}
\tilde{e}_i=e_i~, \quad \tilde{f}_i=d_i f_i~, \quad \tilde h_i=d_i h_i~.
\end{equation}

In the given presentation (\ref{Y1}--\ref{Y4}) of $Y_{\hbar}(\frakg)$, the underlying Hopf algebra structure (Theorem \ref{Hopf}) is not immediately apparent. However, an alternative presentation exists that explicitly exhibits the Hopf algebra structure. The detail is not relevant here so we refer the reader to the comprehensive treatment provided in \cite[\S12.1]{chari1995guide}.  

\subsection{Quantum affine algebras}

Let us recall the affine Lie algebra $\dt\frakg$. As a vector space, it is isomorphic to the central extension of the loop algebra $L\mathfrak{g}=\mathfrak{g}[u, u^{-1}]$
\be 
\dt\frakg\cong \mathfrak{g}[u, u^{-1}]\oplus \bC c~.
\ee 
We add the derivation $d$ that counts the degree of $u$ as a generator so that the relations of the affine Lie algebra $\dt\frakg$ are given by 
\begin{align}
&{[\dt{\frakg}, c]=0} \\
&{[X \otimes u^r, Y \otimes u^s]=[X, Y] \otimes u^{r+s}+r \delta_{r+s, 0}(X, Y) c} \qquad \textrm{ for } X,Y\in \frakg\\ 
&[d, X \otimes u^r]=r X \otimes u^r
\end{align}
The affine Lie algebra appears as the symmetry of the $G$-WZW model with level $c$.

Quantum affine algebra $U_q(\dt{\mathfrak{g}})$ is the deformation of the universal enveloping algebra $U({\mathfrak{g}})$ of affine Lie algebra \cite{drinfeld1987new,drinfeld1986degenerate}. 
\begin{definition}\label{def:QA} Quantum affine algebra $U_q(\dt{\mathfrak{g}})$ is  a Hopf algebra over  $\bC\llbracket \hbar\rrbracket$ generated by $\sfh_{i,r},\sfx_{i,r}^\pm$ ($i\in I,r\in\bZ$) and 
a central element $C=e^{\hbar c}$ with relations
\begin{align}
&[\sfH_{i,0}, \sfH_{j,s}]=0 \label{HH}\\
&
[\mathsf{H}_{i,r}, \mathsf{H}_{j, s}]=\delta_{r,-s} \frac{1}{r}[r\mathtt{a}_{ij}]_{q^{d_i}} \frac{C^r-C^{-r}}{q^{d_j}-q^{-d_j}} \\
&
[\mathsf{H}_{i,0}, \mathsf{X}_{j, s}^{\pm}]=\pm\mathtt{a}_{ij} \sfX_{j, s}^{\pm}, \label{HX} \\
&
{[\mathsf{H}_{i,r}, \mathsf{X}_{j, s}^{ \pm}]= \pm \frac{1}{r}[r \mathtt{a}_{ij}]_{q^{d_i}} C^{\mp|r| / 2} \mathsf{X}_{j, r+s}^{ \pm},} \\
&
\mathsf{X}_{i,r+1}^{ \pm} \mathsf{X}_{j, s}^{ \pm}-q^{ \pm d_i\mathtt{a}_{ij}} \mathsf{X}_{j, s}^{ \pm} \mathsf{X}_{i,r+1}^{ \pm}=q^{ \pm d_i\mathtt{a}_{ij}} \mathsf{X}_{i,r}^{ \pm} \mathsf{X}_{j, s+1}^{ \pm}-\mathsf{X}_{j, s+1}^{ \pm} \mathsf{X}_{i,r}^{ \pm} \\
&
{[\mathsf{X}_{i,r}^{+}, \mathsf{X}_{j, s}^{-}]=\delta_{i,j} \frac{C^{(r-s) / 2} \Phi_{i,r+s}^{+}-C^{-(r-s) / 2} \Phi_{i,r+s}^{-}}{q^{d_i}-q^{-d_i}}}\\
& \mathrm{Sym}_{\pi\in \frakS_{m}}\sum_{k=0}^m(-1)^k\begin{bmatrix}m \\ 
k\end{bmatrix}_{q^{d_i}} \mathsf{X}_{i,r_{\pi(1)}}^{ \pm} \cdots \mathsf{X}_{i,r_{\pi(k)}}^{ \pm} \mathsf{X}_{j, s}^{ \pm} \mathsf{X}_{i,r_{\pi(k+1)}}^{ \pm} \cdots \mathsf{X}_{i,r_{\pi(m)}}^{ \pm}=0 , \quad m:=1-\mathtt{a}_{ij} \label{Serre}
\end{align}
where
\begin{equation}
\Phi_{i}^\pm(z):=\sum_{r=0}^{\infty} \Phi_{i,\pm r}^{ \pm} z^{ \mp r}=e^{ \pm \hbar d_i\mathsf{H}_{i,0}} \exp ( \pm(q^{d_i}-q^{-d_i}) \sum_{s=1}^{\infty} \mathsf{H}_{i,\pm s} z^{ \mp s})
\end{equation}
\end{definition}
If we define $\mathsf{K}_i:=e^{\hbar d_i\mathsf{H}_{i,0}}$, then \eqref{HH} and \eqref{HX} can be replaced by
 \begin{align}
&
{[\mathsf{K}_i, \mathsf{K}_j]=[\mathsf{K}_i, \mathsf{H}_{j, r}]=0}~, \\
&
\mathsf{K}_i \mathsf{X}_{j, r}^{ \pm} \mathsf{K}_i^{-1}=q^{ \pm d_i\mathtt{a}_{ij}} \mathsf{X}_{j, r}^{ \pm}~. 
\end{align}
As $q \rightarrow 1$, the aforementioned identification $\lim _{q \rightarrow 1}U_q(\dt{\mathfrak{g}}) \to U(\dt{\mathfrak{g}})$ sends
\begin{equation}
 \sfX_{i,k}^+ \mapsto \tilde{e}_i \otimes u^k, \quad \sfX_{i,k}^- \longmapsto \tilde{f}_i \otimes u^k, \quad  \sfH_{i k} \longmapsto \tilde{h}_i \oplus u^k .
\end{equation}
The quantum affine algebra with the trivial center $C=1$ is called the quantum loop algebra $U_q(L\mathfrak{g})$.  One way to present $U_q(L\mathfrak{g})$ is by using the relations (\ref{HH}--\ref{Serre}) at $C=1$. Another presentation for $U_q(L\mathfrak{g})$ involves defining generating currents and their associated relations, which are as follows. Let us define the generating current 
\begin{align}
X_i^\pm(z)=\sum_{k \in \mathbb{Z}}\sfX_{i,k}^\pm z^{-k}~,
\end{align}
and a function that is a close cousin of \eqref{bond}
\begin{equation}
\varphi^{j\Leftarrow i}(z,w)=\frac{q^{d_i\mathtt{a}_{ij}}z-w}{z-q^{d_i\mathtt{a}_{ij}}w}~.
\end{equation}
Then, the relations of $U_q(L\mathfrak{g})$ can be repackaged into 
\bea\label{QL}
& X_i^\pm(z) X_j^\pm(w)=\varphi^{j\Leftarrow i}(z,w)^\pm X_j^\pm(w) X_i^\pm(z) \\
&[X_i^{+}(z),X_j^{-}(w)]=\frac{\delta_{ij}}{q^{d_i}-q^{-d_i}}\Big[\delta(\frac{w}{z}) \Phi_i^{+}(w)-\delta(\frac{z}{w})\Phi_i^{-}(z)\Big] \\
& \Phi_i^{ s}(z) X_j^\pm(w)=\varphi^{j\Leftarrow i}(z,w)^{\pm 1} X_j^\pm(w) \Phi_i^{ s}(z) \\
&\mathrm{Sym}_{\frakS_m}[X^\pm_i(z_1),\ldots,[X^\pm_i(z_m),X^\pm_j(w)]_{q^{d_i}}\ldots]_{q^{d_i}}=0~, \quad m:=1-\mathtt{a}_{ij}~,
\eea
where $s=\pm$, and the $q$-commutator is given by $[a,b]_q=ab-qba$.

The Hopf algebra structure of $U_q(L\mathfrak{g})$ is given as follows. First, the coproduct $\Delta$ is defined as 
\bea
\Delta(X_i^{+}(z))=&X_i^{+}(z) \otimes 1+\Phi^-_i(z) \otimes X_i^{+}(z )\cr
\Delta(X_i^{-}(z))=&1 \otimes X_i^{-}(z)+X_i^{-}(z ) \otimes \Phi^+_i(z )\cr
\Delta(\Phi^-_i(z))=&\Phi^-_i(z ) \otimes \Phi^-_i(z )\cr
\Delta(\Phi^+_i(z))=&\Phi^+_i(z ) \otimes \Phi^+_i(z )
\eea
Also, the counit and antipode are given by
\begin{equation}
\begin{aligned}
S(X^+(z)) =&-(\Phi^{-}(z))^{-1} X^+( z) \\
S(X^-(z)) =&-X^-(z)(\Phi^{+}(z))^{-1} \\
S(\Phi^{\pm}(z)) =&(\Phi^{\pm}(z))^{-1}~.
\end{aligned}
\end{equation}

The Yangian $Y_\hbar(\mathfrak{g})$ is the rational degeneration of the quantum loop algebra $U_q(L\mathfrak{g})$. (This is stated in \cite{Drinfeld1987} while the explicit proof is given in \cite{guay2012quantum}.) Furthermore, a deeper relationship between (a completion of) $Y_\hbar(\mathfrak{g})$ and $U_q(L\mathfrak{g})$ can be found in \cite{gautam2013yangians}.

\subsection{Toroidal generalizations}
Finally, we turn our attention to toroidal (double affine, or double loop) generalizations of quantum groups, which constitute the main focus of this note. Here, we briefly introduce affine Yangians and quantum toroidal algebras because the comprehensive exploration is given in the main text.

The introduction of the quantum toroidal algebra of $\fraksl_N$ can be traced back to its initial appearance in \cite{ginzburg1995langlands} while the affine Yangians emerged approximately a decade later in \cite{guay2005cherednik,guay2007affine}.

To transition from the Yangian in Definition \ref{def:Yangian} to the affine Yangian, we incorporate the use of the affine Cartan matrix $\dt{\mathtt{a}}_{ij}$. Additionally, we introduce deformations to the relations \eqref{Y2} and \eqref{Y3}. For the sake of clarity, we present an explicit relation specific to the type $A$ case in the following.

\begin{definition}
    The affine Yangian $Y_{\epsilon_1, \epsilon_2}(\dt{\mathfrak{sl}}_N)$ ($N\ge3$) is the algebra over $\mathbb{C}\llbracket\epsilon_1,\epsilon_2\rrbracket$ generated by $\sfx_{i,r}^{\pm}, \sfh_{i,r}(i \in \mathbb{Z}_N, r \in \mathbb{Z}_{\ge0})$ subject to the relations:
\begin{align}
&[\sfh_{i,r}, \sfh_{j, s}]=0, \\ 
&[\sfh_{i,0}, \sfx_{j, r}^{ \pm}]= \pm \dt{\mathtt{a}}_{ij} \sfx_{j, r}^{ \pm}, \\
&[\sfx_{i,r}^{+}, \sfx_{j, s}^{-}]=\delta_{ij} \sfh_{i,r+s}, \\
&{[\sfh_{i,r+1}, \sfx_{j, s}^{ \pm}]-[\sfh_{i,r}, \sfx_{j, s+1}^{ \pm}]= \pm \dt{\mathtt{a}}_{ij} \frac{\epsilon_1-\epsilon_2}{4}\{\sfh_{i,r}, \sfx_{j, s}^{ \pm}\}-\mathtt{m}_{ij} \frac{\epsilon_1+\epsilon_2}{4}[\sfh_{i,r}, \sfx_{j, s}^{ \pm}],} \\
&{[\sfx_{i,r+1}^{ \pm}, \sfx_{j, s}^{ \pm}]-[\sfx_{i,r}^{ \pm}, \sfx_{j, s+1}^{ \pm}]= \pm \dt{\mathtt{a}}_{ij} \frac{\epsilon_1-\epsilon_2}{4}\{\sfx_{i,r}^{ \pm}, \sfx_{j, s}^{ \pm}\}-\mathtt{m}_{ij} \frac{\epsilon_1+\epsilon_2}{4}[\sfx_{i,r}^{ \pm}, \sfx_{j, s}^{ \pm}],} \\
& \mathrm{Sym}_{\pi\in \frakS_{m}}[\sfx_{i,r_{\pi(1)}}^{ \pm},[\sfx_{i,r_{\pi(2)}}^{ \pm}, \ldots,[\sfx_{i,r_{\pi(m)}}^{ \pm}, \sfx_{j, s}^{ \pm}] \cdots]]=0~, \quad m:=1-\dt{\mathtt{a}}_{ij},
\end{align}
where
$$
 \mathtt{m}_{ij}= \begin{cases}1 & \textrm{if } j=i-1 \\
-1 & \textrm{if } j=i+1 \\
0 & \textrm{otherwise }\end{cases}
$$
\end{definition} 
The coproduct structure of affine Yangian is provided in \cite{Guay:2017exp}.

In a similar vein, when transitioning from the quantum loop algebra to the quantum toroidal algebra, we introduce the replacement of the Cartan matrix with its corresponding affine counterpart and deform certain relations in \eqref{QL}. To construct the quantum toroidal algebra from the corresponding affine Dynkin diagram, a detailed prescription is outlined in \S\ref{sec:quiver-QTA}. Notably, the generators in the relation \eqref{quiver-QTA} can be identified as follows, taking into account suitable normalizations: 
\be 
X_i^+(z) \rightsquigarrow E_i(z) \quad X_i^-(z)  \rightsquigarrow F_i(z) \quad \Phi_i^+(z)  \rightsquigarrow K_i^\pm(z)~.
\ee 
Besides, a more comprehensive exposition of the explicit relations and the underlying Hopf algebra structure for the type $A$ case can be found in \cite[\S3.1]{Tsymbaliuk:2022bqx}.

\section{Symmetric functions}\label{app:symmetric-functions}
Quantum toroidal algebras and exactly solvable modes treated in this note are intertwined with the theory of symmetric functions through representations.
In this appendix, we therefore review special functions such as Macdonald, Jack, and Schur functions, and consider the meaning of exact solvability. For more detail, we refer to \cite{macdonald1998symmetric}.

Let us introduce some notations.  Let $\mathbb{C}(q,t)$ be the field of rational functions in $q,t\in \bC^\times$.  Let $\scR^{q,t}_{N}$ be the ring of $N$-variables symmetric polynomials with coefficients in $\mathbb{C}(q,t)$:
\begin{align*}
    \scR^{q,t}_{N}:=\mathbb{C}(q,t)[X_{1},\ldots,X_{N}]^{\mathfrak{S}_{N}}.
\end{align*}
We have a natural surjective homomorphism
\begin{align*}
    \scR^{q,t}_{N+1}\rightarrow \scR^{q,t}_{N}
\end{align*}
by setting $X_{N+1}=0$. We can also take the inverse limit and define
\begin{align*}
    \scR^{q,t}\coloneqq \varprojlim_{N}\scR^{q,t}_{N}
\end{align*}
which gives symmetric functions with infinite variables. In \S\ref{app:Jack}, we will change the ground field to the field $\mathbb{C}(\beta)$ of rational functions. Then, we will denote the ring of symmetric functions by $\scR^\beta$.

We also introduce the multi-index notation as
\begin{align*}
    X^{\lambda}=X_{1}^{\lambda_{1}}X_{2}^{\lambda_{2}}\cdots
\end{align*}
for $X=(X_{1},X_{2},\ldots)$ and $\lambda=(\lambda_{1},\lambda_{2},\ldots)$. The set of all partitions is denoted $\scP$.

We also introduce a partial ordering on $\scP$  by
\begin{align}
    \lambda\geq \mu \Leftrightarrow |\lambda|=|\mu|,\quad \text{and}\quad  \sum_{i=1}^{r}\lambda_{i}\geq \sum_{i=1}^{r}\mu_{i},\quad (\forall r\geq 1)
\end{align}
which is called the dominance ordering.

Basic symmetric functions are
\begin{enumerate}
    \item monomial symmetric functions $m_{\lambda}(X)$:
    \begin{align}
        m_{\lambda}(X)=\sum_{(i_{1},\ldots,i_{\ell(\lambda)})\in I_{\lambda}}X_{i_{1}}^{\lambda_{1}}\cdots X_{i_{\ell(\lambda)}}^{\lambda_{\ell(\lambda)}},\quad I_{\lambda}=\{(i_{1},\ldots,i_{\ell(\lambda)})\in\mathbb{Z}^{\ell(\lambda)}_{>0}\,|\,i_{j}\neq i_{k}\,(j\neq k)\}.\label{eq:monomial-symm}
    \end{align}
    \item power sum $p_{\lambda}(X)$:
    \begin{align}
        p_{\lambda}(X)=\prod_{i}p_{\lambda_{i}}(X),\quad p_{n}(X)=\sum_{i}X_{i}^{n}.\label{eq:powersum}
    \end{align}
    \item elementary symmetric functions $e_{\lambda}(X)$:
    \begin{align}
        e_{\lambda}(X)=\prod_{i}e_{\lambda_{i}}(X),\quad \sum_{n}e_{n}(X)Y^{n}\coloneqq\exp\left(-\sum_{n>0}\frac{1}{n}p_{n}(X)(-Y)^{n}\right).\label{eq:elem-symm}
    \end{align}
    \item $g_{\lambda}(X)$:
    \begin{align}
        g_{\lambda}(X)=\prod_{i}g_{\lambda_{i}}(X),\quad \sum_{n}g_{n}(X)Y^{n}=\exp\left(\sum_{n>0}\frac{1}{n}\frac{1-t^{n}}{1-q^{n}}p_{n}(X)Y^{n}\right).
    \end{align}
\end{enumerate}

\subsection{Macdonald functions}\label{app:Macdonald}

This subsection focuses on Macdonald symmetric polynomials of type $\frakgl_N$ and their stable limit to an infinite number of variables. Macdonald symmetric \emph{polynomials} of type $\frakgl_N$ have $N$ variables and form a basis of the space $\scR^{q,t}_N$. They play a central role in the representation theory of (spherical) DAHA, as seen in \eqref{DAHA-pol}. Moreover, their stable limits to an infinite number of variables, called Macdonald symmetric \emph{functions}, are also essential in the representation theory of quantum toroidal algebras, as discussed in \S\ref{sec:QTA-Mac}.

In this subsection, we first review the definition and properties of Macdonald symmetric polynomials. We then explore their stable limits and relate them to free bosons by taking the number of variables to infinity. This approach allows us to better understand the connection of Macdonald symmetric functions to quantum toroidal algebras, integrable systems and mathematical physics. Overall, the study of Macdonald symmetric functions and their stable limits has profound implications in various areas of mathematics and physics, and we refer \cite{macdonald1998symmetric} for more details.

First, we introduce an inner product $\langle \cdot,\cdot\rangle_{q,t}$ in $\scR^{q,t}_N$ by
\begin{align}\label{eq:Macdinnprod}
    \ev{p_{\lambda},p_{\mu}}_{q,t}=\delta_{\lambda,\mu}\prod_{i=1}^{\ell(\lambda)}\frac{1-q^{\lambda_{i}}}{1-t^{\lambda_{i}}},\quad z_{\lambda}=\prod_{i\geq 1}i^{m_{i}}\cdot m_{i}!
\end{align}
where $m_{i}$ is the number of entries in $\lambda$ equal to $i$.

Macdonald polynomials are determined from a set of difference operators
\begin{align}\label{Macdonald-diff}
\begin{split}
    D_{N}^{(k)}(X;q,t)=\sum_{\substack{I\subset \{1,\ldots,N\}\\ |I|=k}}\prod_{\substack{i \in I\\ j\not\in I}}\frac{t^{\frac12}X_{i}-t^{-\frac12}X_{j}}{X_{i}-X_{j}}\prod_{i\in I}T_{q,X_{i}},\quad k=0,1,\ldots,N
\end{split}
\end{align}
where  $T_{q,X_{i}}$ is the $q$-shift operator with respect to $X_{i}$ so that  $T_{q,X_{i}}X_j=q^{\delta_{ij}}X_j$. These operators are called \emph{Macdonald difference operators}. The Macdonald difference operators have the following properties:
\begin{itemize}
\item The operators commute with each other:
\begin{align}
    [D_{N}^{(k)}(X;q,t),D_{N}^{(l)}(X;q,t)]=0,\quad k,l\in\{1,\ldots,N\}.
\end{align}
\item They are self-adjoint in the inner product
\begin{align}
    \langle f,D_{N}^{(k)}g \rangle_{q,t}=\langle D^{(k)}_{N}f,g \rangle_{q,t},\quad f,g\in\scR^{q,t}_{N}
\end{align}
\item They are upper triangular on the basis of the monomial symmetric polynomials:
\begin{align}
    D^{(k)}_{N}m_{\lambda}=\sum_{\mu\leq \lambda}c^{k}_{\lambda\mu}(q,t)m_{\mu}
\end{align}
where $c^{k}_{\lambda\mu}(q,t)$ is some coefficient.
\end{itemize}
Since the difference operators commute with each other, we can consider their simultaneous eigenfunctions:
\begin{align}
    D_{N}^{(k)}(X;q,t)P_{\lambda}(X;q,t)=e_{k}(t^{\rho}q^{\lambda})P_{\lambda}(X;q,t),\quad k=0,1,\ldots,N
\end{align}
where $e_{k}$ is the elementary symmetric polynomial (\ref{eq:elem-symm}), and $\rho$ is the Weyl vector of $\frakgl_N$. The eigenfunctions $P_{\lambda}(X;q,t)\in \scR^{q,t}_{N}$  are uniquely determined by imposing
\begin{equation}
    P_{\lambda}(X; q,t)=m_{\lambda}(X)+\sum_{\mu<\lambda}\alpha_{\lambda,\mu}(q,t)m_{\mu}(X),\quad \alpha_{\lambda,\mu}(q,t)\in \mathbb{C}(q,t)
\end{equation}
These polynomials are called \emph{Macdonald symmetric polynomials} of type $\frakgl_N$.
From the definition, one can show Macdonald polynomials are orthogonal with respect to the inner product
\begin{equation}\label{k_lambda}
\left\langle P_{\lambda}, P_{\mu}\right\rangle_{q, t}=\delta_{\lambda\mu}\prod_{(i,j) \in \lambda} \frac{1-q^{a_\lambda(i,j)+1} t^{l_\lambda(i,j)}}{1-q^{a_\lambda(i,j)} t^{l_\lambda(i,j)+1}}:=\delta_{\lambda\mu}\frac{1}{k_{\lambda}(q, t)}
\end{equation}
Therefore, we define the dual Macdonald polynomials $Q_{\lambda}(X;q,t)$ as
\begin{equation}\label{dual-Macdonald}
    Q_{\lambda}(X;q,t)=k_\lambda(q,t)P_{\lambda}(X;q,t)~
\end{equation}
with
so that
\begin{align}
    \ev{P_{\lambda},Q_{\mu}}_{q,t}=\delta_{\lambda\mu}~.
\end{align}

One can explicitly construct Macdonald functions by using the polynomial representation of DAHA in \S\ref{sec:DAHA}. We define the following $q$-difference operators \cite{kirillov1998affine}
\begin{equation}
B_m=\sum_{r=0}^m(-1)^r t^{(2m-n+1) r/2} \sum_{\substack{I\subset \{1,\ldots,N\} \\|I|=r}} X_I e_{m-r}(X_{ \{1,\ldots,N\} \backslash I})f_I(X) T_{q, X}^I
\end{equation}
where
\begin{equation}
X_I=\prod_{i \in I} X_i, \quad T_{q, X}^I=\prod_{i \in I} T_{q, X_i}, \quad f_I(X)=\prod_{\substack{i \in I \\ j \notin I}}\frac{t^{\frac12}X_{i}-t^{-\frac12}X_{j}}{X_{i}-X_{j}}~.
\end{equation}
Then, Macdonald function for $\lambda=(\lambda_1\ge \cdots \ge \lambda_n>0)$ can be constructed by acting these operators on $1$ recursively as
\begin{equation}
P_\lambda(X)=\frac{1}{c_\lambda(q,t)}\left(B_n\right)^{\lambda_n}\left(B_{n-1}\right)^{\lambda_{n-1}-\lambda_n} \cdots\left(B_1\right)^{\lambda_1-\lambda_2}(1)
\end{equation}
where
\begin{equation}
 c_\lambda(q,t):=\prod_{s \in \lambda}(1-t^{\ell(s)+1} q^{a(s)})~.
\end{equation}
See \cite{kirillov1998affine} for more detail. Here we list explicit expressions for Young diagrams up to three boxes
\begin{equation}\label{Macdonald-ex}
\begin{aligned}
P_{(1)} = &p_1\cr
P_{(2)} =& \frac{(1-t)(1+q)}{(1-tq)} \frac{p_1^2}{2} + \frac{(1+t)(1-q)}{(1-tq)} \frac{p_2}{2}, \cr
P_{(1^2)}=&\frac{p_1^2}{2} - \frac{p_2}{2}\cr
P_{(3)} =& \frac{(1+q)(1-q^3)(1-t)^2}{(1-q)(1-tq)(1-tq^2)} \frac{p_1^3}{6} + \frac{(1-q)(1-t^2)(1-q^3)}{(1-q)(1-tq)(1-tq^2)} \frac{p_1p_2}{2} +
\frac{(1-q)(1-q^2)(1-t^3)}{(1-t)(1-tq)(1-tq^2)} \frac{p_3}{3}\cr
P_{(2,1)} =& \frac{(1-t)(2qt + q + t + 2)}{1-qt^2} \frac{p_1^3}{6} + \frac{(1+t)(t-q)}{1-qt^2} \frac{p_1p_2}{2} -
\frac{(1-q)(1-t^3)}{(1-t)(1-qt^2)} \frac{p_3}{3}\cr
P_{(1^3)} =& \frac{p_1^3}{6} - \frac{p_2p_1}{2} + \frac{p_3}{3}.
\end{aligned}
\end{equation}

\subsubsection*{Relation to free $q$-boson}

It is well-known that the ring $\scR$ of symmetric functions admits a free-field realization \cite{andric1983large,jing1991vertex,jing1992formula,avan1992algebraic,jing1994q,iso1995collective,Awata:1994xd,Awata:1995eh}.
Although the free-field realization is a simple rewriting of $\scR$, it is quite interesting as a technique to efficiently deal with operators acting on $\scR$. As an example, we present a free-field representation of the Macdonald difference operator here.

We consider $q$-Heisenberg algebra
\be \label{q-Heisenberg}
\left[\sfa_m, \sfa_n\right]=m \frac{1-q^{|m|}}{1-t^{|m|}} \delta_{m+n, 0}~.
\ee
Here we change the normalization of $q$-Heisenberg modes from \eqref{eq:commuteqboson}.
We define the Fock space and the dual Fock space
$$
\mathcal{F}=\mathbb{C}\left[\sfa_{-1}, \sfa_{-2}, \ldots\right]|0\rangle~, \qquad
\mathcal{F}^*= \langle 0| \mathbb{C}\left[\sfa_{1}, \sfa_{2}, \ldots\right]
$$
We also denote
$$
\begin{aligned}
\sfa_{-\lambda}|0\rangle=\sfa_{-\lambda_1} \sfa_{-\lambda_2} \cdots|0\rangle~,  \qquad
\langle 0| \sfa_\lambda=\langle 0| \sfa_{\lambda_1} \sfa_{\lambda_2} \cdots
\end{aligned}
$$
If we introduce a vertex operator
\be\label{Vertex-qt}
V^{q,t}_+=\exp \left(\sum_{n>0} \frac{1-t^n}{1-q^n} \frac{\sfa_{n}}{n} p_n\right)~
\ee
then \eqref{q-Heisenberg} provides $[V^{q,t}_+, \sfa_{-n}]=p_n V^{q,t}_+$ so that we have
$$
\langle 0|V^{q,t}_+ \sfa_{-\lambda}| 0\rangle=p_\lambda\langle 0|V^{q,t}_+| 0\rangle=p_\lambda~.
$$
Therefore, the multiplication of $\langle 0| V$ from the left gives an isomorphism between the Fock space and the ring of symmetric functions
\be\label{Fock-Ring}
\langle 0| V: \mathcal{F} \xrightarrow{\sim} \scR^{q,t}; \
\sfa_{-\lambda}|0\rangle \mapsto p_\lambda ~.
\ee
Similarly, if we introduce a vertex operator
\be\label{dual-Vertex-qt}
V^{q,t}_-=\exp \left(\sum_{n>0} \frac{1-t^n}{1-q^n} \frac{\sfa_{-n}}{n} p_n\right)~
\ee
we have an isomorphism $\iota^*$ by a multiplication of $V^{q,t}_-|0\rangle$ from the right
\be 
V^{q,t}_-|0\rangle:\mathcal{F}^* \xrightarrow{\sim} \scR^{q,t};\
\langle 0| \sfa_\lambda \mapsto p_\lambda~.
\ee
Also, the inner product \eqref{eq:Macdinnprod} can be interpreted as
$$
\left\langle p_\lambda, p_\mu\right\rangle=\left\langle 0\left|\sfa_\lambda \sfa_{-\mu}\right| 0\right\rangle~.
$$

In the subsequent analysis, we focus on this isomorphism by confining it to $N$ variables. For this purpose, we represent the vertex operators \eqref{Vertex-qt} and \eqref{dual-Vertex-qt} using the notation $[V^{q,t}_\pm]_N$, which incorporates power-sum polynomials with $N$ variables. Consequently, we examine the state that corresponds to the Macdonald polynomial with $N$ variables in equation \eqref{dual-Macdonald}.
$$
\left|P_\lambda\right\rangle=\frac{1}{k_\lambda(q,t)}\textrm{C.T.}\prod_{1 \leq i<j \leq N}\left(1-X_j / X_i\right) \cdot P_\lambda(X^{-1}; q, q / t) [V^{q,t}_-]_N|0\rangle
$$
for $\lambda \in P$ with $N \geq \ell(\lambda)$. Here $\textrm{C.T.}$ stands for a constant term with respect to $X$. For example, in this operation, the power-sum functions $p_{n}$ in \eqref{Macdonald-ex} are replaced by the creation operators $\sfa_{-n}$.

The first Macdonald difference operator $D_N^{(1)}$ in \eqref{Macdonald-diff} can be also defined by using $q$-Heisenberg modes. This is done by introducing the vertex operator
$$\widehat D(z):=\exp \left(\sum_{n=1}^{\infty} \frac{1-t^{-n}}{n} \sfa_{-n} z^n\right) \exp \left(-\sum_{n=1}^{\infty} \frac{1-t^n}{n} \sfa_n z^{-n}\right)~.$$

\begin{fact}\cite{Awata:1994xd}
In the $N$-variable version of the isomorphism, $\langle 0|[V^{q,t}_+]_N|u\rangle\in \scR^{q,t}_N$ for $|u\rangle\in \cF$. Then, the action of Macdonald difference operator  $D_N^{(1)}$ is equivalent to
\begin{equation}
\frac{1}{t-1}\langle 0|[V^{q,t}_+]_N \oint \frac{d z}{2 \pi iz}(\widehat D(z)-1)|u\rangle=D_N^{(1)}\langle 0|[V^{q,t}_+]_N|u\rangle~.\label{eq:Macdonalop_vertexop}
\end{equation}
\end{fact}
To generate Macdonald polynomials with a finite number of variables, the topological vertex can also be utilized, as detailed in \cite[Proposition A.1]{Fukuda:2019ywe}.

\subsubsection*{Pieri rules and skew Macdonald polynomials}
In the study of symmetric functions, the Pieri rules provide a way to expand a product of a Macdonald function and either $g_r$ or $e_r$ as a linear combination of Macdonald functions corresponding to Young diagrams obtained by adding $r$ boxes to a given diagram. 

The coefficients in the expansion are known as Pieri coefficients and are expressed as products of the following functions:
\begin{align}
    b_{\lambda}(x)=\begin{dcases}
   k_{\lambda}(q, t)\frac{1-q^{a_{\lambda}(x)}t^{l_{\lambda}(x)+1}}{1-q^{a_{\lambda}(x)+1}t^{l_{\lambda}(x)}}& x\in\lambda\cr
    1& \text{otherwise}
    \end{dcases}
\end{align}
where $k_\lambda$ is defined in \eqref{k_lambda}.
For $\mu\subset\lambda\in \scP$, the Pieri coefficients are defined by using the function $b_\lambda(x)$ as
\bea
  \varphi_{\lambda/\mu}=&\prod_{x\in C_{\lambda/\mu}}\frac{b_{\lambda}(x)}{b_{\mu}(x)},\cr
  \psi_{\lambda/\mu}=&\prod_{x\in R_{\lambda/\mu}-C_{\lambda/\mu}}\frac{b_{\mu}(x)}{b_{\lambda}(x)},\cr
  \varphi'_{\lambda/\mu}=&\prod_{x\in R_{\lambda/\mu}}\frac{b_{\mu}(x)}{b_{\lambda}(x)},\cr
  \psi'_{\lambda/\mu}=&\prod_{x\in C_{\lambda/\mu}-R_{\lambda/\mu}}\frac{b_{\lambda}(x)}{b_{\mu}(x)}
\eea
where $R_{\lambda/\mu}$ (resp. $C_{\lambda/\mu}$) is the union of the rows (resp. columns) that intersect $\lambda-\mu$.  For $\lambda \supset \mu$ where $\lambda-\mu$ is a horizontal strip, the first two are explicitly written as
\begin{align}
    \varphi_{\lambda/\mu}=&\prod_{1\leq i\leq j \leq \ell(\lambda)}\frac{f(q^{\lambda_{i}-\lambda_{j}}t^{j-i})f(q^{\mu_{i}-\mu_{j+1}}t^{j-i})}{f(q^{\lambda_{i}-\mu_{j}}t^{j-i})f(q^{\mu_{i}-\lambda_{j+1}}t^{j-i})}~,\cr
    \psi_{\lambda/\mu}=&\prod_{1\leq i\leq j\leq \ell(\mu)}\frac{f(q^{\mu_{i}-\mu_{j}}t^{j-i})f(q^{\lambda_{i}-\lambda_{j+1}}t^{j-i})}{f(q^{\lambda_{i}-\mu_{j}}t^{j-i})f(q^{\mu_{i}-\lambda_{j+1}}t^{j-i})}~,
\end{align}
where $f(X)=(tX;q)_{\infty}/(qX;q)_{\infty}$. For $\lambda \supset \mu$ where $\lambda-\mu$ is a vertical strip, the third coefficients has an explicit form as
\begin{align}
    \varphi'_{\lambda/\mu}=\prod_{(i,j)}\frac{(1-q^{\lambda_{i}-\lambda_{j}}t^{j-i-1})(1-q^{\mu_{i}-\mu_{j}}t^{j-i+1})}{(1-q^{\lambda_{i}-\lambda_{j}}t^{j-i})(1-q^{\mu_{i}-\mu_{j}}t^{j-i})}
\end{align}
where the product is taken over all $(i,j)$ such that $j>i$, $\lambda_{i}>\mu_{i}$, $\lambda_
{j}=\mu_{j}$. Also, we can express the last coefficients as
\begin{align}
        \psi'_{\lambda/\mu}=\prod_{(i,j)}\frac{(1-q^{\mu_{i}-\mu_{j}}t^{j-i-1})(1-q^{\lambda_{i}-\lambda_{j}}t^{j-i+1})}{(1-q^{\mu_{i}-\mu_{j}}t^{j-i})(1-q^{\lambda_{i}-\lambda_{j}}t^{j-i})}
\end{align}
where the product is taken over all $(i,j)$ such that $j>i$,  $\lambda_{i}=\mu_{i}$ and  $\lambda_{j}=\mu_{j}+1$.
Then, we have the following Pieri rules.
\begin{fact}(Pieri rule)
The Macdonald polynomials are closed under multiplications of symmetric polynomials:
\bea\label{Pieri}
g_{r}P_{\mu}=&\sum_{\lambda}\varphi_{\lambda/\mu}P_{\lambda},\cr
g_{r}Q_{\mu}=&\sum_{\lambda}\psi_{\lambda/\mu} Q_{\lambda},\cr
e_{r}P_{\mu}=&\sum_{\lambda}\psi'_{\lambda/\mu}P_{\lambda},\cr
e_{r}Q_{\mu}=&\sum_{\lambda}\varphi'_{\lambda/\mu}Q_{\lambda}.
\eea
The first two summations are over partitions $\lambda$ such that $\lambda/\mu$ is a horizontal $r$-strip while the last two summations are over the partition such that $\lambda/\mu$ is a vertical $r$-strip.
\end{fact}

For partitions $\mu\subset \lambda$, the skew Macdonald polynomial is defined as
\bea\label{skew-Macdonald}
    P_{\lambda/\mu}=\sum_{\nu}f^{\lambda}_{\mu\nu}P_{\nu}~,\quad \textrm{where} \quad
    f^{\lambda}_{\mu\nu}=\ev{P_{\lambda},Q_{\mu}Q_{\nu}}~.
\eea
The skew Macdonald polynomials have the following properties
\bea
P_{\lambda / \emptyset}(X; q, t)=&P_{\lambda}(X; q, t)~, \cr
P_{\lambda / \mu}(c X; q, t)=&c^{|\lambda|-|\mu|} P_{\lambda / \mu}(X; q, t), \quad c \in \mathbb{C}~, \cr
P_{\lambda / \mu}(X; q^{-1}, t^{-1})=&P_{\lambda / \mu}(X; q, t)~, \cr
P_{\lambda^{t} / \mu^{t}}( X; t, q)=&\frac{k_{\lambda}(q, t)}{k_{\mu}(q, t)} \omega_{q, t} \left(P_{\lambda / \mu}(X; q, t)\right)~
\eea
where $\omega_{q, t}:\scR^{q,t}\to \scR^{q,t}$ is defined by using the power sum polynomials as
\be \label{endmorphism}
\omega_{q, t}:p_{n}\mapsto (-1)^{n-1} \frac{1-q^{n}}{1-t^{n}} p_{n}~.
\ee

\subsubsection*{Cauchy formulas and kernel function}
In topological string theory, Cauchy formulas play an important role.
One of the Cauchy formulas for Macdonald polynomials is defined by
\begin{align}\label{Cauchy-1}
    \Pi(X,Y;q,t):=\sum_{\lambda} k_{\lambda}(q, t) P_{\lambda}(X; q, t) P_{\lambda}(Y; q, t)=&\exp\left(\sum_{n>0}\frac{1}{n}\frac{1-t^{n}}{1-q^{n}}p_{n}(X)p_{n}(Y)\right)\cr
    =&\prod_{i,j}\frac{(tX_{i}Y_{j};q)_{\infty}}{(X_{i}Y_{j};q)_{\infty}}~.
\end{align}
where $\Pi(X,Y;q,t)$ is called the \emph{kernel function}.
It has the following properties:
\begin{fact}
For $x=(X_{1},\ldots,X_{N})$ and $y=(Y_{1},\ldots,Y_{M})$, the kernel function satisfies
\begin{align}
    D_{N}^{(k)}(X;q,t)\Pi(X,Y;q,t)=D_{M}^{(k)}(Y;q,t)\Pi(X,Y;q,t),\quad (k=0,1,\ldots, \text{min}(N,M)).
\end{align}
\end{fact}
\begin{fact}
Another expansion of the kernel function is obtained by using the monomial symmetric function and its dual symmetric function with respect to the inner product \eqref{eq:Macdinnprod}:
\begin{align}
    \Pi(X,Y;q,t)=\sum_{\lambda}g_{\lambda}(X)m_{\lambda}(Y).
\end{align}
\end{fact}

The other Cauchy formula is also important
\begin{align}
\Pi_{0}(x, y) :=\sum_{\lambda} P_{\lambda}(X; q, t) P_{\lambda^{t}}(Y; t, q)=&\exp \left(\sum_{n>0} \frac{(-1)^{n-1}}{n} p_{n}(X) p_{n}(Y)\right) \cr
=&\prod_{i,j}\left(1+X_{i} Y_{j}\right) .
\end{align}
The Cauchy formulas for the skew Macdonald function are
\begin{align}\label{Cauchy-Macdonald}
\sum_{\lambda} \frac{k_{\lambda}(q, t)}{k_{\mu}(q, t)} P_{\lambda / \mu}(X; q, t) P_{\lambda / \nu}(Y; q, t) =&\Pi(X, Y) \sum_{\lambda} P_{\mu / \lambda}(Y; q, t) P_{\nu / \lambda}(X; q, t) \frac{k_{\nu}(q, t)}{k_{\lambda}(q, t)}, \cr
\sum_{\lambda} P_{\lambda / \mu}(X; q, t) P_{\lambda^{t} / \nu^{t}}(Y; t, q) =&\Pi_{0}(X,Y) \sum_{\lambda}^{\lambda} P_{\mu^{t} / \lambda^{t}}(Y; t, q) P_{\nu / \lambda}(X; q, t) .
\end{align}
If we denote by $\omega_{q, t}^{X}$ the endomorphism \eqref{endmorphism} on variables $X$, then
\be
\Pi( X, Y; q, t)=\Pi( X, Y; q^{-1}, t^{-1})=\omega_{t, q}^{X} \omega_{t, q}^{Y} \Pi( X, Y; t, q) .
\ee

\paragraph{Limits of Macdonald functions}
By examining the limits of the Macdonald symmetric functions (polynomials), we can derive a range of other symmetric functions (polynomials). These special limits include the Jack, Hall-Littlewood, and Schur functions, which hold great significance in numerous mathematical and physical applications. 
\begin{figure}[ht]\centering
\includegraphics{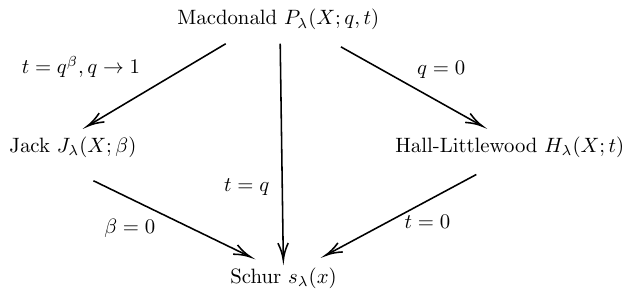}
\end{figure}

Let us list down some special limits that we will use in the main text.
\begin{itemize}
    \item The $t=q^\beta, q \rightarrow 1$ limit leads to Jack functions $J_{\lambda}(X)$
    \item The $q=0$ limit leads to the Hall-Littlewood functions $H_{\lambda}(X)$
       \item The $t=q$ limit leads to the Schur functions $s_{\lambda}(X)$
 \item Elementary symmetric function
    \begin{align}
        \lim_{q\rightarrow 1}P_{\lambda}(X;q,t)=e_{\lambda^{t}}(X).
    \end{align}
    \item Monomial symmetric function
    \begin{align}
        \lim_{t\rightarrow 1}P_{\lambda}(X;q,t)=m_{\lambda}(X).
    \end{align}
\end{itemize}
Among them, Jack and Schur functions play roles in the main text so that we will elaborate on them in the following subsections.

\subsection{Jack functions}\label{app:Jack}

Jack functions naturally show up in the study of the Calogero-Sutherland model \cite{sutherland1971exact,sutherland1972exact}. In the Calogero-Sutherland model, $N$ particles are bounded on a circle of circumference length $L$ with a potential of $\sin^{-2}$ type interaction. Since the particles are bounded on the circle, they have only a discrete energy spectrum. The  Hamiltonian is given by
$$
\mathcal{H}=\frac12 \sum_{j=1}^N \left(\frac{1}{i}\frac{\partial}{\partial x_j}\right)^2+\frac{\pi^2}{2L^2}\sum_{i\neq j}\frac{\beta(\beta-1)}{\sin^2\frac{\pi}{L}(x_i-x_j)}
$$
where $x_i\left(0 \leq x_i \leq L\right)$ are the coordinates of the $i$-th particle, $p_i=\frac{1}{i} \frac{\partial}{\partial x_i}$ is the momentum of the $i$-th particle and $\beta$ is the coupling constant. It is well-known \cite{sutherland1971exact} that the ground state wave function is
$$
\psi_0\left(x_1, \cdots, x_N\right)=\frac{1}{\sqrt{N !}}(2 \sqrt{-1})^{\beta \frac{N(N-1)}{2}}\left(\prod_{i<j} \sin \frac{\pi}{L}\left(x_i-x_j\right)\right)^\beta
$$
with energy eigenvalue
$E_0=\beta^2\left(\frac{\pi}{L}\right)^2 \frac{(N+1) N(N-1)}{6}$.
By writing $X_j=\exp(2\pi i x_j/L)$, we rewrite the Hamiltonian by taking the conjugation with respect to the ground state wave function
\be\label{CS-Hamiltonian}
H_\beta=(\psi_0)^{-1} \cdot\mathcal{H}\cdot \psi_0
= \sum_{i=1}^{N}\left(X_{i} \frac{\partial}{\partial X_{i}}\right)^{2}+\beta \sum_{i<j} \frac{X_{i}+X_{j}}{X_{i}-X_{j}}\left(X_{i} \frac{\partial}{\partial X_{i}}-X_{j} \frac{\partial}{\partial X_{j}}\right)
\ee

There is a unique eigenfunction of the Hamiltonian that takes the form
\be J_{\lambda}(X)=m_{\lambda}(X)+\sum_{\mu < \lambda} u_{\lambda, \mu}(\beta) m_{\mu}(X), \quad u_{\lambda, \lambda}(\beta)\in \bQ(\beta)~.\ee
where the eigenvalue is
\be\label{CS-Jack} H_{\beta} J_{\lambda}(X)=\epsilon_{\lambda} J_{\lambda}(X)~,\quad
\textrm{where}
\quad
\epsilon_{\lambda}=\sum_{i=1}^{N}\left(\lambda_{i}^{2}+\beta(N+1-2 i) \lambda_{i}\right)
\ee
They are called \emph{Jack polynomials}.
As the degeneration  of \eqref{eq:Macdinnprod}, we introduce an inner product in the ring $\scR^\beta_N$ by
\begin{equation}
\lim _{\substack{t=q^\beta \\ q \rightarrow 1}}\langle f, g\rangle_{q, t}=\langle f, g\rangle_{\beta} ~,\qquad\left\langle p_{\lambda}, p_{\mu}\right\rangle_{\beta}=\delta_{\lambda, \mu} z_{\lambda} \beta^{-\ell(\lambda)}, \quad z_{\lambda} \equiv \prod_{i \geq 1} i^{m_{i}} m_{i} !~.
\end{equation}
Jack functions are orthogonal with respect to this inner product
\begin{equation}
\left\langle J_{\lambda}, J_{\mu}\right\rangle_\beta=\delta_{\lambda, \mu} \prod_{(i,j)\in \lambda} \frac{a(i,j)+\beta l(i,j)+1}{a(i,j)+\beta l(i,j)+\beta}~.
\end{equation}

We define Dunkl operators as
$$
\nabla_{i}=\frac{\partial}{\partial X_{i}}+\beta \sum_{j(\neq i)} \frac{1}{X_{i}-X_{j}}\left(1-s_{ij}\right)
$$
where $s_{ij}$ is the permutation of the variables $X_i$ and $X_j$. Then, they satisfy the following properties
\begin{align}
\left[\nabla_{i}, \nabla_{j}\right]=0 ~, \quad
s_{ij} \nabla_{j}=\nabla_{i} s_{ij} ~, \quad
\left[\nabla_{i}, X_{j}\right]=\delta_{i,j}(1+\beta \sum_{l=1}^{N} s_{i l})-\beta s_{ij}~.
\end{align}
In fact, if we restrict the action
$X_i \nabla_i$ on the ring $\scR^{\beta}_N$ of symmetric functions, it is the same as the action of $y_i\in \dH_N$
under the polynomial representation \eqref{SdH-pol} up to a constant. In fact
 $H_\beta$ in \eqref{CS-Hamiltonian} takes a simple form:
$$
H_\beta= p_2(X_1 \nabla_1,\ldots,X_N \nabla_N)\Big|_{\scR^{\beta}_N~~~}
$$
where $p_2$ is the second power-sum polynomial. Moreover, $H_\beta$ belongs to a set of $N$ mutually commuting operators
$$H_\beta^{(k)}:= p_k(X_1 \nabla_1,\ldots,X_N \nabla_N), \qquad (k=1, \ldots, N)$$
on $\scR^{\beta}_N$. Jack functions are eigenfunctions of these commuting operators where eigenvalues depend on $N$ as in \eqref{CS-Jack}. Generators $\sfD_{0,k}$ in spherical degenerate DAHA $\SdH_N$ are variants of these commuting operators where eigenvalues are independent of the rank $N$ as in \eqref{y-pol}. Therefore, the Calogero-Sutherland model is completely integrable, and spherical degenerate DAHA $\SdH_N$ is the symmetry of this integrable system.

As we see in \eqref{Fock-Ring}, Macdonald functions admit the free-field realization by free $q$-boson. Here let us consider the free-field realization of Jack functions by the free boson.

First, we define the eigenstate of the zero mode $\sfJ_0$
\begin{equation}
|t\rangle=e^{t \sfq_0}|0\rangle \qquad \sfJ_0|t\rangle=t|t\rangle
\end{equation}
where $\sfq_0$ is the center position in the mode expansion \eqref{boson-modes} of the free boson. We consider the Fock space generated by the Heisenberg modes from $|t\rangle$
\begin{equation}
    \cF_t= \bC[\sfJ_{-1},\sfJ_{-2},\ldots]|t\rangle~.
\end{equation}
Degeneration versions of the vertex operators \eqref{Vertex-qt} and \eqref{dual-Vertex-qt} are
\be\label{VO-beta}
V^\beta_+=\exp \left(\sqrt{\beta} \sum_{n=1}^{\infty} \frac{1}{n} p_n \sfJ_{n}\right)~,\quad V^\beta_-=\exp \left(\frac{1}{\sqrt{\beta}} \sum_{n=1}^{\infty} \frac{1}{n} p_n \sfJ_{-n}\right)~
\ee
which are close cousins of \eqref{vertex-op-pm}.
Then, there is an isomorphism between the Fock space $\cF_t$ and the ring $\scR^\beta$ of symmetric functions
\begin{equation}\label{beta-iso}
\langle t|V: \cF_t \xrightarrow{\sim} \scR^\beta; \ |u\rangle \mapsto \langle t|V^\beta|u\rangle
\end{equation}
Under this isomorphism, the power-sum functions can be expressed as
\begin{equation}\label{J-p}
\langle t|V^\beta\sfJ_{-n}|t\rangle=\sqrt{\beta} p_n
\end{equation}
Consequently, the action of the Calogero-Sutherland Hamiltonian \eqref{CS-Hamiltonian} on $\scR^{\beta}_N$ is equivalent to the action of the following operator on $\cF_t$
\begin{equation}\label{CS-free-field}
\widehat{H}_\beta =\sqrt{\beta}\sum_{n, m \geq 1}\left( \sfJ_{-n-m} \sfJ_n \sfJ_m+ \sfJ_{-n} \sfJ_{-m} \sfJ_{n+m}\right)+\sum_{n=1}^{\infty}((1-\beta) n+N \beta) \sfJ_{-n} \sfJ_n~.
\end{equation}
Namely, we have the following commutative diagram:
\[ \begin{tikzcd}
\cF_t \arrow{r}{\widehat{H}_\beta} \arrow[swap]{d}{\langle t| V} & \cF_t \arrow{d}{\langle t| V} \\
\scR^{\beta} \arrow{r}{{H}_\beta}& \scR^{\beta}
\end{tikzcd}
\]

It was found in \cite{r:MY} that there is a relation between singular states \eqref{singular-state} of the Virasoro minimal model and Jack functions. To see that, we shift the correspondence \eqref{J-p} between the power-sum functions and the Heisenberg modes by $\sqrt{2}$:
\be
\tau: \scR^\beta \to \cF_t: \ p_n \mapsto \sqrt{\frac{2}{\beta}} \sfJ_{-n}~.
\ee
Then, the singular state \eqref{singular-state} can be constructed (up to a constant) by acting Jack functions with a rectangular Young diagram $(s^r)$
\be\label{singular-state2}
|\Delta_{r,s}\rangle \sim \tau(J_{(s^r)})|t_{r,s}\rangle
\ee
where $t=t_{r,s}$ is given in \eqref{trs}. Here $\sim$ means up to a constant.
As it is well-known, the singular states can be constructed by the Kac determinant. Since the energy-momentum tensor $T(z)$ admits the Miura transformation \eqref{Virasoro_g}, the Virasoro modes can be written in terms of the Heisenberg modes
\be
\sfL_n=\frac{1}{2} \sum_{k \in \boldsymbol{Z}}: \sfJ_{n-k} \sfJ_k:-\frac{Q}{\sqrt{2}}(n+1) \sfJ_n
\ee
Using this free-field (Feigin-Fuks) presentation of the Virasoro modes, we can write the singular states in terms of the Heisenberg modes. This coincides with \eqref{singular-state2} up to a constant.
Here we give a few examples for an illustrative purpose:
\be\begin{array}{rlrl}
J_{(1)}=& p_1,& |\chi_{1,1}\rangle \sim&\sfJ_{-1}\left|t_{1,1}\right\rangle\\
J_{(2)}=& \frac{1}{1+\beta} p_2+\frac{\beta}{1+\beta} p_1^2, &|\chi_{1,2}\rangle \sim&\left( \sfJ_{-2}+\sqrt{2\beta} \sfJ_{-1}^2\right)\left|t_{1,2}\right\rangle\\
J_{(1^2)}=&-\frac{1}{2} p_2+\frac{1}{2} p_1^2,  &|\chi_{2,1}\rangle \sim&\left(- \sfJ_{-2}+\sqrt{\frac{2}{\beta}}\sfJ_{-1} \sfJ_{-1}\right)\left|t_{2,1}\right\rangle\\
J_{(3)}=& \frac{2}{(1+\beta)(2+\beta)} p_3+\frac{3 \beta}{(1+\beta)(2+\beta)} p_2 p_1 & |\chi_{1,3}\rangle \sim&\left(\frac{1}{\beta} \sfJ_{-3}+\frac{3}{\sqrt{2 \beta}} \sfJ_{-2} \sfJ_{-1}+\sfJ_{-1}^3\right) | t_{1,3}\rangle\\
&+\frac{\beta^2}{(1+\beta)(2+\beta)} p_1^3, \\
J_{\left(1^3\right)}=& \frac{1}{3} p_3-\frac{1}{2} p_2 p_1+\frac{1}{6} p_1^3 &|\chi_{3,1}\rangle \sim&\left(\beta \sfJ_{-3}-3 \sqrt{\frac{\beta}{2}} \sfJ_{-2} \sfJ_{-1}+\sfJ_{-1}^3\right)\left|t_{3,1}\right\rangle.
\end{array}\ee

\paragraph{Generalized Jack}
In the proof of the AGT conjecture for $\SU(N)$ gauge group \cite{Schiffmann:2012aa}, it becomes crucial to employ the coproduct of the affine Yangian for the purpose of augmenting the rank, $N$.  In the Calogero-Sutherland context, it corresponds to the definition of the Hamiltonian for $N$-sets of the Calogero-Sutherland systems
\begin{align}
    \widehat{H}_\beta &= \sum_{i=1}^N \widehat{H}^{(i)}_\beta +(1-\beta)\sum_{i>j} \sum_{n} n \sfJ^{(i)}_{-n}\sfJ^{(j)}_n
\end{align}
where $\widehat{H}^{(i)}_\beta$ is (\ref{CS-free-field}) where $\sfJ_n$ is replaced by $\sfJ^{(i)}_n$. It coincides with the mutually commuting conserved charge discussed in (\ref{GeneralizedCS}).
An eigenfunction of the generalized Hamiltonian is sometimes referred to as a generalized Jack polynomial. It is labeled by $N$-tuple Young diagrams $(\lambda_1,\cdots,\lambda_N)$ and is related to a tensor product of the Fock representation in \S\ref{sec:tensor_product}. Some of its properties, such as the Selberg-integral, were proposed in \cite{Morozov:2013rma}.

\subsection{Schur functions}\label{app:Schur}

As another important limit, we can take the $t=q$ limit of Macdonald functions, which become functions of only $X$, independent of $q$. The resulting function is called a Schur function  $s_\lambda(X)$. Or equivalently, they can be obtained by the $\beta=1$ limit of Jack functions. More directly, the Schur function for a Young diagram $\lambda$ can be defined by the so-called bialternant formula
\be
s_\lambda(X_1,\dots,X_N):=\frac{\det(X_i^{\lambda_j+n-j})}{\det(X_i^{n-j})}~.
\ee
Likewise, the $t=q$ limit of the \eqref{skew-Macdonald} become functions of only $X$, which are called skew Schur functions. We can directly define the skew Schur functions by taking the product of two Schur functions
\be
s_\mu s_\nu=\sum_\lambda c^\lambda_{\mu\nu}s_\lambda~.
\ee
Then, the skew Schur function is given by
\be\label{skew-schur}
s_{\lambda/\mu}:=\sum_\nu c^\lambda_{\mu\nu}s_\nu~.
\ee

Moreover, the connection of (skew) Schur functions to the free boson and fermion \cite{kac1990infinite,macdonald1998symmetric} become much more evident than Macdonald and Jack functions. Let $\uppsi_n,\uppsi_n^\ast$ be the free fermion modes, namely
\be
\psi(z)=\sum_{\a\in \bZ+1/2}\uppsi_\a z^{-\a}~,\qquad \overline\psi(z)=\sum_{\a\in \bZ+1/2}\uppsi_\a^\ast z^{-\a}~.
\ee
Then, the free boson modes are expressed as the normal-ordered products of the free fermion modes
\be
\sfJ_n:=\sum_{\a\in \mathbb{Z}+1/2}:\uppsi_{-\a}\uppsi^\ast_{\a+n}:~.
\ee
They satisfy the following (anti-)commutation relations
\begin{align}\label{free-boson-fermion}
&\{\uppsi_\a,\uppsi_\b\}=\{\uppsi_\a^\ast,\uppsi_\b^\ast\}=0,\quad \{\uppsi_\a,\uppsi^\ast_\b\}=\delta_{\a+\b,0}, \cr
&\lt[\sfJ_n,\uppsi_\a\rt]=\uppsi_{n+\a},\quad \lt[\sfJ_n,\uppsi^\ast_\a\rt]=-\uppsi^\ast_{n+\a},\quad \lt[\sfJ_n,\sfJ_m\rt]=n\delta_{n+m,0}.
\end{align}
Now let us define a state $|\lambda\rangle$ in the fermionic Fock space with the label of the Frobenius presentation  \eqref{Frobenius} of a Young diagram $\lambda$. When the Frobenius presentation of a Young diagram $\lambda$ is $\lambda=(\alpha_1,\alpha_2,\dots|\beta_1,\beta_2\dots)$, then $\ket{\lambda}$ is given by
\be\label{lambda-state}
|\lambda\rangle=(-1)^{\beta_1+\beta_2+\cdots+\beta_s+\frac{s}{2}} \uppsi_{-\beta_1}^* \uppsi_{-\beta_2}^* \ldots \uppsi_{-\beta_s}^* \uppsi_{-\alpha_s} \uppsi_{-\alpha_{(s-1)}} \ldots \uppsi_{-\alpha_1}|0\rangle~,
\ee
where $s$ is the number of boxes on the diagonal line in $\lambda$ and the vacuum state $\ket{0}$ satisfies $\uppsi_\alpha\ket{0}=\uppsi^\ast_\beta\ket{0}=0$ for any $\alpha>0$, $\beta>0$.

Now we introduce the $\beta=1$ limit of the vertex operators \eqref{VO-beta}
\be\label{Vpm}
V_{\pm}(X)=\exp \left(\sum_{n=1}^{\infty} \frac{1}{n} p_n(X) \sfJ_{\pm n}\right)~.
\ee
Then, the skew Schur function is defined by sandwiching the vertex operators with the fermionic states
\begin{equation}\label{skew-Schur-fermion}
s_{\lambda / \mu}(X)=\left\langle\mu\left|V_{+}(X)\right| \lambda\right\rangle=\left\langle\lambda\left|V_{-}(X)\right| \mu\right\rangle~.
\end{equation}
In the Schur limit, the two Cauchy formulas \eqref{Cauchy-Macdonald} become
\begin{align}
    \sum_\lambda s_{\lambda/\mu}(X)s_{\lambda/\nu}(Y)=\prod_{i,j}(1-X_iY_j)^{-1}\sum_\sigma s_{\nu/\sigma}(X)s_{\mu/\sigma}(Y),\label{Schur-nor-id}\\
\sum_\lambda s_{\lambda/\mu^t}(X)s_{\lambda^t/\nu}(Y)=\prod_{i,j}(1+X_iY_j)\sum_\sigma s_{\nu^t/\sigma}(X)s_{\mu/\sigma^t}(Y)\label{Schur-twist-id}.
\end{align}
These formulas can be nicely interpreted by using the definition \eqref{skew-Schur-fermion} of skew Schur functions by the vertex operators. The vertex operators satisfy the commutation relation
\be
V_+(X)V_-(Y)=\prod_{i,j}\frac{1}{1-X_iY_j}V_-(Y)V_+(X)~.\label{V-V+}
\ee
Therefore, this yields the first identify \eqref{Schur-nor-id} as
\be
\sum_\lambda \bra{\mu}V_+(X)\ket{\lambda}\bra{\lambda}V_-(Y)\ket{\nu} =\prod_{i,j}\left(1-X_i Y_j\right)^{-1} \sum_\sigma\left\langle\mu\left|V_{-}(Y)\right| \sigma\right\rangle\left\langle\sigma\left|V_{+}(X)\right|\nu\right\rangle~
\ee
where we insert the identity operator $1=\sum_\lambda \ket{\lambda}\bra{\lambda}$.

To derive the second identity \eqref{Schur-twist-id}, we introduce the involution of the algebra \eqref{free-boson-fermion}:
\be
\iota: \begin{cases}\uppsi_\a \mapsto \uppsi^*_\a \\ \uppsi_\a^* \mapsto \uppsi_\a  \\ \sfJ_n\mapsto -\sfJ_n\end{cases}~.
\ee
As the involution $\iota$ exchanges $\{\alpha_i\}$ and $\{\beta_i\}$ in \eqref{lambda-state}, we have
\begin{equation}
|\lambda^t\rangle=(-1)^{|\lambda|} \iota(|\lambda\rangle) .
\end{equation}
Since the vertex operators are mapped as $\iota: V_\pm(X)\to V_\pm^{-1}(X)$, we have
$$
\begin{aligned}
\langle\lambda^t|V_{-}(Y)| \nu\rangle =&(-1)^{|\lambda|} \iota(\langle\lambda|) V_{-}(Y)|\nu\rangle \\
=&(-1)^{|\lambda|} \iota^2(\langle\lambda|) \iota(V_{-}(Y)) \iota(|\nu\rangle) \\
=&(-1)^{|\lambda|-|\nu|}\langle\lambda|\iota(V_{-}(Y))| \nu^t\rangle \\
=&\langle\lambda|V_{-}^{-1}(-Y)| \nu^t\rangle~.
\end{aligned}
$$
where we use $\iota^2=\textrm{id}$.
Therefore, the commutation relation of the vertex operators
$$
V_{+}(X) V_{-}^{-1}(-Y)=\prod_{i,j}(1+X_i Y_j) V_{-}^{-1}(-Y) V_{+}(X)
$$
provide the second Cauchy identity \eqref{Schur-twist-id} as follows:
$$
\begin{aligned}
\sum_\lambda\langle\mu^t|V_{+}(X)| \lambda\rangle\langle\lambda^t|V_{-}(Y)| \nu\rangle =&\sum_\lambda\langle\mu^t|V_{+}(X)| \lambda\rangle\langle\lambda|V_{-}^{-1}(-Y)| \nu^t\rangle \\
=&\langle\mu^t|V_{+}(X) V_{-}^{-1}(-Y)| \nu^t\rangle \\
=&\prod_{i,j}(1+X_i Y_j)\langle\mu^t|V_{-}^{-1}(-Y) V_{+}(X)| \nu^t\rangle \\
=&\prod_{i,j}(1+X_i Y_j) \sum_\uppsi\langle\mu^t|V_{-}^{-1}(-Y)| \uppsi\rangle\langle\uppsi|V_{+}(X)| \nu^t\rangle \\
=&\prod_{i,j}(1+X_i Y_j) \sum_\uppsi\langle\mu|V_{-}(Y)| \uppsi^t\rangle\langle\uppsi|V_{+}(X)| \nu^t\rangle~.
\end{aligned}
$$

\section{Computations associated to singular vectors}\label{a:sing}

In this Appendix, we present a more detailed review on the construction of singular vectors in the Coulomb gas representation, which compensates the argument used in \eqref{singular_state} and \eqref{onshell_c}. We mainly follow the discussion presented in \cite{Kato:1985vq}. 

We first define a vertex operator, which later will be shown to generate singular vectors in the Virasoro algebra, 
\begin{equation}
    \Psi^\pm_r(t,z):=\oint_{C_r}{\rm d}z_r\int_z^{z_r}{\rm d}z_{r-1}\dots\int_z^{z_2}{\rm d}z_{1}\cV(t_\pm,z_r)\dots\cV(t_\pm,z_1)\cV(t,z).
\end{equation}
Using Wick's theorem and the OPE relation \eqref{prodVV}, it can be rewritten to 
\bea
    \Psi^\pm_r(t,z)
    =&\oint_{C_r}{\rm d}z_r\int_z^{z_r}{\rm d}z_{r-1}\dots\int_z^{z_2}{\rm d}z_{1}f^\pm_r(t;z,\{z_i\})\cr
    &\times:\exp\lt(-t_\pm\sum_{i=1}^r\sum_{\ell=0}^\infty\frac{(z_i-z)^{\ell+1}}{(\ell+1)!} \partial^\ell[\partial\phi](z)\rt)\cV(rt_\pm+t,z): ~ ,\label{ex-sing}
\eea
where
\bea
    f^\pm_r(t;z,\{z_i\})=\prod_{i<j}^r(z_j-z_i)^{t_\pm^2}\prod_{i=1}^r(z_i-z)^{t_\pm t},\label{contract-const}\\
    [\partial\phi](z)=\sum_{n=-\infty}^{\infty}\sfJ_nz^{-n-1}~.
\eea
Let us first derive the so-called on-shell condition for the case of $r=3$ as an example, and the discussion can then easily be generalized to the general cases. We expand the exponential in \eqref{ex-sing}, and obtain a sum of infinite pieces of integrals, with each of them taking the form,
\begin{equation}
    \oint_{C_3}{\rm d}z_3\int_z^{z_3}{\rm d}z_2\int_z^{z_2}{\rm d}z_1\ \prod_{i<j}(z_j-z_i)^{t_\pm^2}(z_3-z)^{t_\pm t+n_3}(z_2-z)^{t_\pm t+n_2}(z_1-z)^{t_\pm t+n_1}
\end{equation}
for some positive integers $n_{1,2,3}\in\mathbb{Z}_+$. When performing the integral over $z_1$, we can further expand terms of the form $(z_3-z_1)^a=(z_3-z_2+z_2-z_1)^a$ around $z_1\sim z_2$ to instead calculate a sum of a bunch of integrals proportional to
\begin{align}
    \oint_{C_3}{\rm d}z_3\int_z^{z_3}{\rm d}z_2\int_z^{z_2}{\rm d}z_1\ (z_3-z_2)^{2t_\pm^2-k_1}(z_2-z_1)^{t_\pm^2+k_1}(z_3-z)^{t_\pm t+n_3}(z_2-z)^{t_\pm t+n_2}(z_1-z)^{t_\pm t+n_1}\cr
    \propto\oint_{C_3}{\rm d}z_3\int_z^{z_3}{\rm d}z_2(z_3-z_2)^{2t_\pm^2-k_1}(z_3-z)^{t_\pm t+n_3}(z_2-z)^{t_\pm t+n_2}(z_2-z)^{1+t^2_\pm+t_\pm t+k_1+n_1}\cr
    \propto\oint_{C_3}{\rm d}z_3(z_3-z)^{3t_\pm^2+3t_\pm t+n_1+n_2+n_3+2}.
\end{align}
This integral can be non-zero only when $3t_\pm^2+3t_\pm t+n_1+n_2+n_3+2=-1$, which reproduces \eqref{onshell_c} for $r=3$ and $N=n_1+n_2+n_3$. For larger $r$, the argument is similar, the term $\prod_{i<j}(z_j-z_i)^{t_\pm^2}$ in \eqref{ex-sing} contributes a power $\frac{r(r-1)}{2}t_\pm^2$, and the product $\prod_{i}(z_i-z)^{t_\pm t}$ gives rise to a power $rt_\pm t$ to the final integrand, and in this way we obtained the general on-shell condition \eqref{onshell_c}
\begin{equation}
    \frac{r(r-1)}{2}t_\pm^2+rt_\pm t+N+r=0
\end{equation}
with $N=\sum_{i=1}^rn_i$.

To check how $\Psi^\pm_r$ transform under the conformal transformation, we first recall that under a finite coordinate transformation
\begin{equation}
    z\rightarrow z'=z+\varepsilon(z)=z+\sum_{n=-\infty}^\infty\epsilon_n z^{n+1}~.
\end{equation}
The vertex operator obeys the transformation rule
\begin{equation}
    \cV(t,z)\rightarrow U\cV(t,z)U^{-1}=\lt(\frac{{\rm d}z'}{{\rm d}z}\rt)^{\Delta_0(t)}\cV(t,z')~,
\end{equation}
where
\begin{equation}
    U=\exp\lt(\sum_{n=-\infty}^\infty\epsilon_n\sfL_n\rt),
\end{equation}
generates a 2d (holomorphic) conformal transformation. In particular, the screening charge is invariant under the conformal transformation as
\begin{equation}
    U\lt(\cV(t_\pm,z){\rm d}z\rt)U^{-1}=\cV(t_\pm,z'){\rm d}z'.
\end{equation}
This immediately implies that the conformal transformation of $\Psi^\pm_r$ is completely controlled by that of $\cV(t,z)$, so
\begin{equation}
    U\Psi^\pm_r(t,z)U^{-1}=\lt(\frac{{\rm d}z'}{{\rm d}z}\rt)^{\Delta_0(t)}\Psi^\pm_r(t,z').
\end{equation}
This further shows that $\Psi^\pm_r$ transforms as a primary operator with conformal dimension $\Delta_0(t)$.

The state $\ket{\Psi^\pm_r}$ that corresponds to $\Psi^\pm_r$ is clearly generated from the highest weight state $\ket{t}$, therefore one expects it to be written as a linear combination of descendant states of $\ket{t}$. On the other hand, it transforms as a primary state or equivalently as a highest weight state. The only way to make these two statements consistent is to have $\ket{\Psi^\pm_r}$ be a null state, i.e. $\Psi^\pm_r$ generate singular vectors in the Virasoro algebra.

\section{Combinatorial computations}\label{sec:appendix-combinatorial}
In this appendix, we give a detailed derivation of the combinatorial computations in \S\ref{sec:verticalFockrep}. Let us prove the following identity:
\begin{align}
    \mathcal{Y}_{\lambda}(z,u)\coloneqq&(1-u/z)\prod_{\sAbox\in\lambda}S(\chi_{\sAbox}/z)
    =\frac{\prod\limits_{\sAbox\in \frakA(\lambda)}(1-\chi_{\sAbox}/z)}{\prod\limits_{\sAbox\in \frakR(\lambda)}(1-q_{3}^{-1}\chi_{\sAbox}/z)}=\prod_{i=1}^{\infty}\frac{1-x_{i}/z}{1-q_{1}x_{i}/z},\cr
    &\chi_{\sAbox}=uq_{1}^{i-1}q_{2}^{j-1},\quad x_{i}=uq_{1}^{i-1}q_{2}^{\lambda_{i}},\quad S(z)=\frac{(1-q_{1}z)(1-q_{2}z)}{(1-z)(1-q_{3}^{-1}z)}
\end{align}
where we set $c=3$ and omit the index. Using the fact that a rational function $\frac{1}{1-x}$ is rewritten as 
\begin{align}
    \frac{1}{1-x}=\exp\left(\sum_{n\in\mathbb{Z}}\frac{1}{n}x^{n}\right)=\text{P.E.}\left[x\right],\quad 1-x=\exp\left(-\sum_{n\in\mathbb{Z}}\frac{x^{n}}{n}\right)=\text{P.E.}\left[-x\right],\quad |x|<1,
\end{align}
the function $\mathcal{Y}_{\lambda}(z,u)$ can be written as
\begin{align}
    \mathcal{Y}_{\lambda}(z,u)=&\exp\left(-\sum_{n\in\mathbb{Z}}\frac{z^{-n}}{n}\left(u^{n}-(1-q_{1}^{n})(1-q_{2}^{n})\sum_{\sAbox\in\lambda}(\chi_{\sAbox})^{n}\right)\right)\cr
    =&\text{P.E.}\left[-\left(u-(1-q_{1})(1-q_{2})\sum_{\sAbox\in\lambda}\chi_{\sAbox}\right)z^{-1}\right]\cr
    =&\text{P.E.}\left[-\left(\sum_{\sAbox\in \frakA(\lambda)}\chi_{\sAbox}-q_{1}q_{2}\sum_{\sAbox\in \frakR(\lambda)}\chi_{\sAbox}\right)z^{-1}\right]\cr
    =&\text{P.E.}\left[-(1-q_{1})\sum_{i=1}^{\infty}x_{i}z^{-1}\right].
\end{align}
Therefore, it is enough the prove the following identities
\begin{align}
    u-(1-q_{1})(1-q_{2})\sum_{\sAbox\in\lambda}\chi_{\sAbox}=\sum_{\sAbox\in \frakA(\lambda)}\chi_{\sAbox}-q_{1}q_{2}\sum_{\sAbox\in \frakR(\lambda)}\chi_{\sAbox}=(1-q_{1})\sum_{i=1}^{\infty}x_{i}.\label{eq:Xvariables1}
\end{align}

Let us show the first equality. We have the following identity 
\begin{align}
    (1-q_{1})(1-q_{2})\sum_{\sAbox\in\lambda}\chi_{\sAbox}=\hspace{-0.7cm}\adjustbox{valign=c}{\includegraphics[width=5.3cm]{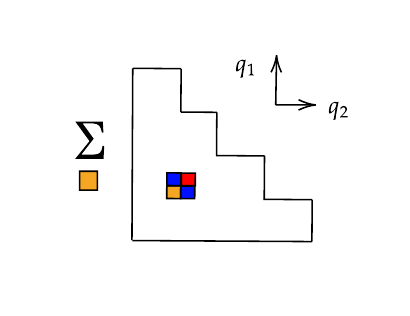}}=\hspace{-0.7cm}\adjustbox{valign=c}{\includegraphics[width=5cm]{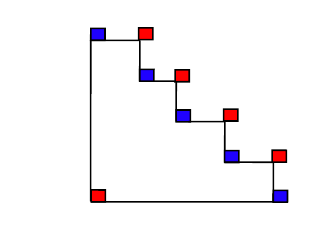}}.
\end{align}
The orange and red boxes correspond to terms with the positive signs in front of them while the blue boxes correspond to terms with negative signs. While moving the orange box inside the Young diagram and summing the terms, the term with positive and negative signs cancel with each other and we obtain the second equality. For example, terms coming when we move the orange box in the first row are 
\begin{align}
\begin{split}
    (1-q_{1})(1-q_{2})\sum_{j=1}^{\lambda_{1}}uq_{2}^{j-1}=&\sum_{j=1}^{\lambda_{1}}u(1-q_{2})q_{2}^{j-1}-q_{1}\sum_{j=1}^{\lambda_{1}}u(1-q_{2})q_{2}^{j-1}\\
    =&u+\cancel{u\sum_{j=2}^{\lambda_{1}}q_{2}^{j-1}}-\cancel{u\sum_{j=2}^{\lambda_{1}}q_{2}^{j}}-uq_{2}^{\lambda_{1}}\\
    &\quad -q_{1}\left(u+\bcancel{u\sum_{j=2}^{\lambda_{1}}q_{2}^{j-1}}-\bcancel{u\sum_{j=2}^{\lambda_{1}}q_{2}^{j}}-uq_{2}^{\lambda_{1}}\right)\\
    =&u(1-q_{2}^{\lambda_{1}})- uq_{1}(1-q_{2}^{\lambda_{1}}).
\end{split}
\end{align}
Similarly, terms coming from the second row are
\begin{align}
    (1-q_{1})(1-q_{2})\sum_{j=1}^{\lambda_{2}}uq_{1}q_{2}^{j-1}=&uq_{1}(1-q_{2}^{\lambda_{2}})-uq_{1}^{2}(1-q_{2}^{\lambda_{2}}).
\end{align}
Summing these terms, the term with $q_{1}$ will disappear when $\lambda_{1}=\lambda_{2}$ but remain when $\lambda_{1}\neq \lambda_{2}$. Doing this procedure order and order with respect to $q_{1}$, one will obtain the equality
\begin{align}
\begin{split}
    u-(1-q_{1})(1-q_{2})\sum_{\sAbox\in\lambda}\chi_{\sAbox}=&\adjustbox{valign=c}{\includegraphics[width=4cm]{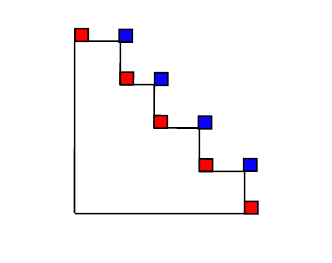}}\\
    =&\sum_{\sAbox\in \frakA(\lambda)}\chi_{\sAbox}-q_{1}q_{2}\sum_{\sAbox\in \frakR(\lambda)}\chi_{\sAbox}.
\end{split}
\end{align}

Let us next show the second equality. We can depict the terms coming from $\sum_{i=1}^{\infty}x_{i}$ as 
\begin{align}
    &\sum_{i=1}^{\infty}x_{i}=\sum_{i=1}^{\infty}uq_{1}^{i-1}q_{2}^{\lambda_{i}}=\adjustbox{valign=c}{\includegraphics[width=4cm]{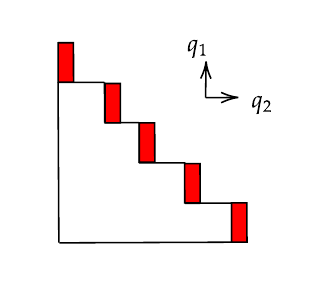}},\quad
    q_{1}\sum_{i=1}^{\infty}x_{i}=\sum_{i=1}uq_{1}^{i}q_{2}^{\lambda_{i}}=\adjustbox{valign=c}{\includegraphics[width=4cm]{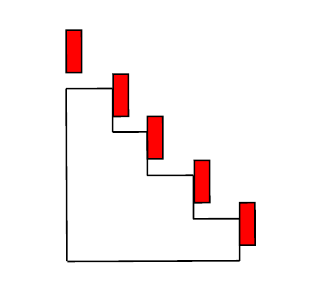}}
\end{align}
which leads to 
\begin{align}
    (1-q_{1})\sum_{i=1}^{\infty}x_{i}=\adjustbox{valign=c}{\includegraphics[width=4cm]{Figures/combinatorial5.pdf}}=\sum_{\sAbox\in \frakA(\lambda)}\chi_{\sAbox}-q_{1}q_{2}\sum_{\sAbox\in \frakR(\lambda)}\chi_{\sAbox}
\end{align}

Actually, we can show the identity 
\begin{align}
 u-(1-q_{1})(1-q_{2})\sum_{\sAbox\in\lambda}\chi_{\sAbox}=(1-q_{1})\sum_{i=1}^{\infty}uq_{1}^{i-1}q_{2}^{\lambda_{i}}\label{eq:Xvariables2}
\end{align}
in a more explicit way as
\begin{align}
\begin{split}
    u-(1-q_{1})(1-q_{2})\sum_{i=1}^{\infty}\sum_{j=1}^{\lambda_{i}}uq_{1}^{i-1}q_{2}^{j-1}=&u-(1-q_{1})(1-q_{2})\sum_{i=1}^{\infty}uq_{1}^{i-1}\frac{1-q_{2}^{\lambda_{i}}}{1-q_{2}}\\
    =&u-(1-q_{1})\sum_{i=1}^{\infty}uq_{1}^{i-1}(1-q_{2}^{\lambda_{i}})\\
    =&u-(1-q_{1})\frac{u}{1-q_{1}}+(1-q_{1})\sum_{i=1}^{\infty}uq_{1}^{i-1}q_{2}^{\lambda_{i}}\\
    =&(1-q_{1})\sum_{i=1}^{\infty}uq_{1}^{i-1}q_{2}^{\lambda_{i}}
\end{split}
\end{align}
where we used $\sum_{i=1}^{\infty}q_{1}^{i-1}=\frac{1}{1-q_{1}}$ for $|q_{1}|<1$.

\section{Equivalent expressions of Nekrasov factor}\label{a:Nekra}

In this appendix, we summarize some useful expressions of the Nekrasov factor that are important in \S\ref{sec:algebraic_topvertex}. 

The Nekrasov factor has various equivalent expressions, which all look different. Its definition is usually given by
\begin{equation}
N_{\lambda\nu}(Q;q_1,q_2):=\prod_{(i,j)\in \lambda}\lt(1-Qq_2^{-\nu_i+j}q_1^{\lambda^t_j-i+1}\rt)\prod_{(i,j)\in \nu}\lt(1-Qq_2^{\lambda_i-j+1}q_1^{-\nu^t_j+i}\rt).\label{def-Nekra}
\end{equation}
The convention used here is that $(i,j)\in\lambda$ is the box in the $i$-th row and $j$-th column of the Young diagram $\lambda$. Several useful equivalent expressions were derived in \cite{Awata:2008ed}:
\bea
N_{\lambda\nu}(Q;q_1,q_2)=&\prod_{i,j=1}^\infty\frac{1-Qq_2^{\lambda_i-j+1}q_1^{i-\nu^t_j}}{1-Qq_2^{-j+1}q_1^{i}}\cr
N_{\lambda\nu}(Q;q_1,q_2)=&\prod_{i,j=1}^\infty\frac{1-Qq_2^{-\nu_j+i}q_1^{\lambda^t_i-j+1}}{1-Qq_2^{i}q_1^{-j+1}}
\eea
which give one of the key identities in the computation of topological string \cite{IKV}, and
\begin{align}
N_{\lambda\nu}(Q;q_1,q_2)=\prod_{i,j=1}^\infty \frac{(Qq_1^{i-j}q_2^{\lambda_i-\nu_j+1};q_2)_\infty}{(Qq_1^{i-j+1}q_2^{\lambda_i-\nu_j+1};q_2)_\infty}\frac{(Qq_1^{i-j+1}q_2;q_2)_\infty}{(Qq_1^{i-j}q_2;q_2)_\infty},\label{Nekra-AK1}\\
N_{\lambda\nu}(Q;q_1,q_2)=\prod_{i,j=1}^\infty \frac{(Qq_2^{i-j}q_1^{\lambda^t_i-\nu^t_j+1};q_1)_\infty}{(Qq_2^{i-j+1}q_1^{\lambda^t_i-\nu^t_j+1};q_1)_\infty}\frac{(Qq_2^{i-j+1}q_1;q_1)_\infty}{(Qq_2^{i-j}q_1;q_1)_\infty}.\label{Nekra-AK2}
\end{align}
Furthermore one can use the sets of variables, $\cX_\lambda=\{u_1 q_1^{k-1}q_2^{\lambda_k}\}^\infty_{k=1}$ and $\cX_\nu=\{u_2 q_1^{k-1}q_2^{\nu_k}\}^\infty_{k=1}$, to rewrite \eqref{Nekra-AK1} to
\begin{equation}
    N_{\lambda\nu}(Q=u_1/u_2;q_1,q_2)=\prod_{(x,x')\in \cX_{\lambda}\times\cX_{\nu}}\frac{\lt(q_2x/x';q_2\rt)_\infty}{\lt(q_1q_2x/x';q_2\rt)_\infty}\prod_{i,j=1}^\infty\frac{\lt(\frac{u_1}{u_2}q_1^{i-j+1}q_2;q_2\rt)_\infty}{\lt(\frac{u_1}{u_2}q_1^{i-j}q_2;q_2\rt)_\infty}
\end{equation}
and with $\cX^\vee_\lambda=\{u_1 q_2^{k-1}q_1^{\lambda^t_k}\}^\infty_{k=1}$ and $\cX^\vee_\nu=\{u_2 q_2^{k-1}q_1^{\nu^t_k}\}^\infty_{k=1}$, \eqref{Nekra-AK2} is rewritten to
\begin{equation}
    N_{\lambda\nu}(Q=u_1/u_2;q_1,q_2)=\prod_{(x,x')\in \cX^\vee_{\lambda}\times\cX^\vee_{\nu}}\frac{\lt(q_1x/x';q_1\rt)_\infty}{\lt(q_1q_2x/x';q_1\rt)_\infty}\prod_{i,j=1}^\infty\frac{\lt(\frac{u_1}{u_2}q_2^{i-j+1}q_1;q_1\rt)_\infty}{\lt(\frac{u_1}{u_2}q_2^{i-j}q_1;q_1\rt)_\infty}
\end{equation}
Then, it is straightforward to derive
\begin{align}
    \frac{N_{(\lambda+s)\nu}(Q;q_1,q_2)}{N_{\lambda\nu}(Q;q_1,q_2)}=\prod_{x'\in \cX_{\nu}}\frac{1-q_1x_s/x'}{1-x_s/x'},\label{Nekra-recursive-1}\\
\frac{N_{\lambda(\nu+s)}(Q;q_1,q_2)}{N_{\lambda\nu}(Q;q_1,q_2)}=\prod_{x\in \cX_{\lambda}}\frac{1-q_2x/x_s}{1-q_1q_2x/x_s},\label{Nekra-recursive-2}
\end{align}
where $x_s=u_{1}q_1^{k-1}q_2^{\lambda_k+1}$ (resp. $x_s=u_{2}q_1^{k-1}q_2^{\nu_k+1}$) for some $s=(k,\lambda_k+1)\in \frakA(\lambda)$ (resp. $s=(k,\nu_k+1)\in \frakA(\nu)$) in the first recursive relation and the second equation respectively.

One can relate the recursive relations above with the $\cY$-function \eqref{def-Y}, which can alternatively expressed as
\begin{equation}
    \cY_\lambda(z,u)=\prod_{x'\in \cX_\lambda}\frac{1-x'/z}{1-q_1x'/z}=\lt(-\frac{u}{z}\rt)\prod_{x'\in \cX_\lambda}\frac{1-z/x'}{1-q_1^{-1}z/x'}.
\end{equation}
Note that $x_s=q_2\chi_s$, we have
\bea\label{eq:rec}
    \frac{N_{(\lambda+s)\nu}(Q;q_1,q_2)}{N_{\lambda\nu}(Q;q_1,q_2)}=&\lt(-\frac{\chi_sq_1q_2}{u_2}\rt)\cY_\nu(\chi_sq_1q_2,u_{2}),\cr
    \frac{N_{\lambda(\nu+s)}(Q;q_1,q_2)}{N_{\lambda\nu}(Q;q_1,q_2)}=&\cY_\lambda(\chi_s,u_{1}).
\eea
One can then use the above recursive relations to write down a useful expression of the Nekrasov factor in this article:
\begin{equation}
    N_{\lambda\nu}(Q;q_1,q_2)=\prod_{\substack{x\in\lambda\\y\in\nu}}S\lt(\chi_x/\chi_y\rt)\prod_{x\in\lambda}\lt(1-\chi_xq_1q_2/u_2\rt)\prod_{y\in\nu}\lt(1-u_1/\chi_x\rt)
\end{equation}
solved with the initial condition $N_{\emptyset\emptyset}(Q;q_1,q_2)=1$. One can also combine two recursive relations \eqref{eq:rec} together to obtain
\begin{align}
    \frac{N_{(\lambda+s)(\lambda+s)}(1;q)}{N_{\lambda\lambda}(1;q)}=&-\chi_sq_1q_2\lim_{z\rightarrow \chi_x}\cY_{\lambda+s}(zq_1q_2,1)\cY_\lambda(z,1).\label{Nekra-rec2}
\end{align}

From (\ref{def-Nekra}), we can further derive two useful identities
\begin{equation}
N_{\lambda\nu}(Q;q_1,q_2)=(-Q)^{|\lambda|+|\nu|}q_1^{n(\lambda)-n(\nu)}q_2^{-n(\nu^t)+n(\lambda^t)}N_{\nu\lambda}\lt((Qq_1q_2)^{-1};q_1,q_2\rt),\label{Nekra-convert-formu}
\end{equation}
and
\begin{equation}
N_{\lambda\nu}(Q;q_1,q_2)=N_{\nu\lambda}(Qq_1^{-1}q_2^{-1};q_1^{-1},q_2^{-1}).\label{q-invert-formu}
\end{equation}

\section{Contraction formulas}\label{appendix:contraction}
We list down the contraction formulas of the intertwiners and the Drinfeld currents. Readers who are interested in the detailed computations in \S\ref{sec:intertwiner} and \S\ref{sec:qq-alg} can use the formulas here to do the computations.
\begin{align}
    \varphi^{+}(z)\Phi_{\emptyset}[v]=&\frac{1-q_{3}v/z}{1-v/z}:\Phi_{\emptyset}[v]\varphi^{+}(z):\cr
    \Phi_{\emptyset}[v]\varphi^{+}(z)=&:\Phi_{\emptyset}[v]\varphi^{+}(z):\cr
    \Phi_{\emptyset}[v]\varphi^{-}(z)=&\frac{1-\gamma^{-1}z/v}{1-\gamma^{-3}z/v}:\Phi_{\emptyset}[v]\varphi^{-}(z):\cr
    \varphi^{-}(z)\Phi_{\emptyset}[v]=&:\varphi^{-}(z)\Phi_{\emptyset}[v]:\cr
    \eta(z)\Phi_{\emptyset}[v]=&\frac{1}{1-v/z}:\eta(z)\Phi_{\emptyset}[v]:\cr
    \Phi_{\emptyset}[v]\eta(z)=&\frac{1}{1-q_{3}^{-1}z/v}:\Phi_{\emptyset}[v]\eta(z):\cr
    \xi(z)\Phi_{\emptyset}[v]=&(1-\gamma v/z):\xi(z)\Phi_{\emptyset}[v]:\cr
    \Phi_{\emptyset}[v]\xi(z)=&(1-\gamma^{-1}z/v):\Phi_{\emptyset}[v]\xi(z):\cr
    \eta(z)\Phi^{(n)}_{\lambda}[u,v]=&\frac{1}{\mathcal{Y}_{\lambda}(z,v)}:\eta(z)\Phi^{(n)}_{\lambda}[u,v]:\cr
    \Phi^{(n)}_{\lambda}[u,v]\eta(z)=&-\frac{vq_{3}}{z}\frac{1}{\mathcal{Y}_{\lambda}(q_{3}^{-1}z,v)}:\Phi^{(n)}_{\lambda}[u,v]\eta(z):\cr
    \xi(z)\Phi_{\lambda}^{(n)}[u,v]=&\mathcal{Y}_{\lambda}(\gamma^{-1}z,v):\xi(z)\Phi_{\lambda}^{(n)}[u,v]:\cr
    \Phi_{\lambda}^{(n)}[u,v]\xi(z)=&-\frac{\gamma^{-1}z}{v}\mathcal{Y}_{\lambda}(\gamma^{-1}z,v):\Phi_{\lambda}^{(n)}[u,v]\xi(z):\cr
    \varphi^{+}(z)\Phi^{*}_{\emptyset}[v]=&\frac{1-\gamma v/z}{1-\gamma^{3} v/z}:\Phi^{*}_{\emptyset}[v]\varphi^{+}(z):\cr
    \Phi_{\emptyset}^{*}[v]\varphi^{+}(z)=&:\Phi_{\emptyset}^{*}[v]\varphi^{+}(z):\cr
    \Phi_{\emptyset}^{*}[v]\varphi^{-}(z)=&\frac{1-q_{3}^{-1}z/v}{1-z/v}:\Phi_{\emptyset}^{*}[v]\varphi^{-}(z):\cr
    \varphi^{-}(z)\Phi^{*}_{\emptyset}[v]=&:\varphi^{-}(z)\Phi^{*}_{\emptyset}[v]:\cr
    \eta(z)\Phi^{*}_{\emptyset}[v]=&\bl(1-\gamma v/z\br):\eta(z)\Phi^{*}_{\emptyset}[v]:\cr
    \Phi^{*}_{\emptyset}[v]\eta(z)=&\bl(1-\gamma^{-1}z/v\br):\Phi^{*}_{\emptyset}[v]\eta(z):\cr
    \xi(z)\Phi^{*}_{\emptyset}[v]=&\frac{1}{1-q_{3}v/z}:\Phi^{*}_{\emptyset}[v]\xi(z):\cr
    \Phi^{*}_{\emptyset}[v]\xi(z)=&\frac{1}{1-z/v}:\Phi^{*}_{\emptyset}[v]\xi(z):\cr
    \eta(z)\Phi_{\lambda}^{*(n)}[u,v]=&\mathcal{Y}_{\lambda}(\gamma^{-1}z,v):\eta(z)\Phi_{\lambda}^{*(n)}[u,v]:\cr
\Phi_{\lambda}^{*(n)}[u,v]\eta(z)=&-\frac{\gamma^{-1}z}{v}\mathcal{Y}_{\lambda}(\gamma^{-1}z,v):\Phi_{\lambda}^{*(n)}[u,v]\eta(z):\cr
\xi(z)\Phi_{\lambda}^{*(n)}[u,v]=&\frac{1}{\mathcal{Y}_{\lambda}(q_{3}^{-1}z,v)}:\xi(z)\Phi_{\lambda}^{*(n)}[u,v]:\cr
\Phi_{\lambda}^{*(n)}[u,v]\xi(z)=&-\frac{v}{z}\frac{1}{\mathcal{Y}_{\lambda}(z,v)}:\Phi_{\lambda}^{*(n)}[u,v]\xi(z):\nonumber
\end{align}

\begin{scriptsize}
\bibliography{references}
\bibliographystyle{hyperamsalpha}
\end{scriptsize}

\end{document}